\documentclass[12pt]{article}

\usepackage{amsthm,amsmath,amsfonts,amssymb}
\usepackage[round, authoryear]{natbib}
\usepackage[colorlinks,linkcolor={blue},urlcolor={blue}, citecolor={blue}, filecolor={black}]{hyperref}
\usepackage{graphicx}
\usepackage{multirow,caption,booktabs,float,mathrsfs}
\usepackage{authblk,enumitem,setspace}
\usepackage[symbol]{footmisc}

\usepackage[top=1in, bottom=1in, left=1in, right=1in]{geometry}

\theoremstyle{plain}
\newtheorem{theorem}{Theorem}[section]
\newtheorem{lemma}[theorem]{Lemma}

\theoremstyle{remark}

\newtheorem{remark}[theorem]{Remark}

\newlist{assumpenum}{enumerate}{1}
\setlist[assumpenum]{leftmargin=3.5em, label={\bf (A.\arabic*)}, ref=(A.\arabic*)}

\newlist{lemmaitemenum}{enumerate}{1}
\setlist[lemmaitemenum]{leftmargin=4.5em, label=(\thelemma.\arabic*), ref=(\thelemma.\arabic*)}

\def\lcrarrow#1{\overset{#1}{\longrightarrow}}
\def\crarrow#1{\overset{#1}{\rightarrow}}

\def\bs#1{\boldsymbol{#1}}
\def\Pn{\mathbb{P}_n}
\def\Pj{\mathbb{P}_j}
\def\Qn{\mathbb{Q}_n}
\def\Qj{\mathbb{Q}_j}
\DeclareMathOperator{\IF}{IF}

\begin{document}
\allowdisplaybreaks
%\doublespacing

\title{Efficient Estimation of the Maximal Association between Multiple Predictors and a Survival Outcome}
\author[1]{Tzu-Jung Huang}
\author[1,2]{Alex Luedtke}
\author[3]{Ian W. McKeague}

\affil[1]{\small Department of Statistics, University of Washington}
\affil[2]{\small Vaccine and Infectious Disease Division, Fred Hutchinson Cancer Research Center}
\affil[3]{\small Department of Biostatistics, Columbia University}

\date{\today}
\maketitle

\fontsize{12}{16pt plus.8pt minus .6pt}\selectfont

\begin{abstract}
This paper develops a new approach to post-selection inference for screening high-dimensional predictors of survival outcomes. Post-selection inference for right-censored outcome data has been investigated in the literature, but much remains to be done to make the methods both reliable and computationally-scalable in high-dimensions.  Machine learning tools are commonly used to provide {\it predictions} of survival outcomes, but the estimated effect of a selected predictor suffers from confirmation bias unless the selection is taken into account.  The new approach involves construction of semi-parametrically efficient estimators of the linear association between the predictors and the survival outcome, which are used to build a test statistic for detecting the presence of an association between any of the predictors and the outcome. Further, a stabilization technique reminiscent of bagging allows a normal calibration for the resulting test statistic, which enables the construction of confidence intervals for the maximal association between predictors and the outcome and also greatly reduces computational cost. Theoretical results show that this testing procedure is valid even when the number of predictors grows superpolynomially with sample size, and our simulations support that this asymptotic guarantee is indicative the performance of the test at moderate sample sizes. The new approach is applied to the problem of identifying patterns in viral gene expression associated with the  potency of an antiviral drug.
\end{abstract}

\section{Introduction}
The problem of identifying associations between high-dimensional predictors and a survival outcome is of great interest in the biomedical sciences.  In virology, for example, the potency of an antiviral drug (in controlling viral replication) is typically assessed in terms of a type of survival time outcome, and it is important to identify associations between patterns of viral gene expression  and the drug's potency \citep{Gilbert2017}. In cancer genomics,  patterns of patients' gene expression can also influence survival time outcomes. Diffuse large B-cell lymphoma, for instance, has been studied with the aim of identifying such patterns from massive collections of gene-expression data \citep{Rosenwald2002,Bovelstad2009}. In earlier work \citep{Huang2019}, we introduced an approach to this general problem based on  marginal accelerated failure time modeling.  In the present paper, we expand this approach to provide a semiparametrically efficient  and more computationally tractable method that can handle the screening of extremely large numbers of predictors (as is typical with gene expression data). 

Our approach is based on marginal screening of the predictors,  and specifying the link between the survival outcome and the predictors by a general semiparametric accelerated failure time (AFT) model that does not make any distributional assumption on the error term. The error term is merely taken to be  uncorrelated with the predictors (i.e., the so-called {\it assumption lean} linear model setting).
Let $T$ be the (log-transformed) time-to-event outcome, and $\bs{U}=(U_1,\ldots,U_p)^T$ denote a
$p$-dimensional vector of predictors. Note that $p=p_n$ can grow with $n$, but we omit the subscript $n$ throughout for notational simplicity unless otherwise stated.
The AFT model takes the form
\begin{align}
  T=\alpha_0 + \bs{U}^T\bs{\beta}_0 +\varepsilon, \label{eq:linear_model}
\end{align}
where $\alpha_0 \in \mathbb{R}$ is an intercept,  $\bs{\beta}_0 \in \mathbb{R}^p$ is a vector of slope parameters, and $\varepsilon$
is the zero-mean error term that is uncorrelated with $\bs{U}$. The transformed survival outcome $T$ is possibly right-censored by $C$, and we only observe $X = \min\{T,C\}$ and $\delta=1(T \le C)$.
The problem is to test the global null  $\bs{\beta}_0=0$. We emphasize that model (\ref{eq:linear_model}) is locally nonparametric \citep{vanderLaan&Robins2003} and holds without  distributional assumptions (such as independent errors) apart from mild moment conditions, by defining the second term as the $L_2$-projection  of $T$ onto the linear span of $U_1,\ldots,U_p$.

An especially attractive feature of the AFT model is that the marginal association between $T$ and each predictor can be represented directly in terms of a correlation, and does not require any structural assumptions.
This allows us to reduce the high-dimensional screening problem (involving all $p$ components of $\bs{\beta}_0$) to a single test of whether the most correlated predictor with $T$ is significant.
A popular approach to the screening of predictors in survival analysis is to use relative or excess conditional hazard function representations of associations.
However, the AFT approach has the advantage that a lack of any marginal correlation implies the absence of all correlation between $T$ and $\bs{U}$ (under the mild assumption that the covariance matrix of $\bs{U}$ is invertible); in the hazard-rate setting, there is no such connection and the semiparametric model needs to hold for testing methods to be useful.

\citet{Koul1981} (henceforth, KSV) introduced the technique of inversely weighting the observed outcomes by the Kaplan--Meier estimator of the censoring distribution, enabling the use of standard least squares estimators from the uncensored linear model.
Subsequently, two additional sophisticated methods were proposed to fit the semiparametric AFT model. The Buckley--James estimator \citep{Buckley1979, Ritov1990} replaces the censored survival outcome by the conditional expectation of $T$ given the data. The rank-based method is an estimating equation approach formulated in terms of the partial likelihood score function \citep{Tsiatis1990, Lai1991, Lai1991a, Ying1993, Jin2003}. A difficulty with the Buckley--James and rank-based methods is that they fail to preserve a direct link with the AFT, which is essential for marginal screening based on correlation.
Our new marginal screening test will rely on finding an asymptotically efficient estimator of each marginal slope parameter; this will have a considerable  advantage  in terms of  efficiency over the marginal screening method based on the KSV estimator \citep{Huang2019}.  

The marginal KSV estimators stem from regressing the estimated synthetic response $Y=\delta X/\hat{G}_n(X)$ on
successive components of $\bs{U}$, where $Y$ is regarded as an inverse probability weighted estimate, and $\hat{G}_n$ is the standard Kaplan--Meier estimator
of the survival function of $C$.
Under independent censoring, the use of least squares estimators, treating $Y$ as a response variable, is justified in view of the uniform consistency of $\hat{G}_n$ under mild  conditions (e.g.,  when
the distribution functions of $T$ and $C$ have no common jumps; see \citealp{Stute1993}). Independent censoring is a common assumption in the high-dimensional screening of predictors for
survival outcomes \citep{He2013, Song2014, Li2016}.
 
We now outline the various novel steps involved in developing  our proposed test.  In \citet{Huang2019} we showed that
$|{\rm Corr}(U_j, T)|$ and  $|{\rm Corr}(U_j, \tilde Y)|$, for $j=1,\ldots,p$, are maximized at a common index $k$. This was used to justify replacing $T$ by $\tilde Y= \delta X/{G}(X)$ (and in turn its estimate $Y$)  in  the  empirical version of the  
slope parameter  $\Psi(P)={\rm Cov}_P(U,T)/{\rm Var}_P(U)$, where $U=U_k$, for use as a test statistic for the global null hypothesis $\bs{\beta}_0=0$. 
Our aim now is to replace this test statistic by one that is more efficient (when $k$ is treated as fixed), and also that is easier to calibrate taking the selection of $k$ into account.   Writing  $\tilde{\mathbb{P}}_n$  as  the empirical distribution of $(U,\tilde Y)$, when $G$ is known, the influence function of $\Psi(\tilde{\mathbb{P}}_n)$  will  be derived from that of the sample correlation coefficient in the uncensored case \citep{Devlin1975}.  This will lead to the influence function $\IF(O|P)$ of the (inefficient) KSV estimator $\hat \Psi(\tilde{\mathbb{P}}_n)$ that replaces the unknown $G$ in $\Psi(\tilde{\mathbb{P}}_n)$ by $\hat G_n$.  This derivation will be based  on the influence function of  $\hat G_n$ and some empirical process and Slutsky-type arguments.   The next step is to project $\IF(\cdot|P)$ onto the tangent space of the observation model to obtain an efficient influence function $\IF^*$, which in turn will lead to an asymptotically efficient one-step estimator of $\Psi(P)$.  This will be accomplished in part using results of \cite{vanderLaan&Robins2003} and \cite{vanderLaanetal2000}.  
The one-step estimator takes the form $S(\hat{P}_n, \mathbb{P}_n) = \Psi(\hat{P}_n) + \mathbb{P}_n \IF^*(\cdot|\hat{P}_n)$, where $\hat{P}_n$ is a plug-in estimator of the various features of  $P$ that appear in $\IF^*$, $\Psi(\hat{P}_n) = \hat \Psi(\tilde{\mathbb{P}}_n)$, and $\mathbb{P}_n$ is the empirical distribution of the data (acting as an expectation operator). Estimation of those features will involve estimation of the function $E[\tilde{Y}|U=u, X\ge s]$ as in \cite{vanderLaan1998},  the empirical distribution of the selected  predictor $U$, and a local Kaplan--Meier estimator of the conditional censoring distribution  given $U$.   

The final step is to develop a method to calibrate the test.  This will be done by introducing a  stabilized version of $S(\hat{P}_n, \mathbb{P}_n)$ that ``smooths out'' the implicit selection of $k$, along the lines of \cite{Luedtke2018} in the uncensored case.  The stabilized version of $S(\hat{P}_n, \mathbb{P}_n)$ is constructed by taking a weighted average over  sub-samples, and is asymptotically equivalent to a martingale sum (provided $\hat P_n$ has no effect asymptotically), which leads to a standard normal limit even under growing dimensions, when $p=p_n \to \infty$ and $\log(p_n)/n^{1/4} \to 0$. Although the stabilized one-step estimator can have slightly diminished power compared with its un-stabilized counterpart, at least in the uncensored case \citep{Luedtke2018}, 
it vastly reduces  computational cost by avoiding the need for a double bootstrap \citep{Huang2019}. The most challenging step in establishing the asymptotic normality involves finding an exponential tail bound for a collection of martingale integrals in which the integrands fall in a class of functions of bounded variation. This is done using bracketing entropy and involves the novel application of a uniform probability inequality bound for a family of counting process integrals due to  \cite{vandeGeer1995}.

In practice, the implementation of our stabilized one-step estimator to screen predictors of dimension $p=10^6$ based on data of $n=500$ on a single-core laptop only takes one minute. Hence our proposed test enjoys both statistical and computational
efficiency. Further, it provides an asymptotically valid confidence interval for the slope parameter of the selected predictor. As far as we know, no other competing method provides all of these features in the setting of high-dimensional marginal screening for  survival outcomes.  % that becomes increasingly crucial for survival analysis in the era of big data.

Variable selection methods for right-censored survival data are widely available, although formal testing procedures
are far less prevalent. For example, variants of  regularized Cox regression have been studied by \citep{Tibshirani1997, Fan2002, Bunea2005, Zhang2007, Bovelstad2009, Engler2009, Antoniadis2010, Binder2011, Wu2012, Sinnott2016}. Penalized AFT models have been
considered by \citep{Huang2006, Datta2007, Johnson2008, Johnson2008a, Cai2009, Huang2010, Bradic2011, Ma2012, Li2014}.
These methods ensure the consistency of variable selection  only (i.e., the oracle property), and do not address the issue of post-selection inference. \cite{Fang2016} have established asymptotically valid confidence intervals for a preconceived regression parameter in a high-dimensional Cox model after variable selection on the remaining predictors,
but this does not apply to marginal screening (where no regression parameter is singled out, a priori). \cite{Yu2021} recently constructed valid  confidence intervals for the regression parameters in   high-dimensional Cox models, but their approach also does not apply to marginal screening because it is predicated on the presence of active predictors (and also pre-selection of parameters of interest). \cite{Zhong2015} have considered the  problem for preconceived regression parameters within a high-dimensional additive risk model. \cite{Taylor2017} proposed a method of finding post-selection corrected p-values and confidence intervals for the Cox model based on conditional testing. However, to the best of our knowledge, their method has not been explored theoretically (except in the uncensored linear regression setting with fixed design and normal errors; see \citealp{Lockhart2014}).
Statistical methods for variable selection based on marginal screening for survival data have been studied by \cite{Fan2010}, who extended sure independence screening to survival outcomes based on the Cox model.
Their method applies to the selection of components of ultra-high dimensional predictors, but no formal testing is available. Other relevant references
include \cite{Zhao2012}, \cite{Gorst-Rasmussen2013}, \cite{He2013}, \cite{Song2014}, \cite{Zhao2014}, \cite{Hong2016}, \cite{Li2016}, \cite{Hong2017}, \cite{Pan2019}, \cite{Xia2019},
\cite{Hong2020} and \cite{LiuChenLi2020}.

The article is organized as follows. In Section \ref{sec:preliminary} we formulate the estimation problem and introduce  background material on semiparametric efficiency. The one-step efficient estimator of the target parameter is developed in Section \ref{sec:eff_est_fixed_d} in the case of a single predictor. In Section \ref{sec:stabilized_onestep} we develop an asymptotic normality result for calibrating the proposed test statistic that takes selection of the predictor into account.
Various competing methods are discussed in Section \ref{sec:competing_methods}. Numerical results reported in Section \ref{sec:numerical} show
that the proposed approach has favorable performance compared with these competing methods. In Section \ref{sec:real_data_application} we present an application using data on  viral gene expression as related to the  potency of an anti-retroviral drug for the treatment of HIV-1. Concluding remarks are given in
Section \ref{sec:discussion}. Proofs are placed in the Appendix.

\section{Preliminaries} 
\label{sec:preliminary}

First we recall the  standard survival analysis model with independent right censorship. Let $T$ and $C$ denote a (log-transformed) survival time and censoring time, respectively, and suppose we observe $n$ i.i.d.\ copies of $O=(X,\delta,\bs{U})\sim P$, where $X=\min\{T,C\}$, $\delta=1(T \leq C)$, and $\bs{U}=(U_k, k= 1, \ldots, p)$ is a $p$-vector of predictors. We denote the joint distribution of $(T,\bs{U})$ by $Q$ and the censoring distribution by $G$, and we also assume throughout that the censoring time $C$ is independent of $(T,\bs{U})$. Though this joint independence assumption will be stronger than needed, it will greatly simplify the developments when $\bs{U}$ is of large dimension relative to sample size. The distribution $P$ belongs to the statistical model $\mathcal{M}$, which is the collection of distributions $P_1$ parameterized by $(Q_1,G_1)$ such that $P_1$ has density with respect to an appropriate dominating measure $\nu$ given by
\begin{align*}
  \frac{dP_1}{d\nu}(x,\delta,\bs{u})&= [q_1(x|\bs{u})G_1(C\ge x)]^{\delta}[Q_1(T\ge x|\bs{U}=\bs{u})g_1(x)]^{1-\delta} q_1(\bs{u}),  
                                        %\label{eq:obsdistr}
\end{align*}
where $q_1$ and $g_1$ are the densities of $Q_1$ and $G_1$ with respect to $\nu$. Let the follow-up period be $\mathcal{T}=(-\infty,\tau]$. The sample space is denoted by $\mathcal{X}=\mathcal{T}\times\{0,1\}\times\mathbb{R}^{p}$ and the empirical distribution on this space is denoted $\Pn$. Moreover, for a distribution $P_1$ on the support of $O$ and a function $f$ mapping from a realization of $O$ to $\mathbb{R}^d$, we let $P_1 f\equiv P_1 f(O)\equiv \int f(o) dP_1(o)$.

Our approach to marginal screening is based on an  estimator of the maximal (absolute) slope parameter from fitting a marginal linear
regression of the survival outcome $T$ against each predictor $U_k$. That is, we target the parameter
\begin{align}\label{eq:def_parameter}
 \Psi(P) \equiv \max_{k = 1,\ldots, p}\left|\Psi_{k}(P)\right|,
\end{align}
where $\Psi_k: \mathcal{M} \rightarrow \mathbb{R}$ is given by
\begin{align}\label{eq:def_parameter.1}
 \Psi_{k}(P)= \frac{{\rm Cov}_P(U_k,T)}{{\rm Var}_{P}(U_k)}.
\end{align}
Throughout we assume that $U_k$ and $T$ have non-degenerate finite second moments.  Further, in order for the target parameter to be proportional  to the maximal absolute (Pearson) correlation, we implicitly assume that all the $U_k$ are pre-standardized to have  unit variance --- this assumption only plays an interpretive  role in the sequel. 
The parameter $\Psi_k(P)$ can be identified in terms of the conditional mean lifetime $E[T|U_k]$ and the marginal distribution of $U_k$. Indeed,
\begin{align}\label{eq:def_parameter.2}
  {\rm Cov}_P(U_k, T) = {\rm Cov}_P( U_k,E[T|U_k]).   
\end{align} 
The proposed one-step estimator of $\Psi_k(P)$ that we will develop also involves estimation of $G$.

We will need some general concepts from semiparametric efficiency theory \citep[e.g.,][]{Pfanzagl1990}. Suppose we observe a general random vector $O\sim P$.
Let $L_0^2(P)$ denote the Hilbert space of $P$-square integrable functions with mean zero. Consider a smooth one-dimensional family of probability
measures $\{ P_\epsilon\} $ passing through $P$ and having score function $k\in L_0^2(P)$ at $\epsilon=0$. The tangent space $\mathbf{T}^{\mathcal{M}}(P)$ is the
$L_0^2(P)$-closure of the linear span of all such score functions $k$. For example, if nothing is known about $P$, then
$P_\epsilon(do)=(1+\epsilon k(o))P(do)$ is such a submodel for any bounded function $k$ with mean zero (provided $\epsilon$ is sufficiently small),
so $\mathbf{T}^{\mathcal{M}}(P)$ is seen to be the whole of $L_0^2(P)$ in this case.

Let $\psi : \mathcal{M}\rightarrow\mathbb{R}$ be a parameter that is pathwise differentiable at $P \colon$ there exists
$g\in L_0^2(P)$ such that $\lim_{\epsilon\to 0} \left ( \psi(P_\epsilon)-\psi(P)\right)/\epsilon =\langle g,k\rangle$, for any smooth submodel
$\{ P_\epsilon\}$ with score function $k$, as above, where $\langle \cdot,\cdot\rangle$ is the inner product in $L_0^2(P)$. The function $g$ is
called a gradient (or influence function) for $\psi$; the projection $\IF_\psi$ of any gradient on the tangent space $\mathbf{T}^{\mathcal{M}}(P)$ is unique and is known as
the canonical gradient (or efficient influence function). The supremum of the Cram\'er--Rao bounds for all submodels (the information bound)
is given by the second moment of $\IF_\psi(O)$.  Furthermore, the influence function of any regular and asymptotically linear estimator must be a gradient \cite[Proposition 2.3 in][]{Pfanzagl1990}.

A one-step estimator is an empirical bias correction of a na\"{i}ve plug-in estimator in the direction of a gradient of the parameter of interest \citep{Pfanzagl1982}; when this gradient is the canonical gradient, then this results in an efficient estimator under some regularity conditions. A one-step estimator for $\psi(P)$  is constructed as follows. First, one obtains an initial estimate $\hat{P}$ of $P$. For any gradient $D(\hat{P})$ of the parameter $\psi$ evaluated at $\hat{P}$, by the definition of the gradient this initial estimate satisfies  
\begin{align*}
  \psi(\hat{P}) - \psi(P)&= -P D(\hat{P}) + {\rm Rem}_{\psi}(\hat{P},P),
\end{align*}
where ${\rm Rem}_{\psi}(\hat{P},P)$ is negligible if $\hat{P}$ is close to $P$ in an appropriate sense. As $D(P)$ has mean zero under $P$, we expect that $P D(\hat{P})$ is close to zero if $D$ is continuous in its argument and $\hat{P}$ is close to $P$. However, the rate of convergence of $P D(\hat{P})$ to zero as sample size grows may be slower than $n^{-1/2}$. The one-step estimator aims to improve $\psi(\hat{P})$ and achieve $n^{1/2}$-consistency and asymptotically normality by adding an empirical estimate $\mathbb{P}_n D(\hat{P})$ of its deviation from $\psi(P)$. By the above, the one-step estimator $\hat{\psi}\equiv \psi(\hat{P}) + \mathbb{P}_n D(\hat{P})$ satisfies the expansion
\begin{align*}
  \hat{\psi} - \psi(P)&= (\mathbb{P}_n-P) D(\hat{P}) + {\rm Rem}_{\psi}(\hat{P},P).
\end{align*}
Under an empirical process and $L^2(P)$ consistency condition on $D(\hat{P})$, the leading term on the right-hand side is asymptotically equivalent to $(\mathbb{P}_n-P) D(P)$, which converges in distribution to a mean-zero Gaussian 
limit with consistently estimable covariance. The construction of this one-step estimator is generally non-unique because there is generally more than one gradient for $\psi$; this is true in our setting when $\psi=\Psi_k$ and we assume that $C$ is independent of $(T,\bs{U})$. To minimize the variance of the Gaussian limit, then $D(\hat{P})$ can generally be chosen to be equal to the canonical gradient of $\psi$ at $\hat{P}$, since under conditions the mean-square limit of the efficient influence function at $\hat{P}$ will be equal to the efficient influence function at $P$.

\section{Slope estimation with a single predictor} \label{sec:eff_est_fixed_d}
Restricting attention to the case of a single predictor $U_k$, in this section we develop an efficient estimator of $\Psi_k(P)$; for notational simplicity we suppress the subscript $k$, and just write $\Psi(P)$ and $U$, both here and in the corresponding proofs in the sequel. 

\subsection{Inefficient estimation with known censoring distribution \texorpdfstring{$G$}{Lg}}\label{sec:ipw}
Let $G(t) \equiv {\rm P}(C \ge t)$ be the unknown survival function of $C$, where throughout we assume 
$G(\tau)>0$. For a given survival function $G_1$, let
\begin{align*}
    \Psi_{G_1}(P) \equiv \frac{{\rm Cov}_P(U,\delta X/G_1(X))}{{\rm Var}_{P}(U)}.
\end{align*}
Also let $\tilde{Y} \equiv \delta X/G(X)$. Noting that $E[T]=E[\tilde Y]$ and $E[UT]=E[U \tilde Y]$, we obtain the useful identity that $\Psi(P)=\Psi_G(P)$. This suggests that, if $G$ is known, then it suffices to try to estimate the value of the inverse-probability-of-censoring weighted parameter $\Psi_G$ evaluated at the distribution that generated the observed data.

To study the parameter $\Psi_G$, it will be useful to use existing arguments from the uncensored case. To do this, we define $f_G(o)\equiv (u,\tilde{y})$, and then let $\tilde{P}$ denote the pushforward measure $\tilde P\equiv P\circ f_G^{-1}$, that is, the distribution of $f_G(O)$ when $O\sim P$. Writing
\begin{align*}
    \Gamma(\tilde{P})\equiv \frac{{\rm Cov}_{\tilde{P}}(U,\tilde{Y})}{{\rm Var}_{\tilde{P}}(U)},
\end{align*}
we then see that the parameter $P\mapsto \Gamma(P\circ f_G^{-1})$ is the same as the parameter $P\mapsto \Psi_G(P)$. Therefore, a first-order expansion of the parameter $\Gamma$ will yield a first-order expansion of the parameter $\Psi_G$. Noting that $\Gamma$ is simply the slope parameter in a linear regression in which both the predictor and the outcome are observed, we can develop an expansion of this parameter using an existing result in \cite{Devlin1975}, who attributed the result to C. L. Mallows. Let
\begin{align*}
  \IF^{ipw}_G(\cdot|P) \colon o \mapsto  \frac{(u-P[U])(\tilde y-P[T])}{{\rm Var}_{P}(U)}
  -\frac{{\rm Cov}_{P}(U, T)}{{\rm Var}_{P}^2(U)}(u-P[U])^2,
\end{align*}
where we note that $\tilde{y}$ is a function of $(x,\delta)$ and therefore of $o$, and the subscript $G$ indicates the implicit dependence of $\tilde y$ on $G$.
Using  $\tilde{\mathbb{P}}_n$ to denote  the empirical distribution of $(U,\tilde Y)$ when $G$ is known, Mallows' result implies that, for a term ${\rm Rem}_G(\Pn,P)$ defined in Appendix Section \ref{sec:canonical_max_slope},
\begin{equation} \label{eq:if_G_fixed_derivation}
\begin{split}
  &\Gamma(\tilde{\mathbb{P}}_n)-\Gamma(\tilde P) \\
  &\hspace{0.1cm} = (\tilde{\mathbb{P}}_n-\tilde P) \left[ \left\{ \frac{(U-\tilde{P}[U])(\tilde Y-\tilde{P}[\tilde Y])}{{\rm Var}_{\tilde P}(U)}
  -\frac{{\rm Cov}_{\tilde P}(U, \tilde Y)}{{\rm Var}_{\tilde P}^2(U)}(U-\tilde{P}[U])^2 \right\} \right]+{\rm Rem}_G(\mathbb{P}_n, P) \\
  &\hspace{0.1cm} = (\mathbb{P}_n-P) \left[ \left\{ \frac{(U-P[U])(\tilde Y-P[T])}{{\rm Var}_{P}(U)}
  -\frac{{\rm Cov}_{P}(U, T)}{{\rm Var}_{P}^2(U)}(U-P[U])^2 \right\} \right]+{\rm Rem}_G(\mathbb{P}_n,P) \\
  &\hspace{0.1cm} = (\Pn-P) \left[\IF^{ipw}_G(O|P)\right]+{\rm Rem}_G({\mathbb{P}}_n, P) \\
  &\hspace{0.1cm} = \Pn \left[\IF^{ipw}_G(O|P)\right]+{\rm Rem}_G({\mathbb{P}}_n, P),
\end{split}
\end{equation}
where the second equality follows from $P[T] = \tilde P[\tilde Y]$, ${\rm Cov}_P(U,T) = {\rm Cov}_{\tilde P}(U,\tilde Y)$ and
the fact that $U$ has the same marginal distribution under $P$ and $\tilde P$, and the last equality follows from $P\big[\IF^{ipw}_G(O|P)\big]=0$.
For more details of the derivation of $\IF^{ipw}_G$, see Appendix Section \ref{sec:canonical_max_slope}. For concise notations, we hereafter omit the subscript $P$ from the functions evaluated at $P$ unless otherwise stated, which simplifies the presentation of the influence function to
\begin{equation} \label{eq:if_ipw}
  \IF^{ipw}_G(\cdot|P) \colon o \mapsto  \frac{(u-P[U])(\tilde y-P[T])}{{\rm Var}(U)}
  -\frac{{\rm Cov}(U, T)}{{\rm Var}^2(U)}(u-P[U])^2.
\end{equation}
Noting that $\tilde{\mathbb{P}}_n=\Pn\circ f_G^{-1}$ and recalling that the parameter $P\mapsto \Gamma(P\circ f_G^{-1})$ is the same as the parameter $P\mapsto \Psi_G(P)$, \eqref{eq:if_G_fixed_derivation} yields that
\begin{align} \label{eq:if_G_fixed}
  \Psi_G(\Pn)-\Psi_G(P)
  = \Pn \big[\IF^{ipw}_G(O|P)\big]+{\rm Rem}_G({\mathbb{P}}_n, P).
\end{align}

\subsection{Efficiency gains through estimation of \texorpdfstring{$G$}{Lg}} \label{sec:estimation_G}
We now consider the effect on the inverse-probability-weighted estimator $\Psi_G(\Pn)$ of replacing $G$ by its Kaplan--Meier (K--M) estimator $\hat{G}_n$. This involves replacing $\tilde Y$ by the estimated (and observed) synthetic response $Y=\delta X/\hat{G}_n(X)$, and results in the following estimator  
of $\Psi(P)$ that now just depends on the empirical distribution $\Pn$ of the observed data:
\begin{align*}
  \Psi_{\hat{G}_n}(\Pn)={\rm Cov}_{\Pn}(U,Y)/{\rm Var}_{\Pn}(U).
                                     %\label{eq:ineff_estimator}
\end{align*}
Importantly, unlike for the estimator presented in Section~\ref{sec:ipw}, constructing this estimator does not rely on having knowledge of $G$.

Under the following conditions, Theorem~\ref{Thm:if_slope} below shows that $\Psi_{\hat{G}_n}(\mathbb{P}_n)$  is asymptotically linear and identifies its influence function. The boundedness part of the first assumption is stronger than we need, but is made to simplify the proof.
\begin{assumpenum}[series=assumptions]
  \item \label{assump:Covariates} The predictor $U$ has bounded support and is non-degenerate. %The predictor $U$ takes its values in $[-1,1]$ and is non-degenerate.
  \item \label{assump:Survival function} The survival function of the censoring, $G$, is continuous and  $G(\tau) > 0$.
  \item \label{assump:At-risk prob} There is a positive probability of a subject  still being at risk at the end of follow-up: ${\rm P}(X \geq \tau) > 0$.
\end{assumpenum}

\begin{theorem} \label{Thm:if_slope}
Given \ref{assump:Covariates}--\ref{assump:At-risk prob},
\begin{align*}
  \Psi_{\hat{G}_n}(\Pn) - \Psi(P)&= \Pn \IF(O|P) + o_p(n^{-1/2}),
\end{align*}
where
\begin{align*}
  \IF(O|P) \equiv \IF^{ipw}_G(O|P) - \IF^{nu}_G(O|P), %\label{eq:if_slope}
\end{align*} 
the influence function $\IF^{ipw}_G(\,\cdot \,|P)$ is given in \eqref{eq:if_ipw}, and
\begin{align*} %\label{eq:if_nu}
  \IF^{nu}_G(\cdot|P) \colon o \mapsto \frac{1}{{\rm Var}(U)} P\left[(U-P[U])\int_{\mathcal{T}}E[T|U,X \geq s]
  \frac{ 1(x \in ds, \delta=0)-1(x \geq s)d\Lambda(s)}{G(s)}\right],
\end{align*}
where $\Lambda$ is the cumulative hazard function corresponding to $G$.
\end{theorem}

Since $P[\IF(\cdot|P)]=0$, Theorem~\ref{Thm:if_slope} implies that $n^{1/2}[\Psi_{\hat{G}_n}(\Pn) - \Psi(P)] \lcrarrow{d} N(0,\sigma^2)$, where $\sigma^2=P[\IF^2(\cdot|P)]$.
The proof of Theorem~\ref{Thm:if_slope} and the relevant Lemmas \ref{lemma:remainder_estG}-\ref{lemma:projection_of_if_ipw} are deferred to Appendix \ref{app-sec:proof_Thm_if_slope}.
In particular, Lemma \ref{lemma:projection_of_if_ipw} shows that $\IF^{nu}_G$ is the projection of $\IF^{ipw}_G$ onto the nuisance tangent space
\begin{align*}
  \mathbf{T}^{nu}(G) = \left\{\int_{\mathcal{T}} H(s)\, dM(s) \,\middle|\,  H\colon \mathcal{T} \rightarrow \mathbb{R} \right\},                                                %\label{eq:Tnu}
\end{align*}
where $H$ is any measurable function for which the integral has finite variance, $M$ is the martingale part of the single-jump counting process for a censored observation: 
$dM(s) = 1(X \in ds, \delta=0)-1(X \geq s)d\Lambda(s)$ and the filtration is $\mathcal{F}_{s} = \sigma\{1(X \leq s'\,, \delta=0), 1(X \geq s'),\, U,\; s' \le s \in \mathcal{T}\}$.

\subsection{One-step estimator}   
\label{sec:one_step_estimator}
In Section~\ref{sec:ipw}, we showed that the inverse probability weighted estimator $\Psi_G(\Pn)$ has influence function $\IF^{ipw}_G(\cdot|P)$. Of course, this estimator can only be evaluated in the case that $G$ is known. In Section~\ref{sec:estimation_G} we showed that plug-in of the Kaplan--Meier estimator of $G$ leads to improved efficiency, even in the case that $G$ is known, and the resulting estimator is  regular and asymptotically linear with influence function $\IF(\cdot|P)$ in the model where $G$ is unknown but $C$ is assumed independent of $(T,U)$. 
However, as will become apparent, $\IF(\cdot|P) \in \mathbf{T}^{nu}(G)^{\perp} $ does not fall in the tangent space $\mathbf{T}^{\mathcal{M}}(P)$ at $P$ in the model $\mathcal{M}$, where $\perp$ denotes the orthogonal complement in $L_0^2(P)$. Therefore we need to project $\IF(\cdot|P)$ onto $\mathbf{T}^{\mathcal{M}}(P)$ to obtain an efficient influence function $\IF^*$. Once we have access to $\IF^*$, it will be then be feasible to construct an asymptotically efficient one-step estimator of $\Psi(P)$. 

To compute this projection, despite our assumption of independent censoring, it is convenient to consider the broader coarsening-at-random (CAR) model $\mathcal{M}^{car}\supseteq \mathcal{M}$. Under $\mathcal{M}^{car}$, $G$ is viewed as a survival function for $C$ conditionally on  $U$, and this survival function may depend nontrivially on $U$. Since we have assumed that $C$ is independent of $(T,U)$ for the particular distribution that generated our data, this conditional survival function is equal to the marginal survival function $G(\cdot)$ for that distribution. This observation slightly simplifies the expression for the  tangent space for $G$ in $\mathcal{M}^{car}$, which is given by
\[
  \mathbf{T}^{car}(G) = \left\{\int_{\mathcal{T}} H(U,s)\,dM(s)\,\middle|\, H \colon \mathbb{R} \times \mathcal{T} \rightarrow \mathbb{R} \right\},
\]
where $H$ is any measurable function for which the integral has finite variance, and $dM(s) = 1(X \in ds, \delta=0)-1(X \geq s)d\Lambda(s)$ with $\Lambda( \cdot )$ as the cumulative hazard function corresponding to $G(\cdot)$ with respect to the filtration $\mathcal{F}_{s} = \sigma\{1(X \leq s'\,, \delta=0)\,, 1(X \geq s')\,, U, \; s' \le s \in \mathcal{T}\}.$
See Example 1.12 in \cite{vanderLaan&Robins2003} for further details. Moreover, since $\mathcal{M}\subseteq \mathcal{M}^{car}$, $\mathbf{T}^{nu}(G)\subseteq \mathbf{T}^{car}(G)$.

To obtain the efficient influence function $\IF^*$, we could project $\IF(\cdot|P)$ onto $\mathbf{T}^{\mathcal{M}}(P)$. To compute this projection, it will be useful to first show that $\mathbf{T}^{car}(G)^{\perp}\subseteq \mathbf{T}^{\mathcal{M}}(P)$. This can be shown as follows. Let $\mathbf{T}^{car}(Q)$ and $\mathbf{T}^{\mathcal{M}}(Q)$ respectively denote the tangent space generated by local fluctuations of $P$ in the CAR and $\mathcal{M}$ models that modify $Q$ but leave $G$ unchanged. Because both CAR and $\mathcal{M}$ could induce a (locally) nonparametric model for $Q$, $\mathbf{T}^{car}(Q)=\mathbf{T}^{\mathcal{M}}(Q)$. Furthermore, because $P$ factorizes as a product of mappings of the variation-independent components $Q$ and $G$, we can write (i) $\mathbf{T}^{car}(P)=\mathbf{T}^{car}(Q)\oplus \mathbf{T}^{car}(G)$ and (ii) $\mathbf{T}^{\mathcal{M}}(P)=\mathbf{T}^{\mathcal{M}}(Q)\oplus \mathbf{T}^{\mathcal{M}}(G)$, as orthogonal sums.  By (i) and the fact that $\mathbf{T}^{car}(P)=L_0^2(P)$, we have $\mathbf{T}^{car}(Q)=\mathbf{T}^{car}(G)^\perp$. Hence, by (ii) and the fact that $\mathbf{T}^{car}(Q)=\mathbf{T}^{\mathcal{M}}(Q)$, we find that $\mathbf{T}^{car}(G)^{\perp}\subseteq \mathbf{T}^{\mathcal{M}}(P)$.

Using $\Pi(\cdot|S)$ to denote the projection operator onto a closed linear subspace $S\subseteq L_0^2(P)$, we have
\begin{align*}
  &\IF^*(\cdot|P)= \Pi( \;  \IF(\cdot|P)\,|\,\mathbf{T}^{\mathcal{M}}(P) \;) \\
  &\quad = \Pi( \; \Pi( \; \IF(\cdot|P)\,|\,\mathbf{T}^{car}(G) \;) + \Pi( \; \IF(\cdot|P)\,|\,\mathbf{T}^{car}(G)^{\perp} \;)\,|\,\mathbf{T}^{\mathcal{M}}(P) \;) \\
  &\quad = \Pi( \; \Pi( \; \IF(\cdot|P)\,|\,\mathbf{T}^{car}(G) \;)\,|\,\mathbf{T}^{\mathcal{M}}(P) \;) + \Pi( \; \Pi( \; \IF(\cdot|P)\,|\,\mathbf{T}^{car}(G)^{\perp} \;)\,|\,\mathbf{T}^{\mathcal{M}}(P) \;) \\
  &\quad = \Pi( \; \IF(\cdot|P)\,|\,\mathbf{T}^{nu}(G) \;) + \Pi( \; \IF(\cdot| P)\,|\,\mathbf{T}^{car}(G)^{\perp} \,),
\end{align*}
where the final equality uses that (1) $\mathbf{T}^{\mathcal{M}}(P)\cap \mathbf{T}^{car}(G)=\mathbf{T}^{nu}(G)$;
(2)  $\mathbf{T}^{car}(G)^{\perp} \subseteq \mathbf{T}^{\mathcal{M}}(P)$. By Lemma~\ref{lemma:projection_of_if_ipw}, we know that $\IF(\cdot|P) \in \mathbf{T}^{nu}(G)^{\perp}$, implying that the first term on the right-hand side is zero. The same lemma tells us that $\IF(\cdot|P)=\Pi(\; \IF^{ipw}_G(\cdot|P) \,|\, \mathbf{T}^{nu}(G)^\perp \;)$, so the above display, along with the fact that $\mathbf{T}^{nu}(G)\subseteq \mathbf{T}^{car}(G)$, implies that
\begin{align*}
  \IF^*(\cdot|P)&= \Pi( \; \IF^{ipw}_G(\cdot|P)\,|\,\mathbf{T}^{car}(G)^{\perp} \,).
\end{align*}
It remains to project $\IF^{ipw}_G(\cdot|P)$ onto $\mathbf{T}^{car}(G)^{\perp}$. The projection of $\IF^{ipw}_G(\cdot|P)$ onto $\mathbf{T}^{car}(G)$ is given by
\begin{equation} \label{eq:if_nu_eff}
\begin{split}
  \IF^{car}_G(\cdot|P)  \colon o \mapsto \frac{(u-P[U])}{{\rm Var}(U)}\int_{\mathcal{T}} E[T|U=u,X \geq s] \frac{1(x \in ds, \delta=0)-1(x \geq s)d\Lambda(s)}{G(s)}, 
\end{split}
\end{equation}
and the projection onto $\mathbf{T}^{car}(G)^{\perp}$ is given by
\begin{align}\label{eq:if_star}
  \IF^*(\cdot|P) = \IF^{ipw}_G(\cdot|P) - \IF^{car}_G(\cdot|P).
\end{align}

In general, $\IF(\cdot|P)$ is not  equivalent to 
$\IF^*(\cdot|P)$, so $\IF(\cdot|P)$ does not  fall in the tangent space $\mathbf{T}^{\mathcal{M}}(P)$ at $P$ for the model $\mathcal{M}$, and therefore also that the estimator from Section~\ref{sec:estimation_G} is  inefficient.  Note that $\IF^{car}_G(\cdot|P)$ in \eqref{eq:if_nu_eff} is similar to the earlier definition of $\IF^{nu}_G(\cdot|P)$, but differs in that the marginal expectation over $P$ is removed, which can lead to an efficiency gain. Indeed, unlike the influence functions for the estimators in the previous two subsections, the above function applied to an observation $O_i$ will make use of the (always) observed $U_i$, regardless of whether that individual's event time was right-censored, both to estimate the $U$ portion of the covariance between $U$ and $T$ and to evaluate the conditional residual life function on the actually observed $U_i$, rather than on an ``average'' $U$.
This can be derived by noting that, under a coarsening at random model, the nuisance tangent space consists of all functions $f\in L^2(P)$ satisfying $E[f(O)|T,U]=0$ almost surely, where the expectation is over the full data distribution whose coarsening gave rise to $P$ \cite[Theorem 1.3 in][]{vanderLaan&Robins2003}. This projection is given in Proposition~5.4 of \cite{vanderLaanetal2000}, and takes the form given in \eqref{eq:if_nu_eff}.

%We also provide an alternative expression for $\IF^{car}_G(\cdot|P)$ that will be useful in the upcoming developments. Recall that $\tilde P$ is the joint distribution of $(U,\tilde Y)$, in which $G$ is known. Using the equality $E[T|U=u,X \geq s] = G(s)E[\tilde{Y}|U=u,X \geq s]$, the display in \eqref{eq:if_nu_eff} turns into
%\begin{equation}
%\begin{split}
%  o \mapsto \frac{(u-P[U])}{{\rm Var}(U)}\int_{\mathcal{T}} E[\tilde{Y}|U=u,X \geq s]\{1(x \in ds, \delta=0)-1(x \geq s)d\Lambda(s)\}.  \label{eq:if_nu_eff.2}
%\end{split}
%\end{equation}

To construct the one-step estimator we need to estimate the various components of $P$ in (\ref{eq:if_star}).  To do that, we introduce the following estimator $\hat{P}_n$ of the various features of $P$ that are needed$\colon$
\begin{itemize} 
\item[] (i) $\mathbb{Q}_{n} \colon$ the empirical distribution used to estimate the marginal distribution of any given single predictor $U$, denoted by $Q_u$.

\item[] (ii)  $\hat{G}_{n}(\cdot|u) \colon$ the Kaplan--Meier estimator of the censoring distribution conditional on $c(U)=c(u)$, with $c$ a fixed, user-defined coarsening so that $c(U)$, $U\sim P$, is a finitely supported discrete random variable. Here $\hat{G}_{n} $ is defined as a maximum likelihood estimator in a model whose tangent space for $G$ is denoted by $\mathbf{T}^*(G)$. The corresponding estimator of the conditional cumulative hazard function is
\[
  \hat{\Lambda}_n(\cdot|u) = \int_0^{\cdot} \frac{1(Y_n(u,s)>0)}{Y_n(u,s)}N_n(u,ds),
\]
where 
\begin{align*}  %\label{eq:agg_processes}
  N_n(u,s)=\sum_{i=1}^n 1(X_i \leq s, \delta_i=0, c(U_i)=c(u)),\;
  Y_n(u,s) = \sum_{i=1}^n 1(X_i \geq s, c(U_i)=c(u))
\end{align*}
denote the stratified basic counting process and the size of the risk set at time $s$, respectively.

\item[] (iii) $\hat{E}_n(u,s) \colon$ an estimator of $E[\tilde{Y}|U=u, X\ge s]$ under $\tilde P$ (the joint distribution of $(U,\tilde Y)$), which is restricted to take values in some $P$-Donsker family of uniformly bounded functions of $(u,s)\in \mathbb{R} \times \mathcal{T}$. Note that along with (ii), this is equivalent to estimating $E[T|U=u,X \geq s]$, according to the equality $E[T|U=u,X \geq s] = G(s)E[\tilde{Y}|U=u,X \geq s]$.
We suppress the argument $s$ if $s=-\infty$, namely, using $\hat{E}_n(u)$ to estimate $E[\tilde{Y}|U=u]$ that is equal to $E[T|U=u]$. See Remark \ref{remark:est_conditional_res_life} below for one possible estimator of $E[\tilde{Y}|U=u, X\ge s]$. Also, we require that, for each $u$, the process $s \mapsto \hat{E}_n(u,s)$ is predictable with respect to the filtration
\begin{align*}
  \sigma\{N_n(\cdot, s')\,,\, Y_n(\cdot, s')\,,\, U_i\,,\, i=1,\ldots,n,\, s' \le s \in \mathcal{T}\}.
\end{align*}
\end{itemize}

In view of (\ref{eq:def_parameter.2}), which can be expressed in terms of  $E[\tilde{Y}|U]$, $Q_u$ and $G$, we see that $\Psi(P)$ can be estimated by plugging-in  $\hat{P}_n$:
\begin{align*}
  \Psi(\hat{P}_n) = \frac{{\rm Cov}_{\mathbb{Q}_{n}}(U, \hat E_n(U))}{{\rm Var}_{\mathbb{Q}_{n}}(U)}.    %\label{eq:est_Psin}
\end{align*}
Moreover, $\IF^{ipw}(\cdot|P)$ in \eqref{eq:if_ipw} and $\IF^{car}$ in \eqref{eq:if_nu_eff} can be re-expressed in terms of these features of $P$ as
\begin{equation} \label{eq:reexpressed_if}
\begin{split}
  &\IF^{ipw}(\cdot|P) \colon o \mapsto  \frac{(u-Q_u[U])\big(\tilde{y}-Q_u[E[\tilde{Y}|U]]\big)}{{\rm Var}_{Q_u}(U)}
  -\frac{{\rm Cov}_{Q_u}(U, E[\tilde{Y}|U])}{{\rm Var}_{Q_u}^2(U)}(u-Q_u[U])^2;\\
  &\IF^{car}(\cdot|P)  \colon o \mapsto \frac{(u-Q_u[U])}{{\rm Var}_{Q_u}(U)}\int_{\mathcal{T}} E[\tilde{Y}|U=u,X \geq s]\{1(x \in ds, \delta=0)-1(x \geq s)d\Lambda(s)\},
\end{split}
\end{equation} 
where the expression of $\IF^{car}(\cdot|P)$ follows by applying the equality $$E[T|U=u,X \geq s] = G(s)E[\tilde{Y}|U=u,X \geq s]$$ to \eqref{eq:if_nu_eff}.
This implies that $\IF^{*}(\cdot|P)=\IF^{ipw}(\cdot|P)-\IF^{car}(\cdot|P)$ is represented by the introduced features of $P$, and
enables the estimation of ${\rm IF}^{*}(\cdot|P)$ in terms of $\hat{P}_n$. 
Note that henceforth the subscript $G$ is suppressed because $G$ is one of the three identified features of $P$.
We also wish to point out that the second term in $\mathbb{P}_n {\rm IF}^{ipw}(\cdot|\hat{P}_n)$ is precisely $\Psi(\hat{P}_n)$.

The one-step estimator is then given by
\begin{equation} \label{eq:if_onestep}
\begin{split} 
  &S(\mathbb{P}_n, \hat{P}_n) = \Psi(\hat{P}_n) + \mathbb{P}_n \IF^*(\cdot|\hat{P}_n) \\
  &\quad =\Psi(\hat{P}_n) +\frac{\mathbb{P}_n(U-\mathbb{Q}_{n}[U])(\delta X / \hat{G}_n(X|U)-\mathbb{Q}_{n}[\hat{E}_n(U)])}{{\rm Var}_{\mathbb{Q}_{n}}(U)}
  -\Psi(\hat{P}_n)- \mathbb{P}_n \IF^{car}(\cdot|\hat P_n) \\    
  & \quad = \frac{\mathbb{P}_n(U-\mathbb{Q}_{n}[U])Y}{{\rm Var}_{\mathbb{Q}_{n}}(U)}
  -\frac{1}{{\rm Var}_{\mathbb{Q}_{n}}(U)}\mathbb{P}_n\left[(U-\mathbb{Q}_{n}[U])\int_{\mathcal{T}} \hat{E}_n(U,s)\hat{M}(ds|U)\right],
  %& \quad = \frac{\mathbb{P}_n(U-\mathbb{Q}_{n}[U])Y}{{\rm Var}_{\mathbb{Q}_{n}}(U)}
  %-\frac{1}{{\rm Var}_{\mathbb{Q}_{n}}(U)}\mathbb{P}_n\left[(U-\mathbb{Q}_{n}[U])\int_{\mathcal{T}} \hat{E}_n(U,s)\hat{M}(ds|U)\right],
\end{split}
\end{equation}
where $\hat{M}(ds|u)=1(X \in ds, \delta=0, c(U)=c(u))-1(X \geq s, c(U)=c(u))d\hat{\Lambda}_n(s|u)$; the third equality holds by
$\mathbb{P}_n(U-\mathbb{Q}_{n}[U])\mathbb{Q}_{n}[\hat{E}_n(U)]=\mathbb{Q}_n(U-\mathbb{Q}_{n}[U])\mathbb{Q}_{n}[\hat{E}_n(U)]=0$ and 
$Y = \delta X / \hat{G}_n(X|U)$. 

We next establish the asymptotic linearity of this one-step estimator, for which we need an extended version of \ref{assump:At-risk prob} and a mild stability condition on the behavior of the estimator $\hat E_n$.
\begin{assumpenum}[resume=assumptions]
  \item \label{assump:At-risk prob_covariate} Assumption \ref{assump:At-risk prob} holds for each subgroup defined by the coarsening of $U$.
  
  \item \label{assump:Conditional mean} There exists a uniformly-bounded, non-random function $(u,s) \mapsto \bar{E}(u,s)$ that is left-continuous in $s$, such that
  %$\mathbb{E}\{\sup_{(u,s)}|\hat{E}_n(u,s) - \bar{E}(u,s)|\} = o(1)$
  $\mathbb{E}\{|\hat{E}_n(u,s) - \bar{E}(u,s)|\} = o(n^{-1/4})$ for each $(u,s)$ and $\sup_{(u,s)}|\hat{E}_n(u,s) - \bar{E}(u,s)|$ is bounded in probability,
  %$n^{1/2}\mathbb{E}\{[\hat{E}_n(u,s) - \bar{E}(u,s)]^2\} \to 0$ and $|\hat{E}_n(u,s) - \bar{E}(u,s)| = O_p(n^{-1/4})$ uniformly on $(u,s)$,
  where $\mathbb{E}$ denotes the expectation over $O_1,\ldots,O_n.$
\end{assumpenum}
When $\bar{E}(U) \not= E[\tilde{Y}|U]$, where
$\bar E(u) = \bar E(u, -\infty)$ and $\bar E$ is from \ref{assump:Conditional mean}, we need a correction to $\IF^*$ given by
\begin{align*} %\label{eq:if_replace_means}
  \IF^{\dagger}(\cdot|\bar{E},P) : o \mapsto \frac{{\rm Cov}_{Q_u}(U, \bar{E}(U)-E[\tilde{Y}|U])}{{\rm Var}^2_{Q_u}(U)}[(u-Q_u[U])^2-{\rm Var}_{Q_u}(U)].
\end{align*}
\begin{theorem}\label{Thm:if_one_step}
Under conditions \ref{assump:Covariates}, \ref{assump:Survival function}, \ref{assump:At-risk prob_covariate} and \ref{assump:Conditional mean},
\begin{align*}
  S(\mathbb{P}_n, \hat{P}_n)-\Psi(P)=[\mathbb{P}_n-P]\,\Pi\big( \IF^*(\cdot|\bar{E}, Q_u, G) + \IF^{\dagger}(\cdot|\bar{E},P)\,|\,\mathbf{T}^*(G)^\perp \big) + o_p(n^{-1/2}),
\end{align*}
where $\mathbf{T}^*(G)$ is the tangent space defined in (iii) above.
\end{theorem}
The proof of this result is given in Appendix \ref{app-sec:proof_Thm_if_one_step}.

\begin{remark}
Suppose that $\hat{E}_n(u,s)$ consistently estimates the true conditional residual life function, that is, $\bar{E}(u,s) = E[\tilde{Y}|U=u, X \geq s]$  for all $(u,s)$, in which case $\IF^{\dagger}(\cdot|\bar{E},P)=0$. Recall that 
$\IF^*(\cdot|\bar{E},Q_u,G) = \Pi(\,\IF^{ipw}(\cdot|P)\,|\,\mathbf{T}^{car}(G)^\perp\,)$, and therefore $\IF^*(\cdot|\bar{E},Q_u,G)\in \mathbf{T}^{car}(G)^\perp$. Furthermore, $\mathbf{T}^*(G) \subseteq \mathbf{T}^{car}(G)$, and so $\mathbf{T}^*(G)^\perp \supseteq \mathbf{T}^{car}(G)^\perp$. Consequently, we know that $\IF^*(\cdot|\bar{E},Q_u,G)\in \mathbf{T}^*(G)^\perp$, which yields that
\[
 \Pi(\,\IF^{*}(\cdot|\bar{E},Q_u,G)\,|\,\mathbf{T}^*(G)^\perp\,) = \IF^{*}(\cdot|\bar{E},Q_u,G) \equiv \IF^*(\cdot|P).
\]
Hence, $S(\mathbb{P}_n, \hat{P}_n)$ is asymptotically linear with the influence function equal to the efficient influence function, that is, $S(\mathbb{P}_n, \hat{P}_n)$ is asymptotically efficient. Furthermore, under regularity conditions, the empirical variance %$\hat{\sigma}^2$ 
of $\IF^*(O|\hat{P}_n)$ converges to the variance of $\IF^{*}(O|P)$ under $O\sim P$, that is, to the variance of the one-step estimator.
\end{remark}

\begin{remark}
If $\bar{E}(u) \neq E[\tilde{Y}|U=u]$ for some $u$, it is still possible that $\IF^{\dagger}(\cdot|\bar{E},P)=0$.
For instance, if we assume that $\Psi(P)$ does not depend on the limit of $\hat{E}_n(u)$ to ensure identifiability, then, recalling the definition of $\Psi(P)$ in \eqref{eq:def_parameter}, we have 
${\rm Cov}_{Q_u}(U,\bar{E}(U))={\rm Cov}_{Q_u}(U,E[\tilde{Y}|U])$ and thus $\IF^{\dagger}(\cdot|\bar{E},P)=0$.
\end{remark}

\begin{remark}
Regardless of whether $\IF^{\dagger}(\cdot|\bar{E},P)=0$ or not, the variance of the influence function of $S(\mathbb{P}_n, \hat{P}_n)$ is no larger than that of $\IF^*(\cdot|\bar{E}, Q_u, G) + \IF^{\dagger}(\cdot|\bar{E},P)$ because projections only decrease the variance.
For any real number $m$, define
\begin{align*} 
  \IF_m^{\dagger}(\cdot|\bar{E},P) : o \mapsto \frac{{\rm Cov}_{Q_u}(U, \bar{E}(U)) - m}{{\rm Var}^{2}_{Q_u}(U)}[(u-Q_u[U])^2-{\rm Var}_{Q_u}(U)].
  %\IF_m^{\dagger}(\cdot|\bar{E},P) : o \mapsto \frac{{\rm Cov}_{Q_u}(U, \bar{E}(U)){\rm Var}^{-1/2}_{Q_u}(U) - m}{{\rm Var}^{3/2}_{Q_u}(U)}[(u-Q_u[U])^2-{\rm Var}_{Q_u}(U)].
\end{align*}
When $m=m^* \equiv {\rm Cov}_{Q_u}(U,E[\tilde{Y}|U])$, we have $\IF_{m^*}^{\dagger}(\cdot|\bar{E},P)=\IF^{\dagger}(\cdot|\bar{E},P)$.
Note also that $|m^*| \le c{\rm Var}^{1/2}_{Q_u}(E[\tilde{Y}|U]) \le M$ since ${\rm Var}^{1/2}_{Q_u}(U) \le c$,
where the upper bounds $c$ and $M$ exist according to \ref{assump:Covariates} and \ref{assump:Survival function}. Therefore, we see the variance of
$\IF^*(\cdot|\bar{E}, Q_u, G) + \IF^{\dagger}(\cdot|\bar{E},P)$ is upper-bounded by
\begin{align} \label{eq:var.ub}
  %P \Big(\big[\IF^*(\cdot|\bar{E}, Q_u, G) + \IF^{\dagger}(\cdot|\bar{E},P)\big]^2\Big) \le 
  \sup_{-M \le m \le M} P \Big( \big[\IF^*(\cdot|\bar{E}, Q_u, G) + \IF_m^{\dagger}(\cdot|\bar{E},P)\big]^2 \Big).
\end{align}
For each given $m$, the variance of $\IF^*(\cdot|\bar{E}, Q_u, G) + \IF_m^{\dagger}(\cdot|\bar{E},P)$ can be estimated via the sample variance of the same quantity but with the unknown parameters replaced by the corresponding estimates. 
Therefore, the quantity in \eqref{eq:var.ub} can be estimated by taking the supremum of these estimates over $m \in [-M,M]$, and we denote this resulting estimate by $\sigma^{\dagger\,2}_n$. Then $\sigma^{\dagger\,2}_n$ is also a valid upper bound of the sample variance of the influence function of $S(\mathbb{P}_n, \hat{P}_n)$, so the Wald-type confidence intervals constructed using $\sigma^{\dagger\,2}_n$ will have conservative coverage asymptotically. An alternative is to construct a (conservative) confidence interval using the bootstrap percentile t-method.
\end{remark}

\begin{remark}[Estimation of the conditional residual life function] \label{remark:est_conditional_res_life}
Condition \ref{assump:Conditional mean} requires that $\hat{E}_n(u,s)$ has some stable limit $\bar{E}(u,s)$.
Along the lines of \cite{vanderLaan1998}, such an estimator can be constructed by regressing $Y$ on $U$ using only a sub-sample $\{O_i=(X_i,\delta_i, \bs{U}_i)\, ,\, i:X_i \geq s\}$ in the fashion of \cite{Koul1981}, leading to
\begin{equation} \label{eq:est_E}
  \hat{E}_n(u,s) = \mathbb{P}_{n}[Y1(X \geq s)] + \frac{{\rm Cov}_{\mathbb{P}_{n}}(U1(X \geq s), Y1(X \geq s))}{{\rm Var}_{\mathbb{P}_{n}}(U1(X \geq s))}(u-\mathbb{P}_{n}[U1(X \geq s)]).
\end{equation}
Empirical process theory can be used to show that $\hat{E}_n$ satisfies \ref{assump:Conditional mean}.
%The assumption $\mathbb{E}\{\hat{E}_n(u,s) - \bar{E}(u,s)\}^2 \to 0$ for $(u,s)\in \mathbb{R} \times \mathcal{T}$ can be verified using Theorem 3.1 of \cite{Koul1981}.
%To show that $\mathbb{E}\{\hat{E}_n(u,s) - \bar{E}(u,s)\}^2$ is uniformly bounded, it suffices to show the uniform boundedness of $\mathbb{E}\{\hat{E}_n(u,s)\}^2$ and $ \bar{E}(u,s)$.
%This is easy to see because $U$ is assumed to be bounded and non-degenerate in \ref{assump:Covariates},  the observed time $X$ is bounded by the end of follow-up $\tau$, and $G(\tau) > 0$ by \ref{assump:Survival function}.  %and  $\hat{G}_n$ converges to $G$ in probability over $(u,s)\in \mathbb{R} \times \mathcal{T}$.
\end{remark}

\section{Stabilized one-step estimator}
\label{sec:stabilized_onestep}
In this section, we adapt the stabilization approach
of \cite{Luedtke2018} to obtain a confidence interval 
of  the target parameter $\Psi(P)$ defined in \eqref{eq:def_parameter}.
The idea is first to randomly order the data, and consider subsamples consisting of the first $j$ observations for $j=q_n,\ldots,n-1$, where $\{q_n\}$ is some positive integer sequence such that both $q_n$ and $n-q_n$ tend to infinity. 
Based on the subsample of size $j$, an estimator of the label of the most informative predictor is
\begin{align} \label{eq:predictor_selection}
  k_{j} = \arg \max_{k = 1,\ldots,p} |\Psi_{\hat{G}_j,k}(\mathbb{P}_j)|
  \equiv \arg \max_{k = 1,\ldots,p} \left| \frac{{\rm  
     Cov}_{\mathbb{P}_{j}}(U_k, \delta X/\hat G_j(X))}{{\rm Var}_{\mathbb{P}_{j}}(U_k)} \right|,
\end{align}
where $\Psi_{\hat{G}_j,k}(\mathbb{P}_j)$ is $\Psi_{\hat{G}_n}(\mathbb{P}_n)$ with $n=j$ and predictor $U_k$. A uniform version of the earlier condition \ref{assump:Covariates}, as well as an extended version $\hat{E}_n(u,s,k)$ of \eqref{eq:est_E}, are now understood to apply to each predictor $U_k$, $k=1,\ldots, p$, and $\hat{G}_j$ is the usual Kaplan--Meier estimator of $G$.
The stabilized one-step estimator of $\Psi(P)$ is then given by
\begin{equation} \label{eq:Sn_star}
  S^*_{n} = \frac{1}{n-q_n}\sum_{j=q_n}^{n-1}w_{nj}m_jS_{k_j}(\delta_{O_{j+1}}, \hat{P}_{nj})\, ,                                             
\end{equation}
where $m_j \in \{-1,1\}$ is the sign of $\Psi_{\hat{G}_j,k_j}(\mathbb{P}_{j})$, $S_{k_j}$ refers to \eqref{eq:if_onestep} with the predictor $U$ now being $U_{k_j}$, and $\hat{P}_{nj}\equiv(\hat{E}_j, \mathbb{Q}_j, \hat{G}_n)$ that refers to $\hat{P}_n$ based on only the first $j$ observations to estimate part of the parameters of $P$.
Here $\delta_{O_{j+1}}$ is the Dirac measure putting unit mass at $O_{j+1}$, $w_{nj} \equiv \bar{\sigma}_n/\hat{\sigma}_{nj}$ with $\bar{\sigma}_n=\{(n-q_n)^{-1}\sum_{j=q_n+1}^{n}(1/\hat{\sigma}_{nj})\}^{-1}$,  
\begin{align*}
  \hat{\sigma}^2_{nj} = \frac{1}{j}\sum_{i=1}^j\bigg\{m_j\IF^*_{k_j}(O_i|\hat{P}_{nj})
  -\frac{1}{j}\sum_{i=1}^jm_j\IF^*_{k_j}(O_i|\hat{P}_{nj})
  \bigg\}^2,                                
\end{align*}
and $\IF^*_{k_j}$ is $\IF^*$ with the predictor taken as $U_{k_j}$.

Note that $\hat{P}_{nj}$ in the stabilized-one-step estimator $S^*_{n}$ involves subsamples. As we will see from simulation studies, however, using the full-sample estimator $\hat{P}_{n}$ instead, considerably improves the  performance of $S^*_{n}$ in small samples.  

The following $95\%$ confidence interval for $\Psi(P)$ is justified by the asymptotic normality of $S_n^*$ given in Theorem \ref{Thm:stab_one_step} below:
\begin{align*} %\label{eq:ci_sn_star}
  [{\rm LB}_n, {\rm UB}_n]=\left[S^*_n \pm 1.96\frac{\bar{\sigma}_n}{\sqrt{n-q_n}}\right],
\end{align*} 
and the two-sided p-value is
\begin{align*}
   2(1-\Phi\left(\,\left|\sqrt{n-q_n}S_n^*/\bar{\sigma}_n\right|\,\right)),
\end{align*}
where $\Phi$ is the cumulative distribution function of $\mathcal{N}(0,1)$.

\begin{theorem}\label{Thm:stab_one_step}
Suppose the number of predictors $p=p_n$ satisfies
$\log(p_n)/n^{1/4} \to 0$, and the smallest subsample size $q_n$ used for stabilization satisfies $n-q_n\to \infty$, $n/q_n = O(1)$, and $q_n^{1/4}/\log( n \lor p_n) \to \infty$.
Assume  \ref{assump:Covariates} holds uniformly for all predictors, \ref{assump:Survival function}, \ref{assump:At-risk prob}, the asymptotic stability conditions \ref{assump:Conditional_mean_E0}-\ref{assump:Signal strength} that are stated just before the proof in Appendix \ref{app-sec:proof_Thm_stab_one_step},
and the non-degeneracy condition
\begin{assumpenum}[resume=assumptions]
%   \item \label{assump:Variances} ${\rm Var}_{Q_u}(U_k)$, ${\rm Var}(U_k1(X \geq s))$ and $\sigma^2_{nj} \equiv {\rm Var}(\IF^*_{k_j}(O|P))$ are uniformly bounded away from zero and infinity as functions of $(k, s, j, n)$, $1 \leq k \leq p_n$, $s \in \mathcal{T}$, $q_n \leq j \leq n$.
  \item \label{assump:Variances} ${\rm Var}_{Q_u}(U_k)$, ${\rm Var}(U_k1(X \geq s))$ and ${\rm Var}(\IF^*_k(O|P))$ are bounded away from zero and infinity, as functions of $k\in \{ 1,\ldots, p_n\}$ and $s\in\mathcal{T}$.
\end{assumpenum}
Then $S_n^*$ is an asymptotically normal estimator of $\Psi(P) \colon$
\begin{align*}
\sqrt{n-q_n}\bar{\sigma}_n^{-1}[S^*_n-\Psi(P)] \lcrarrow{d} \mathcal{N}(0,1).
\end{align*}
\end{theorem}
The proof is postponed to Appendix \ref{app-sec:proof_Thm_stab_one_step}. Note that condition \ref{assump:Conditional_mean_E0} in Appendix \ref{app-sec:proof_Thm_stab_one_step} removes the need to include $\IF^{\dagger}$ in $\IF^*$ when constructing $S_n^*$.
In practice, it is advisable to  pre-standardize each predictor (as is commonly recommended in the variable selection literature) to provide scale-invariance; the above result is given in terms of the unstandardized predictors for simplicity of presentation. %The last part of \ref{assump:Variances} is a mild ``signal strength" condition, needed to ensure identifiability of an effect when $p_n\to \infty$.

The  stabilized one-step estimator is reminiscent of bagging, the aggregation of  multiple weak learners constructed from subsets of the data (in this case, $S_{k_j}$ for  $j\ge q_n$). The value of $q_n$ determines how many weak learners are collected ($n-q_n$ of them) and plays the role of a tuning parameter. Taking a smaller $q_n$ is expected to reduce  variability in the performance of $S_n^*$, but taking too small value of $q_n$ leads to overfitting.
In practice we recommend setting $q_n = n/2$ (which satisfies the conditions in Theorem \ref{Thm:stab_one_step}) as a reasonable trade-off, although in practice it is advisable to run the analysis for a few values of $q_n$ and compare the results.

\section{Competing methods} \label{sec:competing_methods}

\subsection*{Marginal one-step estimators with Bonferroni correction (Bonferroni One-Step)}
For each predictor $U_k, k=1,\ldots,p$, the marginal test statistic is
$B_k \equiv \sqrt{n}S_{k}(\mathbb{P}_n, \hat{P}_n)/\hat{\sigma}_{k}$, where $\hat{\sigma}_k^2$ is the sample second moment of $\IF^*_k(O|\hat{P}_{n})$. 
Marginal testing over all $k$ with Bonferroni correction controls the family-wise error rate of the global null hypothesis $\bs{\beta}_0 = \bs{0}$.
This method is theoretically supported by Theorem \ref{Thm:if_one_step}.

\subsection*{Oracle one-step estimator (Oracle One-Step)}
In this case, the label $k$ of the most correlated predictor $U_k$ is given, and the test statistic is simply $B_k$, which has an asymptotically standard normal null distribution.  Assuming knowledge of $k$ is of course unrealistic, but this estimator serves as a benchmark against which the other methods can be compared.

\section{Simulation results}
\label{sec:numerical}
In this section we report the results of simulation studies evaluating the  performance of the stabilized one-step estimator (with $q_n=n/2$) in comparison with the competing  methods in Section \ref{sec:competing_methods}. The log-transformed survival times are generated under one of the following AFT  scenarios:
\begin{itemize}
\itemsep=-1pt
\item[] {\bf Model N:} $\; T=\varepsilon$;
\item[] {\bf Model A1:} $\; T=U_1/4+\varepsilon$;
\item[] {\bf Model A2:} $\; T=\sum_{j=1}^p\beta_jU_j+\varepsilon$ with $\beta_1=\ldots=\beta_5=0.15,$ $\beta_6=\ldots=\beta_{10}=-0.1,$\\ $\beta_j=0$ for $j \geq 11$.
\end{itemize}
The noise $\varepsilon$ is distributed as $\mathcal{N}(0,1)$ (independently of $\bs{U}$) or $\mathcal{N}(0,0.7(|U_1|+0.7))$ (conditionally on $\bs{U}$). The predictors $\bs{U}$ have a $p$-dimensional multivariate normal distribution with an exchangeable correlation structure such that ${\rm Corr}(U_j, U_k)=0.75, j \neq k$. In {Model N} there is no active predictor, while there is only a single active predictor in {Model A1}. In {Model A2} there are ten active predictors, each having weaker influence than the single predictor in {Model A1}; the most correlated predictor is not unique in this model. The censoring time $C$ is taken to be the log of an exponential  random variable with rate parameters that give either light censoring ($10\%$) or heavy censoring ($30\%$). Here we just consider light censoring; results for the heavy censoring case are given in Appendix section \ref{sec:simulation_data_analysis_results}. For each data generating scenario, we fix the sample size at $n=500$, and consider 4 or 5 values of $p$ of the form $10^a$ (for $a=2,3, \ldots$). A nominal significance level of $5\%$ is used throughout. 
The Kaplan--Meier estimator $\hat{G}_n$ is used in $S_n^*$, as justified by the independent censoring assumption; although a more sophisticated conditional Kaplan--Meier estimator could be used instead, doing so would involve an additional computational cost.

Empirical rejection rates based on $1000$ Monte Carlo replications under  the various scenarious are  displayed in Figures \ref{fig:empirical_rejectionrate_lightcensoring_partialsample} and \ref{fig:empirical_rejectionrate_lightcensoring_wholesample}. The panels for Null show type I error rates under {Model N} for independent and dependent errors, respectively, with the nominal level of $5 \%$ shown by the horizontal dashed line. Similarly, the panels for (Alternative A1, Alternative A2) show the power under {Model A1} and {Model A2}, with independent and dependent errors, respectively.

The left panels of Figure \ref{fig:empirical_rejectionrate_lightcensoring_partialsample} show the results for independent errors.
The stabilized one-step estimator provides the closest-to-nominal type-I error (apart from the oracle one-step estimator). The right panels of Figure \ref{fig:empirical_rejectionrate_lightcensoring_partialsample} give the results in the case of dependent errors, and show that the stabilized one-step estimator outperforms the Bonferroni one-step method in power when $p$ is larger than $1000$. The Bonferroni one-step estimator is favored in terms of power over the stabilized one-step estimator in the case of independent errors, but not in the case of dependent errors.

In Figure \ref{fig:empirical_rejectionrate_lightcensoring_wholesample}, we see that using the full-sample-based $\hat{P}_n$ improves the performance of the stabilized one-step estimator, but has little effect on the Bonferroni one-step estimator. 
In the right panels, the power of the stabilized one-step estimator is much better than that of the Bonferroni one-step estimator, and it maintains good control of type-I error as well, even with one million predictors. The computational cost of the Bonferroni one-step estimator is prohibitive for $p=10^6$ and it is not included.

A general issue with real data is variation in the results due to the ordering of the data. Ordering the data in a different way can change the value of $S_n^*$ and the resulting p-value, making the result difficult to reproduce. One way to address this issue is to use a Bonferroni correction of the minimal $p$-value resulting from $R\ge 1$ random orderings of the data, taking into account the trade-off in terms of computational cost (which grows proportionally to $R$). 
%Due to the added computational cost of implementing random orderings in the simulation, we restrict attention to $p\le 10^4$ in our further numerical studies with respect to this issue. 
Then the null is rejected if the minimum of the p-values obtained from the $R$ random orderings is less than $5 \%$ (after Bonferroni correction for $R$-fold multiple testing). 

Figure \ref{fig:empirical_rejectionrate_lightcensoring_randomordering_wholesample} examines the effect of multiple random orderings on the performance of the stabilized one-step estimator (taking $R=10$). 
%as reported in 
% Figure \ref{fig:empirical_rejectionrate_lightcensoring_wholesample}. 
The type-I error is now always below $5\%$, and in fact the test has become somewhat conservative. Comparing
Figures \ref{fig:empirical_rejectionrate_lightcensoring_wholesample}-\ref{fig:empirical_rejectionrate_lightcensoring_randomordering_wholesample} shows that the combination of multiple random orderings and the full-sample estimator $\hat P_n$ maintains the power of the stabilized one-step estimator at the cost of being slightly more conservative.
The R code used to conduct this simulation study will be furnished upon request.

\begin{figure}
\begin{center}
 \includegraphics[width=0.9\textwidth]{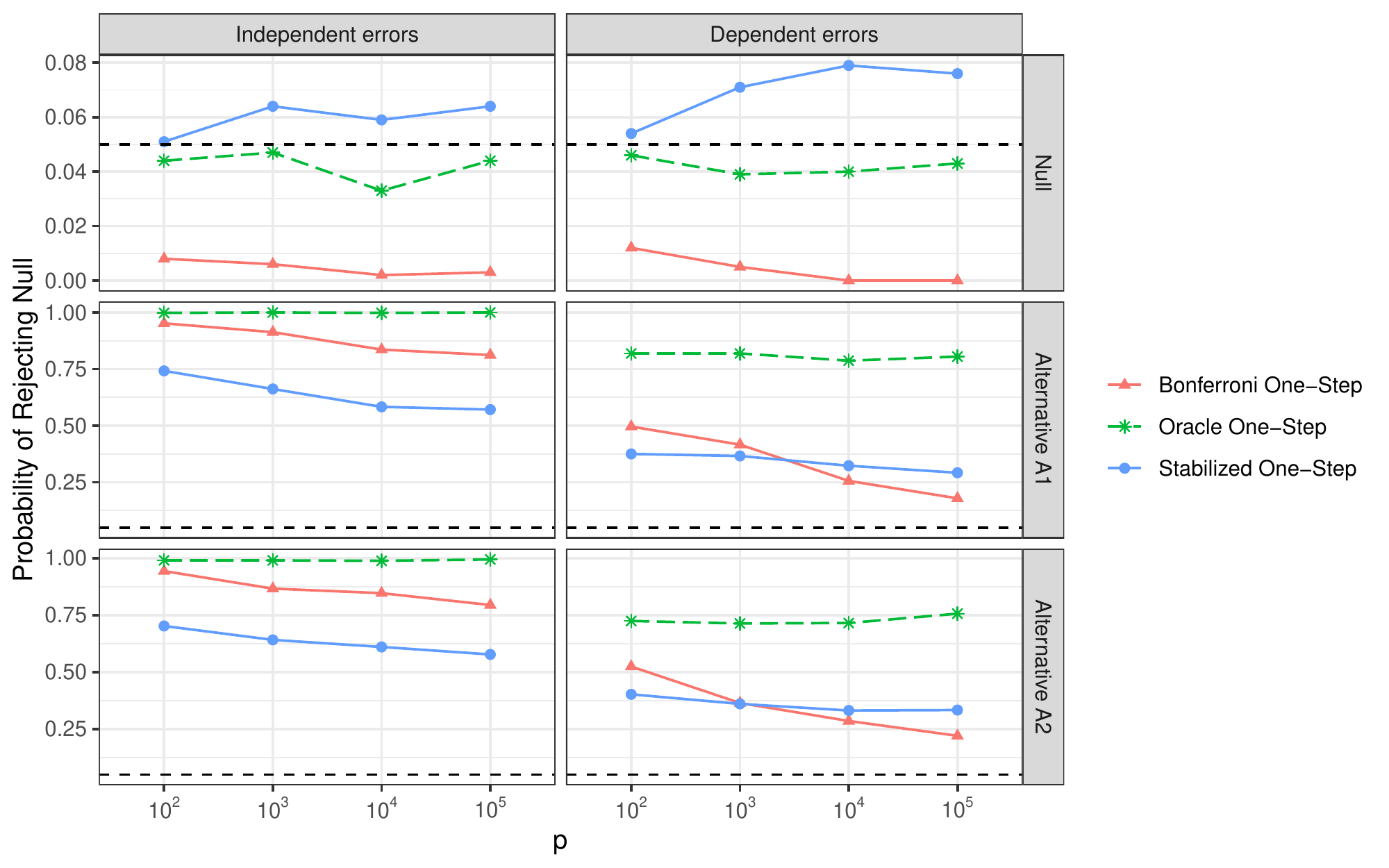}
 \caption{Empirical rejection rates based on $1000$ samples of $n=500$ generated from models with independent and dependent errors under light censoring ($10 \%$), for $p$ in the range $10^2$--$10^5$. The panels tagged by Null give the results under the null model, while those tagged by Alternative A1 and Alternative A2 display the results under alternative models.}
 \label{fig:empirical_rejectionrate_lightcensoring_partialsample}
 \end{center}
\end{figure}

\begin{figure}
\begin{center}
 \includegraphics[width=0.9\textwidth]{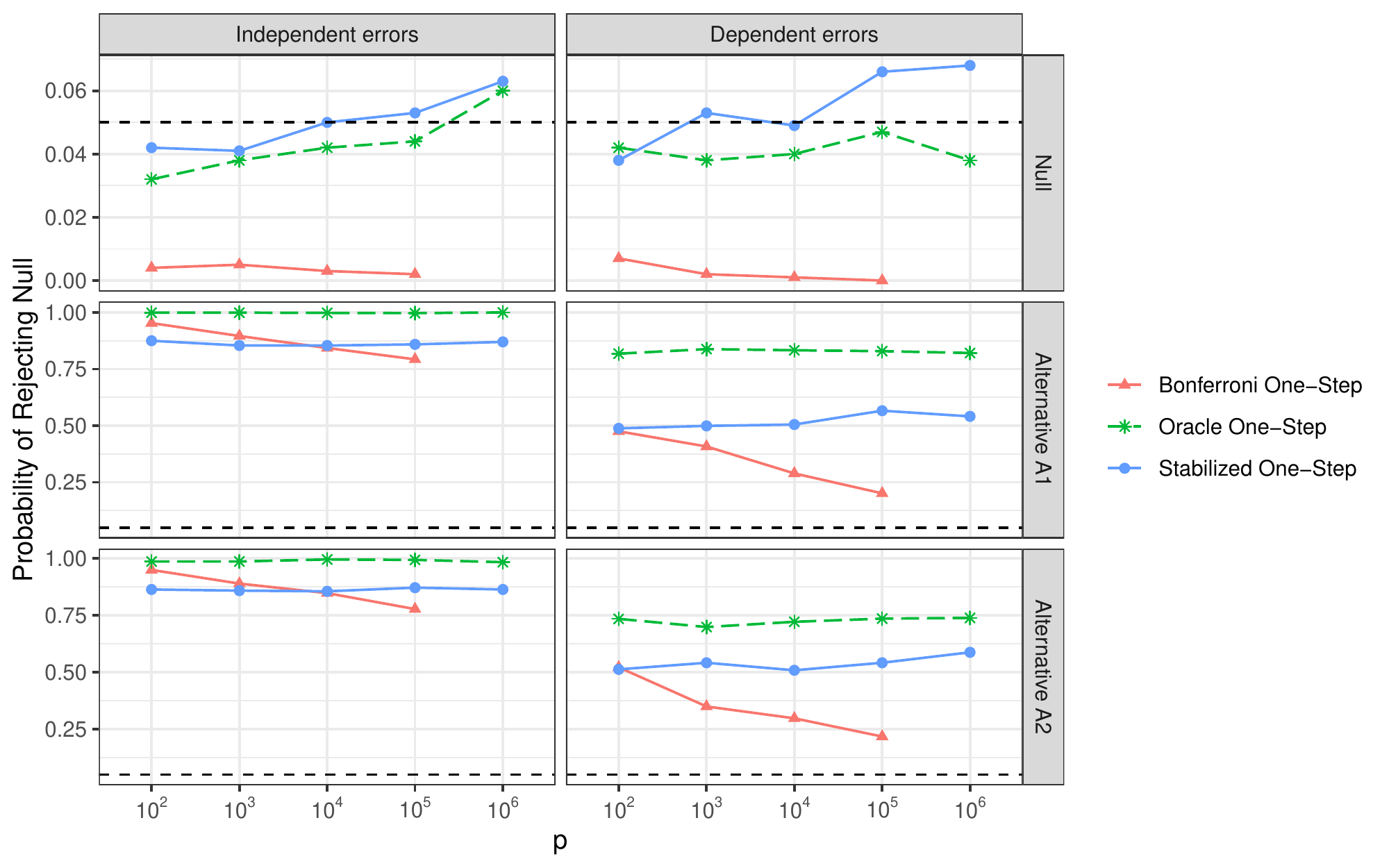}
 \caption{As in Figure \ref{fig:empirical_rejectionrate_lightcensoring_partialsample}, except for using the full sample to estimate the parameters of $P$ that are estimated by the partial sample $q_n=n/2$ in Fig \ref{fig:empirical_rejectionrate_lightcensoring_partialsample}.}
 \label{fig:empirical_rejectionrate_lightcensoring_wholesample}
 \end{center}
\end{figure}

\begin{figure}
\begin{center}
 \includegraphics[width=0.9\textwidth]{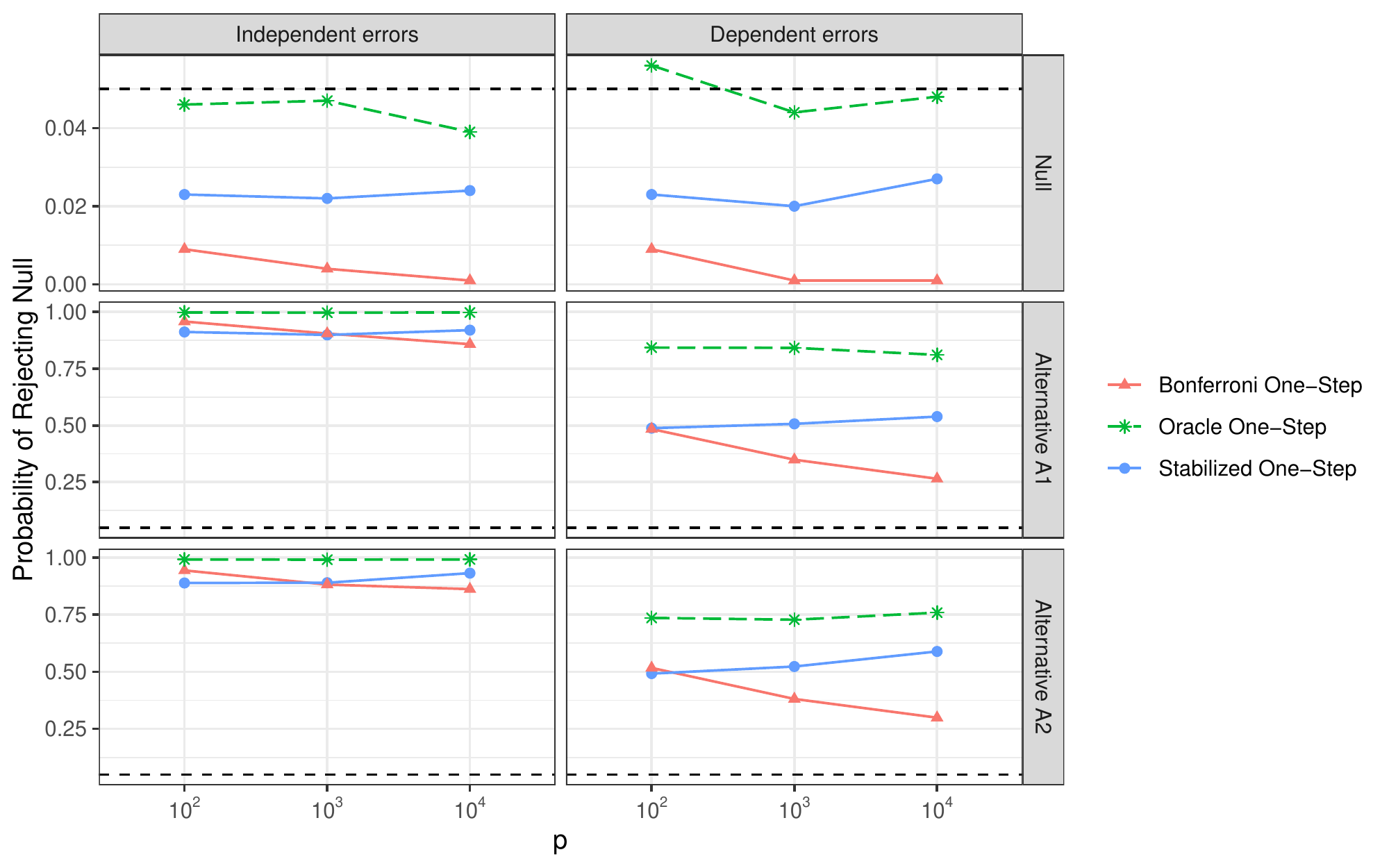}
 \caption{As in Figure \ref{fig:empirical_rejectionrate_lightcensoring_wholesample}, except that the decision for the stabilized one-step estimator uses the (Bonferroni-corrected) minimal p-value from $R=10$ random orderings of the data.}
 \label{fig:empirical_rejectionrate_lightcensoring_randomordering_wholesample}
 \end{center}
\end{figure}

\section{Application to viral replication data} \label{sec:real_data_application}
A widely used measure of the potency of an antiviral drug is the concentration needed  to achieve a $50 \%$ reduction of the (in vitro) rate of viral replication ($\mbox{IC}_{50}$, in units of ${\rm \mu g/mL}$). In this application we treat $\mbox{IC}_{50}$ as a survival time outcome of interest. If the virus is highly resistant, then a $50 \%$ reduction in viral replication rate may not be observed, resulting in a right-censored outcome. We consider the antiviral VRC01, an antibody for HIV-1 that is currently being evaluated in a Phase 2b trial for the prevention of HIV-1 infection \citep{Gilbert2017, Magaret2019}.
In this case, a reduction in viral replication is thought to be caused by a VRC01-mediated neutralization, and the lower $\mbox{IC}_{50}$, the more sensitive the virus is to VRC01-mediated neutralization. %A higher value of $\mbox{IC}_{50}$ implies that the virus is relatively resistant to the neutralization.

Data on a total of $624$ pseudoviruses were retrieved from the CATNAP database \citep{Yoon2015}. We restrict attention to a subgroup of size ($n=611$) after removing $13$ pseudoviruses with unreliable $\mbox{IC}_{50}$ measurements. The censoring rate of $\mbox{IC}_{50}$ is $16 \%$ in the analyzed data set. 
The $611$ pseudoviruses are of $24$ subtypes, where Subtype B and C are predominant over others: $293$ of them ($48.0 \%$) belong to Subtype C and $81$ ($13.3 \%$) are of Subtype B. In terms of geographic regions where the viruses originate, $126$ of $611$ pseudoviruses are from Asia ($20.6 \%$), $96$ from Europe or the Americas ($15.7 \%$), $170$ from Northern Africa ($27.8 \%$), and $219$ from Southern Africa ($35.9 \%$). Pseudoviruses of different subtypes may present varied gene expression, so we analyze the data for Subtypes B and C separately. We set the end of follow-up $\tau$ to the $90$th percentile of $\mbox{IC}_{50}$ in each data set with the corresponding sample size as indicated in Table \ref{tab:features_type_distribution}. 
We aim to investigate whether the potency of VRC01 depends on HIV-1 proteomic characteristics that are presented in lieu of the envelope (Env) amino acid (AA) sequence features. Data on $817$ features are available: 
\begin{enumerate}
\item Binary features:  indicating whether a particular AA sequence appears at a particular position, or whether a position is the starting site of some given enzymatic process. There are  $799$  features of this type.

\item Count features: representing  total numbers of enzyme-directed chemical reactions observed to take place within a given region, or the total length of the aligned sequences over a region. There are $18$ count features.
\end{enumerate}

To simplify interpretation of the effects of binary or count features, we carry out separate analyses for each type. Binary features and binary interactions are included when their incidence rates fall in the range $5$--$95 \%$. All count features are first standardized,  then all pairwise interactions of these predictors are also standardized. 
The total number of predictors included in the analysis varies by data type as well as viral subtype, as given in Table \ref{tab:features_type_distribution}.
Detailed feature names are provided in Appendix Section \ref{sec:simulation_data_analysis_results}.

\begin{table}[ht]
\centering
\caption{Numbers of  binary and count variables for Subtypes B and C; proportion included among all possible are given in parentheses.}
\begin{tabular}[t]{p{3.5cm}cccccc} %  
\toprule
\multirow{2}{*}{\parbox{3.5cm}{}} &  &
\multicolumn{2}{c}{Binary effects} &  &
\multicolumn{2}{c}{Count effects} \\ 
\cmidrule{3-4} \cmidrule{6-7}
 & & {\centering Main} & {Interaction} &
 & {\centering Main} & {Interaction} \\
\midrule
\parbox{3.5cm}{Subtype B $(n=81)$}
 & & $220$ $(28 \%)$ & $11642$ $(48 \%)$ &  & $17$ $(94 \%)$ & 
     $136$ $(100 \%)$ \\
\midrule 
\parbox{3.5cm}{Subtype C $(n=293)$}
 & & $252$ $(32 \%)$ & $12698$ $(40 \%)$ &  & $17$ $(94 \%)$ & $136$ $(100 \%)$\\
\bottomrule
\end{tabular}
\label{tab:features_type_distribution}
\end{table}

As discussed at the end of Section \ref{sec:numerical}, in implementing the stabilized one-step estimator we recommend using the full sample to estimate the parameters of $P$, along with $R$ random orderings of the data. Histograms of p-values based on $1000$ random orderings of the data with respect to different virus subtypes are given in Figures \ref{fig:histogram_1000pvalues_subtypeB}-\ref{fig:histogram_1000pvalues_subtypeC} in the appendix. These histograms show the strong dependence of the p-value on the random ordering.
In Table \ref{tab:stratified_data_analysis_results}, we report the p-values of the stabilized one-step estimator and $95 \%$ confidence intervals using $q_n = n/2$ and  $R=200$ random orderings of the data (separated by virus subtype). 
%That is, we report the minimal p-value multiplied by $2 \times R$, representing a Bonferroni correction to the results from $2 \times R$ versions of the stabilized one-step estimator.
The reported confidence interval corresponds to the minimal p-value.
This table also presents the results of applying the Bonferroni one-step estimator to the original data. Though the Bonferroni one-step estimator also returns significant results, except in the case of count features and Subtype C,
this method generally yields more conservative conclusions than the stabilized one-step estimator, as expected.
%whereas the Bonferroni KSV method is less likely to find significantly associated predictor with $\mbox{IC}_{50}$, except  in the case of binary features and Subtype C.

The confidence intervals based on the stabilized one-step estimator in Table 2  represent changes in $\mbox{IC}_{50}$ (in units of ${\rm \mu g/mL}$)  due to the presence of the identified  binary feature, or due to a unit increase in the identified count feature. 
Genetic descriptions of these identified  features  are provided in Table \ref{tab:most_correlated_features} of Appendix Section \ref{sec:simulation_data_analysis_results}. 

\begin{table}[ht]
\centering 
\caption{Results of applying the Bonferroni one-step estimator and the stabilized one-step estimator to data on Subtypes B and C; the numbers of binary predictors and count predictors are denoted $p_{bin}$ and  $p_{count}$, respectively.}
\begin{tabular}[t]{p{2.8cm}ccc} %  
\toprule
\multirow{2}{*}{\parbox{2.8cm}{$(n, p_{bin}, p_{count})$}} & 
\multirow{2}{*}{Method} &
Binary effects & Count effects \\ 
\cmidrule{3-3} \cmidrule{4-4}
 &  & {\centering p-value ($95 \%$ CI)} & {p-value ($95 \%$ CI)}  \\
\midrule
\multirow{2}{*}{\parbox{2.8cm}{Subtype B \\ $(81, 11862, 153)$}}
 %& Bonferroni KSV      & $0.07$  &  $0.20$ \\
 & Bonferroni One-Step & $0.01$  &  $0.04$ \\ 
 & Stabilized One-Step & $<0.001$ $(7.8, 9.5)$ & $<0.001$ $(12.7, 13.5)$ \\
\midrule 
\multirow{2}{*}{\parbox{2.8cm}{Subtype C \\ $(293, 12950, 153)$}}
 %& Bonferroni KSV      & $0.01$   &  $0.24$  \\    
 & Bonferroni One-Step & $<0.001$ &  $0.91$   \\
 & Stabilized One-Step & $<0.001$ $(10.4, 23.1)$ & $<0.001$ $(3.3, 5.0)$ \\
\bottomrule
\end{tabular}
\label{tab:stratified_data_analysis_results}
\end{table}

%\begin{figure}
%\begin{center}
% \includegraphics[scale=1, width=1\textwidth, height=1\textwidth]{Bonferroni_pvalues_by_number_randomorderings.pdf}
% \caption{Bonferroni corrected minimal p-values obtained from various numbers of random orderings of data, where each point presents the minimal p-value multiplied by the number of random ordering.}
% \label{fig:pvalue_patterns}
%\end{center}
%\end{figure}

\section{Discussion} \label{sec:discussion}

Though we have focused on using the correlation as the marginal association measure, our results can be extended to a wide range of other measures. The key requirements on such a measure are that (a) it be pathwise differentiable in the full data model where censoring is not present, and (b) the efficient influence function be non-degenerate in this full data setting. Requirement (a) is sufficient for the association measure to be pathwise differentiable even in the presence of right censoring, thereby allowing the construction of a one-step estimator for each marginal association. Requirement (b) enables the use of first-order asymptotics to study the behavior of these one-step estimators; without this condition, the centered estimator, scaled by the square root of sample size, would converge weakly to zero rather than to a non-degenerate mean-zero normal distributed random variable. Requirement (b) also ensures that the variance estimates used for standardization in the stabilized one-step estimator do not converge to zero, a condition that is required by existing theory for this estimator. Examples of parameters satisfying (a) and (b) include the Spearman correlation, odds ratio, and model-agnostic hazard ratio \citep{whitney2019comment}. Examples of parameters satisfying (a), but not (b), include the nonparametric $R^2$, distance correlation \citep{szekely2007measuring}, and maximum mean discrepancy \citep{smola2006maximum}.

%In future work, it would be interesting to develop theory for the case that number of covariates $p$ grows with $n$. \citep{Luedtke2018} showed that, in the absence of censoring, the stabilized one-step estimator is asymptotically normal even in the case that $\log p = o(n^{1/2})$. Though establishing this kind of result is considerably more challenging in the presence of right censoring, our simulations suggest that establishing a result for the case that $p\gg n$ in this setting is likely possible. Indeed, our simulations indicate both control of the type 1 error when $n=100$ at a range of covariate dimensions, ranging from the moderate-$p$ setting where $p=100$ to the large-$p$ setting where $p=10^6$. Moreover, the power of the test is essentially constant across all $p$ considered.

\newpage
\section*{Appendix}
\begin{appendix}
\section{Proof of Theorem \ref{Thm:if_slope}}
\label{app-sec:proof_Thm_if_slope}
%\begin{proof}
This theorem implies $\Psi_{\hat{G}_n}(\Pn)$ is asymptotically linear with its influence function $\IF^{ipw}_G(\cdot|P)-\IF^{nu}_G(\cdot|P)$, where $\IF^{ipw}_G(\cdot|P)$ is the influence function of an inverse probability weighted estimator in the setting where $G$ is known, and $\IF^{nu}_G(\cdot|P)$ accounts for the fact that, in truth, $G$ was estimated by the Kaplan--Meier estimator $\hat{G}_n$. Below we follow Lemmas \ref{lemma:remainder_estG}-\ref{lemma:if_estG} to validate the substitution of $G$ for $\hat{G}_n$ in the asymptotically linear representation of $\Psi_{\hat{G}_n}(\Pn)-\Psi(P)$, while Lemma \ref{lemma:projection_of_if_ipw} further indicates that $\IF^{nu}_G(\cdot|P)$ is the projection of $\IF^{ipw}_G(\cdot|P)$ onto the nuisance tangent space $\mathbf{T}^{nu}(G)$.

To obtain the asymptotic linear representation of $\Psi_{\hat{G}_n}(\Pn)-\Psi(P)$, we first have that
\begin{equation}\label{eq:expression_Psi_hatG}
\begin{split}
  &\Psi_{\hat{G}_n}(\Pn)-\Psi(P)=[\Psi_{\hat{G}_n}(\Pn)-\Psi_{\hat{G}_n}(P)]+[\Psi_{\hat{G}_n}(P)-\Psi(P)]\\
  &\quad =\Pn\big[\IF^{ipw}_{\hat{G}_n}(O|P)\big]+[\Psi_{\hat{G}_n}(P)-\Psi(P)]-{\rm Rem}_{\hat{G}_n}(\Pn, P)\\
  &\quad =\Pn\big[\IF^{ipw}_G(O|P)\big]+[\Psi_{\hat{G}_n}(P)-\Psi(P)]-{\rm Rem}_{\hat{G}_n}(\Pn, P)+o_p(n^{-1/2}),  
\end{split}
\end{equation}
where the second equality follows by
replacing $G$ by $\hat{G}_n$ in \eqref{eq:if_G_fixed}, leading to
\begin{align*}
  &\Psi_{\hat{G}_n}(\Pn)-\Psi_{\hat{G}_n}(P)\\
  &\quad = \Pn \left[\left\{ \frac{(U-P[U])(Y-P[T])}{{\rm Var}(U)}
  -\frac{{\rm Cov}(U,T)}{{\rm Var}^2(U)}(U-P[U])^2 \right\} \right]-{\rm Rem}_{\hat{G}_n}(\Pn, P)\\
  &\quad \equiv \Pn\big[\IF^{ipw}_{\hat{G}_n}(O|P)\big]-{\rm Rem}_{\hat{G}_n}(\Pn, P),
\end{align*}
and the last equality in \eqref{eq:expression_Psi_hatG} holds by the asymptotic equivalence between $\Pn[\IF^{ipw}_{\hat{G}_n}(O|P)]$ and $\Pn[\IF^{ipw}_G(O|P)]$ that
is a consequence of Lemma 19.24 in \cite{Vaart1998} and stated in Lemma \ref{lemma:estG_by_G_for_if_ipw}.
Moreover, ${\rm Rem}_{\hat{G}_n}(\Pn, P)$ is shown asymptotically negligible in Lemma \ref{lemma:remainder_estG} in the appendix.

In addition, a first-order Taylor expansion of the middle term on the right-hand side of \eqref{eq:expression_Psi_hatG} around $G$ yields
the approximation
\begin{equation} \label{eq:expression_Psi_hatG.1}
\begin{split}
  &\Psi_{\hat{G}_n}(P)-\Psi(P)
  =\frac{-1}{{\rm Var}(U)}P\left[(U-P[U])\frac{\delta X}{G^2(X)}[\hat{G}_n(X)-G(X)]\right]+o_p(n^{-1/2})\\
  &\quad =\frac{-1}{{\rm Var}(U)}P\left[(U-P[U])\tilde{Y}\left[\frac{\hat{G}_n(X)}{G(X)}-1\right]\right]+o_p(n^{-1/2}),                                                                 
\end{split}
\end{equation}
which can be further expressed in terms of the influence function of $\hat{G}_n$. %(denoted by $\IF^{ipw}_{G}$).
Lemma \ref{lemma:if_estG} shows that the influence function of $\hat{G}_n(t)$ is $-G(t)\IF_{G}(O)(t)$, where
\begin{align}
  \IF_{G}(O)(t) = \int_{-\infty}^t\left[\frac{1(X \in ds, \delta=0)}{P1(X \geq s)}-\frac{1(X \geq s)d\Lambda(s)}{{P1(X \geq s)}}\right]
  = \int_{-\infty}^{t} \frac{dM(s)}{P1(X \geq s)}                                                              \label{if_G}
\end{align}
for $ t\in \mathcal{T}$; $\Lambda$ is the cumulative hazard function of $G$, and
$dM(s) \equiv 1(X \in ds, \delta=0)-1(X \geq s)d\Lambda(s)$ is a martingale with respect to the filtration generated by the counting process
$1(X \leq s, \delta=0)$. This suggests that $\hat{G}_n(t)-G(t) = -G(t)[n^{-1}\sum_{i=1}^n \IF_{G}(O_i)(t)] + o_p(n^{-1/2})$, which further implies that
\[
  \frac{\hat{G}_n(t)}{G(t)}-1 = \frac{-1}{n} \sum_{i=1}^n \int_{-\infty}^{t} \frac{dM_i(s)}{P1(X \geq s)} = -\Pn\IF_G(O)(t)
\]
with $dM_i$ is the replication of $dM$ based on the observation of the $i$-th subject.
Let $\mathbb{L}_n \equiv \Pn\IF_{G}(O)$ and define $\varphi \colon \ell^{\infty}_{\tau} \rightarrow \mathbb{R}$
by $\varphi(g)=P[(U-P[U])\tilde{Y}g(X)]$,
where $ \ell^{\infty}_{\tau}$ is the space of uniformly bounded functions on $\mathcal{T}$.
Then \eqref{eq:expression_Psi_hatG.1} can be expressed in terms of $\mathbb{L}_n$ and $\varphi$ as follows:
\begin{align*}
  \Psi_{\hat{G}_n}(P)-\Psi(P)=\frac{-1}{{\rm Var}(U)}\varphi\left(\mathbb{L}_n\right)+o_p(n^{-1/2})
  =\Pn\left[\frac{-1}{{\rm Var}(U)}\varphi(\IF_{G}(O))\right]+o_p(n^{-1/2}),
\end{align*}
where the second equality holds by the linearity of $\varphi(g)$ in $g$.

Evaluating $\varphi$ at $\IF_{G}(O)$ in \eqref{if_G}, this leads to
\begin{align*}
  &\varphi\left(\IF_{G}(O)\right)=P\left[(U-P[U])\tilde{Y}\int_{-\infty}^{X} \frac{dM(s)}{P1(X \geq s)}\right]\\
  %&\quad = P\left[(U-P[U])P\left[\tilde{Y}\int_{\mathcal{T}} \frac{dM(s)}{P1(X \geq s)}|U\right]\right]\\
  &\quad =P\left[(U-P[U])\int_{\mathcal{T}} P\left[\tilde{Y}|U,X \geq s\right]P1(X \geq s)\frac{dM(s)}{P1(X \geq s)}\right],
\end{align*}
in which
\begin{align*}
  &P\left[\tilde{Y}|U,X \geq s\right]
  =\int_{\mathcal{T}} \frac{x}{G(x)}P1(\delta=1,X \in dx|U, X \geq s)\\
  &\quad =\int_{\mathcal{T}} \frac{x1(x \geq s)}{G(x)}\frac{P1(U,C \geq x,T \in dx)}{P1(U,C \geq s, T \geq s)}
  =\int_{\mathcal{T}} \frac{x1(x \geq s)}{G(x)}\frac{G(x)P1(U,C \geq s,T \in dx)}{G(s)P1(U,C \geq s, T \geq s)}\\
  &\quad =\int_{\mathcal{T}} \frac{x}{G(s)}\frac{P1(U,X \geq s,T \in dx)}{P1(U,X \geq s)}=\frac{E[T|U,X \geq s]}{G(s)}.
\end{align*}
Therefore we have that
\begin{align*}
  \varphi\left(\IF_{G}(O)\right)
  =P\left[(U-P[U])\int_{\mathcal{T}} E[T|U,X \geq s]\frac{dM(s)}{G(s)}\right],
\end{align*}
where the expectation with respect to $P$ only applies to $(U,X)$, not to $M$, and it implies that
the expression in \eqref{eq:expression_Psi_hatG} becomes
\begin{align} \label{eq:expression_Psi_hatG.2}
  &\Psi_{\hat{G}_n}(\Pn)-\Psi(P)
  =\Pn\big[\IF^{ipw}_G(O|P)\big] \\
  & -\Pn\left[\frac{1}{{\rm Var}(U)} P\left[(U-P[U])\int_{\mathcal{T}} E[T|U,X \geq s]\frac{dM(s)}{G(s)}\right]\right]
  -{\rm Rem}_{\hat{G}_n}(\Pn, P) + o_p(n^{-1/2}) \nonumber \\
  & \equiv \Pn \big[\IF^{ipw}_G(O|P)\big]-\Pn\big[\IF^{nu}_G(O|P)\big]
  -{\rm Rem}_{\hat{G}_n}(\Pn, P)+o_p(n^{-1/2}). \nonumber  
\end{align}
%\end{proof}

\section{Proof of Theorem \ref{Thm:if_one_step}}
\label{app-sec:proof_Thm_if_one_step}
%\begin{proof}
In this theorem, our objective is to show that, under regularity conditions, $S(\mathbb{P}_n,\hat{P}_n)$ is asymptotically linear with influence function
\vspace{-5pt}
\begin{align*}
  \Pi\left\{\IF^*(\cdot|\bar{E},Q_u,G)|\mathbf{T}^*(G)^\perp\right\}.
\end{align*}
Let $\hat{P}_n' = (\hat{E}_n, \mathbb{Q}_{n}, G)$ denote the estimate of $P$ but with $\hat{G}_n$ replaced by the true censoring distribution $G$. Note also that $S(\mathbb{P}_n, \hat{P}'_n)=\Psi(\hat{P}_n) + \mathbb{P}_n \IF^*(\cdot|\hat{E}_n, \mathbb{Q}_{n}, G)$, where Equations \eqref{eq:def_parameter}-\eqref{eq:def_parameter.2} imply that $\Psi$ does not depend on the censoring distribution so that $\Psi(\hat{P}_n)=\Psi(\hat{P}_n')$. We have that
\begin{align} \label{eq:Sn_decomp}
  &S(\mathbb{P}_n, \hat{P}_n) - \Psi(P) = S(\mathbb{P}_n, \hat{P}'_n) - \Psi(P) + S(\mathbb{P}_n, \hat{P}_n) - S(\mathbb{P}_n, \hat{P}'_n) \nonumber \\
  &\quad = S(\mathbb{P}_n, \hat{P}_n') - \Psi(P) + \mathbb{P}_n [\IF^*(\cdot|\hat{E}_n,\mathbb{Q}_{n},\hat{G}_n)-\IF^*(\cdot|\hat{E}_n,\mathbb{Q}_{n},G)] \nonumber \\
  &\quad = S(\mathbb{P}_n, \hat{P}'_n) - \Psi(P) + P[\IF^*(\cdot|\hat{E}_n,\mathbb{Q}_{n},\hat{G}_n)-\IF^*(\cdot|\hat{E}_n,\mathbb{Q}_{n},G)] \nonumber \\
  &\qquad + [\mathbb{P}_n-P] [\IF^*(\cdot|\hat{E}_n,\mathbb{Q}_{n},\hat{G}_n)-\IF^*(\cdot|\hat{E}_n,\mathbb{Q}_{n},G)].
\end{align}
The last term on the right-hand side of \eqref{eq:Sn_decomp} is $o_p(n^{-1/2})$  
by \ref{eq:process_conv.2} of Lemma \ref{lemma:convergence_processes}, with the required conditions verified in Lemmas \ref{lemma:conditions_for_Thm2.1}-- \ref{lemma:L2_conv.2}. Moreover, Lemma \ref{lemma:remainder} shows that 
\[
  P[\IF^*(\cdot|\hat{E}_n,\mathbb{Q}_{n},\hat{G}_n)- \IF^*(\cdot|\bar{E},Q_u,\hat{G}_n)  + \IF^*(\cdot|\bar{E},Q_u,G)-\IF^*(\cdot|\hat{E}_n,\mathbb{Q}_{n},G)] = o_p(n^{-1/2}),
\]
which would simplify the middle term in \eqref{eq:Sn_decomp} as follows. 
At last, Lemma \ref{lemma:if_SnG} gives the asymptotic linearity of the term $S(\mathbb{P}_n, \hat{P}'_n)-\Psi(P)$ in
\eqref{eq:Sn_decomp}.

To simplify the middle term on the right-hand side in \eqref{eq:Sn_decomp}, we further decompose it as
\begin{align*}
  P&[\IF^*(\cdot|\hat{E}_n,\mathbb{Q}_{n},\hat{G}_n)-\IF^*(\cdot|\hat{E}_n,\mathbb{Q}_{n},G)] \\
  &= P[\IF^*(\cdot|\bar{E},Q_u,\hat{G}_n)-\IF^*(\cdot|\bar{E},Q_u,G)] \\
  &\hspace{1em} + P[\IF^*(\cdot|\hat{E}_n,\mathbb{Q}_{n},\hat{G}_n)- \IF^*(\cdot|\bar{E},Q_u,\hat{G}_n)  + \IF^*(\cdot|\bar{E},Q_u,G)-\IF^*(\cdot|\hat{E}_n,\mathbb{Q}_{n},G)],
\end{align*}
with the last line of the above display shown as $o_p(n^{-1/2})$ by Lemma \ref{lemma:remainder} as mentioned earlier. 
Expressing the first term of the above display using the notation $\Phi(\tilde G)= P[\IF^*(\cdot|\bar{E},Q_u,\tilde G)]$, for $\tilde G$ equal to $\hat{G}_n$ or $G$, upon inserting them back into \eqref{eq:Sn_decomp} we have that
\begin{align*}
  S(\mathbb{P}_n, \hat{P}_n) - \Psi(P) = S(\mathbb{P}_n, \hat{P}'_n)- \Psi(P) + \Phi(\hat{G}_n) - \Phi(G)+ o_p(n^{-1/2}),        
\end{align*}
together with the previously-developed results. 

In Lemma \ref{lemma:if_SnG}, $S(\mathbb{P}_n, \hat{P}'_n)$ is shown to be a regular asymptotically linear estimator of $\Psi(P)$ with influence function ${\rm IF}^* + \IF^{\dagger}$ in the model ${\cal M}(G)$ with $G$ known.
Further, because we specified that  $\hat{G}_n$ is estimated via maximum likelihood, the delta method can be used to show that $\Phi(\hat{G}_n)$ is an asymptotically efficient estimator of $\Phi(G)$ in the model used for $G$ (with tangent space given by $\mathbf{T}^*(G)$). %with $G$ known (and tangent space $\mathbf{T}^*(G)$). 
Combining the above results, we have verified all the required conditions of Theorem~2.3 in \cite{vanderLaan&Robins2003}, from which we conclude that
\begin{align*}
  S(\mathbb{P}_n, \hat{P}_n) - \Psi(P)&=
  [\mathbb{P}_n-P]\Pi\left\{\IF^*(\cdot|\bar{E},Q_u,G) + \IF^{\dagger}(\cdot|\bar{E},P)|\mathbf{T}^*(G)^\perp\right\} + o_p(n^{-1/2}),
\end{align*}
and the proof is complete.
%\end{proof}

\section{Proof of Theorem \ref{Thm:stab_one_step}}
\label{app-sec:proof_Thm_stab_one_step}
 Before proceeding to the proof we make the following asymptotic stability assumptions:
\begin{assumpenum}[resume=assumptions]
  \item \label{assump:Conditional_mean_E0}
  $\hat{E}_n$ defined in \eqref{eq:est_E} 
  for a given predictor $U_k$ consistently estimates (pointwise in its arguments) the true conditional mean residual life function $E_0(u,s,k) \equiv E[\tilde{Y}|U_k=u, X \ge s]$, which is assumed to be uniformly-bounded and left-continuous in $s$, and
   with $k_j$ defined in \eqref{eq:predictor_selection},
  \begin{align*}
   E\bigg[\sup_{(j,s) \in \{q_n,\ldots,n\}
  \times {\cal T}}\big|E_0(U_{k_{j}},s,k_{j})-E_0(U_{k_{j-1}},s,k_{j-1})\big|\bigg] = o(n^{-1/2}).
  \end{align*}
  %   The following are of order $o(n^{-1/2})$ uniformly in $(j,s) \in \{q_n,\ldots,n\}$
  %   $\times {\cal T}$:
  %   \begin{align*}
  %   &E\{[{\rm Cov}(U_{k_j}1(X \ge s),\tilde{Y}1(X \ge s))-{\rm Cov}(U_{k_{j-1}}1(X \ge s),\tilde{Y}1(X \ge s))]^2\},\\
  %   &E\{[{\rm Var}(U_{k_j}1(X \ge s))-{\rm Var}(U_{k_{j-1}}1(X \ge s))]^2\},\\
  %   &E\{[U_{k_j}1(X \ge s)-P[U_{k_j}1(X \ge s)]-\{U_{k_{j-1}}1(X \ge s)-P[U_{k_{j-1}}1(X \ge s)]\}]^2\}.
  % \end{align*}
  
  \item \label{assump:Signal strength} There exist a sufficiently large $c > 0$ and a sequence of non-empty subsets $\mathcal{K}_n^* \subseteq \mathcal{K}_n=\{1,\ldots,p_n\}$ such that
  \begin{align*}
    \inf_{k \in \mathcal{K}_n^*}\bigg|\frac{{\rm Cov}(U_{k}, T)}{{\rm Var}(U_{k})}\bigg| - \sup_{l \in \mathcal{K}_n \setminus \mathcal{K}_n^*}\bigg|\frac{{\rm Cov}(U_{l}, T)}{{\rm Var}(U_{l})}\bigg| \ge c\sqrt{\frac{\log(n \lor p_n)}{q_n}},
  \end{align*}
  where the supremum over $l \in \mathcal{K}_n \setminus \mathcal{K}_n^*$ is  defined to be 0 if $\mathcal{K}_n^*=\mathcal{K}_n$, and
  \[
  {\rm Diam}(\mathcal{K}^*_n) \equiv \sup_{k, l \in \mathcal{K}^*_n}\left|\,\left|\frac{{\rm Cov}(U_{k},T)}{{\rm Var}(U_{k})}\right|-\left|\frac{{\rm Cov}(U_{l},T)}{{\rm Var}(U_{l})}\right|\,\right| = o(n^{-1/2}).
  \]
\end{assumpenum}
The last condition allows the most correlated predictors with $T$ to be non-unique under the alternative, but requires that they are contained in some subset $\mathcal{K}^*_n$ with stronger association with $T$ than the other predictors in $\mathcal{K}_n$. In addition, the variation in signal strength of the predictors in $\mathcal{K}^*_n$ is assumed to be of order $o(n^{-1/2})$.

The proof below is developed in a special case of taking a user-defined coarsening $c$ so that $c(U)$ is a degenerate random variable, which reasonably reduces $\hat{G}_n(\cdot|u)$ to $\hat{G}_n(\cdot)$, and this is supported by the independent censoring assumption.
% In addition, we assume that $\bar{E}(u,s) = E_0(u,s)$ for any $(u,s)$, where $\bar{E}$ is as defined in \ref{assump:Conditional mean} and 
% $E_0(u,s)$ denotes the true conditional mean residual life function $E[\tilde Y|U=u, X \geq s]$ throughout this section for notational simplification. This implies that the estimator $\hat{E}_n$ consistently estimates the true conditional mean residual life function (pointwise in its arguments).

\begin{proof}
We will show the asymptotic normality of $\sqrt{n-q_n}\bar{\sigma}_n^{-1}[S^*_n-\Psi(P)]$, where $\Psi(P)$ is defined in \eqref{eq:def_parameter}.
From the expression $S^*_n$ in \eqref{eq:Sn_star},
the desired result will follow from the limiting distribution of
\[
  \frac{1}{\sqrt{n-q_n}}\sum_{j=q_n}^{n-1}\hat{\sigma}_{nj}^{-1}m_j[S_{k_j}(\delta_{O_{j+1}}, P)-\Psi_{k_j}(P)],
\]
with  certain remainder terms shown to be asymptotically negligible in Lemmas \ref{lemma:asymptotic_negligibility_(I)}--\ref{lemma:asymptotic_negligibility_(V)}, based on concentration results and supportive preliminaries developed in Lemmas
\ref{lemma:preservation_BV}--\ref{lemma:convergence_sum_Dnj.2}

We start by introducing a decomposition of the stabilized one-step estimator.
The distribution of $P$ is identified by $(E_0, Q_u, G)$, where $E_0(u,s,k) \equiv E[\tilde Y|U_k=u, X \geq s]$.
Replacing in various ways each feature of $P$ by its estimator introduced in Section \ref{sec:one_step_estimator} gives $\hat{P}_{nj}=(\hat{E}_j, \mathbb{Q}_j, \hat{G}_n)$;
$\hat{P}^{''}_{nj}=(E_0, \mathbb{Q}_{j}, \hat{G}_n)$ and 
$\hat{P}^{'''}_{nj}=(E_0, \mathbb{Q}_{j}, G)$.
Therefore we are able to decompose the statistic of interest as
\begin{align} \label{eq:decomposition_rootn_Sn_star}
  &\sqrt{n-q_n}\bar{\sigma}_n^{-1}[S^*_n - \Psi(P)] \;=\; \frac{1}{\sqrt{n-q_n}}\sum_{j=q_n}^{n-1}\hat{\sigma}_{nj}^{-1}m_j[S_{k_j}(\delta_{O_{j+1}}, \hat{P}_{nj})-S_{k_j}(\delta_{O_{j+1}}, \hat{P}^{''}_{nj})] \\
  &\hspace{10pt} + \frac{1}{\sqrt{n-q_n}}\sum_{j=q_n}^{n-1}\hat{\sigma}_{nj}^{-1}m_j[S_{k_j}(\delta_{O_{j+1}}, \hat{P}^{''}_{nj})-S_{k_j}(\delta_{O_{j+1}}, \hat{P}^{'''}_{nj})] \nonumber\\
  &\hspace{10pt} + \frac{1}{\sqrt{n-q_n}}\sum_{j=q_n}^{n-1}\hat{\sigma}_{nj}^{-1}m_j[S_{k_j}(\delta_{O_{j+1}}, \hat{P}^{'''}_{nj})-S_{k_j}(\delta_{O_{j+1}}, P)] \nonumber\\
  &\hspace{10pt} + \frac{1}{\sqrt{n-q_n}}\sum_{j=q_n}^{n-1}\hat{\sigma}_{nj}^{-1}m_j[S_{k_j}(\delta_{O_{j+1}}, P)-\Psi_{k_j}(P)] \nonumber\\
  &\hspace{10pt} + \frac{1}{\sqrt{n-q_n}}\sum_{j=q_n}^{n-1}\hat{\sigma}_{nj}^{-1}m_j[\Psi_{k_j}(P)-\Psi(P)]
  \equiv \mbox{(I) + (II) + (III) + (IV) + (V)}. \nonumber
\end{align}

Using Lemmas \ref{lemma:asymptotic_negligibility_(I)}--\ref{lemma:asymptotic_negligibility_(V)} mentioned above, in conjunction with \ref{assump:Conditional_mean_E0}--\ref{assump:Signal strength}, gives the asymptotic negligibility of (I), (II), (III) and (V). Therefore the remaining task is to show the asymptotic normality of (IV). Modifying the expression in \eqref{eq:if_onestep}, with $(\mathbb{P}_n,\hat{P}_n)$ replaced by $(\delta_{O_{j+1}}, P)$ and with the predictor taken as $U_{k_j}$, gives that $S_{k_j}(\delta_{O_{j+1}}, P) = \Psi_{k_j}(P) + \IF^*_{k_j}(O_{j+1}|P)$. This further implies that
\begin{align*}
  \mbox{(IV)} = \frac{1}{\sqrt{n-q_n}}\sum_{j=q_n}^{n-1}\hat{\sigma}_{nj}^{-1}m_j\IF^*_{k_j}(O_{j+1}|P),    
\end{align*}
where $\IF_{k_j}^*(\cdot|P)$ is $\IF^*(\cdot|P)$ with the predictor taken as $U_{k_j}$. Decompose (IV) into two terms:
\begin{equation}
\begin{split}
  \frac{1}{\sqrt{n-q_n}}\sum_{j=q_n}^{n-1}\left[\frac{\sigma_{nj}}{\hat{\sigma}_{nj}}-1\right]\frac{m_j}{\sigma_{nj}}\IF^*_{k_j}(O_{j+1}|P)
  + \frac{1}{\sqrt{n-q_n}}\sum_{j=q_n}^{n-1}\frac{m_j}{\sigma_{nj}}\IF^*_{k_j}(O_{j+1}|P).
                             \label{IV_decomposition}
\end{split}
\end{equation}
Note that $\hat{\sigma}_{nj}$ in the first term above involves the partial-sample estimator $\hat{P}_{nj}$.

Let $\lesssim$ denote ``bounded above up to a universal multiplicative constant that does not depend on $(j,n)$,'' and  $\sigma^2_{nj} \equiv \int \IF^*_{k_j}(o|P)^2 dP(o) = {\rm Var}(\IF^*_{k_j}(O|P))$.
Note that $\sigma^2_{nj}$ is bounded away from zero by \ref{assump:Variances}, which implies that $\min_{k \in \mathbb{N}} {\rm Var}(\IF^*_{k}(O|P))$ is bounded away from zero; namely, there exists some constant $\epsilon > 0$ so that 
$\sigma^2_{nj} \geq \min_{k \in \mathbb{N}} {\rm Var}(\IF^*_{k}(O|P)) \geq \epsilon$. 
The  first term in the decomposition of (IV) in \eqref{IV_decomposition} is seen to be of order $o_p(1)$ as follows.
Let ${\cal O}_{nj} \equiv \sigma \big(\{O_1,\ldots,O_j\},\hat{G}_n\big)$, and 
\begin{align*}
 H_{nj} \equiv \frac{1}{\sqrt{n-q_n}}\left[\frac{\sigma_{n,\,j-1}}{\hat{\sigma}_{n,\,j-1}}-1\right]\frac{m_{j-1}}{\sigma_{n,\,j-1}}\IF^*_{k_{j-1}}(O_{j}|P);
\end{align*}
the first term in the decomposition of (IV) in \eqref{IV_decomposition} is equal to $\sum_{j=q_n}^{n-1}H_{n,\,j+1}$.
Note that $E|H_{nj}| < \infty$, using the fact that $\sigma^2_{nj}$ is bounded away from zero by \ref{assump:Variances} and so is $\hat{\sigma}^2_{nj}$ by \ref{eq:event_En_1} of Lemma \ref{lemma:An_Bn_Cn_Dn_En_prob_to_one},  
and also that $H_{nj}$ is ${\cal O}_{nj}$-measurable. Along with 
\begin{align*}
  E[H_{n,\,j+1}|{\cal O}_{nj}] = \frac{1}{\sqrt{n-q_n}}E\Big[\Big(\frac{\sigma_{nj}}{\hat{\sigma}_{nj}}-1\Big)\frac{m_j}{\sigma_{nj}}E\big[\IF^*_{k_j}(O_{j+1}|P)\big|{\cal O}_{nj}\big]\Big]=0,
\end{align*}
$\{(H_{nj}, {\cal O}_{nj}), j=q_n+1,\ldots,n\}$ is a martingale difference sequence. Moreover, we have that
\begin{align}
  &|H_{n,\,j+1}| \lesssim \sqrt{\frac{\log(n \lor p_n)}{q_n(n-q_n)}} \equiv B_n. \label{eq:Bn}
\end{align}
Then Chebyshev's inequality implies that for $\varepsilon > 0$,
\begin{align*}
  & {\rm P}\bigg(\,\bigg|\sum_{j=q_n}^{n-1}H_{n,\,j+1}\bigg| \ge \varepsilon \bigg) \le \varepsilon^{-2}E\bigg[\bigg(\sum_{j=q_n}^{n-1}H_{n,\,j+1}\bigg)^2\,\bigg] \\
  %&= \varepsilon^{-2}\bigg(\sum_{j=q_n}^{n-1}E\big[H_{n,\,j+1}^2\big] + 2\sum_{q_n\le i< j\le n-1}E\big[H_{n,\,j+1}H_{n,\,i+1}\big]\bigg) \\
  &= \varepsilon^{-2}\bigg(\sum_{j=q_n}^{n-1}E\big[H_{n,\,j+1}^2\big] + 2\sum_{q_n\le i< j\le n-1}E\Big[H_{n,\,i+1}E\Big[H_{n,\,j+1} \Big| \mathcal{O}_{nj}\Big]\Big]\bigg) \\
  &= \varepsilon^{-2}\sum_{j=q_n}^{n-1}E\big[ H_{n,\,j+1}^2 \big] \le \varepsilon^{-2}(n-q_n)B_n^2 \to 0;
\end{align*}
this result disposes of the first term in the decomposition of (IV) in \eqref{IV_decomposition}.

Observe that the second term in \eqref{IV_decomposition} is a sum of martingale differences because $$E[m_j\IF^*_{k_j}(O_{j+1}|P)|O_{1},\ldots,O_{j}]=0.$$ Therefore it converges in distribution to standard normal by the martingale central limit theorem for triangular arrays \cite[e.g., Theorem 2 in][]{Gaenssler1978}, under the following conditions 
\begin{align*}
  &\frac{1}{n-q_n}\sum_{j=q_n}^{n-1}E\bigg[\frac{[\IF^*_{k_j}(O_{j+1}|P)]^2}{\sigma^2_{nj}} \bigg| O_1,\ldots,O_j \bigg] \lcrarrow{p} 1;\\
  &\frac{1}{n-q_n}\sum_{j=q_n}^{n-1}E\bigg[\frac{[\IF^*_{k_j}(O_{j+1}|P)]^2}{\sigma^2_{nj}}1\bigg(\,\bigg|\frac{\IF^*_{k_j}(O_{j+1}|P)}{\sigma_{nj}}\bigg| > \epsilon_0\sqrt{n-q_n} \bigg) \bigg| O_1,\ldots,O_j \bigg] \lcrarrow{p} 0
\end{align*}
for every $\epsilon_0 > 0$. The first condition follows from the definition of $\sigma^2_{nj}$, which implies that each term in the summation is identically equal to 1. The second condition holds because $\IF^*_k(\cdot|P)$ is uniformly bounded over $k$ in view of \eqref{eq:if_star} and \eqref{eq:reexpressed_if} (giving the expression for $\IF^*_k$)
and \ref{assump:Covariates}--\ref{assump:At-risk prob} and \ref{assump:Variances}--\ref{assump:Conditional_mean_E0}; further,  $\sigma_{nj}^2$
is assumed to be uniformly bounded away from zero by \ref{assump:Variances}.
\end{proof}

\section{Notation} \label{sec:notation}
For convenient reference, below we collect the notation that appears in the main text and will be frequently used in the upcoming proofs.
\begin{itemize}\setlength\itemsep{1em}
  \item $M$ is the martingale part of the single-jump counting process for a censored observation, that is, $dM(s) = 1(X \in ds, \delta=0)-1(X \geq s)d\Lambda(s)$ with respect to the filtration
  \begin{align} \label{eq:filtration_generic}
    \mathcal{F}_{s} \equiv \sigma(\{1(X \leq s', \delta=0), 1(X \geq s'),\, U:\, s' \le s \in \mathcal{T}\}),
  \end{align}
  where $\Lambda$ is the cumulative hazard function corresponding to $G$.
  
  \item Based on a fixed and user-defined coarsening $c$ so that $c(U)$, $U \sim P$, is a finitely supported discrete random variable, $$M(u,ds) = 1(X \in ds, \delta=0, c(U)=c(u))-1(X \geq s, c(U)=c(u))d\Lambda(s)$$ is a martingale with respect to the filtration $\mathcal{F}_{s}$.
  
  \item Let $N_n(u,s)=\sum_{i=1}^n 1(X_i \leq s, \delta_i=0, c(U_i)=c(u))$ and $Y_n(u,s)=\sum_{i=1}^n 1(X_i \geq s, c(U_i)=c(u))$. The conditional cumulative hazard function is estimated by
  $$\hat{\Lambda}_n(\cdot|u) = \int_{-\infty}^{\cdot} \frac{1(Y_n(u,s)>0)}{Y_n(u,s)}N_n(u,ds).$$ With $\Lambda(s)$ estimated by
  $\hat{\Lambda}_n(s|u)$,
  $\hat{M}(u,ds)=1(X \in ds, \delta=0, c(U)=c(u))-1(X \geq s, c(U)=c(u))d\hat{\Lambda}_n(s|u)$.
  
  \item $\bar{M}$ is the martingale part of the counting process for censored observations with the predictor $U$ that is coarsened at $c(U)=c(u)$; namely,
  $\bar{M}(u,ds) = N_n(u,ds) - Y_n(u,s) d\Lambda(s)$ is a local square integrable martingale with respect to the aggregated filtration
  \begin{align} \label{eq:agg_filtration}
    \bar{\mathcal{F}}_{s} \equiv \sigma(\{N_n(\cdot,s'),\, Y_n(\cdot,s'),\, U_i\,:\, i=1,\ldots,n,\, s' \le s \in \mathcal{T}\}).
 \end{align}
 
  %\item Suppose that $M'(s)$ is a local square integrable martingale with respect to the filtration ${\cal F}'_{s}$.
  %Following the notation of martingale theory in \cite{Andersen1993}, $\langle M' \rangle$ denotes the predictable variation process of $M'$, and $d\langle M' \rangle(s) = E[d(M'(s))^2|{\cal F}'_{s}] = E[(dM'(s))^2|{\cal F}'_{s}] = {\rm Var}(dM'(s)|{\cal F}'_{s})$.
  
  \item Starting from Section \ref{sec:one_step_estimator} in the main text and Section \ref{sec:proof_lemmas_of_Thm_if_one_step} in this supplementary, the distribution $P$ is characterized by three features, namely $P = (E_0, Q_u, G)$. For any given $k$, let $U \equiv U_k$.
  Based on observations $\{O_1,\ldots,O_j\}$, $\hat{E}_j(u,s)$ is an estimator of $E_0(u,s) \equiv E[\tilde{Y}|U=u, X\ge s]$, and $\mathbb{Q}_{j}$ is the empirical distribution of any given single predictor whose distribution is $Q_u$, for $j=q_n,\ldots,n$. Therefore $\hat{P}_{nj} \equiv (\hat{E}_j, \Qj, \hat{G}_n)$, and $\hat{P}'_j \equiv (\hat{E}_j, \Qj, G)$, where $\hat{G}_{n}(\cdot|u)$ is the Kaplan--Meier estimator of the censoring distribution $G(\cdot)$ conditional on $c(U)=c(u)$.
  Note that $\hat{E}_j(u) = \hat{E}_j(u, \infty)$ to estimate $E[\tilde{Y}|U=u]$, and $\hat{G}_{n}(\cdot|u)$ reduces to a standard Kaplan--Meier estimator $\hat{G}_{n}(\cdot)$ when $c(U)$ is a degenerate random variable.
  
  \item For a function $h$ mapping from a realization of $O$ to $\mathbb{R}^d$, we use $P[h] \equiv P[h(O)] \equiv \int h(o)dP(o)$. The expectation $\mathbb{E}$ applies to a function of $O$ and $O_1, \ldots, O_n$, regarding $O$ as fixed, in contrast to the expectation under $P$ that applies only to $O \sim P$ and not to any estimator based on $O_1, \ldots, O_n$. In addition, $E$ denotes the expectation of both $O$ and $O_1, \ldots, O_n$.
  To simplify the notation, moreover, we sometimes omit the subscript $P$ from the functions evaluated under $P$ unless otherwise stated, for instance, ${\rm Var}_P$ as ${\rm Var}$ and ${\rm Cov}_P$ as ${\rm Cov}$. Moreover, we use ${\rm P}(A)$ to denote the probability of event $A$ concerning the behaviors of the statistics based on $O$, and ${\rm P}(A_n)$ to denote the probability of event $A_n$ concerning those of the statistics derived from $O_1, \ldots, O_n$.
  
  %\item Consider an array of sub-$\sigma$-fields $\{{\cal O}_{nj}:j=q_n,\ldots,n-1, n \in \mathbb{N}\}$ on the given probability space $(\Omega, {\cal O},P)$ satisfying ${\cal O}_{nq_n} \subset {\cal O}_{n(q_n+1)} \subset \cdots \subset {\cal O}_{nn}$ for $n \in \mathbb{N}$. Let $h$ be a random function that may depend on $n$ and maps a realization of $O$ to $\mathbb{R}$. We use $P_{nj}h$ to denote the conditional expectation of the random element $h$ given ${\cal O}_{nj}$. 
\end{itemize}

\section{Canonical gradient of the slope parameter} \label{sec:canonical_max_slope}

The influence function of the sample Pearson correlation coefficient was first reported by \cite{Devlin1975}, who attributed the result to
C. L. Mallows. This influence function corresponds to the canonical gradient of the correlation coefficient in a locally nonparametric model. Here we use similar arguments to find the canonical gradient of the slope parameter $\Gamma$ in this same model. As in Section~\ref*{sec:eff_est_fixed_d} in the main text, here we focus on the case that there is only a single predictor.

Define a path $\{\tilde{P}_{\epsilon}=(1-\epsilon)\tilde{P} + \epsilon \delta_{(u,\tilde{y})}, \epsilon \in [0,1] \}$, where $\delta_{(u,\tilde{y})}$ is the Dirac measure putting unit mass at $(u,\tilde{y})$. Then
\begin{align*}
  \Gamma(\tilde{P}_{\epsilon})= \frac{{\rm Cov}_{\tilde{P}_{\epsilon}}(U,\tilde Y)}{{\rm Var}_{\tilde{P}_{\epsilon}}(U)},
\end{align*}
and
\begin{align*}
  &{\rm Cov}_{\tilde{P}_{\epsilon}}(U,\tilde Y)=\tilde{P}_{\epsilon}[U \tilde Y]-\tilde{P}_{\epsilon}[U]\tilde{P}_{\epsilon}[\tilde Y]\\
  &\;=(1-\epsilon)\tilde{P}[U \tilde Y]+\epsilon[u \tilde y]-(1-\epsilon)^2\tilde{P}[U]\tilde{P}[\tilde Y]-\epsilon(1-\epsilon)\tilde{P}[U]\tilde y
  -\epsilon(1-\epsilon)\tilde{P}[\tilde Y]u-\epsilon^2u \tilde y;\\
  &{\rm Var}_{\tilde{P}_{\epsilon}}(U)=\tilde{P}_{\epsilon}[U^2]-(\tilde{P}_{\epsilon}[U])^2=
  (1-\epsilon)\tilde{P}[U^2]+\epsilon u^2 - (1-\epsilon)^2(\tilde{P}[U])^2\\
  &\hspace{2cm} -2\epsilon(1-\epsilon)u\tilde{P}[U]
   -\epsilon^2u^2,
\end{align*}
where $\tilde y=\delta x/G(x)$.
Let $D(\tilde{P})$ denote the canonical gradient of $\Gamma$ at $\tilde{P}$ in a locally nonparametric model. The evaluation of this function at the chosen $(u,\tilde{y})$ is equal to $\left.\frac{d}{d\epsilon} \Gamma(\tilde{P}_{\epsilon})\right|_{\epsilon=0}$, and so
\begin{align*}
  D(\tilde{P})(u,\tilde{y})= \left\{\frac{(u-\tilde{P}[U])(\tilde y-\tilde{P}[\tilde Y])}{{\rm Var}_{\tilde{P}}(U)}
  -\frac{{\rm Cov}_{\tilde{P}}(U,\tilde Y)}{{\rm Var}^2_{\tilde{P}}(U)}(u-\tilde{P}[U])^2\right\}.
\end{align*}
% Noting that
% \begin{align*}
%   &{\tilde P}[U-{P}[U]]^2={\rm Var}_{\tilde  P}(U)+({\tilde P}[U]-P[U])^2,\\
%   &\tilde  P[(U-P[U])(\tilde Y-{P}[\tilde{Y}])]={\rm Cov}_{\tilde  P}(U,\tilde Y) + ({\tilde P}[U]-P[U])({\tilde P}[\tilde  Y]-P[\tilde{Y}]),
% \end{align*}
% and that $\tilde{P}[\tilde Y] = P[\tilde{Y}]$ and
% ${\rm Cov}_{\tilde{P}}(U,\tilde Y) = {\rm Cov}(U,\tilde{Y})$,
Straightforward calculations show that, for distributions $\tilde{P}_1$ and $\tilde{P}$ for the random variable $(U,\tilde{Y})$, it holds that
\begin{align*}
  \int D(\tilde{P})(u,\tilde{y})\, d\tilde{P}_1(u,\tilde{y})
  =&\left(\frac{{\rm Cov}_{\tilde{P}_1}(U,\tilde Y) + (\tilde{P}_1[U]-\tilde{P}[U])(\tilde{P}_1[\tilde Y]-\tilde{P}[\tilde{Y}])}{{\rm Var}_{\tilde{P}}(U)} \right.\\
   &\qquad\left.-\frac{{\rm Cov}_{\tilde{P}}(U,\tilde{Y})}{{\rm Var}_{\tilde{P}}^2(U)}
   \left[{\rm Var}_{\tilde{P}_1}(U)+(\tilde{P}_1[U]-\tilde{P}[U])^2\right]\right).
\end{align*}
Recall that $\tilde{P}\equiv P\circ f_G^{-1}$ for the distribution $P$ that generated each of the variates $O_1,\ldots,O_n$. Suppose also that $\tilde{P}_1=P_1\circ f_G^{-1}$ for a distribution $P_1$ of the random variable $O$. The above then shows that
\begin{align*}
%   {\rm Rem}({P},\tilde P)&=
  &{\rm Rem}_G(P_1,P)
  \equiv \Gamma(\tilde{P})-\Gamma(\tilde{P}_1)+\int D(\tilde{P})(u,\tilde{y})\, d\tilde{P}_1(u,\tilde{y})\\
  &=  \left\{\frac{{\rm Cov}_{\tilde{P}}(U,\tilde{Y})}{{\rm Var}_{\tilde{P}}(U)}-\frac{{\rm Cov}_{\tilde{P}_1}(U,\tilde Y)}{{\rm Var}_{\tilde{P}_1}(U)}
  +\frac{(\tilde{P}_1[U]-\tilde{P}[U])(\tilde{P}_1[\tilde Y]-\tilde{P}[\tilde{Y}])}{{\rm Var}_{\tilde{P}}(U)}\right.\\
  &\left.\qquad +\frac{{\rm Cov}_{\tilde{P}_1}(U,\tilde Y)}{{\rm Var}_{\tilde{P}}(U)}
  -\frac{{\rm Cov}_{\tilde{P}}(U,\tilde{Y})}{{\rm Var}_{\tilde{P}}^2(U)}[{\rm Var}_{\tilde{P}_1}(U)+(\tilde{P}_1[U]-\tilde{P}[U])^2]\right\} \\ %,
% \end{align*}
% which can be expressed as
% \begin{align*}
  &= \left\{\left[\frac{{\rm Cov}_{\tilde{P}}(U,\tilde{Y})}{{\rm Var}_{\tilde{P}}^2(U)}-\frac{{\rm Cov}_{\tilde{P}_1}(U,\tilde Y)}{{\rm Var}_{\tilde{P}_1}(U){\rm Var}_{\tilde{P}}(U)}\right]
  ({\rm Var}_{\tilde{P}}(U)-{\rm Var}_{\tilde{P}_1}(U))
   \right.\\
  &\quad \left.\quad-\frac{{\rm Cov}_{\tilde{P}}(U,\tilde{Y})}{{\rm Var}_{\tilde{P}}^2(U)}(\tilde{P}_1[U]-\tilde{P}[U])^2 +\frac{1}{{\rm Var}_{\tilde{P}}(U)}(\tilde{P}_1[U]-\tilde{P}[U])(\tilde{P}_1[\tilde Y]-\tilde{P}[\tilde{Y}]) \right\}.
\end{align*}
Note that the remainder term in the linear expansion \eqref{eq:if_G_fixed_derivation} is equal to $-{\rm Rem}_G(\Pn,P)$. The following lemma shows that this
term is asymptotically negligible.
\begin{lemma} \label{lemma:remainder_G}
Under \ref{assump:Covariates} (so that ${\rm Var}_{\tilde{P}}(U)>0$), we have
${\rm Rem}_G(\Pn, P) = o_p(n^{-1/2})$.
\end{lemma}

The proof uses the CLT, the weak law of large numbers and the triangle inequality:
\begin{align*}
  &\sqrt{n}|{\rm Rem}_G(\Pn, P)|\\
  & \leq [{\rm Var}_{\tilde{\mathbb{P}}_n}(U) \wedge {\rm Var}_{\tilde P}(U)]^{-2}\left\{\left|\sqrt{n}({\rm Cov}_{\tilde{\mathbb{P}}_n}(U, \tilde Y)-{\rm Cov}_{\tilde P}(U,\tilde{Y}))\right|
  |{\rm Var}_{\tilde{\mathbb{P}}_n}(U)-{\rm Var}_{\tilde P}(U)| \right.\\
  &\hspace{0.12cm} \left.+\left|{\rm Cov}_{\tilde P}(U,\tilde{Y})\right|\sqrt{n}(\tilde{\mathbb{P}}_n[U]-\tilde{P}[U])^2+{\rm Var}_{\tilde P}(U)\left|\sqrt{n}(\tilde{\mathbb{P}}_n[U]-\tilde{P}[U])(\tilde{\mathbb{P}}_n[\tilde Y]-\tilde{P}[\tilde{Y}])\right|\right\}.
\end{align*}

\section{Proof of lemmas for Theorem \ref{Thm:if_slope}}
\label{sec:proof_lemmas_of_Thm_if_slope}

\begin{lemma} \label{lemma:remainder_estG}
  If \ref{assump:Covariates} and \ref{assump:Survival function} hold, then ${\rm Rem}_{\hat{G}_n}(\Pn, P) = o_p(n^{-1/2}).$
\end{lemma}
\begin{proof}
  Similar to Lemma \ref{lemma:remainder_G}, by the triangle inequality and the equivalence between $\Pn$ and $\tilde{\mathbb{P}}_n$ as they operate on $(U,\tilde Y)$ marginally or jointly,
  \begin{align*}
    &\sqrt{n}|{\rm Rem}_{\hat{G}_n}(\Pn, P)|
    \leq [{\rm Var}_{\Pn}(U) \wedge {\rm Var}_{\tilde P}(U)]^{-2}\Big\{\big|{\rm Cov}_{\Pn}(U, Y)-{\rm Cov}_{\Pn}(U,\tilde{Y})\big|\\
    &\quad + \big|{\rm Cov}_{\Pn}(U, \tilde Y)-{\rm Cov}_{\tilde P}(U,\tilde{Y})\big|\Big\}
    \big|\sqrt{n}({\rm Var}_{\Pn}(U)-{\rm Var}_{\tilde P}(U))\big|\\
    &\quad +\big|{\rm Cov}_{\tilde P}(U,\tilde{Y})\big|\sqrt{n}\big\{(\Pn-\tilde{P})[U]\big\}^2\\
    &\quad +{\rm Var}_{\tilde P}(U)\big|\sqrt{n}(\Pn-\tilde{P})[U]\big|
    \Big\{\big|\Pn[Y]-\Pn[\tilde{Y}]\big|+\big|(\Pn-\tilde{P})[\tilde{Y}]\big|\Big\}.
  \end{align*}
Each term above is now shown to be $o_p(1)$. For the  first term, note that ${\rm Var}_{\tilde{P}}(U)>0$, ${\rm Var}_{\Pn}(U) \crarrow{p} {\rm Var}_{\tilde{P}}(U) > 0$,
${\rm Cov}_{\Pn}(U, Y)-{\rm Cov}_{\Pn}(U,\tilde{Y}) = o_p(1)$ by \ref{assump:Survival function} and the uniform consistency of $\hat{G}_n$ to $G$,
the convergence of $|{\rm Cov}_{\Pn}(U, \tilde{Y})-{\rm Cov}_{\tilde{P}}(U,\tilde{Y})|$ to zero in probability by
the weak law of large numbers, and $\sqrt{n}({\rm Var}_{\Pn}(U)-{\rm Var}_{\tilde P}(U))=O_p(1)$ by the CLT.
The second term is $o_p(1)$ because
$\sqrt{n}(\Pn-\tilde{P})[U]=O_p(1)$ and $(\Pn-\tilde{P})[U] \crarrow{p} 0$. Similarly, the last term is $o_p(1)$, which is a consequence of
$\sqrt{n}(\Pn-\tilde{P})[U]=O_p(1)$,
$\Pn[\tilde Y] \crarrow{p} \tilde{P}[\tilde{Y}]$,
and $\Pn[Y]-\Pn[\tilde{Y}] = o_p(1)$ that follows by \ref{assump:Survival function} and the uniform consistency of $\hat{G}_n$ as an estimator of $G$
\end{proof}

\begin{lemma} \label{lemma:estG_by_G_for_if_ipw}
Suppose \ref{assump:Covariates} and \ref{assump:Survival function} hold, so $U$ has a finite fourth moment and  ${\rm Var}(U)>0$. Then 
\begin{align*}
  \Pn\big[\IF^{ipw}_{\hat{G}_n}(O|P)\big] = \Pn\big[\IF^{ipw}_G(O|P)\big] + o_p(n^{-1/2}).
\end{align*}
\end{lemma}
\begin{proof}
Following \ref{assump:Survival function}, fix an $\tilde{\epsilon}$ such that $0 < \tilde{\epsilon} < G(\tau)$.
Let $\mathcal{G}$ be a collection of monotone nonincreasing c\`{a}dl\`{a}g functions $\tilde{G}:\mathcal{T} \to [0,1]$ such that $\tilde{G}(\tau) > \tilde{\epsilon}$, which is seen to be P-Donsker \cite[Example 19.11 in][]{Vaart1998}. The influence function $\psi(\tilde{G})=\IF^{ipw}_{\tilde{G}}(\cdot|P)$, and it is continuous of $\tilde{G}$ because $\tilde{G}(\tau) > \tilde{\epsilon} > 0$. We further have that $\|\psi(\tilde{G})\|_{P,2} < \infty$ because $X$ is bounded by $\tau$;
${\rm Var}(U)>0$ and the finite fourth moment condition of $U$ given in \ref{assump:Covariates} and $\tilde{G}(\tau)>0$, implying that the class $\psi(\mathcal{G})$ is also P-Donsker by the Donsker preservation properties \cite[Section 9.4 in][]{Kosorok2008}.
Since $\psi(G) \in \psi(\mathcal{G})$, $\psi(\hat{G}_n) \in \psi(\mathcal{G})$ and the uniform consistency of Kaplan--Meier estimator implies that
\[
  \int (\psi(\hat{G}_n)-\psi(G))^2 dP \lcrarrow{p} 0,
\]
it follows that $\sqrt{n}(\Pn-P)(\psi(\hat{G}_n)-\psi(G)) \crarrow{p} 0$ by Lemma 19.24 of \cite{Vaart1998}.
Equivalently we have that
\begin{align*}
  \Pn\big[\IF^{ipw}_{\hat{G}_n}(O|P)\big] = \Pn\big[\IF^{ipw}_G(O|P)\big] + o_p(n^{-1/2})
\end{align*}
because $P[\psi(G)]=P\big[\IF^{ipw}_G(O|P)\big]=0$ and $P[\psi(\hat{G}_n)]=P\big[\IF^{ipw}_{\hat{G}_n}(O|P)\big]=0$, the latter implied by the expectation with respect to $P$ not operating on $\hat{G}_n$.
\end{proof}

Standard results concerning the Kaplan--Meier estimator give
\begin{lemma} \label{lemma:if_estG}
Given \ref{assump:At-risk prob} and for $ t\in \mathcal{T}$, the influence function of $\hat{G}_n(t)$ is 
\begin{align*}
  &\IF_{G}(O|t) = -G(t)\int_{-\infty}^t\left[\frac{1(X \in ds, \delta=0)}{P1(X \geq s)}-\frac{1(X \geq s)d\Lambda(s)}{{P1(X \geq s)}}\right]\\
  &\quad = -G(t)\int_{-\infty}^{t} \frac{dM(s)}{P1(X \geq s)}.
\end{align*}
\end{lemma}

\begin{lemma}  \label{lemma:projection_of_if_ipw}
Given \ref{assump:Covariates}--\ref{assump:At-risk prob},
the function $\IF^{nu}_G(O|P)$ is the projection of $\IF^{ipw}_G(O|P)$ onto the nuisance tangent space $\mathbf{T}^{nu}(G)$, where $\IF^{ipw}_G(O|P)$, $\IF^{nu}_G(O|P)$ and $\mathbf{T}^{nu}(G)$ are defined in the main text.
\end{lemma}
\begin{proof}
 The projection onto $\mathbf{T}^{nu}(G)$ is given by 
\[
  \Pi(H'(O)|\mathbf{T}^{nu}(G))=\int_{\mathcal{T}} \left\{P[H'(O)|X=s]- P[H'(O)|X \geq s]\right\}dM(s),
\]
for any bounded measurable function $H': \mathcal{T} \times \{0,1\} \times \mathbb{R} \rightarrow \mathbb{R}$.
The projection of $\IF^{ipw}_G(O|P)$ onto $\mathbf{T}^{nu}(G)$ is therefore
\begin{equation} \label{eq:projection_if_ipw}
\begin{split}
  &\Pi(\IF^{ipw}_G(O|P)|\mathbf{T}^{nu}(G))
  = \int_{\mathcal{T}}  \left\{P[\IF^{ipw}_G(O|P)|X=s]-P[\IF^{ipw}_G(O|P)|X \geq s]\right\}dM(s)\\
  & \quad = \int_{\mathcal{T}} \bigg\{P\bigg[\frac{(U-P[U])\tilde{Y}}{{\rm Var}(U)}\bigg|X=s\bigg]-P\bigg[\frac{(U-P[U])\tilde{Y}}{{\rm Var}(U)}\bigg|X \geq s\bigg]\bigg\}dM(s).
\end{split}
\end{equation}
The first term on the right-hand-side of \eqref{eq:projection_if_ipw} is
\[
  \frac{1}{{\rm Var}(U)} \int_{\mathcal{T}} P[(U-P[U])\tilde{Y}|X=s]dM(s)=0,
\]
and it is easy to see that the second term is
\begin{align*}
  &\frac{1}{{\rm Var}(U)} \int_{\mathcal{T}} P[(U-P[U])\tilde{Y}|X \geq s]dM(s)\\
  &\quad = \frac{1}{{\rm Var}(U)} P\left[(U-P[U])\int_{\mathcal{T}} E[T|U,X \geq s]\frac{dM(s)}{G(s)}\right],
\end{align*}
applying arguments used for deriving $\IF_{G}(O)$.
Combining these two terms, we can see that the projection of $\IF^{ipw}_G(O|P)$ onto $\mathbf{T}^{nu}(G)$ is equal to
$\IF^{nu}_G(O|P)$ that is defined in \eqref{eq:expression_Psi_hatG.2} in the main text.
For further details see the proof of Theorem~2.3 of \cite{vanderLaan&Robins2003}.
\end{proof}

\section{Proof of lemmas for Theorem \ref{Thm:if_one_step}}
\label{sec:proof_lemmas_of_Thm_if_one_step}
The first two lemmas in this section are given in order to apply Theorem 2.1 of \cite{vanderVaart&Wellner2007}, as done in the third lemma.  First of all, we need the following notation.

By \ref{assump:Survival function}, we can fix an $\tilde{\epsilon}$ such that $0 < \tilde{\epsilon} < G(\tau)$. Let $\mathcal{G}$ be the collection of monotone nonincreasing c\`{a}dl\`{a}g functions $\tilde{G} \colon \mathcal{T} \rightarrow [0,1]$ such that $\tilde{G}(\tau) > \tilde{\epsilon}$, and define 
\begin{align*}
  \mathcal{G}_0 = \Big\{\tilde{G}: \tilde{G} \in {\cal G};\, \sup_{s\in\mathcal{T}}\Big|\frac{\tilde{G}(s)}{G(s)}-1\Big| \le 1 \Big\}.
\end{align*}
Let $\mathcal{Q}^*$ be the collection of monotone nondecreasing c\`{a}dl\`{a}g functions $\tilde{Q} \colon \mathbb{R} \rightarrow [0,1]$. Note that $\tilde{Q}$ could be the c.d.f. of $U \sim Q_u$. By \ref{assump:Covariates} that $U$ is non-degenerate, there exists $\nu$ such that $0 < \nu < {\rm Var}_{Q_u}(U)$.  Define $\mathcal{Q} = \{\tilde{Q}: \tilde{Q} \in \mathcal{Q}^*\,;\;  |\tilde{Q}[U]| < \infty\,;\; \nu < {\rm Var}_{\tilde{Q}}(U) < \infty \}$. Let $\mathcal{E}$ be the collection of functions $(u,s) \mapsto \tilde{E}(u,s)$ that are uniformly bounded and left-continuous in $s$, and
$\mathcal{H} = \{(u,s) \mapsto (\tilde{Q}(u),\tilde{E}(u,s)) : \tilde{Q} \in \mathcal{Q}; \tilde{E} \in \mathcal{E}\}$. Fix $\epsilon> 0$, and define
$\mathcal{H}_0 \equiv \mathcal{Q}_0 \times \mathcal{E}_0 \subset \mathcal{H}$,
where 
\begin{align*}
  \mathcal{Q}_0 = \big\{&\tilde{Q}1(|\tilde{Q}[U]-Q_u[U]| \le \epsilon, |{\rm Var}_{\tilde{Q}}(U)-{\rm Var}_{Q_u}(U)| \le \epsilon):\tilde{Q} \in \mathcal{Q}\big\}
\end{align*}
and $\mathcal{E}_0 = \{\tilde{E}1(|\tilde{E}-\bar{E}| \le \epsilon):\tilde{E} \in \mathcal{E}\}$.

\begin{lemma} \label{lemma:conditions_for_Thm2.1}
Given \ref{assump:Covariates}, \ref{assump:Survival function}, \ref{assump:At-risk prob_covariate} and \ref{assump:Conditional mean}, let $\mathcal{H}_0$ and  $\mathcal{G}_0$ be defined above, and let $\IF^*$ be as defined in the main text. Then
\begin{lemmaitemenum}
  \item \label{eq:estimators_in_classes} ${\rm P}((\hat{E}_n,\mathbb{Q}_n) \in \mathcal{H}_0) \to 1$ and ${\rm P}(\hat{G}_n \in \mathcal{G}_0) \to 1$;
  \item \label{eq:Donsker_class} $\{\IF^*(\cdot|\tilde{E},\tilde{Q},\tilde{G}):(\tilde{E},\tilde{Q}) \in \mathcal{H}_0; \tilde{G} \in \mathcal{G}_0\}$ is $P$-Donsker.
\end{lemmaitemenum}
\end{lemma}
\begin{proof}
First note that $\mathbb{Q}_n[U]$ is bounded almost surely, and
${\rm Var}_{\mathbb{Q}_n}(U)$ is finite and larger than $\nu$ almost surely, by the strong law of large numbers and \ref{assump:Covariates} that $U$ is bounded and non-degenerate.
In addition, $\hat{E}_n(u,s)$ is uniformly bounded over $(u,s)$ on $\mathbb{R} \times \mathcal{T}$ and left-continuous in $s$ with probability tending to one, by \ref{assump:Conditional mean} that $\mathbb{E}[|\hat{E}_n(u,s)-\bar{E}(u,s)|^2] \to 0$ uniformly and $\bar{E}$ as a left-continuous funciton in $s$.
Thus for sufficiently large $n$, 
\begin{align*}
  & {\rm P}((\hat{E}_n,\mathbb{Q}_n) \notin \mathcal{H}_0) 
  \le {\rm P}\big(\,\big|\mathbb{Q}_n[U]-Q_u[U]\big| > \epsilon\big) + {\rm P}\big(\,\big|{\rm Var}_{\mathbb{Q}_n}(U)-{\rm Var}_{Q_u}(U)\big| > \epsilon\big)\\ 
  & \hspace{3.4cm} {\rm P}\big(\big|\hat{E}_n(u,s)-\bar{E}(u,s)\big| > \epsilon \big)\\
  & \le \mathbb{E}\big[\big|\mathbb{Q}_n[U]-Q_u[U]\big|\big] + \mathbb{E}\big[\big|{\rm Var}_{\mathbb{Q}_n}(U)-{\rm Var}_{Q_u}(U)\big|\big] + \mathbb{E}\big[\big|\hat{E}_n(u,s)-\bar{E}(u,s)\big|^2\big] \to 0
\end{align*}
by the strong law of large numbers, \ref{assump:Conditional mean} along with the above arguments. Note also that ${\rm P}(\hat{G}_n \in \mathcal{G}_0) \to 1$ by the uniform consistency of the Kaplan--Meier estimator. This gives \ref{eq:estimators_in_classes}.

Moreover, we see that $\mathcal{G}$ and $\mathcal{Q}^*$ are $P$-Donsker \cite[Example 19.11 in][]{Vaart1998}, and so are $\mathcal{G}_0$ and $\mathcal{Q}$ because $\mathcal{G}_0 \subset \mathcal{G}$ and $\mathcal{Q} \subset \mathcal{Q}^*$ \cite[Theorem 2.10.1 in][]{Vaart1996}.
Moreover, $\mathcal{H}_0$ is $P$-Donsker because $\mathcal{Q}_0$ and $\mathcal{E}_0$ are the classes of indicator functions multiplied by some uniformly bounded functions and known as $P$-Donsker \cite[Corollary 9.32 $(i)$ and $(v)$ in][]{Kosorok2008}.

Now we start to show \ref{eq:Donsker_class}.
Based on the properties of the above classes, below
we show the class $\mathcal{F} = \{f_1-f_2-f_3: f_1 \in \mathcal{F}_1, f_2 \in \mathcal{F}_2, f_3 \in \mathcal{F}_3\}$
is a $P$-Donsker class of functions on $\mathcal{T} \times \{0,1\} \times \mathbb{R}$, where
\begin{align*}
  \mathcal{F}_1 &= \bigg\{(x,\delta,u) \mapsto
  \frac{(u-\tilde{Q}[U])}{{\rm Var}_{\tilde{Q}}(U)}\Big(\frac{\delta x}{\tilde{G}(x)}-\tilde{Q}[\tilde{E}(U)]\Big) :\,(\tilde{E},\tilde{Q}) \in \mathcal{H}_0;\,
  \bar{G} \in {\cal G}_0\bigg\};\\
  \mathcal{F}_2 &= \bigg\{(x,\delta,u) \mapsto \bigg[\frac{(u-\tilde{Q}[U])}{{\rm Var}_{\tilde{Q}}(U)}\bigg]^2{\rm Cov}_{\tilde{Q}}(U,\tilde{E}(U)):\,(\tilde{E},\tilde{Q}) \in \mathcal{H}_0 \bigg\};\\
  \mathcal{F}_3 &= \bigg\{(x,\delta,u) \mapsto \frac{(u-\tilde{Q}[U])}{{\rm Var}_{\tilde{Q}}(U)}\int_{\mathcal{T}}\tilde{E}(u, s)\big[1(x \in ds, \delta=0)-1(x \ge s)d\tilde{\Lambda}(s) \big] :\\
  &\hspace{3cm} \tilde{\Lambda}(s)=-\log(\tilde{G}(s));\,
  (\tilde{E},\tilde{Q}) \in \mathcal{H}_0;\,\tilde{G} \in \mathcal{G}_0 \bigg\}.
\end{align*}
We can see that $\IF^*(\cdot|\tilde{E},\tilde{Q},\tilde{G}) \in \mathcal{F}$, so $\{\IF^*(\cdot|\tilde{E},\tilde{Q},\tilde{G}):(\tilde{E},\tilde{Q}) \in \mathcal{H}_0; \tilde{G} \in \mathcal{G}_0\} \subseteq \mathcal{F}$. 
Note also that $\tilde{\Lambda}$ is a non-decreasing $c\grave{a}dl\grave{a}g$ function on $\mathcal{T}$, which could be a cumulative hazard function corresponding to $\tilde{G}$ if $\tilde{G}$ is a survival function.
To show \ref{eq:Donsker_class}, it suffices to show $\mathcal{F}_1$, $\mathcal{F}_2$ and $\mathcal{F}_3$ are $P$-Donsker. 

By Corollary 9.32 $(iv)$ in \cite{Kosorok2008}, the class 
$\{(x,\delta,u) \mapsto u-\tilde{Q}[U] : \tilde{Q} \in \mathcal{Q}_0\}$ is $P$-Donsker.
Together with \ref{assump:Covariates} that $U$ is non-degenerate,
the class $\mathcal{F}_{4} \equiv \{(x,\delta,u) \mapsto [(u-\tilde{Q}[U])/{\rm Var}_{\tilde{Q}}(U)]^r : \tilde{Q} \in \mathcal {Q}_0,\, r \in \{1,2\}\}$
is uniformly bounded and $P$-Donsker \cite[Corollary 9.32 $(iii)$ in][]{Kosorok2008}.
Note also that $\tilde{Q}[\tilde{E}(U)]$ and ${\rm Cov}_{\tilde{Q}}(U,\tilde{E}(U))$ are constant functions of $(x,\delta,u)$, so the classes $\mathcal{F}_{5} \equiv \{(x,\delta,u) \mapsto \tilde{Q}[\tilde{E}(U)] : (\tilde{E},\tilde{Q}) \in \mathcal{H}_0\}$ and 
$\mathcal{F}_{6} \equiv \{(x,\delta,u) \mapsto {\rm Cov}_{\tilde{Q}}(U,\tilde{E}(U)) : (\tilde{E},\tilde{Q}) \in \mathcal{H}_0\}$ are uniformly bounded and $P$-Donsker.
Moreover, because $\tilde{G}(\tau)>\tilde{\epsilon} > 0$ for all $\tilde{G} \in {\cal G}$ and $|X| \le \tau$, the aforementioned Donsker preservation properties imply that the class $\mathcal{F}_{7}\equiv\{(x,\delta,u) \mapsto \delta x/\tilde{G}(x), \tilde{G} \in \mathcal{G}_0\}$ is uniformly bounded and $P$-Donsker.
Therefore by Donsker preservation properties, we have the below classes are $P$-Donsker: 
\begin{align*}
  \mathcal{F}_1 &\subseteq \{f_{11} \cdot (f_{12}-f_{13}) : f_{11} \in \mathcal{F}_{4}, f_{12} \in \mathcal{F}_{7}, f_{13} \in \mathcal{F}_{5}\};\\
  \mathcal{F}_2 &\subseteq \{f_{21} \cdot f_{22} : f_{21} \in \mathcal{F}_{4}, f_{22} \in \mathcal{F}_{6}\}.
\end{align*}

Below we show that $\mathcal{F}_3$ is $P$-Donsker.
Let $\mathcal{F}_{8} = \{(x,\delta,u) \mapsto \int_{\mathcal{T}}\tilde{E}(u, s)[1(x \in ds, \delta=0)-1(x \ge s)d\tilde{\Lambda}(s)] :
\tilde{\Lambda}(s)=-\log(\tilde{G}(s));\,\tilde{G} \in {\cal G}_0;\,
\tilde{E} \in \mathcal{E}_0 \subset \mathcal{H}_0\}.$
Let $\{t_1 < t_2 < \ldots < t_m\}$ be an arbitrary partition of ${\cal T}$ with uniform increments $\Delta_t = t_{j+1}- t_{j}$ for all $j$.
Note that $\tilde{E} \in \mathcal{E}_0$ is uniformly bounded: $|\tilde{E}| \leq K_1$ for some positive constant $K_1$.
Let $K' = 4K_1\sum_{j=1}^{m}|\tilde{\Lambda}(t_j+\Delta_t)|$.
Following that $\tilde{E}$ is also left-continuous in $s$, any function in $\mathcal{F}_{8}$ is the scalar multiple (by $K'$) of the uniform limit of the sequence 
\begin{align*}
  &\sum_{j=1}^{m}\frac{\tilde{E}(u,t_{j+1})}{K'}\Big[1_{[t_j,t_{j+1})}(x)(1-\delta)-\tilde{\Lambda}(t_{j+1})1(x \geq t_{j+1}) + \tilde{\Lambda}(t_{j})1(x \geq t_{j})\Big]\\
  & = \sum_{j=1}^{m}\frac{\tilde{E}(u,t_j+\Delta_t)}{K'}\Big[1_{[t_j,t_{j}+\Delta_t)}(x)(1-\delta)-\tilde{\Lambda}(t_{j}+\Delta_t)1(x \geq t_{j}+\Delta_t)
  + \tilde{\Lambda}(t_{j})1(x \geq t_{j})\Big]\\
  & = \sum_{j=1}^m \sum_{r=1}^4 \bigg\{\frac{\tilde{E}(u,t_{j}+\Delta_t)}{K'}(-1)^{r}\tilde{\Lambda}(t_{j}+\Delta_{t}1_{r=3})^{1_{r \ge 3}}\bigg\}(1-\delta)^{1_{r \le 2}}1\big(x \ge t_j + \Delta_{t}1_{r \in \{2,3\}}\big)\\
  &\equiv \sum_{j=1}^m \sum_{r=1}^4 \alpha_{jr}(1-\delta)^{1_{r \le 2}}1\big(x \ge t_j + \Delta_{t}1_{r \in \{2,3\}}\big),      
\end{align*}
where $\sum_{j=1}^m \sum_{r=1}^4 |\alpha_{jr}| \le 1$ by
\begin{align*}
  &\sum_{j=1}^m \sum_{r=1}^4 |\alpha_{jr}| \le \sum_{j=1}^{m}\sum_{r=1}^{4}\bigg|\frac{\tilde{E}(u,t_j+\Delta_t)}{K'}(-1)^{r}\tilde{\Lambda}(t_{j}+\Delta_{t}1_{r=3})^{1_{r \ge 3}}\bigg|\\  
  &\le \sum_{j=1}^{m}\sum_{r=1}^{4}\,\bigg|\frac{\tilde{E}(u,t_j+\Delta_t)}{K'}\tilde{\Lambda}(t_j+\Delta_t)\bigg| \le \frac{4K_1}{K'}\sum_{j=1}^{m}\big|\tilde{\Lambda}(t_j+\Delta_t)\big|
  = 1.
\end{align*}
Then any function in $\mathcal{F}_{8}$ is in a class contained in the scalar-multiplied symmetric convex hull of the VC-subgraph class $\left\{(x, \delta) \mapsto \delta1(x\ge s) : s\in\mathcal{T} \right\}$, which is a class of indicator functions. Therefore $\mathcal{F}_{8}$ is $P$-Donsker \cite[Theorem 2.10.3 in][]{Vaart1996}, and then $\mathcal{F}_3 \subseteq \{f_{31} \cdot f_{32} : f_{31} \in \mathcal{F}_{4}, f_{32} \in \mathcal{F}_{8}\}$ is also $P$-Donsker,  \cite[Corollary 9.32 $(iii)$ in][]{Kosorok2008}.
Hence, $\mathcal{F}$ is $P$-Donsker by Corollary 9.32 $(i)$ of \cite{Kosorok2008}, which gives \ref{eq:Donsker_class}.
\end{proof}

\begin{lemma} \label{lemma:L2_conv.1}
Given \ref{assump:Covariates}, \ref{assump:Survival function}, \ref{assump:At-risk prob_covariate} and \ref{assump:Conditional mean}, and let $\mathcal{G}_0$ be defined above and $\IF^*$ be as defined in the main text,
$\sup_{\tilde{G} \in \mathcal{G}_0} P[(\IF^*(\cdot|\hat{E}_n,\mathbb{Q}_n,\tilde{G})-\IF^*(\cdot|\bar{E},Q_u,\tilde{G}))^2] \lcrarrow{p} 0.$
\end{lemma}
\begin{proof}
To give the proof, we apply Chebyshev's inequality, so that it suffices to show the first moments of the relevant mean-squared quantities converge to zero. Recall that $P$ denotes the expectation of a generic variable, one of the conventional notations for empirical processes. 
For any $\tilde{G} \in \mathcal{G}_0$,
\begin{align} \label{eq:target_convergence.1}
  &P[\IF^*(\cdot|\hat{E}_n,\mathbb{Q}_n,\tilde{G})-\IF^*(\cdot|\bar{E},Q_u,\tilde{G})]^2 \lesssim (i) + (ii) + (iii) + (iv) + (v), \mbox{ where}\\
  &\quad (i) \equiv P\bigg[ \bigg\{\bigg(\frac{(U-\mathbb{Q}_n[U])}{{\rm Var}_{\mathbb{Q}_n}(U)}-\frac{(U-Q_u[U])}{{\rm Var}_{Q_u}(U)}\bigg)\frac{\delta X}{\tilde{G}(X)} \bigg\}^2\bigg]; \nonumber \\
  &\quad (ii) \equiv P\bigg[\bigg\{\frac{(U-\mathbb{Q}_n[U])}{{\rm Var}_{\mathbb{Q}_n}(U)}\mathbb{Q}_n[\hat{E}_n(U)]-\frac{(U-Q_u[U])}{{\rm Var}_{Q_u}(U)}Q_u[\bar{E}(U)] \bigg\}^2\bigg]; \nonumber \\
  &\quad (iii) \equiv P\bigg[\bigg\{\frac{(U-\mathbb{Q}_n[U])^2}{{\rm Var}^2_{\mathbb{Q}_n}(U)}{\rm Cov}_{\mathbb{Q}_{n}}(U, \hat{E}_n(U))-\frac{(U-Q_u[U])^2}{{\rm Var}_{Q_u}^2(U)}{\rm Cov}_{Q_u}(U, \bar{E}(U))\bigg\}^2\bigg]; \nonumber \\
  &\quad (iv) \equiv P\bigg[\bigg\{\bigg(\frac{(U-\mathbb{Q}_n[U])}{{\rm Var}_{\mathbb{Q}_n}(U)} - \frac{(U-Q_u[U])}{{\rm Var}_{Q_u}(U)}\bigg) \int_{\mathcal{T}}\hat{E}_n(U,s)\tilde{M}(U,ds) \bigg\}^2\bigg]; \nonumber \\ 
  &\quad (v) \equiv P\bigg[\bigg\{\frac{(U-Q_u[U])}{{\rm Var}_{Q_u}(U)} \int_{\mathcal{T}}\Big(\hat{E}_n(U,s)-\bar{E}(U,s)\Big)\tilde{M}(U,ds)\bigg\}^2 \bigg]. \nonumber
\end{align}
Then showing \ref{lemma:L2_conv.1} is equivalent to showing that the quantities $(i)$--$(v)$ on the right-hand-side of \eqref{eq:target_convergence.1} converge to zero in probability. First along with \ref{assump:Covariates} that $U$ is non-degenerate, the strong law of large numbers and the continuous mapping theorem give that
\begin{align} \label{eq:convergence_U_func}
  \left|\frac{(U-\mathbb{Q}_n[U])}{{\rm Var}_{\mathbb{Q}_n}(U)}-\frac{(U-Q_u[U])}{{\rm Var}_{Q_u}(U)}\right| \to 0 \:\mbox{ a.s.}    
\end{align}
By \eqref{eq:convergence_U_func}, together with $\tilde{G}(\tau) > \tilde{\epsilon} > 0$ for $\tilde{G} \in \mathcal{G}_0$ and $|X| \le \tau$ such that $|\delta X/\tilde{G}(X)| \le \tau/\tilde{\epsilon}$, we have the quantity $(i)$ converging to zero in probability. 

The quantity $(ii)$ of \eqref{eq:target_convergence.1} is
\begin{align} \label{eq:term2.upper_bound}
  \lesssim & \; P\bigg[\left\{\left[\frac{(U-\mathbb{Q}_n[U])}{{\rm Var}_{\mathbb{Q}_n}(U)}-\frac{(U-Q_u[U])}{{\rm Var}_{Q_u}(U)}\right]Q_u[\bar{E}(U)]\right\}^2\bigg] \\
  & + P\bigg[\left\{\frac{(U-\mathbb{Q}_n[U])}{{\rm Var}_{\mathbb{Q}_n}(U)}\Qn\big[\hat{E}_n(U)-\bar{E}(U)\big]\right\}^2\bigg] \nonumber \\
  & + P\bigg[\left\{\frac{(U-\mathbb{Q}_n[U])}{{\rm Var}_{\mathbb{Q}_n}(U)}\left(\big(\Qn-Q_u\big)[\bar{E}(U)]\right)\right\}^2\bigg]. \nonumber
\end{align}
The first quantity of \eqref{eq:term2.upper_bound} converges to zero in probability by \eqref{eq:convergence_U_func} and $Q_u[\bar{E}(U)]$ is bounded. The second quantity of \eqref{eq:term2.upper_bound} is
\begin{align*}
  & \lesssim P\bigg[\left\{\left[\frac{(U-\mathbb{Q}_n[U])}{{\rm Var}_{\mathbb{Q}_n}(U)}-\frac{(U-Q_u[U])}{{\rm Var}_{Q_u}(U)}\right]\Qn\big[\hat{E}_n(U)-\bar{E}(U)\big]\right\}^2\bigg]\\
  &\quad + P\bigg[\left\{\frac{(U-Q_u[U])}{{\rm Var}_{Q_u}(U)}\Qn\big[\hat{E}_n(U)-\bar{E}(U)\big]\right\}^2\bigg],    
\end{align*}
where by \ref{assump:Covariates} that $U$ is bounded and non-degenerate, the dominated convergence theorem implies that the first moment of
the first part converges to zero, using \eqref{eq:convergence_U_func} and $\sup_{u \in \mathbb{R}}|\hat{E}_n(u)-\bar{E}(u)|$ is bounded in probability implied by \ref{assump:Conditional mean}. Therefore applying Chebyshev's inequality, the first part converges to zero in probability,
and so does the second part by
\begin{align*}
  &\mathbb{E}\bigg\{P\bigg[\left\{\frac{(U-Q_u[U])}{{\rm Var}_{Q_u}(U)}\Qn\big[\hat{E}_n(U)-\bar{E}(U)\big]\right\}^2\bigg]\bigg\}\\
  & = P\bigg[\frac{(U-Q_u[U])^2}{{\rm Var}^2_{Q_u}(U)}\frac{1}{n^2}\bigg\{\sum_{i=1}^n\mathbb{E}\Big\{\big[\hat{E}_n(U_i)-\bar{E}(U_i)\big]^2\Big\}\\
  & \hspace{2cm} + \sum_{\{i \not= j: i=1,\ldots,n\} }\sum_{j=1}^n\mathbb{E}\Big\{\big[\hat{E}_n(U_i)-\bar{E}(U_i)\big]\big[\hat{E}_n(U_j)-\bar{E}(U_j)\big]\Big\}\bigg\}\bigg]\\
  & \lesssim P\bigg[\frac{1}{n}\mathbb{E}\big\{\big[\hat{E}_n(U_1)-\bar{E}(U_1)\big]^2\big\} + \frac{(n-1)}{n}\mathbb{E}\big\{\big[\hat{E}_n(U_1)-\bar{E}(U_1)\big]\big[\hat{E}_n(U_2)-\bar{E}(U_2)\big]\big\}\bigg]\\
  & \to 0,
\end{align*}
where the penultimate line follows \ref{assump:Covariates} that implies that $(U-Q_u[U])^2/{\rm Var}^2_{Q_u}(U)$ is bounded almost surely, and
the convergence to zero holds by the dominated convergence theorem, along with
\ref{assump:Conditional mean} that implies that $\sup_{(u,s)}|\hat{E}_n(u,s)-\bar{E}(u,s)|$ is bounded in probability and
$\mathbb{E}\{|\hat{E}_n(u)-\bar{E}(u)|\} = o(n^{-1/4})$ for each $u$. Therefore the second quantity of \eqref{eq:term2.upper_bound} converges to zero in probability. The last quantity of \eqref{eq:term2.upper_bound} converges to zero in probability, following \ref{assump:Covariates} so that $(U-Q_u[U])/{\rm Var}_{Q_u}(U)$ is bounded almost surely and using $(\Qn-Q_u)[\bar{E}(U)] \crarrow{p} 0$.

The quantity $(iii)$ of \eqref{eq:target_convergence.1} is
\begin{align} \label{eq:term3.upper_bound}
  & \lesssim P\bigg[\bigg\{\bigg[\frac{(U-\mathbb{Q}_n[U])^2}{{\rm Var}^2_{\mathbb{Q}_n}(U)}-\frac{(U-Q_u[U])^2}{{\rm Var}_{Q_u}^2(U)}\bigg]{\rm Cov}_{Q_u}(U, \bar{E}(U))\bigg\}^2\bigg] \\
  &\quad + P\bigg[\bigg\{\frac{(U-\mathbb{Q}_n[U])^2}{{\rm Var}^2_{\mathbb{Q}_n}(U)}\Big[ {\rm Cov}_{\mathbb{Q}_{n}}(U, \hat{E}_n(U))-{\rm Cov}_{Q_u}(U, \bar{E}(U))\Big]\bigg\}^2\bigg]. \nonumber
\end{align}
The first quantity of \eqref{eq:term3.upper_bound} converges to zero in probability by \eqref{eq:convergence_U_func} and
that ${\rm Cov}_{Q_u}(U, \bar{E}(U))$ is bounded (implied by \ref{assump:Covariates} and \ref{assump:Conditional mean}). The second quantity of \eqref{eq:term3.upper_bound}
goes to zero in probability, following \ref{assump:Covariates} that implies $(U-\mathbb{Q}_n[U])/{\rm Var}_{\mathbb{Q}_n}(U)$ is bounded almost surely and
\begin{align*} %\label{eq:convergence_cov}
  &\big|{\rm Cov}_{\mathbb{Q}_{n}}(U, \hat{E}_n(U)) - {\rm Cov}_{Q_u}(U, \bar{E}(U))\big|
  \leq \big|{\rm Cov}_{\mathbb{Q}_{n}}(U, \hat{E}_n(U))-{\rm Cov}_{\mathbb{Q}_{n}}(U, \bar{E}(U))\big|\\
  &\quad + \big|{\rm Cov}_{\mathbb{Q}_{n}}(U, \bar{E}(U)) - {\rm Cov}_{Q_u}(U, \bar{E}(U))\big| = o_p(1).  
\end{align*}
The above display results from the first term on the right-hand-side converging to zero in probability by $\mathbb{E}\{|\hat{E}_n(u) - \bar{E}(u)|\} = o(n^{-1/4})$ for each $u$ in \ref{assump:Conditional mean} along with \ref{assump:Covariates},
and the second term converging to zero in probability by the law of large numbers.

The quantity $(iv)$ of \eqref{eq:target_convergence.1} is
\begin{align} \label{eq:equiv_convergence_condition.11}
   \lesssim & \; P\bigg[\left\{ \left(\frac{(U-\mathbb{Q}_n[U])}{{\rm Var}_{\mathbb{Q}_n}(U)}
  - \frac{(U-Q_u[U])}{{\rm Var}_{Q_u}(U)}\right)\int_{\mathcal{T}}\Big(\hat{E}_n(U,s)-\bar{E}(U,s)\Big)\tilde{M}(U,ds) \right\}^2 \bigg] \\   
  & + P\bigg[\left\{ \left(\frac{(U-\mathbb{Q}_n[U])}{{\rm Var}_{\mathbb{Q}_n}(U)}
  - \frac{(U-Q_u[U])}{{\rm Var}_{Q_u}(U)}\right)\int_{\mathcal{T}}\bar{E}(U,s)\tilde{M}(U,ds) \right\}^2 \bigg]. \nonumber
\end{align}
To show that the quantity $(iv)$ of \eqref{eq:target_convergence.1} converges to zero in probability, it suffices to give the convergence to zero in probability of this upper bound in \eqref{eq:equiv_convergence_condition.11}. Note that $\tilde{M}(U,ds) = 1(X \in ds, \delta=0)-1(X \ge s)d\tilde{\Lambda}(s|U)$ and
\begin{align} \label{eq:upper_bound_martint}
  &\Big|\int_{\mathcal{T}}\Big(\hat{E}_n(U,s)-\bar{E}(U,s)\Big)\tilde{M}(U,ds)\Big|^2
  \le \sup_{(u,s)}\Big|\hat{E}_n(u,s)-\bar{E}(u,s)\Big|^2\sup_u\big(1+\tilde{\Lambda}(\tau|u)\big)^2 \\
  &\lesssim \sup_{(u,s)}\big|\hat{E}_n(u,s)-\bar{E}(u,s)\big|^2, \nonumber
\end{align}
so that
\begin{align*}
  &E\bigg\{P\bigg[\left\{\left(\frac{(U-\mathbb{Q}_n[U])}{{\rm Var}_{\mathbb{Q}_n}(U)}
  - \frac{(U-Q_u[U])}{{\rm Var}_{Q_u}(U)}\right)\int_{\mathcal{T}}\Big(\hat{E}_n(U,s)-\bar{E}(U,s)\Big)\tilde{M}(U,ds)\right\}^2\bigg]\bigg\} \\
  &= \mathbb{E}\bigg\{P\bigg[\left\{\left(\frac{(U-\mathbb{Q}_n[U])}{{\rm Var}_{\mathbb{Q}_n}(U)}
  - \frac{(U-Q_u[U])}{{\rm Var}_{Q_u}(U)}\right)\int_{\mathcal{T}}\Big(\hat{E}_n(U,s)-\bar{E}(U,s)\Big)\tilde{M}(U,ds)\right\}^2\bigg]\bigg\} \\
  &\lesssim P\bigg[\mathbb{E}\bigg\{\left(\frac{(U-\mathbb{Q}_n[U])}{{\rm Var}_{\mathbb{Q}_n}(U)} - \frac{(U-Q_u[U])}{{\rm Var}_{Q_u}(U)}\right)^2\sup_{(u,s)}\big|\hat{E}_n(u,s)-\bar{E}(u,s)\big|^2\bigg\}\bigg] \to 0 \; \mbox{ a.s.},
\end{align*}
following the dominated convergence theorem, along with \eqref{eq:convergence_U_func}, \ref{assump:Covariates} that $U$ is bounded and non-degenerate, and \ref{assump:Conditional mean} that $\sup_{(u,s)}|\hat{E}_n(u,s)-\bar{E}(u,s)|$ is bounded in probability.
Therefore, we show that the first term in \eqref{eq:equiv_convergence_condition.11} converges to zero in probability. We continue dealing with the second term in \eqref{eq:equiv_convergence_condition.11}. Similarly, we first upper-bounds 
$|\int_{\mathcal{T}}\bar{E}(U,s)\tilde{M}(U,ds)|^2$ by $\sup_{(u,s)}|\bar{E}(u,s)|^2\sup_u(1+\tilde{\Lambda}(\tau|u))^2$, which further gives that
\begin{align*}
  &E\bigg\{P\bigg[\left\{\left(\frac{(U-\mathbb{Q}_n[U])}{{\rm Var}_{\mathbb{Q}_n}(U)}
  - \frac{(U-Q_u[U])}{{\rm Var}_{Q_u}(U)}\right)\int_{\mathcal{T}}\bar{E}(U,s)\tilde{M}(U,ds)\right\}^2\bigg]\bigg\} \\
  &\quad =\mathbb{E}\bigg\{P\bigg[\left\{\left(\frac{(U-\mathbb{Q}_n[U])}{{\rm Var}_{\mathbb{Q}_n}(U)}
  - \frac{(U-Q_u[U])}{{\rm Var}_{Q_u}(U)}\right)\int_{\mathcal{T}}\bar{E}(U,s)\tilde{M}(U,ds)\right\}^2\bigg]\bigg\} \\
  &\quad \lesssim P\bigg[\mathbb{E}\bigg\{\left(\frac{(U-\mathbb{Q}_n[U])}{{\rm Var}_{\mathbb{Q}_n}(U)} - \frac{(U-Q_u[U])}{{\rm Var}_{Q_u}(U)}\right)^2\sup_{(u,s)}\big|\bar{E}(u,s)\big|^2\bigg\}\bigg] \to 0 \; \mbox{ a.s.},
\end{align*}
following the analogous arguments to those for the first term in \eqref{eq:equiv_convergence_condition.11}, along with \ref{assump:Conditional mean} that $\bar{E}$ is uniformly bounded.
Therefore, we show that the quantity $(iv)$ of \eqref{eq:target_convergence.1} converges to zero in probability.

Now we deal with the last quantity $(v)$ of \eqref{eq:target_convergence.1}.
Applying \eqref{eq:upper_bound_martint} and the above arguments, 
\begin{align*}
  &E\bigg\{P\bigg[\left\{\frac{(U-Q_u[U])}{{\rm Var}_{Q_u}(U)}\int_{\mathcal{T}}\Big(\hat{E}_n(U,s)-\bar{E}(U,s)\Big)\tilde{M}(U,ds)\right\}^2\bigg]\bigg\} \\
  &\quad = \mathbb{E}\bigg\{P\bigg[\left\{\frac{(U-Q_u[U])}{{\rm Var}_{Q_u}(U)}\int_{\mathcal{T}}\Big(\hat{E}_n(U,s)-\bar{E}(U,s)\Big)\tilde{M}(U,ds)\right\}^2\bigg]\bigg\} \\
  &\quad \lesssim P\bigg[\left(\frac{(U-Q_u[U])}{{\rm Var}_{Q_u}(U)}\right)^2\mathbb{E}\Big\{\sup_{(u,s)}\big|\hat{E}_n(u,s)-\bar{E}(u,s)\big|^2\Big\}\bigg] \to 0,
\end{align*}
where the convergence follows the dominated convergence theorem, \ref{assump:Covariates} that $U$ is bounded and non-degenerate, and
\ref{assump:Conditional mean} that $\sup_{(u,s)}|\hat{E}_n(u,s)-\bar{E}(u,s)|$ is bounded in probability and $\mathbb{E}\{|\hat{E}_n(u,s)-\bar{E}(u,s)|\}= o(n^{-1/4})$ for each $(u,s)$. Hence, we conclude the proof. 
\end{proof}

\begin{lemma} \label{lemma:L2_conv.2}
Suppose that \ref{assump:Covariates}, \ref{assump:Survival function}, \ref{assump:At-risk prob_covariate} and \ref{assump:Conditional mean} hold. Let $\mathcal{H}_0$ be as defined above and $\IF^*$ be as defined in the main text,
$\sup_{(\tilde{E},\tilde{Q}) \in \mathcal{H}_0} P[(\IF^*(\cdot|\tilde{E},\tilde{Q},\hat{G}_n)-\IF^*(\cdot|\tilde{E},\tilde{Q},G))^2] \lcrarrow{p} 0.$
\end{lemma}
\begin{proof}
For any $(\tilde{E},\tilde{Q}) \in \mathcal{H}_0$,
\begin{align} \label{eq:target_convergence.2}
  & P[\IF^*(\cdot|\tilde{E},\tilde{Q},\hat{G}_n)-\IF^*(\cdot|\tilde{E},\tilde{Q},G)]^2\\
  & \leq 2\, \bigg\{ P\bigg[ \bigg\{\frac{(U-\tilde{Q}[U])\delta X}{{\rm Var}_{\tilde{Q}}(U)}\left(\frac{1}{\hat{G}_n(X|U)}-\frac{1}{G(X)}\right)\bigg\}^2 \bigg] \nonumber \\
  & \hspace{1cm} + P\bigg[ \bigg\{\frac{(U-\tilde{Q}[U])}{{\rm Var}_{\tilde{Q}}(U)} \int_{\mathcal{T}}\tilde{E}(U,s)\Big(\hat{M}(U,ds)-M(U,ds)\Big) \bigg\}^2\bigg]  \bigg\}, \nonumber 
\end{align}
where $\hat{M}(U,\cdot)$ is the martingale residual with $\hat{\Lambda}_n(\cdot|U)$ corresponding to $\hat{G}_n(\cdot|U)$.
Applying Taylor expansion on 
$(1/\hat{G}_n-1/G)$, the first term on the right-hand-side of \eqref{eq:target_convergence.2} is dominated by 
\begin{align*}
  &P\bigg[\bigg\{\frac{(U-\tilde{Q}[U])\delta X}{{\rm Var}_{\tilde{Q}}(U)}\left(\frac{1}{G^2(X)}\Big(\hat{G}_n(X|U)-G(X)\Big)\right)\bigg\}^2\bigg]\\
  &\quad \leq \sup_{(u,s) \in \mathbb{R} \times \mathcal{T}}\Big(\hat{G}_n(s|u)-G(s)\Big)^2 P\bigg[\bigg\{\frac{(U-\tilde{Q}[U])\delta X}{{G^2(X)\rm Var}_{\tilde{Q}}(U)} \bigg\}^2\bigg] = o_p(1)
\end{align*}
because $(U-\tilde{Q}[U])\delta X / {\rm Var}_{\tilde{Q}}(U) = O_p(1)$ that is implied by \ref{assump:Covariates} and $|X| \le \tau$; $G(\tau) > 0$ as stated in \ref{assump:Survival function} and the uniform convergence of $\hat{G}_n$.

Then we show that the second term on the right-hand-side of \eqref{eq:target_convergence.2} converges to zero in probability.
As we have seen that $(U-\tilde{Q}[U])/{\rm Var}_{\tilde{Q}}(U) = O_p(1)$ by \ref{assump:Covariates}, it suffices to show that 
\begin{align}\label{eq:equiv_convergence_condition.21}  
  P\bigg[\bigg\{\int_{\mathcal{T}}\tilde{E}(U,s)\Big(\hat{M}(U,ds)-M(U,ds)\Big)\bigg\}^2\bigg] \lcrarrow{p} 0.
\end{align} 
The decomposition $\hat{M}(u,ds) = M(u,ds) + 1(X \geq s)(d\Lambda(s) - d\hat{\Lambda}_n(s|u))$ further reduces proving \eqref{eq:equiv_convergence_condition.21} to showing that
\begin{align}\label{eq:equiv_convergence_condition.22}
  %&P\left[ \int_{\mathcal{T}}\bigg(\hat{E}_n(U,s)- \bar{E}(U,s)\bigg)\{M(U,ds) + 1(X \geq s)(d\Lambda(s) - d\hat{\Lambda}_n(s|U))\} \right]^2 \lcrarrow{p} 0\,;\;\\
  P\bigg[\bigg\{\int_{\mathcal{T}} \tilde{E}(U,s)1(X \geq s)\Big(d\Lambda(s) - d\hat{\Lambda}_n(s|U)\Big) \bigg\}^2\bigg] \lcrarrow{p} 0.
\end{align}
Following $N_n(u,s)$ and $Y_n(u,s)$ as defined in Section \ref{sec:notation}, we easily see that $\bar{M}(u,ds) = N_n(u,ds) - Y_n(u,s) d\Lambda(s)$ is a local martingale with respect to the aggregated filtration that is defined in \eqref{eq:agg_filtration} from Section \ref{sec:notation}.
Note also that
\begin{align} \label{eq:martingale_est_CHF}
  &\hat{\Lambda}_n(t|u)-\Lambda(t) = \int_{-\infty}^{t} \frac{1(Y_n(u,s)>0)}{Y_n(u,s)}\bar{M}(u,ds) + \int_{-\infty}^{t} [1(Y_n(u,s)>0)-1]d\Lambda(s)\\ 
  &\quad = \int_{-\infty}^{t} \frac{1(Y_n(u,s)>0)}{Y_n(u,s)}\bar{M}(u,ds) - \int_{-\infty}^{t} 1(Y_n(u,s)=0)\,d\Lambda(s). \nonumber      
\end{align}
Inserting the decomposition in \eqref{eq:martingale_est_CHF} back to \eqref{eq:equiv_convergence_condition.22}, along with $(a+b)^2 \leq 2\, (a^2+b^2)$, gives that showing \eqref{eq:equiv_convergence_condition.22} is equivalent to showing 
\begin{align} 
  &P\bigg[ \bigg\{\int_{\mathcal{T}}\tilde{E}(U,s)1(X \geq s)\frac{1(Y_n(U,s)>0)}{Y_n(U,s)}\bar{M}(U,ds) \bigg\}^2\bigg] \lcrarrow{p} 0\;;   \label{eq:equiv_convergence_condition.23}\\
  &P\bigg[ \bigg\{\int_{\mathcal{T}}\tilde{E}(U,s)1(X \geq s)1(Y_n(U,s)=0)d\Lambda(s) \bigg\}^2\bigg] \lcrarrow{p} 0.  \label{eq:equiv_convergence_condition.24}
\end{align}
Recall that $\mathbb{E}$ denotes the expectation over $O_1, \ldots, O_n$, regarding $O$ as fixed, in contrast to the expectation $P$ that applies to $O$. Note also that $\mathbb{E}[(\bar{M}(u,ds))^2|\bar{{\cal F}}_{s}]=Y_n(u,s)d\Lambda(s)$.
Then for each $u$, $\tilde{E}(u,s)$ is left-continuous in $s$ and adapted to the filtration $\bar{\mathcal{F}}_{s}$, so
we have the display in \eqref{eq:equiv_convergence_condition.23} by
\begin{align*}
  &E\bigg\{ P\bigg[ \bigg\{\int_{\mathcal{T}} \tilde{E}(U,s)1(X \geq s)\frac{1(Y_n(U,s)>0)}{Y_n(U,s)}\bar{M}(U,ds) \bigg\}^2 \bigg] \bigg\}\\
  &\quad = \mathbb{E} \bigg\{ P\bigg[ \bigg\{\int_{\mathcal{T}} \tilde{E}(U,s)1(X \geq s)\frac{1(Y_n(U,s)>0)}{Y_n(U,s)}\bar{M}(U,ds) \bigg\}^2 \bigg] \bigg\}\\
  &\quad = P\bigg[ \mathbb{E} \bigg\{ \bigg[\int_{\mathcal{T}} \tilde{E}(U,s)1(X \geq s)\frac{1(Y_n(U,s)>0)}{Y_n(U,s)}\bar{M}(U,ds) \bigg]^2 \bigg\} \bigg] \\
  &\quad = P \bigg[ \int_{\mathcal{T}} \tilde{E}^2(U,s)1(X \geq s)\mathbb{E}\bigg\{\frac{1(Y_n(U,s)>0)}{Y_n(U,s)} \bigg\}d\Lambda(s) \bigg] \to 0,
\end{align*}
where the convergence to zero in the last line follows that $\tilde{E} \in \mathcal{E}_0$ is uniformly bounded, that $\inf_{(u,s) \,\in \,\mathbb{R} \times \mathcal{T}} Y_n(u,s) \to \infty$, and the dominated convergence theorem. Similarly, the display in \eqref{eq:equiv_convergence_condition.24} is an immediate consequence of the uniform boundedness of
$\tilde{E}$, $1(Y_n(u,s)=0)=o_p(1)$ for each $(u,s)$ and the dominated convergence theorem.
Hence, we conclude this proof.  
\end{proof}

\begin{lemma} \label{lemma:convergence_processes}
Given \ref{assump:Covariates}, \ref{assump:Survival function}, \ref{assump:At-risk prob_covariate} and \ref{assump:Conditional mean} and $\IF^*$ defined in the main text, we have that
\begin{lemmaitemenum}
  \item \label{eq:process_conv.1} $[\mathbb{P}_n-P] [\IF^*(\cdot|\hat{E}_n,\mathbb{Q}_n,G)-\IF^*(\cdot|\bar{E},Q_u,G)] = o_p(n^{-1/2})$;
  \item \label{eq:process_conv.2} $[\mathbb{P}_n-P] [\IF^*(\cdot|\hat{E}_n,\mathbb{Q}_n,\hat{G}_n)-\IF^*(\cdot|\hat{E}_n,\mathbb{Q}_n,G)] = o_p(n^{-1/2})$.
\end{lemmaitemenum}
\end{lemma}
\begin{proof}
As the core of this proof relies on applying Theorem 2.1 of \cite{vanderVaart&Wellner2007}, we first relate our notation to  theirs. We take the functional $\IF^*(\cdot|\tilde{E}, \tilde{Q}, \tilde{G})$ to be  $f_{\theta, \eta}$ in their notation, where $(\theta, \eta)$ could be either $(\tilde{G}, (\tilde{E}, \tilde{Q}))$ or $((\tilde{E}, \tilde{Q}), \tilde{G})$. Further, if we take $\theta \equiv \tilde{G}$ and $\eta \equiv (\tilde{E}, \tilde{Q})$, which corresponds to the situation in \ref{eq:process_conv.1}, we have in their notation $\eta_n \equiv (\hat{E}_n, \mathbb{Q}_n)$, $\eta_0 \equiv (\bar{E}, Q_u)$, $H_0 \equiv {\cal H}_0$, and $\Theta \equiv {\cal G}_0$.
Alternatively, if we take $\theta \equiv (\tilde{E}, \tilde{Q})$ and $\eta \equiv\tilde{G}$, as in \ref{eq:process_conv.2}, then
$\eta_n \equiv \hat{G}_n$, $\eta_0 \equiv G$, $H_0 \equiv {\cal G}_0$, and $\Theta \equiv {\cal H}_0$.  Note that the condition ${\rm P}(\eta_n \in H_0)\to 1$ and the  $P$-Donsker condition of their theorem are satisfied by our Lemma \ref{lemma:conditions_for_Thm2.1}.  The main step is to check their condition (3), namely that $\sup_{\theta \in \Theta} {\rm P}(f_{\theta, \eta_n}- f_{\theta, \eta_0})^2\to_p 0$, in the instances arising here. 

We first show \ref{eq:process_conv.1}.
For any $\epsilon > 0$,
\begin{align*}
  &{\rm P}\Big(\,\Big|\sqrt{n}[\mathbb{P}_n-P] [\IF^*(\cdot|\hat{E}_n,\mathbb{Q}_n,G)-\IF^*(\cdot|\bar{E},Q_u,G)]\Big| > \epsilon \Big)
  \le {\rm P}\Big((\hat{E}_n,\mathbb{Q}_n) \notin \mathcal{H}_0\Big) \\   
  &\quad + {\rm P}\Big(G \notin \mathcal{G}_0\Big) + {\rm P}\Big((\bar{E},Q_u) \notin \mathcal{H}_0\Big) \\
  &\quad + {\rm P}\bigg(\,\sup_{\tilde{G} \in \mathcal{G}_0}\Big|\sqrt{n}[\mathbb{P}_n-P] [\IF^*(\cdot|\hat{E}_n,\mathbb{Q}_n, \tilde{G})-\IF^*(\cdot|\bar{E},Q_u,\tilde{G})]\Big| > \epsilon \bigg) \to 0,
\end{align*}
where the first probability on the right-hand-side goes to zero by \ref{eq:estimators_in_classes} of Lemma \ref{lemma:conditions_for_Thm2.1}, the second and third probability are trivially zero by the definitions of ${\cal G}_0$ and ${\cal H}_0$, and the last probability converges to zero by checking their condition (3) using Lemma \ref{lemma:L2_conv.1}.

Similarly, \ref{eq:process_conv.2} holds by
\begin{align*}
  &{\rm P}\Big(\,\Big|\sqrt{n}[\mathbb{P}_n-P] [\IF^*(\cdot|\hat{E}_n,\mathbb{Q}_n,\hat{G}_n)-\IF^*(\cdot|\hat{E}_n,\mathbb{Q}_n,G)]\Big| > \epsilon \Big)\\
  & \le {\rm P}\Big((\hat{E}_n,\mathbb{Q}_n) \notin \mathcal{H}_0\Big) + {\rm P}\Big( \hat{G}_n \notin \mathcal{G}_0 \Big) + {\rm P}\Big(G \notin \mathcal{G}_0\Big)\\
  & \quad + {\rm P}\bigg(\,\sup_{(\tilde{E},\tilde{Q}) \in \mathcal{H}_0}\Big|\sqrt{n}[\mathbb{P}_n-P] [\IF^*(\cdot|\tilde{E},\tilde{Q},\hat{G}_n)-\IF^*(\cdot|\tilde{E},\tilde{Q},G)]\Big| > \epsilon \bigg) \to 0,
\end{align*}
where the first two probabilities on the right-hand-side converge to zero by \ref{eq:estimators_in_classes} of Lemma \ref{lemma:conditions_for_Thm2.1}, the third probability is obviously zero by the definition of $\mathcal{G}_0$, and the last probability converges to zero by checking the condition (3) in this instance using Lemma \ref{lemma:L2_conv.2}
\end{proof}

Before we proceed with the next lemma, we list some properties that will be repeatedly used later:
\begin{align}
  & \label{eq:prelim.1} \frac{(U-\Qn[U])}{{\rm Var}_{\Qn}(U)}-\frac{(U-Q_u[U])}{{\rm Var}_{Q_u}(U)}
  = \frac{(Q_u-\Qn)[U]}{{\rm Var}_{\Qn}(U)} + (U-Q_u[U])\Big[\frac{1}{{\rm Var}_{\Qn}(U)}-\frac{1}{{\rm Var}_{Q_u}(U)}\Big] \\
  & \quad = o_p(1); \nonumber \\
  & \label{eq:prelim.2} \sqrt{n}\bigg\{\frac{(U-\Qn[U])}{{\rm Var}_{\Qn}(U)}-\frac{(U-Q_u[U])}{{\rm Var}_{Q_u}(U)}\bigg\} \\
  &\quad = \frac{\sqrt{n}(Q_u-\Qn)[U]}{{\rm Var}_{\Qn}(U)} + (U-Q_u[U])\sqrt{n}\Big[\frac{1}{{\rm Var}_{\Qn}(U)}-\frac{1}{{\rm Var}_{Q_u}(U)}\Big]
  = O_p(1), \nonumber  
\end{align}
which follow empirical process theories along with \ref{assump:Covariates} that $U$ is bounded and non-degenerate. In addition, we observe that $Y_n(u,\tau) \sim \mbox{Binomial}(n, p_*)$ with $p_*={\rm P}(X \geq \tau, c(U)=c(u))$ that is positive by \ref{assump:At-risk prob_covariate}.
Along with the monotonicity of $Y_n(u,s)$ in $s \in \mathcal{T}$, Hoeffding's inequality gives that as $n \to \infty$, ${\rm P}(\inf_{s \in \mathcal{T}}Y_n(u,s) \leq \sqrt{n}) = {\rm P}(Y_n(u,\tau) \leq \sqrt{n}) \leq \exp(-2(\sqrt{n}p_*-1)^2) \rightarrow 0$.
Therefore, we have that as $n \to \infty$,
\begin{align}
 &\label{eq:Yn_larger_than_rootn} {\rm P}(\inf_{s \in \mathcal{T}}Y_n(u,s) > \sqrt{n}) \rightarrow 1 \mbox{ for each } u\,;\\
 &\label{eq:conv_prob_Yn} \sqrt{n}{\rm P}\big(Y_n(u,s)=0 \big) \le 
 \sqrt{n}{\rm P}\big(Y_n(u,\tau)=0\big) = \sqrt{n}(1-p_*)^n \rightarrow 0,
 \mbox{ for each } (u,s).
\end{align}

\begin{lemma} \label{lemma:remainder}
Given \ref{assump:Covariates}, \ref{assump:Survival function}, \ref{assump:At-risk prob_covariate}, \ref{assump:Conditional mean} and $\IF^*$ as defined in the main text,
\begin{align*}
  &P[\IF^*(\cdot|\hat{E}_n,\mathbb{Q}_{n},\hat{G}_n)- \IF^*(\cdot|\bar{E},Q_u,\hat{G}_n)  + \IF^*(\cdot|\bar{E},Q_u,G)-\IF^*(\cdot|\hat{E}_n,\mathbb{Q}_{n},G)]\\
  &\quad =o_p(n^{-1/2}).
\end{align*}
\end{lemma}
\begin{proof}
Observe that 
\begin{align*}
  P[\IF^*(\cdot|\hat{E}_n,\mathbb{Q}_{n},\hat{G}_n)] = &\frac{P[(U-\mathbb{Q}_{n}[U])(Y-\Qn[\hat{E}_n(U)])]}{{\rm Var}_{\mathbb{Q}_{n}}(U)}-\frac{{\rm Cov}_{\mathbb{Q}_{n}}(U, \hat{E}_n(U))}{ {\rm Var}_{\mathbb{Q}_{n}}(U)}\\
  & - P[\IF^{CAR}(\cdot|\hat{E}_n,\mathbb{Q}_{n},\hat{G}_n)];\\
  P[\IF^*(\cdot|\hat{E}_n,\mathbb{Q}_{n},G)] = &\frac{P[(U-\mathbb{Q}_{n}[U])(\tilde{Y}-\Qn[\hat{E}_n(U)])]}{{\rm Var}_{\mathbb{Q}_{n}}(U)}-\frac{{\rm Cov}_{\mathbb{Q}_{n}}(U, \hat{E}_n(U))}{ {\rm Var}_{\mathbb{Q}_{n}}(U)}\\
  &-P[\IF^{CAR}(\cdot|\hat{E}_n,\mathbb{Q}_{n},G)]; \\
  P[\IF^*(\cdot|\bar{E},Q_{u},\hat{G}_n)]=&\frac{P[(U-Q_u[U])(Y-Q_u[\bar{E}(U)])]}{{\rm Var}_{Q_u}(U)}-\frac{{\rm Cov}_{Q_u}(U, \bar{E}(U))}{ {\rm Var}_{Q_u}(U)}\\
  &-P[\IF^{CAR}(\cdot|\bar{E},Q_{u},\hat{G}_n)];\\
  P[\IF^*(\cdot|\bar{E},Q_{u},G)]
   =&\frac{P[(U-Q_u[U])(\tilde{Y}-Q_u[\bar{E}(U)])]}{{\rm Var}_{Q_u}(U)}-\frac{{\rm Cov}_{Q_u}(U, \bar{E}(U))}{ {\rm Var}_{Q_u}(U)}\\
   &-P[\IF^{CAR}(\cdot|\bar{E},Q_{u},G)],
\end{align*}
which implies that
\begin{align} \label{eq:decomp_four_terms}
  &P[\IF^*(\cdot|\hat{E}_n,\mathbb{Q}_{n},\hat{G}_n)- \IF^*(\cdot|\bar{E},Q_u,\hat{G}_n)  + \IF^*(\cdot|\bar{E},Q_u,G)-\IF^*(\cdot|\hat{E}_n,\mathbb{Q}_{n},G)] \\
  & = P\bigg[\bigg\{\frac{(U-\mathbb{Q}_{n}[U])}{{\rm Var}_{\mathbb{Q}_{n}}(U)}-\frac{(U-Q_u[U])}{{\rm Var}_{Q_u}(U)}\bigg\}(Y-\tilde{Y})\bigg] \nonumber\\
  &\quad - P\left[\frac{(U-Q_u[U])}{{\rm Var}_{Q_u}(U)}\int_{\mathcal{T}}\big(\hat{E}_n(U,s)-\bar{E}(U,s)\big)\big[\hat{M}(U,ds)-M(U,ds)\big]\right] \nonumber \\
  &\quad - P\bigg[\bigg[\frac{(U-\Qn[U])}{{\rm Var}_{\Qn}(U)}-\frac{(U-Q_u[U])}{{\rm Var}_{Q_u}(U)}\bigg]\int_{\mathcal{T}}\hat{E}_n(U,s)\big[\hat{M}(U,ds)-M(U,ds)\big]\bigg]. \nonumber
\end{align}

By the decomposition on the left-hand-side of \eqref{eq:prelim.1}, along with $Y = \delta X/\hat{G}_n(X)$ and $\tilde{Y} = \delta X/G(X)$,
we have the first quantity on the right-hand-side of \eqref{eq:decomp_four_terms} as
\begin{align*}
  P\bigg[\bigg\{\frac{(Q_u-\Qn)[U]}{{\rm Var}_{\Qn}(U)} + (U-Q_u[U])\bigg[\frac{1}{{\rm Var}_{\Qn}(U)}-\frac{1}{{\rm Var}_{Q_u}(U)}\bigg]\bigg\}\bigg\{\frac{\delta X}{\hat{G}_n(X)}-\frac{\delta X}{G(X)}\bigg\}\bigg],
\end{align*}
which is $o_p(n^{-1/2})$ by \eqref{eq:prelim.2} and the uniform consistency of $\hat{G}_n$. 

Then we deal with the last two terms on the right-hand-side of \eqref{eq:decomp_four_terms}.
By the decomposition $\hat{M}(U,ds)-M(U,ds) = 1(X \geq s)(d\Lambda(s) - d\hat{\Lambda}_n(s|U))$, the two terms turn into
\begin{align} \label{eq:decomp_four_terms.1}
  &-P\left[\frac{(U-Q_u[U])}{{\rm Var}_{Q_u}(U)}\int_{\mathcal{T}} \big(\hat{E}_n(U,s)-\bar{E}(U,s)\big)1(X \geq s)(d\Lambda(s) - d\hat{\Lambda}_n(s|U))\right] \\
  & - P\left[\bigg[\frac{(U-\Qn[U])}{{\rm Var}_{\Qn}(U)}-\frac{(U-Q_u[U])}{{\rm Var}_{Q_u}(U)}\bigg]\int_{\mathcal{T}} \hat{E}_n(U,s)1(X \geq s)(d\Lambda(s) - d\hat{\Lambda}_n(s|U))\right]. \nonumber
\end{align}
Now we tackle the first term of \eqref{eq:decomp_four_terms.1}, and apply similar techniques to the second term. According to the decomposition in \eqref{eq:martingale_est_CHF}, the first term of \eqref{eq:decomp_four_terms.1} is further expressed as
\begin{align}
  &P\left[\frac{(U-Q_u[U])}{{\rm Var}_{Q_u}(U)}\int_{\mathcal{T}} \big(\hat{E}_n(U,s)-\bar{E}(U,s)\big)1(X \geq s)\frac{1(Y_n(U,s)>0)}{Y_n(U,s)}\bar{M}(U,ds)\right] \label{eq:equiv_convergence_conditions_11}\\
  & - P\left[\frac{(U-Q_u[U])}{{\rm Var}_{Q_u}(U)}\int_{\mathcal{T}}\big(\hat{E}_n(U,s)-\bar{E}(U,s)\big)1(X \geq s)1(Y_n(U,s)=0)\,d\Lambda(s) \right].   \label{eq:equiv_convergence_conditions_12}
\end{align}
The quantity \eqref{eq:equiv_convergence_conditions_11} is $o_p(n^{-1/2})$ as shown in what follows. First by Jensen's inequality, the second moment of the quantity \eqref{eq:equiv_convergence_conditions_11} (multiplied by $\sqrt{n}$) is bounded by
\begin{align} \label{eq:second_moments_upper_bound}
  &E\bigg\{P\bigg[\bigg(\frac{(U-Q_u[U])}{{\rm Var}_{Q_u}(U)}\int_{\mathcal{T}} \sqrt{n}\big(\hat{E}_n(U,s)-\bar{E}(U,s)\big)\frac{1(Y_n(U,s)>0)}{Y_n(U,s)}\bar{M}(U,ds)\bigg)^2\bigg]\bigg\}\\
  & = P\bigg[\frac{(U-Q_u[U])^2}{{\rm Var}_{Q_u}^2(U)}\int_{\mathcal{T}} n\mathbb{E}\Big\{\big(\hat{E}_n(U,s)-\bar{E}(U,s)\big)^2\frac{1(Y_n(U,s)>0)}{Y_n^2(U,s)}\mathbb{E}\big[\big(\bar{M}(U,ds)\big)^2\big|\bar{{\cal F}}_s\big]\Big\}\bigg] \nonumber \\
  & = P\bigg[\frac{(U-Q_u[U])^2}{{\rm Var}_{Q_u}^2(U)}\int_{\mathcal{T}} n\mathbb{E}\Big\{\big(\hat{E}_n(U,s)-\bar{E}(U,s)\big)^2\frac{1(Y_n(U,s)>0)}{Y_n(U,s)}\Big\}d\Lambda(s)\bigg] \nonumber \\
  & \le P\bigg[\frac{(U-Q_u[U])^2}{{\rm Var}_{Q_u}^2(U)}\int_{\mathcal{T}} n\mathbb{E}\Big\{\big(\hat{E}_n(U,s)-\bar{E}(U,s)\big)^2\frac{1}{\sqrt{n}}\Big\}d\Lambda(s)\bigg] \to 0. \nonumber
\end{align}
The second line of \eqref{eq:second_moments_upper_bound} holds by the fact that $\bar{M}(u,ds)$ is a local martingale with respect to the aggregated filtration $\bar{\mathcal{F}}_s$ (defined in \eqref{eq:agg_filtration} from Section \ref{sec:notation}), and that $\hat{E}_n$ and $Y_n$ are predictable with respect to $\bar{\mathcal{F}}_s$ (see in the main text for the details of $\hat{E}_n$). Moreover, the inequality in \eqref{eq:second_moments_upper_bound} holds by $\inf_{s \in {\cal T}}Y_n(u,s) > \sqrt{n}$ with probability tending to one for each $u$ as given in \eqref{eq:Yn_larger_than_rootn}, 
while the final convergence to zero follows \ref{assump:Conditional mean} that $\mathbb{E}\{|\hat{E}_n(u,s)-\bar{E}(u,s)|\}=o(n^{-1/4})$ for each $(u,s)$, along with using the dominated convergence theorem.
Therefore by Chebyshev's inequality, \eqref{eq:second_moments_upper_bound} implies that the quantity \eqref{eq:equiv_convergence_conditions_11} is $o_p(n^{-1/2})$.  

The quantity \eqref{eq:equiv_convergence_conditions_12} is $o_p(n^{-1/2})$ because
\begin{align*}
  &\sqrt{n}E\bigg|-P\bigg[\frac{(U-Q_u[U])}{{\rm Var}_{Q_u}(U)}\int_{\mathcal{T}}\big(\hat{E}_n(U,s)-\bar{E}(U,s)\big)1(X \geq s)1(Y_n(U,s)=0)\,d\Lambda(s) \bigg]\bigg|\\
  & \le \sqrt{n}\mathbb{E}\bigg\{P\left[\frac{|U-Q_u[U]|}{{\rm Var}_{Q_u}(U)}\int_{\mathcal{T}} \big|\hat{E}_n(U,s)-\bar{E}(U,s)\big|1(X \geq s)1(Y_n(U,s)=0)d\Lambda(s) \right]\bigg\}\\
  & = \sqrt{n}P\left[\frac{|U-Q_u[U]|}{{\rm Var}_{Q_u}(U)}\int_{\mathcal{T}} 1(X \geq s)\mathbb{E}\Big\{\big|\hat{E}_n(U,s)-\bar{E}(U,s)\big|1(Y_n(U,s)=0)\Big\}d\Lambda(s) \right]\\
  & \le \sqrt{n}P\left[\frac{|U-Q_u[U]|}{{\rm Var}_{Q_u}(U)}\int_{\mathcal{T}} \sqrt{\mathbb{E}\Big\{\big|\hat{E}_n(U,s)-\bar{E}(U,s)\big|^2\Big\}\mathbb{E}\big\{1(Y_n(U,s)=0)\big\}}d\Lambda(s) \right] \to 0,
\end{align*}
where the last inequality holds by the fact that $\Lambda$ is nondecreasing and $1(X \ge s) \le 1$, and the Cauchy--Schwarz inequality.
The final convergence to zero follows because \ref{assump:Conditional mean} that $\mathbb{E}\{|\hat{E}_n(u,s)-\bar{E}(u,s)|\}=o(n^{-1/4})$ for each $(u,s)$, and by \eqref{eq:conv_prob_Yn} that $$\sqrt{n}\mathbb{E}\{1(Y_n(U,s)=0)\} = \sqrt{n}{\rm P}(Y_n(U,s)=0) \le \sqrt{n}{\rm P}(Y_n(U,\tau)=0) \to 0,$$ together with \ref{assump:Covariates} that $U$ is bounded and non-degenerate and using the dominated convergence theorem. 
By Chebyshev's inequality, along with the above displays, the quantity \eqref{eq:equiv_convergence_conditions_12} is $o_p(n^{-1/2})$.  

Similarly, the second term of \eqref{eq:decomp_four_terms.1} is expressed as
\begin{align}
  &P\left[\bigg[\frac{(U-\Qn[U])}{{\rm Var}_{\Qn}(U)}-\frac{(U-Q_u[U])}{{\rm Var}_{Q_u}(U)}\bigg]\int_{\mathcal{T}} \hat{E}_n(U,s)1(X \geq s)\frac{1(Y_n(U,s)>0)}{Y_n(U,s)}\bar{M}(U,ds)\right] \label{eq:equiv_convergence_conditions_21}\\
  & - P\left[\bigg[\frac{(U-\Qn[U])}{{\rm Var}_{\Qn}(U)}-\frac{(U-Q_u[U])}{{\rm Var}_{Q_u}(U)}\bigg]\int_{\mathcal{T}}\hat{E}_n(U,s)1(X \geq s)1(Y_n(U,s)=0)\,d\Lambda(s) \right].   \label{eq:equiv_convergence_conditions_22}
\end{align}
The quantity \eqref{eq:equiv_convergence_conditions_21} is $o_p(n^{-1/2})$, applying
similar arguments to those used for \eqref{eq:equiv_convergence_conditions_11}. 
It therefore suffices to show that the second moment of the quantity \eqref{eq:equiv_convergence_conditions_21} (multiplied by $\sqrt{n}$) converges to zero in probability. 
As in \eqref{eq:second_moments_upper_bound}, we see the second moment of the quantity \eqref{eq:equiv_convergence_conditions_21} (multiplied by $\sqrt{n}$) is bounded by
\begin{align*}
  &E\bigg\{nP\bigg[\left(\bigg[\frac{(U-\Qn[U])}{{\rm Var}_{\Qn}(U)}-\frac{(U-Q_u[U])}{{\rm Var}_{Q_u}(U)}\bigg]\int_{\mathcal{T}} \hat{E}_n(U,s)\frac{1(Y_n(U,s)>0)}{Y_n(U,s)}\bar{M}(U,ds)\right)^2\bigg]\bigg\}\\
  & = P\bigg[n\mathbb{E}\bigg\{\bigg[\frac{(U-\Qn[U])}{{\rm Var}_{\Qn}(U)}-\frac{(U-Q_u[U])}{{\rm Var}_{Q_u}(U)}\bigg]^2\int_{\mathcal{T}} \hat{E}^2_n(U,s)\frac{1(Y_n(U,s)>0)}{Y_n(U,s)}d\Lambda(s)\bigg\}\bigg] \to 0,
\end{align*}
where the convergence to zero follows \eqref{eq:prelim.2}, that $\hat{E}_n(u,s)$ is bounded in probability for each $(u,s)$ (see in the main text for the details of $\hat{E}_n$), and $\inf_{(u,s) \,\in \,\mathbb{R} \times \mathcal{T}} Y_n(u,s) \to \infty$. Thus, we have that the quantity \eqref{eq:equiv_convergence_conditions_21} is $o_p(n^{-1/2})$, by Chebyshev's inequality.  
Analogously, we see the quantity \eqref{eq:equiv_convergence_conditions_22} is $o_p(n^{-1/2})$ as follows. The expectation of the absolute value of this quantity (multiplied by $\sqrt{n}$) is bounded by
\begin{align*}
  \sqrt{n}E\bigg\{P\bigg[\,\bigg|\frac{(U-\Qn[U])}{{\rm Var}_{\Qn}(U)}-\frac{(U-Q_u[U])}{{\rm Var}_{Q_u}(U)}\bigg|\int_{\mathcal{T}}\big|\hat{E}_n(U,s)\big|1(Y_n(U,s)=0)\,d\Lambda(s) \bigg]\bigg\} \to 0,
  %& = \sqrt{n}P\bigg[\,\bigg|\frac{(U-\Qn[U])}{{\rm Var}_{\Qn}(U)}-\frac{(U-Q_u[U])}{{\rm Var}_{Q_u}(U)}\bigg|\int_{\mathcal{T}} \mathbb{E}\Big\{\big|\hat{E}_n(U,s)\big|1(Y_n(U,s)=0)\Big\}d\Lambda(s) \bigg] \to 0,
\end{align*}
using \eqref{eq:prelim.2}, that $\hat{E}_n(u,s)$ is bounded in probability and as implied by \eqref{eq:conv_prob_Yn} that $1(Y_n(u,s)=0) = o_p(n^{-1/2})$, and the dominated convergence theorem. Along with Chebyshev's inequality, this gives the quantity \eqref{eq:equiv_convergence_conditions_22} is $o_p(n^{-1/2})$.
Hence, we complete the proof.
\end{proof}

\begin{lemma}  \label{lemma:if_SnG}
Let $\hat{P}_n' = (\hat{E}_n, \mathbb{Q}_{n}, G)$, $\IF^*$ and $\IF^{\dagger}$ be as respectively defined in the main text. Given \ref{assump:Covariates}, \ref{assump:Survival function}, \ref{assump:At-risk prob_covariate} and \ref{assump:Conditional mean},
\begin{align*}
    S(\Pn, \hat{P}'_n)-\Psi(P) = [\Pn-P]\big\{\IF^*(\cdot|\bar{E},Q_u,G)+\IF^{\dagger}(\cdot|\bar{E},P)\big\} + o_p(n^{-1/2}).
\end{align*}
\end{lemma}
\begin{proof}
\begin{align*}
  &S(\Pn, \hat{P}'_n)-\Psi(P) = \Psi(\hat{P}'_n) + \Pn \IF^*(\cdot|\hat{P}'_n) - \Psi(P) \\
  &\hspace{0.1cm} = [\Pn-P] \IF^*(\cdot|\hat{P}'_n) + \left[\Psi(\hat{P}'_n)-\Psi(P) + P\IF^*(\cdot|\hat{P}'_n)\right] \\
  &\hspace{0.1cm} = [\Pn-P] \IF^*(\cdot|\hat{P}'_n) + [\Pn-P]\IF^{\dagger}(\cdot|\bar{E},P) + o_p(n^{-1/2})\\
  &\hspace{0.1cm} = [\Pn-P]\big\{\IF^*(\cdot|\bar{E},Q_u,G)+\IF^{\dagger}(\cdot|\bar{E},P)\big\} + [\Pn-P] [\IF^*(\cdot|\hat{P}'_n)-\IF^*(\cdot|\bar{E},Q_u,G)]\\
  &\hspace{0.5cm} + o_p(n^{-1/2})\\
  &\hspace{0.1cm} = [\Pn-P]\big\{\IF^*(\cdot|\bar{E},Q_u,G)+\IF^{\dagger}(\cdot|\bar{E},P)\big\} + o_p(n^{-1/2}), 
\end{align*}
where the last equality follows by \ref{eq:process_conv.1} of Lemma \ref{lemma:convergence_processes}. The third equality is shown below.

Recalling the definition in \eqref{eq:reexpressed_if} of the main text,
\begin{align*} 
  &\Psi(\hat{P}'_n)-\Psi(P)+P\IF^*(\cdot|\hat{P}'_n)=\frac{{\rm Cov}_{\mathbb{Q}_{n}}(U, \hat E_n(U))}{{\rm Var}_{\mathbb{Q}_{n}}(U)}-\frac{{\rm Cov}_{Q_u}(U,E[\tilde{Y}|U])}{{\rm Var}_{Q_u}(U)}\\
  & +\frac{1}{{\rm Var}_{\mathbb{Q}_{n}}(U)}P(U-\mathbb{Q}_{n}[U])\big(\delta X/G(X)-\mathbb{Q}_{n}[\hat{E}_n(U)]\big) \nonumber \\
  & -\frac{1}{{\rm Var}_{\mathbb{Q}_{n}}^2(U)}
  P(U-\mathbb{Q}_{n}[U])^2{\rm Cov}_{\mathbb{Q}_{n}}(U, \hat{E}_n(U))
  -P\left[\frac{(U-\mathbb{Q}_{n}[U])}{{\rm Var}_{\mathbb{Q}_{n}}(U)}\int_{\mathcal{T}} \hat{E}_n(U,s)dM(s)\right].\nonumber
\end{align*}
Inserting $\tilde{Y}=\delta X/G(X)$, $\mathbb{Q}_{n}[U] = n^{-1}\sum_{i=1}^nU_i$ and $\mathbb{Q}_{n}[\hat{E}_n(U)]=n^{-1}\sum_{i=1}^n\hat{E}_n(U_i)$ back into the third term of the above display implies that 
\begin{align*}
  & P(U-\mathbb{Q}_{n}[U])\big(\delta X/G(X)-\mathbb{Q}_{n}[\hat{E}_n(U)]\big)\\
  & = P\Big[\Big\{\Big(U-Q_u[U]\Big)+\Big(Q_u[U]-\frac{1}{n}\sum_{i=1}^nU_i\Big)\Big\}\tilde{Y}\Big]-\frac{1}{n}\sum_{i=1}^nP\Big[\Big(U-\frac{1}{n}\sum_{i=1}^nU_i\Big)\hat{E}_n(U_i)\Big].
\end{align*}
Recall also that $P$ denotes the expectation that applies only to $O \sim P$ and not to any estimator composed by $\{O_1,\ldots,O_n\}$ and let $\tilde{Y} = \tilde{Y}_1$ and $U = U_1$ without loss of generality, so we therefore see that
\begin{align*}
  &P\Big[\Big(Q_u[U]-\frac{1}{n}\sum_{i=1}^nU_i\Big)\tilde{Y}\Big] = Q_u[U]P[\tilde{Y}] - \frac{1}{n}P[U_1\tilde{Y}_1]  - \frac{1}{n}\sum_{i\not=1}^n Q_u[U_i]P[\tilde{Y}_1]\\ 
  & = \frac{1}{n}Q_u[U]P[\tilde{Y}] - \frac{1}{n}P[U\tilde{Y}],
\end{align*}
and
\begin{align*}
  &\frac{1}{n}\sum_{i=1}^nP\Big[\Big(U-\frac{1}{n}\sum_{i=1}^nU_i\Big)\hat{E}_n(U_i)\Big]  
  = \frac{1}{n}P[U_1\hat{E}_n(U_1)] + \frac{1}{n}\sum_{i\not=1}^n Q_u[U_1]P[\hat{E}_n(U_i)]\\
  & \quad - \frac{1}{n^2}\sum_{i=1}^nP[U_i\hat{E}_n(U_i)] - \frac{1}{n^2}\sum_{\{i:\,i=1,\ldots,n;\, i\not=j\}} \sum_{j=1}^nQ_u[U_i]P[\hat{E}_n(U_j)]=0.
\end{align*}
Hence combining the results in the above three displays, we have that
\begin{align*}
  & \frac{1}{{\rm Var}_{\mathbb{Q}_{n}}(U)}P(U-\mathbb{Q}_{n}[U])\big(\delta X/G(X)-\mathbb{Q}_{n}[\hat{E}_n(U)]\big)\\
  & = \frac{1}{{\rm Var}_{\mathbb{Q}_{n}}(U)}\bigg\{P\big[(U-Q_u[U])\tilde{Y}\big] + \frac{1}{n}Q_u[U]P[\tilde{Y}] - \frac{1}{n}P[U\tilde{Y}]\bigg\}\\
  &=\frac{{\rm Cov}_{Q_u}
  (U,E[\tilde{Y}|U])}{{\rm Var}_{\mathbb{Q}_{n}}(U)} + o_p(n^{-1/2}),
\end{align*}
following \ref{assump:Covariates} and \ref{assump:Survival function} so that
$Q_u[U]$, $P[\tilde{Y}]$ and $P[U\tilde{Y}]$ are bounded, and ${\rm Var}_{\mathbb{Q}_{n}}(U)$ is bounded away from zero almost surely.

Let
\begin{align*}  
  (i) & \equiv \frac{{\rm Cov}_{\mathbb{Q}_{n}}(U, \hat E_n(U))}{{\rm Var}_{\mathbb{Q}_{n}}(U)};\; 
  (ii) \equiv -\frac{{\rm Cov}_{Q_u}(U,E[\tilde{Y}|U])}{{\rm Var}_{Q_u}(U)}; \; (iii) \equiv \frac{{\rm Cov}_{Q_u}
  (U,E[\tilde{Y}|U])}{{\rm Var}_{\mathbb{Q}_{n}}(U)}; \\ 
  (iv) & \equiv -\frac{1}{{\rm Var}_{\mathbb{Q}_{n}}^2(U)}
  P(U-\mathbb{Q}_{n}[U])^2{\rm Cov}_{\mathbb{Q}_{n}}(U, \hat{E}_n(U)); \nonumber \\
  (v) &\equiv -P\left[\frac{(U-\mathbb{Q}_{n}[U])}{{\rm Var}_{\mathbb{Q}_{n}}(U)}\int_{\mathcal{T}} \hat{E}_n(U,s)dM(s)\right]. \nonumber
\end{align*}
so we have 
\begin{align} \label{eq:remainder_decomp}
  \Psi(\hat{P}'_n)-\Psi(P)+P\IF^*(\cdot|\hat{P}'_n) = (i) + (ii) + (iii) + (iv) + (v).    
\end{align}
The quantity $(iv)$ could be simplified as
\begin{align*}
  (iv) &= -\frac{1}{{\rm Var}_{\mathbb{Q}_{n}}^2(U)}
  P(U-Q_u[U]+Q_u[U]-\mathbb{Q}_{n}[U])^2{\rm Cov}_{\mathbb{Q}_{n}}(U, \hat{E}_n(U))\\
  & = -\frac{{\rm Cov}_{\mathbb{Q}_{n}}(U, \hat{E}_n(U))}{{\rm Var}_{\mathbb{Q}_{n}}^2(U)}\bigg[{\rm Var}_{Q_u}(U)+
  P(Q_u[U]-\mathbb{Q}_{n}[U])^2 - \frac{2}{n}{\rm Var}_{Q_u}(U)\bigg]\\ 
  & = -\frac{{\rm Cov}_{\mathbb{Q}_{n}}(U, \hat{E}_n(U))}{{\rm Var}_{\mathbb{Q}_{n}}^2(U)}{\rm Var}_{Q_u}(U) + o_p(n^{-1/2}),
\end{align*}
following the fact that $\sqrt{n}(Q_u[U]-\mathbb{Q}_{n}[U])^2 = o_p(1)$, ${\rm Cov}_{\mathbb{Q}_{n}}(U, \hat{E}_n(U))/{\rm Var}_{\mathbb{Q}_{n}}^2(U) = O_p(1)$ and ${\rm Var}_{Q_u}(U)$ is bounded,
by assumptions \ref{assump:Covariates} and the details of $\hat{E}_n$ in the main text.
Moreover, 
\begin{align*}
  (v) = -P\left[\frac{(U-\Qn[U])}{{\rm Var}_{\Qn}(U)}\int_{\mathcal{T}} \hat{E}_n(U,s)P[dM(s)|U]\right] = 0, 
\end{align*}
following $dM(s) = 1(X \in ds,\delta=0)-1(X \ge s)d\Lambda(s)$ and $d\Lambda(s)={\rm P}(C \in ds)/{\rm P}(C \ge s)$, so that $P[dM(s)|U]=P[1(T \ge s)|U]P[1(C \in ds)-1(C \ge s)d\Lambda(s)] = 0$, together with the independent censoring assumption.

Inserting the above results into \eqref{eq:remainder_decomp}, we have that
\begin{align} \label{eq:remainder_decomp.1}
  &\Psi(\hat{P}'_n)-\Psi(P)+P\IF^*(\cdot|\hat{P}'_n)\\
  & = \frac{{\rm Cov}_{\mathbb{Q}_{n}}(U, \hat E_n(U))}{{\rm Var}_{\mathbb{Q}_{n}}(U)}-\frac{{\rm Cov}_{Q_u}(U,E[\tilde{Y}|U])}{{\rm Var}_{Q_u}(U)} + \frac{{\rm Cov}_{Q_u}
  (U,E[\tilde{Y}|U])}{{\rm Var}_{\mathbb{Q}_{n}}(U)} \nonumber \\
  &\quad - \frac{{\rm Cov}_{\mathbb{Q}_{n}}(U, \hat{E}_n(U))}{{\rm Var}_{\mathbb{Q}_{n}}^2(U)}{\rm Var}_{Q_u}(U) + o_p(n^{-1/2}) \nonumber \\ 
  & = \big[{\rm Var}_{\mathbb{Q}_{n}}(U)-{\rm Var}_{Q_u}(U)\big]\frac{1}{{\rm Var}_{\mathbb{Q}_{n}}(U)}
  {\rm Cov}_{Q_u}(U, E[\tilde{Y}|U])\Big[\frac{1}{{{\rm Var}_{\mathbb{Q}_{n}}(U)}}-\frac{1}{{{\rm Var}_{Q_u}(U)}}\Big] \nonumber \\
  &\quad + \big[{\rm Var}_{\mathbb{Q}_{n}}(U)-{\rm Var}_{Q_u}(U)\big]\bigg\{\frac{1}{{\rm Var}^2_{\mathbb{Q}_{n}}(U)}-\frac{1}{{\rm Var}^2_{Q_u}(U)}\bigg\} \nonumber \\
  & \qquad \times \Big[{\rm Cov}_{\mathbb{Q}_{n}}(U, \hat E_n(U))-{\rm Cov}_{Q_u}(U, E[\tilde{Y}|U])\Big] \nonumber \\
  &\quad + \big[{\rm Var}_{\mathbb{Q}_{n}}(U)-{\rm Var}_{Q_u}(U)\big]\frac{1}{{\rm Var}^2_{Q_u}(U)}\Big[{\rm Cov}_{\mathbb{Q}_{n}}(U, \hat E_n(U))-{\rm Cov}_{Q_u}(U, E[\tilde{Y}|U])\Big] \nonumber \\
  &\quad + o_p(n^{-1/2}). \nonumber
\end{align}
Because ${\rm Var}_{\mathbb{Q}_{n}}(U)$ is bounded away from zero almost surely by \ref{assump:Covariates} that $U$ is non-degenerate, along with ${\rm Cov}_{Q_u}(U, E[\tilde{Y}|U])$ is bounded by \ref{assump:Covariates} and the fact that $E[\tilde{Y}|U=u,X \ge s]$ is uniformly bounded over $(u,s)$, the first quantity on the right-hand-side of \eqref{eq:remainder_decomp.1} is $o_p(n^{-1/2})$, following
$\sqrt{n}[{\rm Var}_{\mathbb{Q}_{n}}(U)-{\rm Var}_{Q_u}(U)] = O_p(1)$ and $1/{\rm Var}_{\mathbb{Q}_{n}}(U)-1/{\rm Var}_{Q_u}(U)=o_p(1)$. 

Note also that using the general properties of $\hat{E}_n$ in \ref{assump:Conditional mean} and the law of large numbers gives
\begin{align} \label{eq:decomp_cov}
  & {\rm Cov}_{\mathbb{Q}_{n}}(U, \hat E_n(U))-{\rm Cov}_{Q_u}(U, E[\tilde{Y}|U]) = {\rm Cov}_{Q_u}(U, \bar{E}(U)-E[\tilde{Y}|U]) \nonumber \\
  & \quad + {\rm Cov}_{\mathbb{Q}_{n}}(U, \hat E_n(U)-\bar{E}(U))
  + \big[{\rm Cov}_{\mathbb{Q}_{n}}(U, \bar{E}(U)) - {\rm Cov}_{Q_u}(U, \bar{E}(U))\big] \\
  & = {\rm Cov}_{Q_u}(U, \bar{E}(U)-E[\tilde{Y}|U]) + o_p(1). \nonumber
\end{align}
By the facts that $\sqrt{n}[{\rm Var}_{\mathbb{Q}_{n}}(U)-{\rm Var}_{Q_u}(U)] = O_p(1)$ and $1/{\rm Var}^2_{\mathbb{Q}_{n}}(U)-1/{\rm Var}^2_{Q_u}(U)=o_p(1)$, the second quantity on the right-hand-side of \eqref{eq:remainder_decomp.1} is also $o_p(n^{-1/2})$.
In addition, we observe that ${\rm Var}_{\mathbb{Q}_{n}}(U)$ is regular asymptotically linear estimator of ${\rm Var}_{Q_u}(U)$ with influence function $o \mapsto (u-Q_u[U])^2$. Combining this fact with
$\sqrt{n}[{\rm Var}_{\mathbb{Q}_{n}}(U)-{\rm Var}_{Q_u}(U)] = O_p(1)$ and the display in \eqref{eq:decomp_cov}, the third quantity on the right-hand-side of \eqref{eq:remainder_decomp.1} turns into
\begin{align*}
  & \big[{\rm Var}_{\mathbb{Q}_{n}}(U)-{\rm Var}_{Q_u}(U)\big]\frac{1}{{\rm Var}^2_{Q_u}(U)}{\rm Cov}_{Q_u}(U, \bar{E}(U)-E[\tilde{Y}|U]) + o_p(n^{-1/2})\\
  & = \frac{{\rm Cov}_{Q_u}(U, \bar{E}(U)-E[\tilde{Y}|U])}{{\rm Var}^2_{Q_u}(U)} \Big\{[\mathbb{P}_n - P](U-Q_u[U])^2\Big\} + o_p(n^{-1/2})\\
  & \equiv \mathbb{P}_n\IF^{\dagger}(\cdot|\bar{E},P) + o_p(n^{-1/2}).
\end{align*}
Referring to \eqref{eq:remainder_decomp.1} and noting that $P\IF^{\dagger}(\cdot|\bar{E},P)=0$, we have completed the proof that $\Psi(\hat{P}'_n)-\Psi(P)+P\IF^*(\cdot|\hat{P}'_n)=[\mathbb{P}_n-P] \IF^{\dagger}(\cdot|\bar{E},P) + o_p(n^{-1/2})$.
\end{proof}

\section{Proof of lemmas for Theorem \ref{Thm:stab_one_step}}
\label{sec:proof_lemmas_AsympNormal_stabilized_onestep}
For any $m_0, m_1 >0$, let ${\cal BV}({\cal T},m_0,m_1)$ be the collection of functions $f \colon {\cal T} \rightarrow [-m_0,m_0]$ with total variation bounded by $m_1$. The lemma below gives preservation properties of these classes.
\begin{lemma}\label{lemma:preservation_BV}
Fix $m_0, \bar{m}_0, m_1, \bar{m}_1>0$ and $\varepsilon\in (0,1)$. For any $f\in {\cal BV}({\cal T},m_0,m_1)$ and $g \in {\cal BV}({\cal T},\bar{m}_0,\bar{m}_1)$, $f+g$ and $f-g$ are contained in ${\cal BV}({\cal T},m_0+\bar{m}_0,m_1+\bar{m}_1)$; $fg$ belongs to ${\cal BV}({\cal T},m_0\bar{m}_0,m_0\bar{m}_1+\bar{m}_0m_1)$; moreover, if $g$ is such that $\inf_{t\in\mathcal{T}} |g(t)|>\varepsilon$, then $f/g$ is contained in ${\cal BV}({\cal T},m_0/\varepsilon, (m_0\bar{m}_1+\bar{m}_0m_1)/\varepsilon^2).$
\end{lemma}
\begin{proof}
Let $\|\tilde{f}\|_{\nu}$ denote the total variation of $s \mapsto \tilde{f}(s), s \in {\cal T}$, and let $\|f\|_\infty$ denote $\sup_{t\in\mathcal{T}} |f(t)|$. As the notation indicates, both $\|\cdot\|_\nu$ and $\|\cdot\|_\infty$ are norms, and therefore satisfy the triangle inequality.

As $f \in {\cal BV}({\cal T},m_0,m_1)$ and $g \in {\cal BV}({\cal T},\bar{m}_0,\bar{m}_1)$, we have $\|f\|_\infty < m_0$, $\|f\|_{\nu} < m_1$, $\|g\|_\infty < \bar{m}_0$ and $\|g\|_{\nu} < \bar{m}_1$.

To see that $f+g\in {\cal BV}({\cal T},m_0+\bar{m}_0,m_1+\bar{m}_1)$, note that $\|f+g\|_\infty\le \|f\|_\infty+\|g\|_\infty< m_0 + \bar{m}_0$ and $\|f+g\|_\nu\le \|f\|_\nu+\|g\|_\nu< m_1 + \bar{m}_1$. The same argument shows that $f-g\in {\cal BV}({\cal T},m_0+\bar{m}_0,m_1+\bar{m}_1)$.

To see that $fg\in {\cal BV}({\cal T},m_0\bar{m}_0,m_0\bar{m}_1+\bar{m}_0m_1)$, note that $\|fg\|_\infty\le \|f\|_\infty\|g\|_\infty<m_0\bar{m}_0$, and, for an arbitrary partition $\{t_1 < t_2 < \ldots < t_m < t_{m+1}\}$ of ${\cal T}$,
\begin{align*}
    &\sum_{j=1}^m |(fg)(t_{j+1})-(fg)(t_j)| \le \sum_{j=1}^m \Big\{|g(t_{j+1})||f(t_{j+1})-f(t_j)|+|f(t_{j})||g(t_{j+1})-g(t_j)|\Big\} \\
    &\quad \le \bar{m}_0\sum_{j=1}^m |f(t_{j+1})-f(t_j)| + m_0\sum_{j=1}^m |g(t_{j+1})-g(t_j)| \le \bar{m}_0\|f\|_{\nu} + m_0\|g\|_{\nu}\\
    &\quad < m_0\bar{m}_1+\bar{m}_0m_1.
\end{align*}
For the ratio of two functions, $\|f/g\|_\infty < m_0/\varepsilon$, and, for an arbitrary partition $\{t_1 < t_2 < \ldots < t_m < t_{m+1}\}$ of ${\cal T}$,
\begin{align*}
  &\sum_{j=1}^m |f(t_{j+1})/g(t_{j+1})-f(t_j)/g(t_j)| 
  = \sum_{j=1}^m \bigg|\frac{g(t_{j})f(t_{j+1})}{g(t_{j})g(t_{j+1})}
  -\frac{g(t_{j+1})f(t_j)}{g(t_{j+1})g(t_j)}\bigg|\\
  &\quad \le \frac{1}{\varepsilon^2} \sum_{j=1}^m |g(t_{j})f(t_{j+1})
  -g(t_{j+1})f(t_j)|\\
  &\quad \le \frac{1}{\varepsilon^2} \sum_{j=1}^m \Big\{|g(t_{j})||f(t_{j+1})-f(t_j)| + |f(t_j)||g(t_{j+1})-g(t_j)|\Big\}\\
  &\quad \le \big(\bar{m}_0\|f\|_v + m_0\|g\|_v\big)/\varepsilon^2
  < (m_0\bar{m}_1+\bar{m}_0m_1)/\varepsilon^2.
\end{align*}
\end{proof}

Note that the present theorem refers to $\hat{E}_n$ given in Remark \ref{remark:est_conditional_res_life} with $n$ replaced by $j$ for $j=q_n,\ldots,n-1$, and $U=U_k$ for a given $k$. That is, with $k$ included and the sample size of $j$ considered, we are now dealing with
\begin{align}\label{eq:Ejlinear}
  \hat{E}_j(u,s,k) \equiv \mathbb{P}_{j}[Y1(X \geq s)] + \frac{{\rm Cov}_{\mathbb{P}_{j}}(U_k1(X \geq s), Y1(X \geq s))}{{\rm Var}_{\mathbb{P}_{j}}(U_k1(X \geq s))}(u-\mathbb{P}_{j}[U_k1(X \geq s)]).
\end{align}
By the weak law of large numbers and the continuous mapping theorem, together with the uniform consistency of $\hat{G}_n$ on $\mathcal{T}$, $\hat{E}_n$ is a pointwise consistent estimator of
\begin{align*}
  E_0(u,s,k) \equiv P[\tilde{Y}1(X \geq s)] + \frac{{\rm Cov}(U_k1(X \geq s),\tilde{Y}1(X \geq s))}{{\rm Var}(U_k1(X \geq s))}(u-P[U_k1(X \geq s)]).
\end{align*}
Note that we suppress the argument $s$ if $s=-\infty$.
The following lemma shows that all of $\hat{E}_j$, $j=\{q_n,\ldots,n\}$, and $E_0$ are asymptotically contained in a class of uniformly bounded functions with uniformly bounded total variation.

\begin{lemma}\label{lemma:inclusion_E0_Ehat}
  Under the assumptions \ref{assump:Covariates}, \ref{assump:Survival function} and \ref{assump:Variances}, there exist positive constants $\tilde{M}_0$ and $\tilde{M}_1$ such that for each $k$, the function $E_0(\cdot,\cdot,k)$ is contained in the class
  \begin{align*}
    \big\{(u,s)\: \mapsto\: a(s) + b(s)u:\: a, b \in {\cal BV}({\cal T},\tilde{M}_0,\tilde{M}_1) \big\},
  \end{align*}
  and moreover, $\hat{E}_j(\cdot,\cdot,k)$ as defined in \eqref{eq:Ejlinear} is contained in this class with probability tending to one, for $j=q_n,\ldots,n$.
\end{lemma}
\begin{proof}
Let the upper bound of all the $|U_k|$ be $\tilde{M}_u \in (0,\infty)$, which is ensured by the assumption \ref{assump:Covariates}. And let a constant $\tilde{M}_{y} > \tau/G(\tau) > 0$, following $G(\tau) > 0$ in \ref{assump:Survival function} such that $\tilde{M}_y \in (0,\infty)$. 
For some $\varepsilon \in (0,1)$, take $\tilde{M}_0 = \tilde{M}_{y}+\sum_{q=1}^2 2\varepsilon^{-1}\tilde{M}_{u}^q\tilde{M}_y$ and $\tilde{M}_1 = \tilde{M}_{y}+\sum_{q=1}^4 6\varepsilon^{-2}\tilde{M}_{u}^q\tilde{M}_y$. Note that $\tilde{M}_u$
and $\tilde{M}_y$ do not depend on $(j,k,n)$, so $\tilde{M}_0$
and $\tilde{M}_1$ are independent of $(j,k,n)$.

To have $E_0(\cdot,\cdot,k)$ in the above-defined class, it suffices to have both
$$P[\tilde{Y}1(X \geq s)] - \frac{{\rm Cov}(U_k1(X \geq s),\tilde{Y}1(X \geq s))}{{\rm Var}(U_k1(X \geq s))}P[U_k1(X \geq s)]$$ and
${\rm Cov}(U_k1(X \geq s),\tilde{Y}1(X \geq s))/{\rm Var}(U_k1(X \geq s))$
belonging to ${\cal BV}({\cal T}, \tilde{M}_0, \tilde{M}_1)$. We start by showing that each of the following functions belongs to an appropriate class of uniformly bounded functions with uniformly bounded total variation:  $P[\tilde{Y}1(X \geq s)]$, $P[U_k1(X \geq s)]$, $P[(U_k-P[U_k1(X \geq s)])\tilde{Y}1(X \geq s)]$ and $1/P[U_k1(X \geq s)-P[U_k1(X \geq s)]]^2$. Specifically, for each of these functions, we will exhibit an $m_0,m_1$ such that the function belongs to ${\cal BV}(\mathcal{T},m_0,m_1)$. 

We see that $s \mapsto P[\tilde{Y}1(X \geq s)]$ is uniformly bounded by $P|\tilde{Y}| = P|\delta X/G(X)| \le \tau/G(\tau) < \tilde{M}_y$ and has total variation bounded by $\tilde{M}_{y}$; therefore $s \mapsto P[\tilde{Y}1(X \geq s)] \in {\cal BV}({\cal T}, \tilde{M}_{y}, \tilde{M}_{y})$. Similarly, $s \mapsto P[U_k1(X \geq s)]$ is uniformly bounded by $\tilde{M}_u$ and has total variation bounded by $\tilde{M}_u$; thus, this function is in
${\cal BV}({\cal T}, \tilde{M}_{u}, \tilde{M}_{u})$.
Moreover, Lemma \ref{lemma:preservation_BV} gives that $s \mapsto P[(U_k-P[U_k1(X \geq s)])\tilde{Y}1(X \geq s)] \in {\cal BV}({\cal T}, 2\tilde{M}_{u}\tilde{M}_{y}, 3\tilde{M}_{u}\tilde{M}_{y})$. Also, $s \mapsto 1/P[U_k1(X \geq s)-P[U_k1(X \geq s)]]^2$ belonging to ${\cal BV}({\cal T},1/\varepsilon, 3\tilde{M}_{u}^2/\varepsilon^2)$, using Lemma \ref{lemma:preservation_BV}, that 
$s \mapsto P[U_k1(X \geq s)-P[U_k1(X \geq s)]]^2$ in ${\cal BV}({\cal T},2\tilde{M}_{u}^2, 3\tilde{M}_{u}^2)$ and that $P[U_k1(X \geq s)-P[U_k1(X \geq s)]]^2={\rm Var}(U_k1(X \ge s)) > \varepsilon > 0$ by \ref{assump:Variances}. Provided the sufficiently large $\tilde{M}_0$ and $\tilde{M}_1$, the above results and Lemma \ref{lemma:preservation_BV} implies that 
\begin{align*}
&P[\tilde{Y}1(X \geq s)] - \frac{{\rm Cov}(U_k1(X \geq s),\tilde{Y}1(X \geq s))}{{\rm Var}(U_k1(X \geq s))}P[U_k1(X \geq s)]\\ 
&\quad \in {\cal BV}\left({\cal T}, \tilde{M}_{y}+\frac{2\tilde{M}_{u}^2\tilde{M}_y}{\varepsilon}, \tilde{M}_{y}+\frac{6\tilde{M}_{u}^4\tilde{M}_y}{\varepsilon^2}+\frac{5\tilde{M}_{u}^2\tilde{M}_y}{\varepsilon}\right) \subset {\cal BV}({\cal T},\tilde{M}_0,\tilde{M}_1);\\
&\frac{{\rm Cov}(U_k1(X \geq s),\tilde{Y}1(X \geq s))}{{\rm Var}(U_k1(X \geq s))}
\in {\cal BV}\left({\cal T}, \frac{2\tilde{M}_{u}\tilde{M}_y}{\varepsilon}, \frac{6\tilde{M}_{u}^3\tilde{M}_y}{\varepsilon^2}+\frac{3\tilde{M}_{u}\tilde{M}_y}{\varepsilon}\right)\\
&\quad \subset {\cal BV}({\cal T},\tilde{M}_0,\tilde{M}_1).
\end{align*} 

Fix $j \in \{q_n,\ldots,n\}$. For $\hat{E}_j(\cdot, \cdot,k)$ to belong to the function class given in the lemma, we would need
$$\Pj[Y1(X \geq s)] - \frac{{\rm Cov}_{\Pj}(U_k1(X \geq s),Y1(X \geq s))}{{\rm Var}_{\Pj}(U_k1(X \geq s))}\Pj[U_k1(X \geq s)]$$ and
${\rm Cov}_{\Pj}(U_k1(X \geq s),Y1(X \geq s))/{\rm Var}_{\Pj}(U_k1(X \geq s))$
belonging to ${\cal BV}({\cal T}, \tilde{M}_0, \tilde{M}_1)$. It thus suffices to show that, with probability tending to one, the following functions $$\mathbb{P}_{j}[Y1(X \geq s)],\  \mathbb{P}_{j}[(U_k-\mathbb{P}_{j}[U_k1(X \geq s)])Y1(X \geq s)],\  \mathbb{P}_{j}[U_k1(X \geq s)],$$
and $1/\mathbb{P}_{j}[U_k1(X \geq s)-\mathbb{P}_{j}[U_k1(X \geq s)]]^2$ all belong to ${\cal BV}(\mathcal{T},m_0,m_1)$, for suitable $m_0$ and $m_1$ that does not depend on $j$. This will be done by appealing to Lemma \ref{lemma:preservation_BV}.

Below we show that  
$|\mathbb{P}_{j}[Y1(X \geq s)]| \leq \tilde{M}_{y}$ for any $s\in {\cal T}$ and  $s \mapsto \mathbb{P}_{j}[Y1(X \geq s)]$ has total variation bounded by $\tilde{M}_{y}$ with probability tending to one.
\begin{align*}
  & |\Pj[Y1(X \ge s)]| \leq \Pj|Y1(X \ge s)| \leq \Pj|Y| \leq \sup_{t \in {\cal T}}\bigg|\frac{1}{\hat{G}_n(t)}-\frac{1}{G(t)}\bigg|\Pj|\delta X| + \Pj \Big|\frac{\delta X}{G(X)}\Big|\\
  &\quad \leq \sup_{t \in {\cal T}}\bigg|\frac{1}{\hat{G}_n(t)}-\frac{1}{G(t)}\bigg|\big\{(\Pj-P)|\delta X| + P|\delta X| \big\} + (\Pj-P)|\tilde{Y}| + P|\tilde{Y}| < \tilde{M}_{y}
\end{align*}
with probability tending to one, using the weak law of large numbers. We also use $|\delta X| \leq \tau$, the uniform consistency of $\hat{G}_n$, $G(\tau)>0$ by \ref{assump:Survival function} and $P|\tilde{Y}| \le \tau/G(\tau)$.
Also, for an arbitrary partition of ${\cal T}$, say $\{t_1 < t_2 < \ldots < t_m < t_{m+1}\}$, 
\begin{align*}
  &\sum_{j=1}^m \Big|\Pj[Y1( t_j \le X < t_{j+1})]\Big|  
  \le \Pj \sum_{j=1}^m\Big|Y1( t_j \le X < t_{j+1})\Big| = \Pj|Y| \\
  &\quad \leq \sup_{t \in {\cal T}}\bigg|\frac{1}{\hat{G}_n(t)}-\frac{1}{G(t)}\bigg|\Pj|\delta X| + \Pj \Big|\frac{\delta X}{G(X)}\Big|\\
  &\quad \leq \sup_{t \in {\cal T}}\bigg|\frac{1}{\hat{G}_n(t)}-\frac{1}{G(t)}\bigg|\big\{(\Pj-P)|\delta X| + P|\delta X| \big\} + (\Pj-P)|\tilde{Y}| + P|\tilde{Y}| < \tilde{M}_{y}
\end{align*}
with probability tending to one, using the same arguments as the above. Taking a supremum over all partitions of ${\cal T}$ shows that the total variation of $s\mapsto \mathbb{P}_{j}[Y1(X \geq s)]$ is bounded by $\tilde{M}_{y}$ with probability tending to one. Thus $\mathbb{P}_{j}[Y1(X \geq s)]$ belongs to ${\cal BV}({\cal T}, \tilde{M}_{y}, \tilde{M}_{y})$ with probability tending to one.

Lemma \ref{lemma:preservation_BV} implies that, with probability tending to one, $s \mapsto \mathbb{P}_{j}[U_k1(X \geq s)] \in {\cal BV}({\cal T}, \tilde{M}_{u}, \tilde{M}_{u})$;
$s \mapsto \mathbb{P}_{j}[U_k1(X \geq s)-\mathbb{P}_{j}[U_k1(X \geq s)]]^2 \in {\cal BV}({\cal T},2\tilde{M}_u^2,3\tilde{M}_u^2)$, and $s \mapsto \mathbb{P}_{j}[(U_k-\mathbb{P}_{j}[U_k1(X \geq s)])Y1(X \geq s)] \in {\cal BV}({\cal T},2\tilde{M}_u\tilde{M}_y,3\tilde{M}_u\tilde{M}_y)$.

Recall that \ref{assump:Variances} assumes ${\rm Var}(U_k1(X \ge s))$ to be uniformly (over $k,s$) bounded away from zero, that is, $\min_{k,s}{\rm Var}(U_k1(X \ge s)) > \varepsilon$ for some $\varepsilon \in (0,1)$. Then we see that with high probability,
$\mathbb{P}_{j}[U_k1(X \geq s)-\mathbb{P}_{j}[U_k1(X \geq s)]]^2 = {\rm Var}_{\mathbb{P}_{j}}(U_k1(X \geq s))$ is larger than or equal to $\varepsilon$ and then bounded away from zero on ${\cal T}$. By Lemma \ref{lemma:preservation_BV}, we have, with probability tending to one, $1/\mathbb{P}_{j}[U_k1(X \geq s)-\mathbb{P}_{j}[U_k1(X \geq s)]]^2 \in {\cal BV}({\cal T},1/\varepsilon,3\tilde{M}_u^2/\varepsilon^2)$.

Together with all the above results, Lemma \ref{lemma:preservation_BV} gives that with probability tending to one,
\begin{align*}
&\Pj[Y1(X \geq s)] - \frac{{\rm Cov}_{\Pj}(U_k1(X \geq s),Y1(X \geq s))}{{\rm Var}_{\Pj}(U_k1(X \geq s))}\Pj[U_k1(X \geq s)]\\ 
&\quad \in {\cal BV}({\cal T}, \tilde{M}_{y}+\frac{2\tilde{M}_{u}^2\tilde{M}_y}{\varepsilon}, \tilde{M}_{y}+\frac{6\tilde{M}_{u}^4\tilde{M}_y}{\varepsilon^2}+\frac{5\tilde{M}_{u}^2\tilde{M}_y}{\varepsilon}) \subset {\cal BV}({\cal T},\tilde{M}_0,\tilde{M}_1);\\
&\frac{{\rm Cov}_{\Pj}(U_k1(X \geq s),Y1(X \geq s))}{{\rm Var}_{\Pj}(U_k1(X \geq s))}
\in {\cal BV}({\cal T}, \frac{2\tilde{M}_{u}\tilde{M}_y}{\varepsilon}, \frac{6\tilde{M}_{u}^3\tilde{M}_y}{\varepsilon^2}+\frac{3\tilde{M}_{u}\tilde{M}_y}{\varepsilon})\\
&\quad \subset {\cal BV}({\cal T},\tilde{M}_0,\tilde{M}_1).
\end{align*} 
\end{proof}

Let $\tilde{\epsilon} > 0$; define $\mathcal{G}$ to be the collection of monotone nonincreasing c\`{a}dl\`{a}g functions $\tilde{G} \colon \mathcal{T} \rightarrow [0,1]$ such that $\tilde{G}(\tau) > \tilde{\epsilon}$, and let $\tilde{M}_0, \tilde{M}_1$ be the constants shown to exist in Lemma~\ref{lemma:inclusion_E0_Ehat}. To simplify the notation, we let ${\cal BV}({\cal T})\equiv{\cal BV}({\cal T},\tilde{M}_0,\tilde{M}_1)$.
Below we use the notation $\bs{u} =(u_1, \ldots,u_p)$. For $k\in\mathbb{N}$, $(q,v,w)\in\{0,1,2\}^3$, and $r\in\{0,1,2,3,4\}$, define the function classes
\begin{equation}\label{eq:class_tilde_F}
\begin{split}
  &\tilde{\mathcal{F}}_{1}(k,q,r) = \bigg\{(x, \delta, \bs{u})\: \mapsto\: u_k^{r}\Big(\frac{\delta x}{\tilde{G}(x)}\Big)^q1(x \geq s):\: \tilde{G}\in\mathcal{G},\: s \in {\cal T} \bigg\}; \\
  &\tilde{{\cal E}}(k,w) = \bigg\{(\bs{u},s)\: \mapsto\: 
  \Big[a(s) + b(s)u_k\Big]^w:\: a, b \in {\cal BV}({\cal T}) \bigg\};\\
  &\tilde{\mathcal{F}}_{2}(k,v) = \bigg\{(x, \delta, \bs{u})\: \mapsto
  \Big[\int_{\mathcal{T}}\big[a(s) + b(s)u_k\big]1(x \geq s)\big\{1(x \in ds, \delta=0) - \tilde{\Lambda}(s)\big\}\Big]^v :\\
  &\hspace{5cm} \tilde{\Lambda}(s) = -\log(\tilde{G}(s));\:
  \tilde{G}\in\mathcal{G};\: a,\, b \in {\cal BV}({\cal T}) \bigg\};\\
  &\tilde{\mathcal{F}}(k,q,r,v,w) =
  \Big[\tilde{\mathcal{F}}_{1}(k,q,r)\tilde{{\cal E}}(k,w)\Big] \cup 
  \Big[\tilde{\mathcal{F}}_{1}(k,q,r)\tilde{\mathcal{F}}_{2}(k,v)\Big],
\end{split}
\end{equation}  
where for two function classes $\mathcal{H}_1$ and $\mathcal{H}_2$, we let $\mathcal{H}_1\mathcal{H}_2=\{h_1(\cdot)h_2(\cdot) : h_1\in\mathcal{H}_1,h_2\in\mathcal{H}_2\}$.
Also, for $k\in\mathbb{N}$, let $\tilde{\mathcal{F}}(k)=\cup_{q=0}^{2}\cup_{r=0}^{4} \cup_{v=0}^{2}\cup_{w=0}^{2}\tilde{\mathcal{F}}(k,q,r,v,w)$.
Let $\mathcal{K}_n = \{1, \ldots, p_n\}$. Henceforth we consider a large class 
$\tilde{\mathcal{F}}_n = \cup_{k\in\mathcal{K}_n}\tilde{\mathcal{F}}(k)$. In view of Lemma \ref{lemma:inclusion_E0_Ehat}, we have $E_0(\cdot,\cdot,k) \in \tilde{{\cal E}}(k,1)$ and with probability tending to one, $\hat{E}_j(\cdot,\cdot,k) \in \tilde{{\cal E}}(k,1)$, $j= q_n,\ldots,n$. 

\begin{lemma}\label{lemma:UB_funcs_tildecalE}
  Suppose that \ref{assump:Covariates} holds and let $\tilde{M}_u$ be a finite constant such that for all $k\in\mathbb{N}$, $|U_k|\le \tilde{M}_u$ with $P$-probability one.
  For any $k\in\mathbb{N}$ and $\tilde{f} \in \tilde{{\cal E}}(k,1)$, we have that $|\tilde{f}| \leq \tilde{M}_0(1+\tilde{M}_u)$.
\end{lemma}
\begin{proof}
  This is an immediate consequence of the definition of $\tilde{{\cal E}}(k,1)$ and $ {\cal BV}({\cal T})$, along with the assumption \ref{assump:Covariates}.
\end{proof}

\begin{lemma}\label{lemma:tildeE_VC_hull}
  Given $k \in {\cal K}_n$, $w \in \{1,2\}$ and the definition of $\tilde{{\cal E}}(k,w)$ above, we have $\tilde{{\cal E}}(k,w)$ is a Vapnik-\v{C}ervonenkis (VC)-hull class for sets.
\end{lemma}
\begin{proof}
Fix $w=1$. We start with showing that ${\cal BV}({\cal T})$ is a VC-hull class for sets as follows. For any $f \in {\cal BV}({\cal T})$ with $|f| \leq \tilde{M}_0$, the Jordan decomposition indicates that $f = f^+-f^-$, where $(f^+, f^-)$ are the positive and negative parts of $f$, and both of them are positive, c\`{a}dl\`{a}g and monotonic increasing on ${\cal T}$. Therefore, we can see $f$ as the scalar multiple (by $\tilde{M}_0$) of the limit of the sequence 
\begin{align*}
  f_{m} = \sum_{j=1}^{m}\frac{1}{m}\Big[1\Big(f^+ > \frac{j}{m}\tilde{M}_0\Big)-1\Big(f^- > \frac{j}{m}\tilde{M}_0\Big)\Big].  
\end{align*}
Then $f(\cdot)/\tilde{M}_0$ is in a class contained in the pointwise sequential closure of the symmetric convex hull of a class of indicator functions $\{1(f^+ > m'), m' \in \mathbb{R}^+\} \cup \{1(f^- > m'), m' \in \mathbb{R}^+\}$, which is a VC-subgraph class because
\begin{align*}
  &\{1(f^+(\cdot) > m'), m' \in \mathbb{R}^+\} \cup \{1(f^-(\cdot) > m'), m' \in \mathbb{R}^+\}\\
  &\quad \subset \{1(f^+(\cdot) > m'), m' \in \mathbb{R}\} \cup \{1(f^-(\cdot) > m'), m' \in \mathbb{R}\}
\end{align*}
and the union of the two classes on the right-hand-side forms a VC-subgraph class. Hence, ${\cal BV}({\cal T})$ is (a $\tilde{M}_0$-fold rescaling of) a VC-hull class for sets. 

For any given $k$, Lemma \ref{lemma:UB_funcs_tildecalE} indicates that any function $\tilde{f} \in \tilde{{\cal E}}(k,1)$ is uniformly bounded: $|\tilde{f}| \leq \tilde{M}_0(1+\tilde{M}_u)$. Following similar arguments to the above implies that $\tilde{f}/(\tilde{M}_0(1+\tilde{M}_u))$ is in a class contained in the pointwise sequential closure of the symmetric convex hull of
\begin{align*}
  \{1(\tilde{f}^+(\cdot) > m'), m' \in \mathbb{R}\} \cup \{1(\tilde{f}^-(\cdot) > m'), m' \in \mathbb{R}\},
\end{align*}
and therefore $\tilde{{\cal E}}(k,1)$ is (a $\tilde{M}_0(1+\tilde{M}_u)$-fold rescaling of) a VC-hull class for sets.
When $w=2$, $\tilde{{\cal E}}(k,2)$ is the square of $\tilde{{\cal E}}(k,1)$, so it is also a VC-hull class for sets \cite[Lemmas 2.6.20 in][]{Vaart1996}.
\end{proof}

As Lemma \ref{lemma:inclusion_E0_Ehat} indicates that for all $j$, $\hat{E}_j(\cdot,\cdot,k) \in \tilde{{\cal E}}(k,1)$ with probability tending to one and $E_0(\cdot,\cdot,k) \in \tilde{{\cal E}}(k,1)$, we could include $\IF_k^{CAR}(\cdot|P)$ in $\tilde{\mathcal{F}}_{2}(k,1)$.
The following lemma shows that $\tilde{\mathcal{F}}_{1}(k,q,r)$,
and $\tilde{\mathcal{F}}_{2}(k,v)$ are VC-hull classes for sets, for any $(k,q,r,v)$.
\begin{lemma} \label{tildeFn_VC_hull}
  For all $k\in\mathcal{K}_n$, $r\in \{0,1,2,3,4\}$ and $(q,v,w)\in \{0,1,2\}^3$, all of the following are VC-hull classes for sets: $$ \tilde{\mathcal{F}}_{1}(k,q,r),\, \tilde{\mathcal{F}}_{2}(k,v),\, \tilde{\mathcal{F}}_{1}(k,q,r)\tilde{{\cal E}}(k,w),\,
  \tilde{\mathcal{F}}_{1}(k,q,r)\tilde{\mathcal{F}}_{2}(k,v)\; \mbox{and}\; \tilde{\mathcal{F}}(k,q,r,v,w).$$ Moreover, $\tilde{\mathcal{F}}(k)$ is a VC-hull class for sets.
\end{lemma}
\begin{proof}
As observed, $\mathcal{G}$ is a VC-hull class for sets because any $\tilde{G} \in \mathcal{G}$ is the pointwise limit of the sequence $\tilde{G}_m = m^{-1}\sum_{j=1}^m 1(\tilde{G} \geq m^{-1}j)$, and a bounded VC-major class as well. 
By Lemma 2.6.19 of \cite{Vaart1996}, $\{x \mapsto (1/\tilde{G}(x))^q : \tilde{G} \in \mathcal{G}\}$ is bounded VC-major. This equivalently implies that $\tilde{\mathcal{F}}_{11}(k,q) = \{(x,\delta,\bs{u}) \mapsto (1/\tilde{G}(x))^q : \tilde{G} \in \mathcal{G}\}$ is bounded VC-major.
Note that $\tilde{\mathcal{F}}_{12}(k,q,r) = \{(x,\delta,\bs{u}) \mapsto u_k^r (\delta x)^q\}$ is bounded VC-major. Moreover, Lemma 2.6.13 of \cite{Vaart1996} implies that
$\tilde{\mathcal{F}}_{11}(k,q)$ and $\tilde{\mathcal{F}}_{12}(k,q,r)$ are VC-hull classes for sets. Let $\tilde{\mathcal{F}}_{13}(k) = \{(x,\delta,\bs{u}) \mapsto 1(x \geq s): s \in \mathcal{T}\},$ which is also a VC-hull class for sets.
Therefore, we have $\tilde{\mathcal{F}}_{1}(k,q,r) = \tilde{\mathcal{F}}_{11}(k,q)\tilde{\mathcal{F}}_{12}(k,q,r)\tilde{\mathcal{F}}_{13}(k)$ as a VC-hull class for sets \cite[Lemma 2.6.20 in][]{Vaart1996}.

Below we show that $\tilde{\mathcal{F}}_{2}(k,v)$ is a VC-hull class for sets.
Let $\{t_1 < t_2 < \ldots < t_m\}$ be an arbitrary partition of ${\cal T}$ with uniform increments $\Delta_t = t_{j+1}- t_{j}$ for all $j$. By Lemma \ref{lemma:UB_funcs_tildecalE}, any function $\tilde{f} \in \tilde{{\cal E}}(k,1)$ is uniformly bounded: $|\tilde{f}| \leq \tilde{M}_0(1+\tilde{M}_u)$.
Let $\tilde{M} = 4\tilde{M}_0(1+\tilde{M}_u)\sum_{j=1}^{m}|\tilde{\Lambda}(t_j+\Delta_t)|$. The integral in $\tilde{\mathcal{F}}_{2}(k,1)$ is the scalar multiple (by $\tilde{M}$) of the limit of the sequence 
\begin{align*}
  &\sum_{j=1}^{m}\frac{\tilde{f}(t_j)}{\tilde{M}}\Big[1_{[t_j,t_{j+1})}(x)(1-\delta)-\tilde{\Lambda}(t_{j+1})1(x \geq t_{j+1}) + \tilde{\Lambda}(t_{j})1(x \geq t_{j})\Big]\\
  & = \sum_{j=1}^{m}\frac{\tilde{f}(t_j)}{\tilde{M}}\Big[1_{[t_j,t_{j}+\Delta_t)}(x)(1-\delta)-\tilde{\Lambda}(t_{j}+\Delta_t)1(x \geq t_{j}+\Delta_t) + \tilde{\Lambda}(t_{j})1(x \geq t_{j})\Big]\\
  & = \sum_{j=1}^m \sum_{r=1}^4 \bigg\{\frac{\tilde{f}(t_j)}{\tilde{M}}(-1)^{r}\tilde{\Lambda}(t_{j}+\Delta_{t}1_{r=3})^{1_{r \ge 3}}\bigg\}(1-\delta)^{1_{r \le 2}}1\big(x \ge t_j + \Delta_{t}1_{r \in \{2,3\}}\big)\\
  &\equiv \sum_{j=1}^m \sum_{r=1}^4 \alpha_{jr}(1-\delta)^{1_{r \le 2}}1\big(x \ge t_j + \Delta_{t}1_{r \in \{2,3\}}\big),      
\end{align*}
where $\sum_{j=1}^m \sum_{r=1}^4 |\alpha_{jr}| \le 1$ by
\begin{align*}
  &\sum_{j=1}^m \sum_{r=1}^4 |\alpha_{jr}| \le \sum_{j=1}^{m}\sum_{r=1}^{4}\bigg|\frac{\tilde{f}(t_j)}{\tilde{M}}(-1)^{r}\tilde{\Lambda}(t_{j}+\Delta_{t}1_{r=3})^{1_{r \ge 3}}\bigg|  
  \le \sum_{j=1}^{m}\sum_{r=1}^{4}\,\bigg|\frac{\tilde{f}(t_j)}{\tilde{M}}\tilde{\Lambda}(t_j+\Delta_t)\bigg|\\
  &\quad \le \frac{4\tilde{M}_0(1+\tilde{M}_u)}{\tilde{M}}\sum_{j=1}^{m}\big|\tilde{\Lambda}(t_j+\Delta_t)\big|
  = 1.
\end{align*}
Then any integral function in $\tilde{\mathcal{F}}_{2}(k,1)$ is in a class contained in the scalar-multiplied pointwise sequential closure of the symmetric convex hull of the VC-subgraph class $$\left\{(x, \delta) \mapsto \delta1(x\ge s) : s\in\mathcal{T} \right\},$$ which is a class of indicator functions. Hence $\tilde{\mathcal{F}}_{2}(k,1)$ is a VC-hull class for sets, and so is $\tilde{\mathcal{F}}_{2}(k,2)$ because it is the square of $\tilde{\mathcal{F}}_{2}(k,1)$ \cite[Lemma 2.6.20 in][]{Vaart1996}.

Together with the fact that $\tilde{{\cal E}}(k,w)$ is shown VC-hull for sets as in Lemma \ref{lemma:tildeE_VC_hull}, repetitively applying Lemma 2.6.20 of \cite{Vaart1996} further indicates that $\tilde{\mathcal{F}}_{1}(k,q,r)\tilde{{\cal E}}(k,w)$ and $\tilde{\mathcal{F}}_{1}(k,q,r)\tilde{\mathcal{F}}_{2}(k,v)$
are VC-hull classes for sets.
Thanks to the preservation properties of VC-hull classes for sets, $\tilde{\mathcal{F}}(k,q,r,v,w)$, the union of 
$\tilde{\mathcal{F}}_{1}(k,q,r)\tilde{{\cal E}}(k,w)$ and $\tilde{\mathcal{F}}_{1}(k,q,r)\tilde{\mathcal{F}}_{2}(k,v)$, is a VC-hull class for sets.
Analogously, $\tilde{\mathcal{F}}(k)$ is also a VC-hull class for sets.
\end{proof}

By the assumption \ref{assump:Covariates} that the $U_k$ are uniformly bounded, there exists a uniform upper bound 
 $\tilde{M}_2 >0$ for $\sum_{r=0}^4|U_k|^r$. As we will now show, the following is an  envelope function for $\tilde{\mathcal{F}}(k)$ for all $k$:
\begin{align*}
  F: (x, \delta, \bs{u}) \mapsto
  \tilde{M}_2\bigg\{\sum_{q=0}^2\Big|\frac{\delta x}{\tilde{\epsilon}}\Big|^q\bigg\}
  \bigg\{\sum_{w=0}^2(\tilde{M}_0(1+\tilde{M}_u))^w\Big[1 + \big[1-\delta -\log(\tilde{\epsilon})\big]^2\Big]\bigg\}.     
\end{align*}
 First we show that $F$ is an envelope function for $\tilde{\mathcal{F}}(k,2,4,2,2)$;  similar arguments  apply to the other classes involving different values of $(q,r,v,w)$. For any function $f \in \tilde{\mathcal{F}}(k,2,4,2,2)$ depending on $\tilde{G} \in {\cal G}$ with $\tilde{\Lambda} = -\log(\tilde{G})$, we see that, for each $(\delta,x,u_k)$, $|f(\delta, x, u_k)|$ is bounded by
\begin{align*}
  &\bigg|\frac{u_k^4(\delta x)^2}{\tilde{G}^2(\tau)}\bigg|\bigg\{[a(\cdot)+b(\cdot)u_k]^2+\Big[\int_{\mathcal{T}}\big[a(s) + b(s)u_k\big]\big\{1(x \in ds, \delta=0) - 1(x \ge s)\tilde{\Lambda}(s)\big\}\Big]^2\bigg\} \\
  &\quad \leq \tilde{M}_2\bigg|\frac{(\delta x)^2}{\tilde{G}^2(\tau)}\bigg|\Big\{(\tilde{M}_0(1+\tilde{M}_u))^2\Big[1+\big[1-\delta + \tilde{\Lambda}(\tau)\big]^2\Big]\Big\} \\
  &\quad \leq \tilde{M}_2\bigg|\frac{\delta x}{\tilde{\epsilon}}\bigg|^2\Big\{(\tilde{M}_0(1+\tilde{M}_u))^2\Big[1+\big[1-\delta - \log(\tilde{\epsilon})\big]^2\Big]\Big\} \leq F,
\end{align*}
where the first inequality holds by seeing that $[a(\cdot)+b(\cdot)u_k] \in \tilde{{\cal E}}(k,1)$ so that $|a(\cdot)+b(\cdot)u_k| \le \tilde{M}_0(1+\tilde{M}_u)$; that $[1-\delta + \tilde{\Lambda}(\tau)]$ is the total variation of the signed measure $1(x \in ds, \delta=0) - 1(x \ge s)\tilde{\Lambda}(ds)$, and that
\begin{align*}
&\Big|\int_{\mathcal{T}}\big[a(s) + b(s)u_k\big]\big\{1(x \in ds, \delta=0) - 1(x \ge s)\tilde{\Lambda}(s)\big\}\Big| \leq \tilde{M}_0(1+\tilde{M}_u)\big[1-\delta + \tilde{\Lambda}(\tau)\big].%\\
% &\quad \leq \tilde{M}_0(1+\tilde{M}_u)\big[1-\delta + \tilde{\Lambda}(\tau)\big]^2,
\end{align*}
%owing to $1 \leq [1-\delta + \tilde{\Lambda}(\tau)]$ such that $[1-\delta + \tilde{\Lambda}(\tau)] \leq [1-\delta + \tilde{\Lambda}(\tau)]^2$.
The second inequality follows by $\tilde{G}(\tau) > \tilde{\epsilon} > 0$.

Because the $U_k$ is uniformly bounded in $P$-probability, $|\delta X| \leq \tau$ $P$-almost surely and $\tilde{G}(\tau) > \tilde{\epsilon} > 0$, $F$ is square-integrable: $\|F\|_{Q,2}^2 = \int F^2 dQ < \infty$ for any probability measure $Q$ on the sample space ${\cal X}$. 
Therefore, for all $k \ge 1$, Theorem 2.6.9 of \cite{Vaart1996} indicates
there exists a universal constant $K$ that does not depend on $k$ and $a\in (0,2)$ so that
\begin{align*} %\label{eq:bounded_uniform_entropy}
  \sup_Q \log N(\epsilon\|F\|_{Q,2},\; \tilde{\mathcal{F}}(k),\; L_2(Q)) \leq K \epsilon^{-a}.
\end{align*}
Moreover, the above display implies that
\begin{align*}
  &N(\epsilon\|F\|_{Q,2},\;\tilde{\mathcal{F}}_n,\;L_2(Q)) \leq \sum_{k \in \mathcal{K}_n} N(\epsilon\|F\|_{Q,2},\;\tilde{\mathcal{F}}(k),\;L_2(Q)) \leq p_n\exp(K \epsilon^{-a}),
\end{align*}
giving that
\begin{align}\label{eq:bounded_uniform_entropy_goal}
  \sup_Q \log N(\epsilon\|F\|_{Q,2},\; \tilde{\mathcal{F}}_n,\; L_2(Q)) \leq \log(p_n) + K \epsilon^{-a}.  
\end{align}

For $j =1,\ldots, n$,
define the empirical process $\{\mathbb{G}_{j}(\tilde{f}) : \tilde{f} \in \tilde{\mathcal{F}}_n\}$ pointwise as follows
\[
  \mathbb{G}_{j}(\tilde{f}) = \frac{1}{\sqrt{j}}\sum_{i=1}^j\Big[\tilde{f}(O_i)-P(\tilde f)\Big] = \sqrt{j}(\mathbb{P}_j-P)\tilde f,
\]
where $\mathbb{P}_j$ denotes the empirical distribution of $O_1, \ldots, O_j$.
Let $\|\mathbb{G}_{j}\|_{\tilde{\mathcal{F}}_n}=\sup_{\tilde{f} \in \tilde{\mathcal{F}}_n}|\mathbb{G}_{j}(\tilde{f})|.$
Following \eqref{eq:bounded_uniform_entropy_goal},
Theorem 2.14.1 of \cite{Vaart1996} gives $\|\mathbb{G}_{j}\|_{\tilde{\mathcal{F}}_n} \lesssim \sqrt{\log(p_n)}$, so that 
\begin{align*}
  \sup_{\tilde{f} \in \tilde{\mathcal{F}}_n} \Big|(\mathbb{P}_j-P)\tilde{f}\Big| \lesssim \sqrt{\log(p_n)/j},
\end{align*}
where $\lesssim$ means ``bounded above up to a universal multiplicative constant that does not depend on $(j,n)$.''
We also need the following lemmas.

\begin{lemma} \label{lemma:replace_Qn_Q}
For any sample size $n$, the event $\mathcal{A}_n$ occurs with probability at least $1-1/n$, where $\mathcal{A}_n=\cap_{j=1}^n\mathcal{A}_{nj}$, $\mathcal{A}_{nj}=\{\sup_{\tilde{f} \in \tilde{\mathcal{F}}_n} \big|(\mathbb{P}_j-P)\tilde{f}\big| \lesssim K_{nj}\}$ with $K_{nj} \equiv \sqrt{\log(n \lor p_n)/j}$ and $n \lor p_n =\sup(n, p_n)$.
\end{lemma}
\begin{proof}
Follow the same argument as in the proof of  Lemma A.4 of \cite{Luedtke2018}, except based on the class $\tilde{\mathcal{F}}_n$.
\end{proof}
Note that  Lemma \ref{lemma:replace_Qn_Q} reduces in the special case of $\cup_{k\in\mathcal{K}_n}\cup_{r=0}^4\tilde{\mathcal{F}}_{1}(k, 0, r, 0, 0) \in \tilde{\mathcal{F}}_n$, to $\mathbb{P}_j$ and $P$ being replaced by $\mathbb{Q}_j$ (the empirical distribution of $\bs{U}_1, \ldots, \bs{U}_j$) and $Q_u$, respectively.

Let $K_{nj} = \sqrt{\log(n \lor p_n)/j}$, for $j=q_n,\ldots,n$.
For $\tilde{K}\in(0,\infty)$, define
\begin{align*}
  \mathcal{B}_n(\tilde{K}) = &\bigg\{\sup_{t \in \mathcal{T}}\Big|\frac{\hat{G}_n(t)}{G(t)}-1\Big| \leq \sqrt{\frac{\log n}{n}},\:\: \sup_{t \in \mathcal{T}}\big|\hat{\Lambda}_n(t)-\Lambda(t)\big| \le \sqrt{\frac{\log n}{n}},\:\:\inf_{s \in \mathcal{T}}Y_n(s) \ge \sqrt{n},\\
  &\quad \sup_{(k,s) \in \mathcal{K}_n \times \mathcal{T}}\big|\hat{E}_j(U_k,s,k)-E_0(U_k,s,k)\big| \leq \tilde{K}K_{nj},\; j= q_n,\ldots,n \bigg\},
\end{align*}
where $Y_n(s) = \sum_{i=1}^n1(X_i \geq s)$ and $\hat{\Lambda}_n(\cdot)=\int_{-\infty}^{\cdot}\big[1(Y_n(s)>0)/Y_n(s)\big]dN_n(s)$ is the estimator of $\Lambda(\cdot)$.
The following lemma concerns the probability of the event $\mathcal{A}_n\,\cap\,\mathcal{B}_n(\tilde{K})$. 

\begin{lemma} \label{lemma:An_Bn_prob_to_one}
Under the conditions of Theorem \ref{Thm:stab_one_step}, there exists a $\tilde{K} \in (0, \infty)$ such that ${\rm P}(\mathcal{A}_n\,\cap\, \mathcal{B}_n(\tilde{K}))\to 1$. 
\end{lemma}
\begin{proof}
Fix $\tilde{K}$. It holds that ${\rm P}\left(\mathcal{A}_n\cap \mathcal{B}_n(\tilde{K})\right)= {\rm P}\left(\mathcal{A}_n\right) - {\rm P}\left( \mathcal{A}_n\cap \mathcal{B}_n(\tilde{K})^c\right)$. Hence, by Lemma~\ref{lemma:replace_Qn_Q}, we have that:
\begin{align}
  {\rm P}\left(\mathcal{A}_n\cap \mathcal{B}_n(\tilde{K})\right)&\ge 1-1/n - {\rm P}\left( \mathcal{A}_n\cap \mathcal{B}_n(\tilde{K})^c\right). \label{eq:initialDecomp}
\end{align}
It therefore remains to show that, for any appropriate choice of $\tilde{K}$, ${\rm P}\left( \mathcal{A}_n\cap \mathcal{B}_n(\tilde{K})^c\right)=o(1)$. To show this, we use that $\mathcal{A}_n\cap \mathcal{B}_n(\tilde{K})^c\subseteq \cup_{j=1}^3 \mathcal{R}_{nj}$, where
\begin{align*}
  \mathcal{R}_{n1} &\equiv  \left\{\sup_{t \in \mathcal{T}}\Big|\frac{\hat{G}_n(t)}{G(t)}-1\Big| > \sqrt{\frac{\log n}{n}}\right\} \cup \left\{\sup_{t \in \mathcal{T}}\big|\hat{\Lambda}_n(t)-\Lambda(t)\big| > \sqrt{\frac{\log n}{n}}\right\}, \\
  \mathcal{R}_{n2} &\equiv \left\{\inf_{s \in \mathcal{T}}Y_n(s) < \sqrt{n}\right\}, \\
  \mathcal{R}_{n3} &\equiv \mathcal{A}_n \cap \Bigg\{\sup_{t \in \mathcal{T}}\Big|\frac{\hat{G}_n(t)}{G(t)}-1\Big| \leq \sqrt{\frac{\log n}{n}} \Bigg\}\\
  &\hspace{2.5em} \cap \bigcup_{j=q_n}^n\bigg\{\sup_{(k,s) \in \mathcal{K}_n \times \mathcal{T}}\big|\hat{E}_j(U_k,s,k)-E_0(U_k,s,k)\big| > \tilde{K}K_{nj} \Bigg\}.
\end{align*}
By a union bound, this yields that ${\rm P}(\mathcal{A}_n\cap \mathcal{B}_n(\tilde{K})^c)\le \sum_{h=1}^3 P(\mathcal{R}_{nh})$. In the remainder of this proof, we will establish the following three facts
\begin{align}\label{eq:event_prob_convergence}
  {\rm P}(\mathcal{R}_{n1}) = o(1),\:\: 
  {\rm P}(\mathcal{R}_{n2}) = o(1)\:\: \mbox{and}\:\:
  {\rm P}(\mathcal{R}_{n3}) = 0.
\end{align}
Combining these facts with \eqref{eq:initialDecomp} will then yield the result.

We first show that ${\rm P}(\mathcal{R}_{n1}) \rightarrow 0$. Observing that $\{\sqrt{n}[\hat{G}_n(t)/G(t)-1] : t\in\mathcal{T}\}$ converges to a tight Gaussian process, and then the continuous mapping theorem gives that $\sqrt{n}\sup_{t\in\mathcal{T}}|\hat{G}_n(t)/G(t)-1|$ converges to the supremum of the absolute value of this Gaussian process. Therefore for any sequence $\varepsilon_n \rightarrow \infty$, ${\rm P}(\sup_{t \in \mathcal{T}}\sqrt{n}|\hat{G}_n(t)/G(t)-1| > \varepsilon_n \big) \rightarrow 0$, in particular, $\varepsilon_n=\sqrt{\log n}$. The identical argument applies to 
yield that ${\rm P}(\sup_{t \in \mathcal{T}}\sqrt{n}|\hat{\Lambda}_n(t)-\Lambda(t)| > \varepsilon_n \big) \rightarrow 0$.
Consequently, ${\rm P}(\mathcal{R}_{n1})=o(1)$.

To obtain ${\rm P}(\mathcal{R}_{n2}) \rightarrow 0$,
observe that $Y_n(\tau) \sim \mbox{Binomial}(n, p_*)$ with $p_*={\rm P}(X \geq \tau)$ and ${\rm P}(X \geq \tau) > 0$ by \ref{assump:At-risk prob}. Along with the monotonicity of $Y_n(s)$ in $s \in \mathcal{T}$, Hoeffding's inequality gives
\begin{align*}
  {\rm P}(\mathcal{R}_{n2}) & \le {\rm P}\bigg( \inf_{s \in \mathcal{T}}Y_n(s) \leq \sqrt{n} \bigg)
  = {\rm P}\bigg( Y_n(\tau) \leq \sqrt{n} \bigg)
  \leq \exp(-2(\sqrt{n}p_*-1)^2) \rightarrow 0.
\end{align*}
In the remainder of the proof, we show that ${\rm P}(\mathcal{R}_{n3}) =0$. It suffices to show that
\begin{align}
    &\mathcal{A}_n \cap \left\{\sup_t\Big|\frac{\hat{G}_n(t)}{G(t)}-1\Big| \le \sqrt{\frac{\log n}{n}}\right\} \nonumber \\
    &\quad \subseteq \bigcap_{j=q_n}^n\left\{\sup_{(k,s) \in \mathcal{K}_n \times \mathcal{T}}\big|\hat{E}_j(U_k,s,k)-E_0(U_k,s,k)\big| \le \tilde{K}K_{nj}\right\}, \label{eq:finalDecomp}
\end{align}
giving that $\mathcal{R}_{n3} = \emptyset$ and then in turn ${\rm P}(\mathcal{R}_{n3}) =0$.

For the rest of the proof, suppose that $\mathcal{A}_n$ occurs and  $\sup_t|\hat{G}_n(t)/G(t)-1| \le \sqrt{\log(n)/n}$. To simplify notation, let $Y^s = Y1(X \geq s)$; $\tilde{Y}^s = \tilde{Y}1(X \geq s)$ and $U_{k}^s = U_k1(X \geq s)$. Taylor expanding $Y^s$ with respect to $\hat{G}_n$ around $G$ gives 
\begin{align}\label{eq:conditions_for_Bn_1}
    & \sup_{s}\Big\{P|Y^s-\tilde{Y}^s|\Big\} \leq  \sum_{r=1}^{\infty}\Big[\sup_{t}\Big|\frac{\hat{G}_n(t)}{G(t)}-1\Big|\,\Big]^r
    \sup_{s}\Big\{P\big|\tilde{Y}^s\big|\Big\} \nonumber \\
    & \quad \leq \sup_{s}\Big\{P\big|\tilde{Y}^s\big|\Big\}\sum_{r=1}^{\infty}\bigg[\sqrt{\frac{\log n}{n}}\,\bigg]^r
    \leq P\big|\tilde{Y}\big|\sum_{r=1}^{\infty}\bigg[\sqrt{\frac{\log n}{n}}\,\bigg]^r \nonumber \\
    & \quad = \sqrt{\frac{\log n}{n}} P\big|\tilde{Y}\big|\sum_{r=1}^{\infty}\bigg[\sqrt{\frac{\log n}{n}}\,\bigg]^{r-1} \lesssim \sqrt{\frac{\log n}{n}}, 
\end{align}
where the last steps holds by $\sqrt{\log n/n} < 1$ such that
$\sum_{r=1}^{\infty}[\sqrt{\log n/n}\;]^{r-1} < \infty$,
and $|\tilde{Y}| = |\delta X/G(X)|$ is a bounded random variable.
Similarly, the triangle inequality gives 
\begin{align}
  &\sup_{(k,s)}\big|{\rm Cov}(U_k^s,Y^s-\tilde{Y}^s) \big| \leq
  \sup_{(k,s)}\Big\{P|U_k^s(Y^s-\tilde{Y}^s)|+P|U_k^s|P|Y^s-\tilde{Y}^s|\Big\} \nonumber\\
  &\quad \leq \sup_{(k,s)}\Big\{P\big|U_k^s\tilde{Y}^s\big| + P|U_k^s|P\big|\tilde{Y}^s\big| \Big\}\sum_{r=1}^{\infty}\Big[\sup_{t}\Big|\frac{\hat{G}_n(t)}{G(t)}-1\Big|\,\Big]^r\nonumber\\
  &\quad \leq \sup_{(k,s)}\Big\{P\big|U_k^s\tilde{Y}^s\big| + P|U_k^s|P\big|\tilde{Y}^s\big| \Big\}\sum_{r=1}^{\infty}\bigg[\sqrt{\frac{\log n}{n}}\,\bigg]^r \nonumber\\
  &\quad \leq \sqrt{\frac{\log n}{n}} \max_{k}\Big\{P\big|U_k\tilde{Y}\big| + P|U_k|P\big|\tilde{Y}\big|\Big\} \sum_{r=1}^{\infty}\bigg[\sqrt{\frac{\log n}{n}}\,\bigg]^{r-1}
  \lesssim \sqrt{\frac{\log n}{n}}. \label{eq:conditions_for_Bn_1point5}
\end{align}
By the triangle inequality,
\begin{align*}
  &\sup_{(k,s)}\big|{\rm Cov}_{\Pj}(U_k^s, Y^s)-{\rm Cov}(U_k^s, Y^s)\big|\\ 
  &\quad = \sup_{(k,s)}\Big\{\big|\Pj[U_k^sY^s]-\Pj[U_k^s]\Pj[Y^s]-P[U_k^sY^s]+P[U_k^s]P[Y^s]\big|\Big\}\\
  &\quad \leq \sup_{(k,s)}\Big\{\big|(\Pj-P)[U_k^sY^s]\big| + \big|(\Pj-P)[U_k^s]\big|\big|(\Pj-P)[Y^s]\big|\\  
  &\hspace{2cm} + \big|P[U_k^s]\big|\big|(\Pj-P)[Y^s]\big| + \big|P[Y^s]\big|\big|(\Pj-P)[U_k^s]\big|\Big\},
\end{align*}
and based on the above display, we have
\begin{align}
  \sup_{(k,s)}\big|{\rm Cov}_{\Pj}(U_k^s, Y^s)-{\rm Cov}(U_k^s, Y^s)\big|  
  \lesssim K_{nj},
                                \label{eq:conditions_for_Bn_2}
\end{align}
when $\mathcal{A}_n$ occurs, since the $U_k$ is assumed to be uniformly bounded by \ref{assump:Covariates} and  $|Y^s|$ is bounded as implied by \ref{assump:Survival function} and $|X| \leq \tau$, along with
$\log(n \lor p_n)/j \leq 1$ so that $\log(n \lor p_n)/j \leq K_{nj}$, for $j=q_n,\ldots,n$. 

Regarding the assumption in \ref{assump:Variances} that ${\rm Var}(U_k^s)$ is uniformly bounded away from zero, there exists $\tilde{\zeta} \in (0, \infty)$ so that
$\min_{(k,s)}{\rm Var}(U_k^s) > \tilde{\zeta} > 0$. Meanwhile, let $\tilde{\eta}$ be the smallest universal positive constant to maintain $\sup_{(k,s)}|{\rm Var}_{\Pj}(U_k^s)-{\rm Var}(U_k^s)| \leq \tilde{\eta} K_{nj}$ that is implied by
the occurrence of $\mathcal{A}_n$. Let $q_n \geq \lceil 4\tilde{\eta}^2\log(n \lor p_n)/\tilde{\zeta}^{2} \rceil$ and
\begin{align*}
  &{\rm Var}_{\Pj}(U_k^s)={\rm Var}(U_k^s)+{\rm Var}_{\Pj}(U_k^s)-{\rm Var}(U_k^s) \\
  &\quad \geq {\rm Var}(U_k^s) - \sup_{(k,s)}|{\rm Var}_{\Pj}(U_k^s)-{\rm Var}(U_k^s)|
  \geq {\rm Var}(U_k^s) - \tilde{\eta} K_{nj}\\
  &\quad \geq \min_{(k,s)}{\rm Var}(U_k^s) - \tilde{\zeta}/2 > \tilde{\zeta}/2,\;\; \forall\; k,\,s.
\end{align*}
This yields 
\begin{equation}
  \sup_{(k,s)}\Big|\frac{1}{{\rm Var}_{\Pj}(U_k^s)}-\frac{1}{{\rm Var}(U_k^s)}\Big| \leq
  \frac{\sup_{(k,s)}|{\rm Var}_{\Pj}(U_k^s)-{\rm Var}(U_k^s)|}{[\min_{(k,s)}{\rm Var}_{\Pj}(U_k^s)][\min_{(k,s)}{\rm Var}(U_k^s)]}
  \lesssim K_{nj}.
                                    \label{eq:conditions_for_Bn_3}
\end{equation}
Using the results in \eqref{eq:conditions_for_Bn_1}-\eqref{eq:conditions_for_Bn_3} gives
\begin{align*}
  &\sup_{(k,s)}\Big|\big(U_k^s-\mathbb{P}_{j}[U_k^s]\big){\rm Cov}_{\mathbb{P}_{j}}(U_k^s, Y^s)
  - \big(U_k^s-P[U_k^s]\big){\rm Cov}(U_k^s, \tilde{Y}^s)\Big| \\
  &\quad \leq \sup_{(k,s)}\bigg\{ \big|U_k^s-P[U_k^s]\big|\big|{\rm Cov}_{\mathbb{P}_{j}}(U_k^s, Y^s)-{\rm Cov}(U_k^s, Y^s)\big|\\
  & \hspace{2cm} + \big|(\mathbb{P}_{j}-P)[U_k^s]\big|\big|{\rm Cov}_{\mathbb{P}_{j}}(U_k^s, Y^s)-{\rm Cov}(U_k^s, Y^s)\big|\\
  & \hspace{2cm} +\big|U_k^s-P[U_k^s]\big|\big|{\rm Cov}(U_k^s, Y^s-\tilde{Y}^s)\big|
  + \big|(\mathbb{P}_{j}-P)[U_k^s]\big|\big|{\rm Cov}(U_k^s, Y^s-\tilde{Y}^s)\big|\\
  &\hspace{2cm} + \big|(\Pj-P)[U_k^s]\big|\big|{\rm Cov}(U_k^s, \tilde{Y}^s)\big| \bigg\}\\
  %&\quad \lesssim \sqrt{\frac{\log(n \lor p_n)}{j}} + \frac{\log(n \lor p_n)}{j} + \sqrt{\frac{\log n}{j}} + \sqrt{\frac{\log(n \lor p_n)}{j}}\sqrt{\frac{\log n}{j}} \\
  &\quad \lesssim K_{nj},
\end{align*}
following the bounded values of $|\tilde{Y}^s|$ and that of $|U_k^s|$ over $k$, which are implied by \ref{assump:Covariates}--\ref{assump:Survival function} and $|X| \leq \tau$.
As stated, suppose that $\hat{E}_j$ is fitted via the linear regression approach described in Remark 3.5 and as defined in \eqref{eq:Ejlinear}. Applying similar arguments and the above results yields
\begin{align*}
  &\sup_{(k,s)}\big|\hat{E}_j(U_k,s,k)-E_0(U_k,s,k)\big| \leq 
  \sup_{(k,s)}\bigg\{\big|(\Pj-P)[Y_s]\big| + P|Y_s-\tilde{Y}_s|\\
  & \qquad + \frac{1}{{\rm Var}(U_{k,s})}\Big|\big(U_{k,s}-\mathbb{P}_{j}[U_{k,s}]\big){\rm Cov}_{\mathbb{P}_{j}}(U_{k,s}, Y_s)
  - \big(U_{k,s}-P[U_{k,s}]\big){\rm Cov}(U_{k,s}, \tilde{Y}_s)\Big|\\
  & \qquad + \Big|\frac{1}{{\rm Var}_{\mathbb{P}_{j}}(U_{k,s})}-\frac{1}{{\rm Var}(U_{k,s})}\Big|\Big|\big(U_{k,s}-\mathbb{P}_{j}[U_{k,s}]\big){\rm Cov}_{\mathbb{P}_{j}}(U_{k,s}, Y_s)\\
  &\hspace{6cm} - \big(U_{k,s}-P[U_{k,s}]\big){\rm Cov}(U_{k,s}, \tilde{Y}_s)\Big|\\
  & \qquad + \Big|\big(U_{k,s}-P[U_{k,s}]\big){\rm Cov}(U_{k,s}, \tilde{Y}_s)\Big|\Big|\frac{1}{{\rm Var}_{\mathbb{P}_{j}}(U_{k,s})}-\frac{1}{{\rm Var}(U_{k,s})}\Big|\bigg\}\\
  & \quad \lesssim K_{nj} + \sqrt{\log n/j}
  \leq \tilde{K}K_{nj},
\end{align*}
for some constant $\tilde{K} \in (1, \infty)$ that does not depend on $(j,n)$, leading to 
\begin{align}
  \sup_{(k,s)}\big|\hat{E}_j(U_k,s,k)-E_0(U_k,s,k)\big| \leq 
  \tilde{K}K_{nj}.
                        \label{eq:upper_bound_for_Ehat_deviation}
\end{align}
Hence we have shown that \eqref{eq:finalDecomp} holds and conclude the proof.
\end{proof}

\begin{lemma} \label{lemma:An_Bn_Cn_prob_to_one}
Let $\tilde{K}$ be given in Lemma \ref{lemma:An_Bn_prob_to_one} and $\mathcal{I}_n\equiv \{(j,k) : j \in \{q_n,\ldots,n\}, k \in {\cal K}_n\}$. Suppose the conditions of Theorem \ref{Thm:stab_one_step} hold, and introduce the event $\mathcal{C}_n$ that occurs when
\begin{lemmaitemenum}
  \item \label{eq:event_Cn_1} $\max_{k \in \mathcal{K}_n}|{\rm Var}_{\mathbb{Q}_j}(U_k)-{\rm Var}_{Q_u}(U_k)| \lesssim K_{nj}$,
  \item \label{eq:event_Cn_2} $\min_{(j,k) \in \mathcal{I}_n}{\rm Var}_{\mathbb{Q}_j}(U_k)$ is bounded away from zero,
  \item \label{eq:event_Cn_3} $\max_{k \in \mathcal{K}_n}\big|(\Qj-Q_u)\big[U^r_{k} \{E_0(U_{k},k)\}^w\big]\big| \lesssim K_{nj}$, $0 \leq r, w \leq 1$, 
  \item \label{eq:event_Cn_4} $\max_{k \in \mathcal{K}_n}\big|1/{\rm Var}_{\Qj}(U_{k})-1/{\rm Var}_{Q_u}(U_{k})\big| \lesssim K_{nj}$,
  \item \label{eq:event_Cn_5}
  $\sup_{(k,s) \in {\cal K}_n \times {\cal T}}\big|\hat{E}_j(U_k,s,k)\big| \lesssim K_{nj} + \sup_{(k,s) \in {\cal K}_n \times {\cal T}}|E_0(U_k,s,k)|$,
  \item \label{eq:event_Cn_6}
  $\max_{k \in {\cal K}_n}|U_k-\Qj[U_k]| \lesssim K_{nj} + \max_{k \in {\cal K}_n}|U_k-Q_u[U_k]|$,
\end{lemmaitemenum}
where each of \ref{eq:event_Cn_1}--\ref{eq:event_Cn_6} relies on appropriately specified constants that do not depend on $(j,n)$. Then such constants exist such that
${\rm P}(\mathcal{A}_n \cap \mathcal{B}_n(\tilde{K})\cap {\cal C}_n) \to 1$.
\end{lemma}
\begin{proof}
When $\mathcal{A}_n$ occurs, the triangle inequality gives that
\begin{align*}
  &\max_{k \in \mathcal{K}_n}|{\rm Var}_{\mathbb{Q}_j}(U_k)-{\rm Var}_{Q_u}(U_k)|\\
  &\quad =\max_{k \in \mathcal{K}_n}\Big|(\mathbb{Q}_j[U_k]^2-Q_u[U_k]^2)-(\mathbb{Q}_j[U_k]+Q_u[U_k])(\mathbb{Q}_j[U_k]-Q_u[U_k])\Big|\\
  &\quad \lesssim \max_{k \in \mathcal{K}_n}\Big|\mathbb{Q}_j[U_k]^2-Q_u[U_k]^2 \Big|+ \max_{k \in \mathcal{K}_n}\Big|\mathbb{Q}_j[U_k]-Q_u[U_k]\Big| \lesssim K_{nj}.
\end{align*}
Let ${\cal C}_{n1}$ correspond to the event \ref{eq:event_Cn_1} using the $(j,n)$-independent constant implied by the above display, which gives that
${\rm P}(\mathcal{A}_n \cap \mathcal{B}_n(\tilde{K}) \cap {\cal C}_{n1}^c) \to 0.$

By \ref{assump:Variances}, there exists $\zeta>0$ such that
$\min_{k \in {\cal K}_n}{\rm Var}_{Q_u}(U_k)\geq \zeta$. Moreover, if ${\cal C}_{n1}$ holds, then there exists $\eta>0$ such that for all $n \in \mathbb{N}$ and $(j,k)\in\mathcal{I}_n$, $|{\rm Var}_{\mathbb{Q}_j}(U_k)-{\rm Var}_{Q_u}(U_k)| \leq \eta K_{nj}$. This yields, for all $n$ and all $(j,k)\in\mathcal{I}_n$, 
\begin{align*}
  {\rm Var}_{\mathbb{Q}_j}(U_k)&={\rm Var}_{Q_u}(U_k)+{\rm Var}_{\mathbb{Q}_j}(U_k)-{\rm Var}_{Q_u}(U_k) \\
  &\geq {\rm Var}_{Q_u}(U_k) - |{\rm Var}_{\mathbb{Q}_j}(U_k)-{\rm Var}_{Q_u}(U_k)| \geq \zeta - \eta K_{nj}.
%   &\quad \geq \min_{k}{\rm Var}_{Q_u}(U_k) - \eta\sqrt{\log(n \lor p_n)/q_n}\; >\, \zeta/2
\end{align*}
By the conditions of Theorem \ref{Thm:stab_one_step}, $\sqrt{q_n}/\log(n \lor p_n)\to\infty$ and let $q_n/\log(n \lor p_n) \geq 4\eta^2/\zeta^2$ for all $n$ sufficiently large. Hence, the above shows that, for all such $n$, ${\rm Var}_{\mathbb{Q}_j}(U_k)\geq \zeta/2$ when $\mathcal{A}_n \cap \mathcal{B}_n(\tilde{K}) \cap {\cal C}_{n1}$ occurs. Letting $\mathcal{C}_{n2}$ denote the event that $\min_{(j,k) \in \mathcal{I}_n}{\rm Var}_{\mathbb{Q}_j}(U_k)\geq\zeta/2$, and we show that 
${\rm P}(\mathcal{A}_n \cap \mathcal{B}_n(\tilde{K}) \cap {\cal C}_{n2}^c) \to 0$.

We could regard \ref{eq:event_Cn_3} as a consequence of the occurrence of $\mathcal{A}_n$ because the function $u_k \mapsto
u^r_{k}\{E_0(u_{k},k)\}^w$ belongs to $\tilde{\mathcal{F}}_n$. 
Let $\mathcal{C}_{n3}$ correspond to the event \ref{eq:event_Cn_3}, and ${\rm P}(\mathcal{A}_n \cap \mathcal{B}_n(\tilde{K}) \cap {\cal C}_{n3}^c) \to 0.$
Along with the fact that $\min_{k \in {\cal K}_n}{\rm Var}_{Q_u}(U_k)$ is bounded away from zero that is implied by \ref{assump:Variances}, using \ref{eq:event_Cn_1} and \ref{eq:event_Cn_2} of Lemma \ref{lemma:An_Bn_Cn_prob_to_one} gives
\begin{align*}
  \max_{k \in \mathcal{K}_n}\bigg|\frac{1}{{\rm Var}_{\Qj}(U_{k})}-\frac{1}{{\rm Var}_{Q_u}(U_{k})}\bigg| = \frac{\max_{k \in \mathcal{K}_n}|{\rm Var}_{\Qj}(U_{k})-{\rm Var}_{Q_u}(U_{k})|}{[\min_{k \in \mathcal{K}_n}{\rm Var}_{\Qj}(U_{k})][\min_{k \in \mathcal{K}_n}{\rm Var}_{Q_u}(U_{k})]} \lesssim K_{nj}.
\end{align*}
Let $\mathcal{C}_{n4}$ correspond to the event \ref{eq:event_Cn_4},
and $\mathcal{C}_{n1} \cap \mathcal{C}_{n2} \subseteq \mathcal{C}_{n4}$, which gives
${\rm P}(\mathcal{A}_n \cap \mathcal{B}_n(\tilde{K}) \cap {\cal C}_{n4}^c) \to 0.$

According to \ref{assump:Conditional_mean_E0}, $\sup_{(k,s)}|E_0(U_k,s,k)|$ is uniformly bounded by some $n$-independent constant. Therefore when ${\cal B}_n(\tilde{K})$ occurs, we have $\sup_{(k,s)}\big|\hat{E}_j(U_k,s,k) - E_0(U_k,s,k)\big| \leq \tilde{K}K_{nj}$ for all $j$, implying that
\begin{align*}
  \sup_{(k,s)}\big|\hat{E}_j(U_k,s,k)\big| \leq
  \tilde{K}K_{nj} + \sup_{(k,s)}\big|E_0(U_k,s,k)\big| \lesssim K_{nj} + \sup_{(k,s)}\big|E_0(U_k,s,k)\big|.
\end{align*}
Thus, letting $\mathcal{C}_{n5}$ denote the event that $$\sup_{(k,s)}|\hat{E}_j(U_k,s,k)| \lesssim K_{nj} + \sup_{(k,s)}\big|E_0(U_k,s,k)\big|,\; j=q_n,\ldots,n,$$ we have 
${\rm P}(\mathcal{A}_n \cap \mathcal{B}_n(\tilde{K}) \cap {\cal C}_{n5}^c) \to 0$.
Similarly, we have that $\max_k|U_k-Q_u[U_k]|$ is uniformly bounded by some finite constant that does not depend on $(k,n)$, by \ref{assump:Covariates}.
Therefore when ${\cal A}_n$ occurs, 
\begin{align*}
  &\max_k|U_k-\Qj[U_k]| \leq \max_k|(\Qj-Q_u)[U_k]| + \max_k|U_k-Q_u[U_k]|\\
  &\quad \lesssim K_{nj} + \max_k|U_k-Q_u[U_k]|.
\end{align*}
Let $\mathcal{C}_{n6}$ denote the event that for $j=q_n,\ldots,n$, $\max_k|U_k-\Qj[U_k]| \lesssim K_{nj} + \max_k|U_k-Q_u[U_k]|$ and
${\rm P}(\mathcal{A}_n \cap \mathcal{B}_n(\tilde{K}) \cap {\cal C}_{n6}^c) \to 0$.
Letting $\mathcal{C}_n\equiv\cap_{q=1}^6\mathcal{C}_{nq}$, we have shown
${\rm P}(\mathcal{A}_n \cap \mathcal{B}_n(\tilde{K}) \cap {\cal C}_n) \geq 1 - \sum_{q=1}^6{\rm P}(\mathcal{A}_n \cap \mathcal{B}_n(\tilde{K}) \cap {\cal C}_{nq}^c) \to 1$.
\end{proof}

For some constant $\tilde{K}$ as given in Lemma \ref{lemma:An_Bn_prob_to_one}, $$d_n(P_1,P_2) \equiv P\Big[\max_{k \in \mathcal{K}_n}|\IF^*_{k}(O|P_1)-\IF^*_{k}(O|P_2)|1_{\mathcal{A}_n \cap \mathcal{B}_n(\tilde{K}) \cap {\cal C}_n}\,\Big].$$
Recall that $\hat{P}_{nj} = (\hat{E}_j, \mathbb{Q}_j, \hat{G}_n)$ and $\hat{P}'_j = (\hat{E}_j, \mathbb{Q}_j, G)$ as defined in Section \ref{sec:notation}.
\begin{lemma} \label{lemma:An_Bn_Cn_Dn_prob_to_one}
Suppose the conditions of Theorem \ref{Thm:stab_one_step} hold. Then with $\tilde{K} \in (0, \infty)$ given in Lemma \ref{lemma:An_Bn_prob_to_one}, there exists $K'\in (1,\infty)$ such that ${\rm P}(\mathcal{A}_n \cap \mathcal{B}_n(\tilde{K}) \cap {\cal C}_n \cap
\mathcal{D}_{n}(K')) \to 1$, where
\begin{align*}
  \mathcal{D}_n(K') = \big\{d_n(\hat{P}_{nj}, \hat{P}_{j}') \lor d_n(\hat{P}_{j}',P) \leq K'K_{nj},\,j=q_n,\ldots,n \big\}.
\end{align*}
\end{lemma}
\begin{proof}
Let
\begin{align} \label{eq:Lnk}
  L_{njk} = \bigg| \frac{(U_k-\Qj[U_k])}{{\rm Var}_{\Qj}(U_k)}\int_{\mathcal{T}}\hat{E}_j(U_k,s,k)(d\hat{M}(s)-dM(s)) \bigg|1_{\mathcal{A}_n \cap\,\mathcal{B}_n(\tilde{K}) \cap\,\mathcal{C}_n}. 
\end{align}
The triangle inequality first gives the decomposition
\begin{align} \label{eq:upper_bound_distance1}
  &|\IF^*_k(O|\hat{P}_{nj})-\IF^*_k(O|\hat{P}'_j)|1_{\mathcal{A}_n \cap \mathcal{B}_n(\tilde{K}) \cap\,\mathcal{C}_n} = \bigg| \frac{(U_k-\Qj[U_k])\delta X}{{\rm Var}_{\Qj}(U_k)}\left(\frac{1}{\hat{G}_n(X)}-\frac{1}{G(X)}\right) \\
  & \hspace{1cm} + \frac{(U_k-\Qj[U_k])}{{\rm Var}_{\Qj}(U_k)} \int_{\mathcal{T}}\hat{E}_j(U_k,s,k)(d\hat{M}(s)-dM(s)) \bigg|1_{\mathcal{A}_n \cap \mathcal{B}_n(\tilde{K}) \cap\,\mathcal{C}_n} \nonumber \\
  & \le \bigg| \frac{(U_k-\Qj[U_k])\delta X}{{\rm Var}_{\Qj}(U_k)}\left(\frac{1}{\hat{G}_n(X)}-\frac{1}{G(X)}\right)
  \bigg|1_{\mathcal{A}_n \cap \mathcal{B}_n(\tilde{K}) \cap\,\mathcal{C}_n} 
  + L_{njk}. \nonumber
\end{align}  
Moreover,
\begin{align} \label{eq:upper_bound_distance1.1}
  &\max_{k \in {\cal K}_n}\bigg| \frac{(U_k-\Qj[U_k])\delta X}{{\rm Var}_{\Qj}(U_k)}\left(\frac{1}{\hat{G}_n(X)}-\frac{1}{G(X)}\right)
  \bigg|1_{\mathcal{A}_n \cap \mathcal{B}_n(\tilde{K}) \cap\,\mathcal{C}_n}\\
  &\leq 1_{\mathcal{A}_n \cap \mathcal{B}_n(\tilde{K}) \cap\, \mathcal{C}_n}\sum_{r=1}^{\infty}\max_{k \in {\cal K}_n}\bigg|\frac{(U_k-\Qj[U_k])}{{\rm Var}_{\Qj}(U_k)}\tilde{Y}\bigg|\Big[\sup_{t}\Big|\frac{\hat{G}_n(t)}{G(t)}-1\Big|\,\Big]^r \nonumber \\
  %&\quad \leq 1_{\mathcal{A}_n \cap \mathcal{B}_n(\tilde{K}) \cap\, \mathcal{C}_n}\sum_{r=1}^{\infty}\max_{k \in {\cal K}_n}\bigg|\frac{(U_k-\Qj[U_k])}{{\rm Var}_{\Qj}(U_k)}\tilde{Y}\bigg|\bigg[\sqrt{\frac{\log n}{n}}\,\bigg]^r \nonumber \\
  &\leq \max_{k \in {\cal K}_n}\bigg|\frac{(U_k-\Qj[U_k])}{{\rm Var}_{\Qj}(U_k)}\tilde{Y}\bigg|1_{\mathcal{A}_n \cap \mathcal{B}_n(\tilde{K}) \cap\, \mathcal{C}_n}\,\sum_{r=1}^{\infty}\bigg[\sqrt{\frac{\log n}{n}}\,\bigg]^{r}
  \lesssim \big[K_{nj}+1\big]\sqrt{\frac{\log n}{n}} \leq K^{''}K_{nj} \nonumber,
\end{align}
where the first inequality holds by Taylor expansion and the triangle inequality, and the second inequality results from
$\sup_{t}|\hat{G}_n(t)/G(t)-1| \leq \sqrt{\log n/n}$ when $\mathcal{B}_n(\tilde{K})$ occurs. The last two steps in the above display hold because there exists some $(j,n)$-independent constant $K^{''} \in (1, \infty)$ such that
\begin{align*}
  \max_{k \in {\cal K}_n}\bigg|\frac{(U_k-\Qj[U_k])}{{\rm Var}_{\Qj}(U_k)}\tilde{Y}\bigg|\,\sum_{r=1}^{\infty}\bigg[\sqrt{\frac{\log n}{n}}\,\bigg]^{r}
  \lesssim \big[K_{nj}+1\big]\sqrt{\frac{\log n}{n}} \leq K^{''}K_{nj},
\end{align*}
following $|\tilde{Y}| = |\delta X/G(X)| < \infty$ and $\max_k\{|U_k-\Qj[U_k]|/{\rm Var}_{\Qj}(U_k)\} \lesssim K_{nj} + 1$, according to \ref{assump:Covariates}, \ref{eq:event_Cn_2} and \ref{eq:event_Cn_6} of Lemma \ref{lemma:An_Bn_Cn_prob_to_one} in view of the occurrence of ${\cal C}_n$, and $\sum_{r=1}^{\infty}[\sqrt{\log n/n}\,]^{r-1} < \infty$.
In addition, below we show ${\rm P}\big(\max_{j}\{P[\max_{k \in \mathcal{K}_n}L_{njk}]\} > 1/\sqrt{\log(n \lor p_n)}\,\big) = o(1)$.
By Markov's inequality and Jensen's inequality,
\begin{align} \label{eq:upper_bound_prob_max_Lnk}
  &{\rm P}\Big(\max_{j}\Big\{P\Big[\max_{k \in \mathcal{K}_n}L_{njk}\Big]\Big\} > 1/\sqrt{\log(n \lor p_n)}\Big)
  \leq \log(n \lor p_n)E\Big\{\max_{j}\Big\{P\Big[\max_{k \in \mathcal{K}_n}L_{njk}\Big]\Big\}^2\Big\} \\
  & = \log(n \lor p_n)\mathbb{E}\Big\{\max_{j}\Big\{P\Big[\max_{k \in \mathcal{K}_n}L_{njk}\Big]\Big\}^2\Big\}
  \leq \log(n \lor p_n)\mathbb{E}\Big\{P\Big[\max_{(j,k)}L^2_{njk}\Big]\Big\}. \nonumber
\end{align}
Recall $N_n(s)$ and $Y_n(s)$ as defined in Section \ref{sec:notation} except for removing $u$; $d\bar{M}(s) \equiv dN_n(s) - Y_n(s)d\Lambda(s)$ is a local martingale with respect to the aggregated filtration
\begin{align}\label{eq:filtration_Fprime.1}
  \mathcal{F}'_{s} = \sigma(\{N_n(s'),\, Y_n(s'),\,s' \le s \in \mathcal{T}\}, \{\bs{U}_i\}_{i=1}^n).
\end{align}
Applying the decomposition $d\hat{M}(s)-dM(s) = 1(X \geq s)(d\Lambda(s)-d\hat{\Lambda}_n(s))$ to the expression of $L_{njk}$ in \eqref{eq:Lnk} and using the inequalities $(a-b)^2 \leq 2(a^2 + b^2)$ and $1_{{\cal A}_n \cap \mathcal{B}_n(\tilde{K}) \cap \mathcal{C}_n} \leq 1_{\mathcal{B}_n(\tilde{K}) \cap\,\mathcal{C}_n}$
further bound $\mathbb{E}\big\{P\big[\max_{(j,k)}L^2_{njk}\big]\big\}$ above by the sum (multiplied by 2) of  
\begin{align*}
  \mathbb{E}\bigg\{P\Big[ &\max_{(j,k)}\Big\{\frac{(U_k-\Qj[U_k])}{{\rm Var}_{\Qj}(U_k)}\int_{\mathcal{T}}\big[\hat{E}_j(U_k,s,k)-E_0(U_k,s,k)\big]1(X \geq s)\{d\hat{\Lambda}_n(s)-d\Lambda(s)\}\Big\}^2\\
  & \times 1_{\mathcal{B}_n(\tilde{K}) \cap\,\mathcal{C}_n} \Big] \bigg\}
\end{align*}  
and
\[
  \mathbb{E}\bigg\{P\Big[ \max_{(j,k)}\Big\{\frac{(U_k-\Qj[U_k])}{{\rm Var}_{\Qj}(U_k)}\int_{\mathcal{T}}E_0(U_k,s,k)1(X \geq s)\{d\hat{\Lambda}_n(s)-d\Lambda(s)\}\Big\}^21_{\mathcal{B}_n(\tilde{K}) \cap\,\mathcal{C}_n} \Big] \bigg\},
\]
as respectively given in \eqref{eq:upper_bound_1_max_expected_Lnk} and \eqref{eq:upper_bound_2_max_expected_Lnk} below.

First note that $\min_{(j,k)}{\rm Var}_{\Qj}(U_k)$ is bounded away from zero in view of the occurrence of ${\cal C}_n$, as shown in \ref{eq:event_Cn_2} of Lemma \ref{lemma:An_Bn_Cn_prob_to_one}. Along with \ref{assump:Covariates} and that the total variation of $\hat{\Lambda}_n-\Lambda$ over ${\cal T}$ is bounded by $|\hat{\Lambda}_n(\tau)+\Lambda(\tau)|$, we have that
\begin{align*}
  &\max_{(j,k)}\Big\{\frac{(U_k-\Qj[U_k])}{{\rm Var}_{\Qj}(U_k)}\int_{\mathcal{T}}\big[\hat{E}_j-E_0\big](U_k,s,k)1(X \geq s)\{d\hat{\Lambda}_n(s)-d\Lambda(s)\}\Big\}^21_{\mathcal{B}_n(\tilde{K}) \cap\,\mathcal{C}_n} \\ 
  & \lesssim \sup_{(j,k,s)}\big|\hat{E}_j(U_k,s,k)-E_0(U_k,s,k)\big|^2\big|\hat{\Lambda}_n(\tau)+\Lambda(\tau)\big|^21_{\mathcal{B}_n(\tilde{K}) \cap\,\mathcal{C}_n}\\
  & \lesssim \sup_{(j,k,s)}\big|\hat{E}_j(U_k,s,k)-E_0(U_k,s,k)\big|^2\big\{\big|\hat{\Lambda}_n(\tau)-\Lambda(\tau)\big|^2+\Lambda^2(\tau)\big\}1_{\mathcal{B}_n(\tilde{K}) \cap\,\mathcal{C}_n}\\
  & \lesssim K^2_{nq_n}\bigg[\frac{\log(n)}{n}+\Lambda^2(\tau)\bigg],
\end{align*}
where the last line follows the occurrence of $\mathcal{B}_n(\tilde{K}) \cap\,\mathcal{C}_n$ that implies for each $j$,  $$\sup_{(k,s)}\big|\hat{E}_j(U_k,s,k)-E_0(U_k,s,k)\big| \le \tilde{K}K_{nj}$$ and $\sup_t|\hat{\Lambda}_n(t)-\Lambda(t)| \le \sqrt{\log(n)/n}$, giving that $\sup_{(j,k,s)}\big|\hat{E}_j(U_k,s,k)-E_0(U_k,s,k)\big| \le \tilde{K}K_{nq_n}$ because $K_{nj} = \sqrt{\log(n \lor p_n)/j}$ for $j \ge q_n$ and $\max_j\{K_{nj}\} = K_{nq_n}$.
Therefore, we have that
\begin{align}\label{eq:upper_bound_1_max_expected_Lnk}
  \mathbb{E}\bigg\{P\Big[ &\max_{(j,k)}\Big\{\frac{(U_k-\Qj[U_k])}{{\rm Var}_{\Qj}(U_k)}\int_{\mathcal{T}}\big[\hat{E}_j(U_k,s,k)-E_0(U_k,s,k)\big]1(X \geq s)\{d\hat{\Lambda}_n(s)-d\Lambda(s)\}\Big\}^2\\
  & \times 1_{\mathcal{B}_n(\tilde{K}) \cap\,\mathcal{C}_n} \Big] \bigg\}
  \lesssim K^2_{nq_n}\bigg[\frac{\log(n)}{n}+\Lambda^2(\tau)\bigg]. \nonumber
\end{align}

Moreover, we observe the decomposition of $\hat{\Lambda}_n-\Lambda$ analogously to \eqref{eq:martingale_est_CHF} without $u$:
\begin{align*}
  d\hat{\Lambda}_n(t)-d\Lambda(t) = 1(Y_n(t)>0)Y_n(t)^{-1}d\bar{M}(t) - 1(Y_n(t)=0)\,d\Lambda(t).
\end{align*}
Note that the occurrence of $\mathcal{B}_n(\tilde{K}) \cap \mathcal{C}_n$ eliminates $1(Y_n(t)=0)$, so
using \ref{assump:Covariates} and the occurrence of $\mathcal{B}_n(\tilde{K}) \cap \mathcal{C}_n$ similarly gives the upper bound:
\begin{align*}
  &\max_{(j,k)}\Big\{\frac{(U_k-\Qj[U_k])}{{\rm Var}_{\Qj}(U_k)}\int_{\mathcal{T}}E_0(U_k,s,k)1(X \geq s)\{d\hat{\Lambda}_n(s)-d\Lambda(s)\}\Big\}^21_{\mathcal{B}_n(\tilde{K}) \cap\,\mathcal{C}_n}\\
  &\lesssim \max_k\Big\{\int_{\mathcal{T}}E_0(U_k,s,k)1(X \geq s)\frac{1(Y_n(s)>0)}{Y_n(s)}d\bar{M}(s)\Big\}^21_{\mathcal{B}_n(\tilde{K}) \cap\,\mathcal{C}_n}.
\end{align*}
Therefore,
\begin{align}\label{eq:upper_bound_2_max_expected_Lnk}
  &\mathbb{E}\bigg\{P\Big[ \max_{(j,k)}\Big\{\frac{(U_k-\Qj[U_k])}{{\rm Var}_{\Qj}(U_k)}\int_{\mathcal{T}}E_0(U_k,s,k)1(X \geq s)\{d\hat{\Lambda}_n(s)-d\Lambda(s)\}\Big\}^21_{\mathcal{B}_n(\tilde{K}) \cap\,\mathcal{C}_n} \Big]\bigg\} \\
  &\lesssim \mathbb{E}\bigg\{P\bigg[ \max_k\Big\{\int_{\mathcal{T}}E_0(U_k,s,k)1(X \geq s)\frac{1(Y_n(s)>0)}{Y_n(s)}d\bar{M}(s)\Big\}^2 \bigg]1_{\mathcal{B}_n(\tilde{K}) \cap\,\mathcal{C}_n}\bigg\} \nonumber\\
  %& = P\bigg[\mathbb{E}\bigg\{ \max_k\Big\{\int_{\mathcal{T}}E_0(U_k,s,k)1(X \geq s)\frac{1(Y_n(s) \ge \sqrt{n})}{Y_n(s)}d\bar{M}(s)\Big\}\bigg\}^2\bigg] \nonumber\\
  & \le P\bigg[\max_k \mathbb{E}\bigg\{\int_{\mathcal{T}}E_0(U_k,s,k)1(X \geq s)\frac{1(Y_n(s) \ge \sqrt{n})}{Y_n(s)}d\bar{M}(s)\bigg\}^2\bigg] \lesssim \frac{\Lambda(\tau)}{\sqrt{n}}, \nonumber
\end{align}  
where along with \ref{assump:Covariates} that $U_k$ is uniformly bounded on $k$ and \ref{assump:Conditional_mean_E0} that $E_0$ is uniformly bounded, the last inequality holds by using the quadratic variation with respect to the filtration $\mathcal{F}'_{s}$ that is defined in \eqref{eq:filtration_Fprime.1}:  
\begin{align*}
  &\mathbb{E}\bigg\{\int_{\mathcal{T}}E_0(U_k,s,k)1(X \geq s)\frac{1(Y_n(s) \ge \sqrt{n})}{Y_n(s)}d\bar{M}(s)\bigg\}^2\\
  %& = \int_{\mathcal{T}}\mathbb{E}\Big\{E^2_0(U_k,s,k)1(X \geq s)\frac{1(Y_n(s) \ge \sqrt{n})}{Y^2_n(s)}\mathbb{E}\big[d(\bar{M}(s))^2 \big| \mathcal{F}'_{s} \big]\Big\}\\
  & = \int_{\mathcal{T}}\mathbb{E}\Big\{E^2_0(U_k,s,k)1(X \geq s)\frac{1(Y_n(s) \ge \sqrt{n})}{Y_n(s)}\Big\}d\Lambda(s) \lesssim \frac{\Lambda(\tau)}{\sqrt{n}}.
\end{align*}
Collecting the results in \eqref{eq:upper_bound_prob_max_Lnk}, \eqref{eq:upper_bound_1_max_expected_Lnk} and \eqref{eq:upper_bound_2_max_expected_Lnk} leads to
\begin{align*}
  & {\rm P}\Big(\max_j\Big\{P\Big[\max_{k\in\mathcal{K}_n}L_{njk}\Big]\Big\} > 1/\sqrt{\log(n \lor p_n)}\, \Big)\\ 
  & \lesssim \log(n \lor p_n)\bigg[K^2_{nq_n}\bigg[\frac{\log(n)}{n}+\Lambda^2(\tau)\bigg] + \frac{\Lambda(\tau)}{\sqrt{n}}\bigg] = o(1),    
\end{align*}
following that $\Lambda(\tau)$ is bounded, $K^2_{nq_n} = \log(n \lor p_n)/q_n$, $q_n^{1/4}/\log(n \lor p_n) \to \infty$ and $n/q_n=O(1)$. 

Taking the maximum over $k \in \mathcal{K}_n$, applying the triangle inequality, and then taking the expectation with respect to $P$ on both sides of \eqref{eq:upper_bound_distance1}, we see that \eqref{eq:upper_bound_distance1.1} implies that 
$d(\hat{P}_{nj}, \hat{P}_j') \leq K^{''}K_{nj} + P[ \max_{k\in\mathcal{K}_n}L_{njk}]$. Hence, by the above,
\begin{align*}
  & {\rm P}\big( {\cal A}_n \cap {\cal B}_n(\tilde{K}) \cap {\cal C}_n \cap \big\{d(\hat{P}_{nj}, \hat{P}_j') \leq K^{''}K_{nj} + 1/\sqrt{\log(n \lor p_n)},\,j = q_n,\ldots,n\big\} \big)\\
  &\ge {\rm P}\Big( {\cal A}_n \cap {\cal B}_n(\tilde{K}) \cap {\cal C}_n \cap \Big\{\max_j\Big\{P\Big[\max_{k \in \mathcal{K}_n}L_{njk}\Big]\Big\} \leq 1/\sqrt{\log(n \lor p_n)} \Big\} \,\Big)\\
  & = {\rm P}\big( {\cal A}_n \cap {\cal B}_n(\tilde{K}) \cap {\cal C}_n \big)\\
  & \quad - {\rm P}\Big( {\cal A}_n \cap {\cal B}_n(\tilde{K}) \cap {\cal C}_n \cap \Big\{\max_j\Big\{P\Big[\max_{k \in \mathcal{K}_n}L_{njk}\Big]\Big\} > 1/\sqrt{\log(n \lor p_n)} \Big\}\Big)\\
  & \geq {\rm P}\big( {\cal A}_n \cap {\cal B}_n(\tilde{K}) \cap {\cal C}_n \big)
  - {\rm P}\Big(\max_j\Big\{P\Big[\max_{k \in \mathcal{K}_n}L_{njk}\Big]\Big\} > 1/\sqrt{\log(n \lor p_n)}\,\Big) \to 1,
\end{align*}
leading to 
\begin{align} \label{eq:prob_term_1}
  {\rm P}\big( {\cal A}_n \cap {\cal B}_n(\tilde{K}) \cap {\cal C}_n \cap \big\{d(\hat{P}_{nj}, \hat{P}_j') \leq K^{''}K_{nj} + 1/\sqrt{\log(n \lor p_n)},\,j=q_n,\ldots,n\big\}^c \big) \to 0.    
\end{align}

In addition for each $j$, we have that
\begin{align*}
   & \max_{k \in \mathcal{K}_n}\big|\IF^*_{k}(O|\hat{P}'_j) - \IF^*_{k}(O|P)\big|1_{{\cal A}_n \cap {\cal B}_n(\tilde{K}) \cap {\cal C}_n}\\
   &\leq 1_{{\cal A}_n \cap {\cal B}_n(\tilde{K}) \cap {\cal C}_n} \Big\{\frac{1}{\min_k{\rm Var}_{\Qj}(U_{k})}
   \Big(\tilde{Y}\max_k\big|(\Qj-Q_u)[U_k]\big|\\
   &\qquad + \max_k|U_k|\max_k\big|(\Qj-Q_u)[\hat{E}_j(U_k,k)]\big|\\
   &\qquad + \max_k|U_k|Q_u\big[\max_k\big|\hat{E}_j(U_k,k)-E_0(U_k,k)\big|\big]\\
   &\qquad + \Qj\big[\max_k\big|\hat{E}_j(U_k,k)\big|\big]\max_k\big|(\Qj-Q_u)[U_k]\big|\\
   &\qquad + \max_k|Q_u[U_k]|\max_k\big|(\Qj-Q_u)[\hat{E}_j(U_k,k)]\big|\\
   &\qquad + \max_k|Q_u[U_k]|Q_u\big[\max_k\big|\hat{E}_j(U_k,k)-E_0(U_k,k)\big|\big]\Big)\\
   &\quad + \max_k\Big|\frac{1}{{\rm Var}_{\Qj}(U_{k})}-\frac{1}{{\rm Var}_{Q_u}(U_{k})}\Big|\max_k\big|(U_k-Q_u[U_k])(\tilde{Y}-Q_u[E_0(U_k,k)])\big|\\
   &\quad +\frac{\max_k|{\rm Cov}_{\Qj}(U_k,\hat{E}_j(U_k,k))|}{\min_k{\rm Var}^2_{\Qj}(U_{k})}\Big(2\max_k|U_k|\max_k|(\Qj-Q_u)[U_k]|\\
   &\qquad +\max_k|(\Qj-Q_u)[U_k]|\max_k|(\Qj+Q_u)[U_k]|\Big)\\
   &\quad +\frac{\max_k(U_k-Q_u[U_k])^2}{\min_k{\rm Var}^2_{\Qj}(U_{k})}\Big(\max_k\big|{\rm Cov}_{\Qj}(U_k,\hat{E}_j(U_k,k))-{\rm Cov}_{Q_u}(U_k, E_0(U_k,k))\big|\Big)\\
   &\quad +\max_k\Big|\frac{1}{{\rm Var}^2_{\Qj}(U_{k})}-\frac{1}{{\rm Var}^2_{Q_u}(U_{k})}\Big|\max_k|{\rm Cov}_{Q_u}(U_k, E_0(U_k,k))|\max_k(U_k-Q_u[U_k])^2\\
   &\quad + \frac{(1+\Lambda(\tau))}{\min{\rm Var}_{\Qj}(U_{k})}\Big( \max_k|U_k-Q_u[U_k]|\sup_{(k,s)}|\hat{E}_j(U_k,s,k)-E_0(U_k,s,k)|\\
   &\qquad + \max_k|(\Qj-Q_u)[U_k]|\max_k|E_0(U_k,s,k)|\Big)\\
   &\quad + \max_k\Big|\frac{1}{{\rm Var}_{\Qj}(U_{k})}-\frac{1}{{\rm Var}_{Q_u}(U_{k})}\Big|
   \max_k|U_k-Q_u[U_k]|\max_k|E_0(U_k,s,k)|(1+\Lambda(\tau))\Big\}\\
   & \leq K^{'''}K_{nj}
\end{align*}
for some $(j,k,s,n)$-independent constant $K^{'''} \in (0, \infty)$ when $\mathcal{A}_n \cap \mathcal{B}_n(\tilde{K}) \cap {\cal C}_n$ occurs,
by \ref{assump:Covariates}--\ref{assump:At-risk prob} and \ref{assump:Variances}--\ref{assump:Conditional_mean_E0}. This gives 
\begin{align}\label{eq:prob_term_2}
  {\rm P}\big(\mathcal{A}_n \cap \mathcal{B}_n(\tilde{K}) \cap {\cal C}_n \cap 
  \big\{d(\hat{P}_j', P) \leq K^{'''}K_{nj},\; j=q_n,\ldots,n\big\}^c\big) \to 0.
\end{align}  
Then taking $K' = (K^{''}+1) \lor K^{'''}$ and using \eqref{eq:prob_term_1} and \eqref{eq:prob_term_2}, we have
${\rm P}(\mathcal{A}_n \cap \mathcal{B}_n(\tilde{K}) \cap {\cal C}_n \cap {\cal D}_n(K')^c) \to 0$.
Hence the result follows by Lemma \ref{lemma:An_Bn_Cn_prob_to_one}.
\end{proof}

In what follows, we let $\tilde{\sigma}_{nj}^2\equiv{\rm Var}(\IF^*_{k_j}(O|\hat{P}_{nj})|O_1,\ldots,O_j)$.
\begin{lemma}\label{lemma:bounded_var_deviance}
  Let $\tilde{K} \in (0, \infty)$ be given  in Lemma \ref{lemma:An_Bn_prob_to_one}. Then, under the conditions of Theorem \ref{Thm:stab_one_step}, on the event
  $\mathcal{A}_n \cap \mathcal{B}_n(\tilde{K})\cap {\cal C}_n$ we have $|\hat{\sigma}^2_{nj}-\tilde{\sigma}_{nj}^2| \lesssim K_{nj}$ almost surely for $j= q_n,\ldots,n$ and $n$ sufficiently large.
\end{lemma}
\begin{proof}
Fix $j \in \{q_n, \ldots, n\}$. We first have that
\begin{align*}
  & |\hat{\sigma}^2_{nj}-\tilde{\sigma}_{nj}^2|1_{{\cal A}_n \cap {\cal B}_n(\tilde{K}) \cap {\cal C}_n}
  \leq \bigg\{\Big|\mathbb{P}_j \big[\IF^*_{k_j}(O|\hat{P}_{nj})\big]^2 - P\big[\IF^*_{k_j}(O|\hat{P}_{nj})\big]^2 \Big| \\
  &\hspace{0.6cm} + \Big|\big[\mathbb{P}_j\IF^*_{k_j}(O|\hat{P}_{nj})\big]^2 - \big[P\IF^*_{k_j}(O|\hat{P}_{nj})\big]^2\Big|\bigg\}1_{{\cal A}_n \cap {\cal B}_n(\tilde{K}) \cap {\cal C}_n} \nonumber \\
  &\hspace{0.2cm} \leq \bigg\{\max_{k \in \mathcal{K}_n}\Big|(\mathbb{P}_j - P)\big[\IF^*_{k}(O|\hat{P}_{nj})\big]^2 \Big|\\
  &\hspace{1cm} + \max_{k \in \mathcal{K}_n}\Big|(\mathbb{P}_j-P)\big[\IF^*_{k}(O|\hat{P}_{nj})\big]\, (\mathbb{P}_j+P)\big[\IF^*_{k}(O|\hat{P}_{nj})\big]\Big|\bigg\}1_{{\cal A}_n \cap {\cal B}_n(\tilde{K}) \cap {\cal C}_n}.
\end{align*}
At the conclusion of this proof, we'll show that there exists a constant $\bar{K}$ that does not depend on $n$ such that 
\begin{align}\label{eq:bounded_IF_star_hatPn}
\max_{(j,k)}\big|\IF^*_{k}(o|\hat{P}_{nj})1_{{\cal A}_n \cap {\cal B}_n(\tilde{K}) \cap {\cal C}_n}\big|\le \bar{K}
\end{align}
for all observations $o$ in the support of $P$. Then this will simplify the above display as
\begin{align} \label{eq:var_decomposition}
 & |\hat{\sigma}^2_{nj}-\tilde{\sigma}_{nj}^2|1_{{\cal A}_n \cap {\cal B}_n(\tilde{K}) \cap {\cal C}_n} \\
 & \quad\lesssim \bigg\{\max_{k \in \mathcal{K}_n}\Big|(\mathbb{P}_j - P)\big[\IF^*_{k}(O|\hat{P}_{nj})\big]^2 \Big| + \max_{k \in \mathcal{K}_n}
  \Big|(\mathbb{P}_j-P)\big[\IF^*_{k}(O|\hat{P}_{nj})\big]\Big|\bigg\}1_{{\cal A}_n \cap {\cal B}_n(\tilde{K}) \cap {\cal C}_n}, \nonumber
\end{align}
from which we continue showing that $|\hat{\sigma}^2_{nj}-\tilde{\sigma}_{nj}^2|1_{{\cal A}_n \cap {\cal B}_n(\tilde{K}) \cap {\cal C}_n} \lesssim K_{nj}$.

The proof below proceeds with assuming (without statement) that the event ${\cal A}_n \cap {\cal B}_n(\tilde{K}) \cap {\cal C}_n$ occurs.
Recall that $\IF^*_{k}(\cdot|E,Q,G)$ denotes the function $\IF^*_k(\cdot|P_{E,Q,G})$, where $P_{E,Q,G}$ is a distribution with the conditional residual life function $E$, the marginal distribution $Q$ of any given single predictor, and the censoring distribution $G$.
For each $(k,j)$, $\IF^*_{k}(o|\hat{E}_j,Q_u,\hat{G}_n)$ belongs to $\tilde{\mathcal{F}}_n$, observing that $\IF^*_{k}(o|\hat{E}_j,Q_u,\hat{G}_n)$ is equal to
\begin{align*}
  &\frac{(u_k-Q_u[U_k])\big(\delta x/\hat{G}_n(x)-Q_u[\hat{E}_j(U_k,k)]\big)}{{\rm Var}_{Q_u}(U_k)} -\frac{{\rm Cov}_{Q_u}(U_k, \hat{E}_j(U_k,k))}{{\rm Var}_{Q_u}^2(U_k)}(u_k-Q_u[U_k])^2 \nonumber\\ 
  &\quad + \frac{(u_k-Q_u[U_k])}{{\rm Var}_{Q_u}(U_k)}\int_{\mathcal{T}} \hat{E}_{j}(u_k,s,k)\{1(x \in ds, \delta=0)-1(x \geq s)d\hat{\Lambda}_n(s)\},
\end{align*}
in which the first two terms are contained in $\cup_{q=0}^1\cup_{r=0}^4\tilde{\mathcal{F}}_{1}(k,q,r)\tilde{{\cal E}}(k,1)$, and the last term is contained in $\cup_{r=0}^4\tilde{\mathcal{F}}_{1}(k,0,r)\tilde{\mathcal{F}}_{2}(k,1)$.
To take advantage of the fact that  $\IF^*_{k}(o|\hat{E}_j,Q_u,\hat{G}_n) \in \mathcal{F}_n$, we replace $\IF^*_{k}(O|\hat{P}_{nj})$ in \eqref{eq:var_decomposition} by $\IF^*_{k}(O|\hat{E}_j,Q_u,\hat{G}_n)$ and then \eqref{eq:var_decomposition} turns into
\begin{align} \label{eq:var_decomposition.2}
  & |\hat{\sigma}^2_{nj}-\tilde{\sigma}_{nj}^2|1_{{\cal A}_n \cap {\cal B}_n(\tilde{K}) \cap {\cal C}_n} \lesssim \bigg\{\max_{k \in \mathcal{K}_n}\Big|(\mathbb{P}_j - P)\big[\IF^*_{k}(O|\hat{E}_j,Q_u,\hat{G}_n)\big]^2 \Big|\\
  &\quad + \max_{k \in \mathcal{K}_n}
  \Big|(\mathbb{P}_j-P)\big[\IF^*_{k}(O|\hat{E}_j,Q_u,\hat{G}_n)\big]\Big|\bigg\}1_{{\cal A}_n \cap {\cal B}_n(\tilde{K}) \cap {\cal C}_n} + K_{nq_n}, \nonumber
\end{align}
following that on the event ${\cal A}_n \cap {\cal B}_n(\tilde{K}) \cap {\cal C}_n$,
\begin{align}\label{eq:upperbound_replace_Qj_Qu}
 \max_{(j,k)}\big|\IF^*_{k}(o|\hat{P}_{nj})-\IF^*_{k}(o|\hat{E}_j,Q_u,\hat{G}_n)\big| \lesssim K_{nq_n},
\end{align}
and that there exists a constant $\bar{K}_0$ that does not depend on $n$ such that 
\begin{align} \label{eq:bounded_IF_star_P_alt}
  \max_{(j,k)}\big|\IF^*_{k}(o|\hat{E}_j,Q_u,\hat{G}_n) 1_{{\cal A}_n \cap {\cal B}_n(\tilde{K}) \cap {\cal C}_n}\big| \le \bar{K}_0 
\end{align}
for all observations $o$ in the support of $P$. We show \eqref{eq:upperbound_replace_Qj_Qu} and \eqref{eq:bounded_IF_star_P_alt} in the sequel.

For \eqref{eq:upperbound_replace_Qj_Qu}, first we have that
\begin{align} \label{eq:upperbound_replace_Qj_Qu.1}
  &\big|\IF^*_{k}(o|\hat{P}_{nj})-\IF^*_{k}(o|\hat{E}_j,Q_u,\hat{G}_n)\big|\\ &=\big|\IF^*_{k}(o|\hat{E}_j,\mathbb{Q}_j,\hat{G}_n)-\IF^*_{k}(o|\hat{E}_j,Q_u,\hat{G}_n)\big| \lesssim
  |(\Qj-Q_u)[U_k]|\,|\tau/G(\tau)| \nonumber \\ 
  &\quad + |(\Qj-Q_u)[U_k]|\,|\tau/\hat{G}_n(\tau)-\tau/G(\tau)| \nonumber \\
  &\quad + \big|(\Qj-Q_u)[U_k]\big|\:\big|\Qj[\hat{E}_j(U_k,k)]\big|
  + \big|(\Qj-Q_u)[\hat{E}_j(U_k,k)]\big|\:\big|Q_u[U_k]\big| \nonumber\\
  &\quad + \Big|{\rm Cov}_{\Qj}(U_k, \hat{E}_j(U_k,k))-{\rm Cov}_{Q_u}(U_k, \hat{E}_j(U_k,k))\Big|(u_k-\Qj[U_k])^2 \nonumber\\
  &\quad + \big|(u_k-\Qj[U_k])^2-(u_k-Q_u[U_k])^2\big|\big|{\rm Cov}_{Q_u}(U_k, \hat{E}_j(U_k,k))\big|\nonumber\\ 
  &\quad +|(\Qj-Q_u)[U_k]|\:\Big|\int_{\mathcal{T}} \hat{E}_{j}(u_k,s,k)\{1(x \in ds, \delta=0)-1(x \geq s)d\hat{\Lambda}_n(s)\}\Big|. \nonumber
  %&\quad \lesssim \big|(\Qj-Q_u)[U_k]\big| + \big|(\Qj-Q_u)[\hat{E}_j(U_k,k)]\big|
  %+ \big|(\Qj-Q_u)[U_k\hat{E}_j(U_k,k)]\big| \nonumber\\
  %&\quad \lesssim K_{nj}. \nonumber
\end{align}
The inequality in \eqref{eq:upperbound_replace_Qj_Qu.1} holds because for all $j$, $\min_k{\rm Var}_{\Qj}(U_k)$ is bounded away from zero given the occurrence of ${\cal A}_n \cap {\cal B}_n(\tilde{K}) \cap {\cal C}_n$ as shown in \ref{eq:event_Cn_2} of Lemma \ref{lemma:An_Bn_Cn_prob_to_one}
and $\min_k {\rm Var}_{Q_u}(U_k) > 0$ by \ref{assump:Variances}. We continue expanding \eqref{eq:upperbound_replace_Qj_Qu} in what follows. Together with \ref{assump:Covariates} that the support of $U_k$ is uniformly bounded and \ref{assump:Survival function} that $G(\tau)>0$, a Taylor expansion of $\hat{G}_n$ around $G$ and the occurrence of ${\cal A}_n \cap {\cal B}_n(\tilde{K}) \cap {\cal C}_n$ imply that
\begin{align} \label{eq:taylor_expansion_hatG}
  \bigg|\frac{\tau}{\hat{G}_n(\tau)}-\frac{\tau}{G(\tau)}\bigg| \leq  \frac{|\tau|}{G(\tau)}\sum_{r=1}^{\infty}\bigg[\sup_{t}\bigg|\frac{\hat{G}_n(t)}{G(t)}-1\bigg|\,\bigg]^r
  \leq \frac{|\tau|}{G(\tau)}\sum_{r=1}^{\infty}\bigg[\sqrt{\frac{\log n}{n}}\bigg]^r
  \lesssim \sqrt{\frac{\log n}{n}}.
\end{align}
From \ref{assump:Covariates} and \ref{assump:Conditional_mean_E0}, we have that $\sup_{(k,s)}\big|E_0(U_k,s,k)\big| \le K_1$ for some positive finite constant $K_1$. Along with \ref{eq:event_Cn_5} of Lemma \ref{lemma:An_Bn_Cn_prob_to_one}, it gives that
$\max_{k}\big|\Qj[\hat{E}_j(U_k,k)]\big| \le \sup_{(k,s)}\big|\hat{E}_j(U_k,s,k)\big| \lesssim K_{nj} + K_1$. Then together with \ref{assump:Covariates}, these two results lead to
\begin{align} \label{eq:expansion_cov}
  &\big|{\rm Cov}_{\Qj}(U_k, \hat{E}_j(U_k,k))-{\rm Cov}_{Q_u}(U_k, \hat{E}_j(U_k,k))\big| = \big|(\Qj-Q_u)[U_k\hat{E}_j(U_k,k)]\big| \\
  &\quad + \big|\Qj[\hat{E}_j(U_k,k)]\big|\:\big|(\Qj-Q_u)[U_k]\big| + \big|Q_u[U_k]\big|\:\big|(\Qj-Q_u)[\hat{E}_j(U_k,k)]\big| \nonumber\\
  & \quad \lesssim \big|(\Qj-Q_u)[U_k\hat{E}_j(U_k,k)]\big|
  + (K_{nj} + 1)\:\big|(\Qj-Q_u)[U_k]\big| +    \big|(\Qj-Q_u)[\hat{E}_j(U_k,k)]\big| \nonumber
\end{align}  
and
\begin{align} \label{eq:expansion_integral} 
  &\Big|\int_{\mathcal{T}}\hat{E}_{j}(u_k,s,k)\{1(x \in ds, \delta=0)-1(x \geq s)d\hat{\Lambda}_n(s)\}\Big| \leq \sup_{(k,s)}\big|\hat{E}_{j}(u_k,s,k)\big|
  \big(\hat{\Lambda}_n(\tau)+1\big) \\
  &\quad \lesssim \big(\hat{\Lambda}_n(\tau)+1\big)(K_{nj} + 1). \nonumber
\end{align}
Inserting \eqref{eq:taylor_expansion_hatG}--\eqref{eq:expansion_integral} back along with using \ref{assump:Covariates} and $\sup_{(k,s)}\big|\hat{E}_j(U_k,s,k)\big| \lesssim K_{nj} + K_1$ again, \eqref{eq:upperbound_replace_Qj_Qu.1} turns into 
\begin{align*}
  &\big|\IF^*_{k}(o|\hat{P}_{nj})-\IF^*_{k}(o|\hat{E}_j,Q_u,\hat{G}_n)\big|=\big|\IF^*_{k}(o|\hat{E}_j,\mathbb{Q}_j,\hat{G}_n)-\IF^*_{k}(o|\hat{E}_j,Q_u,\hat{G}_n)\big| \nonumber\\
  &\quad \lesssim \bigg[\sqrt{\frac{\log n}{n}}+(K_{nj}+1)(\hat{\Lambda}_n(\tau)+1) \bigg]\big|(\Qj-Q_u)[U_k]\big| + \big|(\Qj-Q_u)[\hat{E}_j(U_k,k)]\big| \nonumber\\
  & \qquad + \big|(\Qj-Q_u)[U_k\hat{E}_j(U_k,k)]\big| \nonumber\\
  &\quad \lesssim \big(\hat{\Lambda}_n(\tau)+1\big)K_{nj} 
  \le \big(\,\big|\hat{\Lambda}_n(\tau)-\Lambda(\tau)\big| + 2\Lambda(\tau) +1 \big)K_{nj}. \nonumber
\end{align*}
Observing that $\sup_{t \in \mathcal{T}}|\hat{\Lambda}_n(t)-\Lambda(t)| \le \sqrt{\log n}/\sqrt{n} < 1$ on the event ${\cal A}_n \cap {\cal B}_n(\tilde{K}) \cap {\cal C}_n$ by Lemma \ref{lemma:An_Bn_prob_to_one}, it follows from the above display that
\begin{align*}
  \max_{(j,k)}\big|\IF^*_{k}(o|\hat{P}_{nj})-\IF^*_{k}(o|\hat{E}_j,Q_u,\hat{G}_n)\big| < 2(\Lambda(\tau)+1)\max_j\{K_{nj}\}
  \lesssim K_{nq_n},
\end{align*}
yielding \eqref{eq:upperbound_replace_Qj_Qu}. The proof of \eqref{eq:bounded_IF_star_P_alt} can be handled using similar arguments that are used to show \eqref{eq:bounded_IF_star_hatPn} and will appear later.
Then following that $\IF^*_{k}(\cdot|\hat{E}_j,Q_u,\hat{G}_n) \in \tilde{{\cal F}}_n$, \eqref{eq:var_decomposition.2} leads to
\begin{align*}
 |\hat{\sigma}^2_{nj}-\tilde{\sigma}_{nj}^2|1_{{\cal A}_n \cap {\cal B}_n(\tilde{K}) \cap {\cal C}_n} 
 \lesssim K_{nj}.
\end{align*}

To complete the proof, we now show that \eqref{eq:bounded_IF_star_hatPn}, and it suffices to show that 
$$\limsup_{n \to \infty}\,\sup_o\,\max_{(j,k)}|\IF^*_{k}(o|\hat{P}_{nj})|1_{{\cal A}_n \cap {\cal B}_n(\tilde{K}) \cap {\cal C}_n} < \infty.$$
In what follows, we assume without statement the occurrence of ${\cal A}_n \cap {\cal B}_n(\tilde{K}) \cap {\cal C}_n$. Note that, for any $o$,
\begin{align*}
  &\max_{k}\big|\IF^*_{k}(o|\hat{P}_{nj})| \\ 
  & = \max_k\bigg|\frac{(u_k-\Qj[U_k])\big(y-\Qj[\hat{E}_j(U_k,k)]\big)}{{\rm Var}_{\Qj}(U_k)}
  -\frac{{\rm Cov}_{\Qj}(U_k, \hat{E}_j(U_k,k))}{{\rm Var}_{\Qj}^2(U_k)}(u_k-\Qj[U_k])^2 \nonumber \\
  &\hspace{0.3cm} + \frac{(u_k-\Qj[U_k])}{{\rm Var}_{\Qj}(U_k)}\int_{\mathcal{T}} \hat{E}_{j}(u_k,s,k)\{1(x \in ds, \delta=0)-1(x \geq s)d\hat{\Lambda}_n(s)\}\bigg| \nonumber \\
  & \lesssim \max_k \big|u_k-\Qj[U_k]\big|\Big[\,\big|y-\tilde{y}\big|
  + \big|\tilde{y}\big| + \Qj\big[\max_k\big|\hat{E}_j(U_k,k)\big|\,\big]\Big] \nonumber\\
  & \quad + \max_k\big|{\rm Cov}_{\Qj}(U_k,\hat{E}_j(U_k,k))\big| \big[\max_k \big|u_k-\Qj[U_k]\big|\,\big]^2 \nonumber\\
  &\quad + \max_k \big|u_k-\Qj[U_k]\big|
  \max_{k}\Big|\int_{\mathcal{T}}\hat{E}_{j}(u_k,s,k)\{1(x \in ds, \delta=0)-1(x \geq s)d\hat{\Lambda}_n(s)\}\Big|, \nonumber
\end{align*}
where the last inequality holds by \ref{eq:event_Cn_2} of Lemma \ref{lemma:An_Bn_Cn_prob_to_one} and the triangle inequality. When ${\cal A}_n \cap {\cal B}_n(\tilde{K}) \cap {\cal C}_n$ occurs, similar techniques to the arguments for \eqref{eq:taylor_expansion_hatG} yields that $$|y-\tilde{y}| \leq |\tilde{y}|\,\sqrt{\log n}/(\sqrt{n}-\sqrt{\log n}\,) \lesssim \sqrt{\log n}/(\sqrt{n}-\sqrt{\log n}\,),$$
following that $|\tilde{y}| = |\delta x/G(x)|$ is bounded by some nonrandom finite constant that does not depend on $(j,n)$ due to \ref{assump:Survival function}.
Meanwhile, by \ref{assump:Conditional_mean_E0} and \ref{assump:Covariates} that the support of $U_k$ is uniformly bounded, there exist positive finite constants $K_0$ and $K_1$ so that $\max_k|u_k-\Qj[U_k]| \leq K_0$; $\max_{k}|E_0(U_k,k)| \leq K_1$ and $\sup_{(k,s)}|E_0(u_k,s,k)| \leq K_1$. Therefore when ${\cal A}_n \cap {\cal B}_n(\tilde{K}) \cap {\cal C}_n$ occurs, \ref{eq:event_Cn_5}--\ref{eq:event_Cn_6} of Lemma \ref{lemma:An_Bn_Cn_prob_to_one} imply that
\begin{align*}
  &\max_k|u_k-\Qj[U_k]|\big[\,\big|y-\tilde{y}\big|
  + \big|\tilde{y}\big| + \Qj\big[\max_k\big|\hat{E}_j(U_k,k)\big|\,\big]\,\big]\\ 
  &\quad \lesssim \,\sqrt{\log n}/(\sqrt{n}-\sqrt{\log n}\,) + K_{nj} + 1;\\
  &\max_k\big|{\rm Cov}_{\Qj}(U_k,\hat{E}_j(U_k,k))\big| \big[\max_k \big|u_k-\Qj[U_k]\big|\,\big]^2 \\
  &\quad \leq \big[\max_k \big|u_k-\Qj[U_k]\big|\,\big]^2 \Qj\Big[ \max_k \big|U_k-\Qj[U_k]\big| \max_k \big| \hat{E}_j(U_k,k) \big| \Big]
  \lesssim K_{nj} + 1;\\
  &\max_k \big|u_k-\Qj[U_k]\big|\max_{k}
  \Big|\int_{\mathcal{T}}\hat{E}_{j}(u_k,s,k)\{1(x \in ds, \delta=0)-1(x \geq s)d\hat{\Lambda}_n(s)\}\Big|\\
  &\quad \leq \max_k \big|u_k-\Qj[U_k]\big|\sup_{(k,s)}\big|\hat{E}_{j}(u_k,s,k)\big|
  \big(1+\hat{\Lambda}_n(\tau)\big)\\
  & \quad \lesssim \big(K_{nj} + K_1 \big)\big(1+\hat{\Lambda}_n(\tau)\big)
  \le \big(K_{nj} + K_1 \big)\big(\big|\hat{\Lambda}_n(\tau)-\Lambda(\tau)\big|+\Lambda(\tau)+1\big)\\
  & \quad < \big(K_{nj} + K_1 \big)\big(2+\Lambda(\tau)\big)
  \lesssim K_{nj} + 1,
\end{align*}
where the last line follows that $\sup_{t \in \mathcal{T}}|\hat{\Lambda}_n(t)-\Lambda(t)| \le \sqrt{\log n}/\sqrt{n} < 1$ on the event ${\cal A}_n \cap {\cal B}_n(\tilde{K}) \cap {\cal C}_n$, using Lemma \ref{lemma:An_Bn_prob_to_one}.
Therefore, the above results ensure that 
\begin{align*}
  \sup_o &\max_{(j,k)}\big|\IF^*_{k}(o|\hat{P}_{nj})\big|1_{{\cal A}_n \cap {\cal B}_n(\tilde{K}) \cap {\cal C}_n}
  \lesssim \sqrt{\log n}/\big(\sqrt{n}-\sqrt{\log n}\,\big) + \max_j\{K_{nj}\} + 1,
\end{align*}
so $\limsup_{n \to \infty} \sup_o \max_{(j,k)}|\IF^*_{k}(o|\hat{P}_{nj})|1_{{\cal A}_n \cap {\cal B}_n(\tilde{K}) \cap {\cal C}_n}$ is bounded by some $(j,n)$-independent constant, following that $\max_j\{K_{nj}\} \le K_{nq_n} \to 0$ by $q_n^{1/4}/\log(n \lor p_n) \to \infty$.
\end{proof}

\begin{lemma} \label{lemma:An_Bn_Cn_Dn_En_prob_to_one}
Suppose the conditions of Theorem \ref{Thm:stab_one_step} hold. Then there exists an event $\mathcal{E}_n$ that corresponds to the intersection of
\begin{lemmaitemenum}
  \item \label{eq:event_En_1} $\min_{j \in \{q_n,\ldots,n\}}\hat{\sigma}^2_{nj}$ is bounded away from zero by a constant;
  \item \label{eq:event_En_2} $|\sigma_{nj}/\hat{\sigma}_{nj}-1| \lesssim K_{nj}$ for all $j =q_n,\ldots,n$, where $\sigma^2_{nj} =\int \IF^*_{k_j}(o|P)^2 dP(o)$,
\end{lemmaitemenum}
where each of \ref{eq:event_En_1} and \ref{eq:event_En_2} relies on appropriately specified (non-random) constants that do not depend on $n$, such that $${\rm P} (\mathcal{A}_n \cap \mathcal{B}_n(\tilde{K})\cap {\cal C}_n \cap {\cal D}_n(K') \cap {\cal E}_n)\to 1,$$
with $\tilde{K}, K' \in (0, \infty)$ given in Lemma \ref{lemma:An_Bn_prob_to_one} and Lemma \ref{lemma:An_Bn_Cn_Dn_prob_to_one}, respectively.
\end{lemma}
\begin{proof}
When $\mathcal{A}_n \cap \mathcal{B}_n(\tilde{K}) \cap {\cal C}_n \cap {\cal D}_n(K')$ occurs, to show that $\hat{\sigma}^2_{nj}$ is uniformly bounded away from zero, it suffices to show that $|\hat{\sigma}^2_{nj}-\sigma^2_{nj}| \lesssim K_{nj}$ for $j=q_n,\ldots,n$ on this event. This gives
\begin{align*}
  \hat{\sigma}^2_{nj} \gtrsim \sigma^2_{nj} - K_{nj} \geq \sigma^2_{nj} - \sqrt{\log(n \lor p_n)/q_n} > \zeta'
\end{align*}  
for some universal constant $\zeta' > 0$, where the final inequality is a direct consequence of the conditions for Theorem \ref{Thm:stab_one_step}: $q_n^{1/4}/\log(n \lor p_n) \to \infty$ and $\sigma_{nj}^2$ is uniformly bounded away from zero by \ref{assump:Variances}.

We now show that $|\hat{\sigma}^2_{nj}-\sigma^2_{nj}| \lesssim K_{nj}$ when $\mathcal{A}_n \cap \mathcal{B}_n(\tilde{K}) \cap {\cal C}_n \cap {\cal D}_n(K')$ occurs. Recall that $\tilde{\sigma}_{nj}^2 = {\rm Var}(\IF^*_{k_j}(O|\hat{P}_{nj})|O_1,\ldots,O_j)$. We first use the triangle inequality to give $|\hat{\sigma}^2_{nj}-\sigma^2_{nj}| \leq |\hat{\sigma}^2_{nj}-\tilde{\sigma}_{nj}^2| + |\tilde{\sigma}_{nj}^2 - \sigma^2_{nj}|$. Because $|\hat{\sigma}^2_{nj}-\tilde{\sigma}_{nj}^2|1_{\mathcal{A}_n \cap \mathcal{B}_n(\tilde{K}) \cap {\cal C}_n \cap {\cal D}_n(K')} \lesssim K_{nj}$ is given by Lemma \ref{lemma:bounded_var_deviance}, it suffices to show that
\begin{align*} %\label{eq:decomp_hat_sigma}
  |\tilde{\sigma}_{nj}^2 - \sigma^2_{nj}|1_{\mathcal{A}_n \cap \mathcal{B}_n(\tilde{K}) \cap {\cal C}_n \cap {\cal D}_n(K')} \lesssim K_{nj}.
\end{align*}
Recalling that $\hat{P}'_j = (\hat{E}_j, \Qj, G)$ and noting that $P\IF^*_{k_j}(O|\hat{P}_j')=0$ and $P\IF^*_{k_j}(O|P)=0$, Jensen's inequality and the triangle inequality give that
\begin{align*}
  &|\tilde{\sigma}_{nj}^2 - \sigma^2_{nj}|1_{\mathcal{A}_n \cap \mathcal{B}_n(\tilde{K}) \cap {\cal C}_n \cap {\cal D}_n(K')}\\ 
  &\quad \le \Big\{P\big|\IF^*_{k_j}(O|\hat{P}_{nj})
  + \IF^*_{k_j}(O|P)\big|\big|\IF^*_{k_j}(O|\hat{P}_{nj}) - \IF^*_{k_j}(O|P)\big|\\ 
  &\qquad + P\big|\IF^*_{k_j}(O|\hat{P}_{nj})+\IF^*_{k_j}(O|\hat{P}_j')\big|P\big|\IF^*_{k_j}(O|\hat{P}_{nj})-\IF^*_{k_j}(O|\hat{P}_j')\big|\Big\}1_{\mathcal{A}_n \cap \mathcal{B}_n(\tilde{K}) \cap {\cal C}_n \cap {\cal D}_n(K')}\\
  &\quad \leq \Big\{P\big|\IF^*_{k_j}(O|\hat{P}_{nj})
  + \IF^*_{k_j}(O|P)\big|\big|\IF^*_{k_j}(O|\hat{P}_{nj}) - \IF^*_{k_j}(O|\hat{P}'_j)\big|\\
  &\hspace{0.8cm} + P\big|\IF^*_{k_j}(O|\hat{P}_{nj}) + \IF^*_{k_j}(O|P)\big|\big|\IF^*_{k_j}(O|\hat{P}'_j) - \IF^*_{k_j}(O|P)\big|\\
  &\hspace{0.8cm} + P\big|\IF^*_{k_j}(O|\hat{P}_{nj}) + \IF^*_{k_j}(O|\hat{P}'_j)\big|P\big|\IF^*_{k_j}(O|\hat{P}_{nj}) - \IF^*_{k_j}(O|\hat{P}'_j)\big|\Big\}1_{\mathcal{A}_n \cap \mathcal{B}_n(\tilde{K}) \cap {\cal C}_n \cap {\cal D}_n(K')}\\
  &\quad \lesssim  P\Big[\max_{k \in \mathcal{K}_n}|\IF^*_{k}(O|\hat{P}_{nj})
  - \IF^*_{k}(O|\hat{P}'_j)|1_{\mathcal{A}_n \cap \mathcal{B}_n(\tilde{K}) \cap {\cal C}_n \cap {\cal D}_n(K')}\Big]\\
  &\qquad + P\Big[\max_{k \in \mathcal{K}_n}|\IF^*_{k}(O|\hat{P}'_j) - \IF^*_{k}(O|P)|1_{\mathcal{A}_n \cap \mathcal{B}_n(\tilde{K}) \cap {\cal C}_n \cap {\cal D}_n(K')}\Big],
\end{align*}
where the last inequality holds because $\max_k|\IF^*_{k}(\cdot|\tilde{P})|1_{\mathcal{A}_n \cap \mathcal{B}_n(\tilde{K}) \cap {\cal C}_n}$ is uniformly bounded by some $(j,n)$-independent constant for $\tilde{P} \in \{\hat{P}_{nj}, \hat{P}'_j, P\}$, using the arguments for \eqref{eq:bounded_IF_star_hatPn} in Lemma \ref{lemma:bounded_var_deviance} together with \ref{assump:Covariates}--\ref{assump:At-risk prob} and \ref{assump:Variances}--\ref{assump:Conditional_mean_E0}. From the above display, Lemma~\ref{lemma:An_Bn_Cn_Dn_prob_to_one} further implies
\begin{align*}
  &|\tilde{\sigma}_{nj}^2 - \sigma^2_{nj}|1_{\mathcal{A}_n \cap \mathcal{B}_n(\tilde{K}) \cap {\cal C}_n \cap {\cal D}_n(K')}
  \lesssim P\Big[\max_{k \in \mathcal{K}_n}|\IF^*_{k}(O|\hat{P}_{nj})
  - \IF^*_{k}(O|\hat{P}'_j)|1_{\mathcal{A}_n \cap \mathcal{B}_n(\tilde{K}) \cap {\cal C}_n}\Big]\\
  &\quad + P\Big[\max_{k \in \mathcal{K}_n}|\IF^*_{k}(O|\hat{P}'_j) - \IF^*_{k}(O|P)|1_{\mathcal{A}_n \cap \mathcal{B}_n(\tilde{K}) \cap {\cal C}_n}\Big] \lesssim K_{nj}.      
\end{align*}
This completes the proof of the fact that when $\mathcal{A}_n \cap \mathcal{B}_n(\tilde{K}) \cap {\cal C}_n \cap {\cal D}_n(K')$ occurs, $\hat{\sigma}^2_{nj}$ is uniformly bounded away from zero by a non-random positive lower bound, for $j=q_n,\ldots,n$. 

When $\mathcal{A}_n \cap \mathcal{B}_n(\tilde{K}) \cap {\cal C}_n \cap {\cal D}_n(K')$ occurs,
the statement in \ref{eq:event_En_2} is an immediate consequence of the already-established \ref{eq:event_En_1}:
\[
  \bigg|\frac{\sigma_{nj}}{\hat{\sigma}_{nj}}-1\bigg| \leq \bigg|\bigg[\frac{\sigma_{nj}}{\hat{\sigma}_{nj}}-1\bigg]\bigg[\frac{\sigma_{nj}}{\hat{\sigma}_{nj}}+1\bigg]\bigg|
  = \bigg|\frac{\sigma^2_{nj}}{\hat{\sigma}^2_{nj}}-1\bigg|= \frac{|\hat{\sigma}^2_{nj}-\sigma^2_{nj}|}{\hat{\sigma}^2_{nj}} \lesssim K_{nj}.
\]
Hence, we have shown that $\mathcal{A}_n \cap \mathcal{B}_n(\tilde{K})\cap {\cal C}_n \cap {\cal D}_n(K')$ implies ${\cal E}_{n}$, where $\mathcal{E}_n$ is the event that  \ref{eq:event_En_1} and \ref{eq:event_En_2} hold with the constants that were shown to exist earlier in this proof. As ${\rm P}(\mathcal{A}_n \cap \mathcal{B}_n(\tilde{K})\cap {\cal C}_n \cap {\cal D}_n(K'))\rightarrow 1$ (Lemma \ref{lemma:An_Bn_Cn_Dn_prob_to_one}), this implies that
\begin{align*}
  {\rm P}(\mathcal{A}_n \cap \mathcal{B}_n(\tilde{K})\cap {\cal C}_n \cap {\cal D}_n(K') \cap {\cal E}_{n})&\rightarrow 1.
\end{align*}
\end{proof}

We also need the lemma below that concerns the probability of the event $\mathcal{H}_{n} = \cap_{j=q_n}^{n-1} \mathcal{H}_{nj}$, where
\begin{align*}
  \mathcal{H}_{nj} = \bigg\{\sup_{t \in \mathcal{T}}\Big|\frac{\hat{G}_j(t)}{G(t)}-1\Big| \le \frac{1}{S(\tau)}\sqrt{\frac{\log n}{j}}\: \bigg\},\; j=q_n,\ldots,n-1
\end{align*}
 with $S(\tau)={\rm P}(T \ge \tau)$. Note that $S(\tau)>0$, following that $G(\tau)>0$ by \ref{assump:Survival function}, that ${\rm P}(X \ge \tau) > 0$ by \ref{assump:At-risk prob}, and the independent censoring assumption that implies ${\rm P}(X \ge \tau)=S(\tau)G(\tau)$.
\begin{lemma} \label{lemma:Hn_prob_to_one}
Under the conditions of Theorem \ref{Thm:stab_one_step}, ${\rm P}(\mathcal{H}_{n}) \to 1$. 
\end{lemma}
\begin{proof}
Fix $\xi \in (0,1)$.
We will use the exponential bound for the Kaplan--Meier estimator that is presented in Theorem 1 of \cite{Wellner2007}. For $\lambda > 0$ and some constant $K>0$, this inequality takes the form
\begin{align*}
  {\rm P}\Big(\sqrt{j}\big\|S(\hat{G}_j - G)\big\|_{\infty} > \lambda\Big) \le 2.5\exp(-2\lambda^2+K\lambda).    
\end{align*}
Noting that
\begin{align*}
  &{\rm P}\Big(\sqrt{j}\big\|S(\hat{G}_j - G)\big\|_{\infty} > \lambda\Big) \ge
  {\rm P}\Big(\sqrt{j}\sup_{t \in {\cal T}}\big|S(t)(\hat{G}_j(t) - G(t))\big| > \lambda\Big)\\
  &\quad \ge {\rm P}\Big(\sqrt{j}S(\tau)\sup_{t \in {\cal T}}\big|\hat{G}_j(t) - G(t)\big| > \lambda\Big)
\end{align*}
and taking $\lambda=\sqrt{\log n}$, we see that
\begin{align*}
    \log {\rm P}\bigg(\sup_{t \in {\cal T}}\big|\hat{G}_j(t) - G(t)\big| > \frac{1}{S(\tau)}\sqrt{\frac{\log n}{j}} \bigg) & \le \log(2.5)-2\log n  + K\sqrt{\log n}\,.
\end{align*}
For all $n$ large enough, $K \le \xi \sqrt{\log n}$. Combining this with the fact that $S(\tau) > 0$ shows that
\begin{align*}
    \log {\rm P}\bigg(\sup_{t \in {\cal T}}\big|\hat{G}_j(t) - G(t)\big| > \frac{1}{S(\tau)}\sqrt{\frac{\log n}{j}} \bigg)& \le \log(2.5)-(2-\xi)\log(n).
\end{align*}
Hence,
\begin{align*}
  {\rm P}&\bigg(\bigcup_{j=q_n}^{n-1}\bigg\{\sup_{t \in {\cal T}}\Big|\frac{\hat{G}_j(t)}{G(t)}-1\Big| > \frac{1}{S(\tau)}\sqrt{\frac{\log n}{j}}\bigg\}\bigg)\\
  & \le \sum_{j=q_n}^{n-1} {\rm P}\bigg(\sup_{t \in {\cal T}}\Big|\frac{\hat{G}_j(t)}{G(t)}-1\Big| > \frac{1}{S(\tau)}\sqrt{\frac{\log n}{j}} \bigg)
  \le 2.5 \sum_{j=q_n}^{n-1} n^{\xi-2} \le 2.5 n^{\xi-1} \to 0.
\end{align*}
\end{proof}

Let ${\cal L}_n =
\mathcal{A}_n \cap \mathcal{B}_n(\tilde{K}) \cap \mathcal{C}_{n} \cap \mathcal{D}_{n}(K') \cap \mathcal{E}_{n} \cap \mathcal{H}_{n}$ and note that, from the above lemmas,
\begin{align}
    {\rm P}({\cal L}_n)\to 1 \label{eq:Ln1}
\end{align}
when the conditions of Theorem \ref{Thm:stab_one_step} hold. The upcoming lemmas give the asymptotic negligibility of {\rm (I),\; (II),\; (III)} and {\rm (V)} in \eqref{eq:decomposition_rootn_Sn_star}. By \eqref{eq:Ln1}, it suffices to show the asymptotically negligibility after multiplication by $1_{{\cal L}_n}$.

To show the lemmas of asymptotic negligibility, we need additional properties that are given below.
\begin{lemma} \label{lemma:bounded_multiplicative_term}
Let ${\cal X}$ be the sample space, $f_n:{\cal X} \to \mathbb{R}$ be a random function that depends on the $n$ observations with $\sup_{o \in {\cal X}}|f_n(o)1_{{\cal L}_n}| \lesssim \sqrt{\log(n \lor p_n)}/q_n^{1/4}$ with probability tending to one, $O_{j+1, k_j} \equiv (X_{j+1}, \delta_{j+1}, U_{j+1, k_j})$ and $\sigma^2_{nj} \equiv {\rm Var}(\IF^*_{k_j}(O|P))$, for $j=q_n,\ldots,n-1$.
Under the conditions of Theorem \ref{Thm:stab_one_step},
\begin{align*}
  &\bigg|\frac{1}{\sqrt{n-q_n}}\sum_{j=q_n}^{n-1} \frac{m_j(U_{j+1,k_j}-\Qj[U_{k_j}])}{\hat{\sigma}_{nj}{\rm Var}_{\Qj}(U_{k_j})}f_n(O_{j+1, k_j})1_{{\cal L}_n}\bigg|\\ 
  &\lesssim \bigg|\frac{1}{\sqrt{n-q_n}}\sum_{j=q_n}^{n-1}\frac{m_j(U_{j+1,k_j}-Q_u[U_{k_j}])}{\sigma_{nj}{\rm Var}_{Q_u}(U_{k_j})}f_n(O_{j+1}, k_j)1_{{\cal L}_n}\bigg| + o_p(1).
\end{align*}
\end{lemma}
\begin{proof}
To show the result, we first observe that
\begin{align} \label{eq:upper_bound}
  &\bigg|\frac{1}{\sqrt{n-q_n}}\sum_{j=q_n}^{n-1} \frac{m_j(U_{j+1,k_j}-\Qj[U_{k_j}])}{\hat{\sigma}_{nj}{\rm Var}_{\Qj}(U_{k_j})}f_n(O_{j+1, k_j})1_{{\cal L}_n}\bigg| \\
  &\quad \leq \bigg|\frac{1}{\sqrt{n-q_n}}\sum_{j=q_n}^{n-1}\frac{m_j(U_{j+1,k_j}-\Qj[U_{k_j}])}{\sigma_{nj}{\rm Var}_{\Qj}(U_{k_j})}f_n(O_{j+1, k_j})1_{{\cal L}_n}\bigg| \nonumber\\
  &\qquad + \bigg|\frac{1}{\sqrt{n-q_n}}\sum_{j=q_n}^{n-1}\frac{m_j(U_{j+1,k_j}-\Qj[U_{k_j}])}{\sigma_{nj}{\rm Var}_{\Qj}(U_{k_j})}\Big[\frac{\sigma_{nj}}{\hat{\sigma}_{nj}}-1\Big]f_n(O_{j+1, k_j})1_{{\cal L}_n}\bigg|. \nonumber
\end{align}
The triangle inequality upper-bounds the first term of the right-hand-side in \eqref{eq:upper_bound} by
\begin{align} \label{eq:upper_bound.1}
  &\bigg|\frac{1}{\sqrt{n-q_n}}\sum_{j=q_n}^{n-1}\frac{m_j(U_{j+1,k_j}-Q_u[U_{k_j}])}{\sigma_{nj}{\rm Var}_{Q_u}(U_{k_j})}f_n(O_{j+1, k_j})1_{{\cal L}_n}\bigg|\\
  &\quad + \bigg|\frac{1}{\sqrt{n-q_n}}\sum_{j=q_n}^{n-1}\frac{m_j(\Qj-Q_u)[U_{k_j}]}{\sigma_{nj}{\rm Var}_{Q_u}(U_{k_j})}f_n(O_{j+1, k_j})1_{{\cal L}_n}\bigg| \nonumber\\
  &\quad +
  \bigg|\frac{1}{\sqrt{n-q_n}}\sum_{j=q_n}^{n-1}\frac{m_j}{\sigma_{nj}}(U_{j+1,k_j}-\Qj[U_{k_j}])\Big[\frac{1}{{\rm Var}_{\Qj}(U_{k_j})}-\frac{1}{{\rm Var}_{Q_u}(U_{k_j})}\Big]f_n(O_{j+1, k_j})1_{{\cal L}_n}\bigg| \nonumber\\
  & \lesssim \bigg|\frac{1}{\sqrt{n-q_n}}\sum_{j=q_n}^{n-1}\frac{m_j(U_{j+1,k_j}-Q_u[U_{k_j}])}{\sigma_{nj}{\rm Var}_{Q_u}(U_{k_j})}f_n(O_{j+1}, k_j)1_{{\cal L}_n}\bigg|
  + \sqrt{n-q_n}\frac{\log(n \lor p_n)}{q_n^{3/4}} + o_p(1), \nonumber
\end{align}
where the last inequality holds because $U_k$ is uniformly bounded on $k$ by \ref{assump:Covariates}; that $\sigma_{nj}$ and ${\rm Var}_{Q_u}(U_k)$ are bounded away from zero uniformly over $(j,k)$ by \ref{assump:Variances}, and by the occurrence of ${\cal A}_n \cap {\cal C}_n \supset {\cal L}_n$, along with
$\sup_{o}|f_n(o)1_{{\cal L}_n}| \lesssim \sqrt{\log(n \lor p_n)}/q_n^{1/4}$ with probability tending to one. Note that $\sqrt{n-q_n}\log(n \lor p_n)/q_n^{3/4} \to 0$ by the conditions of Theorem \ref{Thm:stab_one_step}.

The second term of the right-hand-side in \eqref{eq:upper_bound} is bounded by
\begin{align*}  
  &\bigg|\frac{1}{\sqrt{n-q_n}}\sum_{j=q_n}^{n-1}\frac{m_j(U_{j+1,k_j}-\Qj[U_{k_j}])}{\sigma_{nj}{\rm Var}_{\Qj}(U_{k_j})}\Big[\frac{\sigma_{nj}}{\hat{\sigma}_{nj}}-1\Big]f_n(O_{j+1, k_j})1_{{\cal L}_n}\bigg|\\
  &\quad \lesssim \sup_{o}|f_n(o)1_{{\cal L}_n}|\frac{1}{\sqrt{n-q_n}}\sum_{j=q_n}^{n-1}\Big|\frac{\sigma_{nj}}{\hat{\sigma}_{nj}}-1\Big|
  \leq \sup_{o}|f_n(o)1_{{\cal L}_n}|\sqrt{n-q_n}\frac{\sqrt{\log(n \lor p_n)}}{\sqrt{q_n}}\\
  &\quad \lesssim \sqrt{n-q_n}\frac{\log(n \lor p_n)}{q_n^{3/4}} + o_p(1)\to 0.
\end{align*}  
In above display, the first inequality holds because
$|m_j(U_{j+1,k_j}-\Qj[U_{k_j}])/{\rm Var}_{\Qj}(U_{k_j})|$ is uniformly bounded by some constant that does not depend on $(j,n)$ almost surely for $j=q_n,\ldots,n-1$, which is a consequence of \ref{assump:Covariates} and \ref{eq:event_Cn_2} of Lemma \ref{lemma:An_Bn_Cn_prob_to_one}. The penultimate inequality follows \ref{eq:event_En_2} of Lemmas \ref{lemma:An_Bn_Cn_Dn_En_prob_to_one}.
\end{proof}

For $j=q_n,\ldots,n-1$, we define random functions $\tilde{e}_{nj}, e_{nj} : \mathbb{R} \times {\cal T} \times {\cal K}_n \to \mathbb{R}$ by
\begin{align} 
  \label{eq:random_func_en}
  \tilde{e}_{nj}(u,s,k) &= E_0(u,s,k)1(X_{j+1} \geq s)1\big(\inf_s Y_n(s) \ge \sqrt{n}\,\big), \mbox{ and}\\
  \label{eq:random_func_enj}
  e_{nj}(u,s,k) &= [E_0(u,s,k)-\hat{E}_j(u,s,k)]1(X_{j+1} \geq s)\\
  &\hspace{0.5cm} \times 1\Big( \,\sup_{(k,s,u)}|E_0(u,s,k)-\hat{E}_j(u,s,k)|\lesssim K_{nj},\, \inf_s Y_n(s) \ge \sqrt{n}\,\Big). \nonumber
\end{align}
Recall that $N_n(s)$ and $Y_n(s)$ are the aggregated counting process for the censored outcomes and the size of the risk set at time $s$, as
defined in Section \ref{sec:notation} (except for removing $u$), and also note that $d\bar{M}(s) \equiv dN_n(s) - Y_n(s) d\Lambda(s)$ is a local martingale with respect to the simpler filtration
\begin{align} \label{eq:filtration_Fprime.2}
  \mathcal{F}'_{s} = \sigma(\{N_n(s'),\, Y_n(s'):\,s' \le s \in \mathcal{T}\}).
\end{align}
Observing the decomposition of $\hat{\Lambda}_n-\Lambda$ analogously to \eqref{eq:martingale_est_CHF} without $u$ gives that
\begin{align} \label{eq:martingale_est_CHF.1}
  \hat{\Lambda}_n(t)-\Lambda(t) = \int_{-\infty}^{t} \frac{1(Y_n(s)>0)}{Y_n(s)}d\bar{M}(s) - \int_{-\infty}^{t} 1(Y_n(s)=0)\,d\Lambda(s).  
\end{align}

In what follows, we will need an exponential inequality for martingales with bounded jumps:
\begin{lemma} \label{lemma:martingale_exp_ineq}
Let $W_n(t), t \in \mathcal{T}$ be a martingale with jumps bounded by a constant $K_n > 0$, and the quadratic variation $\langle W_n \rangle (t) \le b_n^2$ for a constant $b_n > 0$ with respect to the filtration
$\sigma(\{N_n(s),\, Y_n(s): s \le t\}, \{\bs{U}_{i}\}_{i=1}^n)$, where both $K_n$ and $b_n$ go to zero for sufficiently large $n$. In particular,
\begin{align*}
  W_n(t) \equiv \int_{-\infty}^t w(s)\frac{1(Y_n(s) \ge \sqrt{n})}{Y_n(s)}d\bar{M}(s),
\end{align*}
where the function $w$ is uniformly bounded and left-continuous in $s$, and adapted to the given filtration.
Let $\epsilon_n$ be any sequence with values in $(0,1)$ and $\epsilon_n \to 0$; then 
$${\rm P}\big(\,|W_{n}(\tau)| \ge \epsilon_n \big) \le 2\exp\bigg(-\frac{\epsilon_n^2}{2(\epsilon_nK_n+b_n^2)}\bigg).$$
\end{lemma}
\begin{proof}
Let $\Delta W_{n}(t)$ be the jump of $W_{n}$ at time $t$:
\begin{align*}
  \Delta W_{n}(t) = W_{n}(t)-W_{n}(t-) = w(t)\frac{1(Y_n(t) \ge \sqrt{n})}{Y_n(t)}\Delta\bar{M}(t),
\end{align*}
where $\Delta\bar{M}(t)=\Delta N_n(t)-Y_n(t)\Delta\Lambda(t)=\Delta N_n(t) $, since \ref{assump:Survival function} implies that $\Lambda$ is continuous. Note also that $|\Delta N_n(t)| \le 1$ because no two individual counting processes that are aggregated in $\{N_n(t): t \in \mathcal{T}\}$ jump at the same time. Therefore, $|\Delta\bar{M}(t)|=|\Delta N_n(t)| \le 1$; along with $w$ being uniformly bounded: $|w(t)| \le K^*$ for all $t$ and a constant $K^*$ that could depend on $n$,
\begin{align*}
  |\Delta W_{n}(t)| \leq \Big|w(t)\frac{1(Y_n(t) \ge \sqrt{n})}{Y_n(t)}\Big| \le \frac{K^*}{\sqrt{n}} \equiv K_n.
\end{align*}
Meanwhile with respect to the given filtration, the predictable quadratic variation of $W_{n}(t)$ is
\begin{align*}
  \langle W_{n} \rangle (t) &= \int_{-\infty}^{t} w^2(s)\frac{1(Y_n(s) \ge \sqrt{n})}{Y_n(s)}d\Lambda(s) \le \frac{(K^*)^2\Lambda(t)}{\sqrt{n}} \equiv b_n^2(t).
\end{align*}
Obviously, $b_n^2(t) \le b_n^2(\tau) \equiv b_n^2$.
By an exponential inequality for martingales with bounded jumps \cite[cf., Lemma 2.1 of][]{vandeGeer1995}, 
\begin{align*}
  & {\rm P}\big(\,|W_{n}(t)| \ge \epsilon_n \mbox{ and } \langle W_{n} \rangle (t) \le b_n^2 \mbox{ for some } t \in {\cal T} \,\big) \le 2\exp\bigg(-\frac{\epsilon_n^2}{2\big(\epsilon_nK_n+b_n^2\big)}\bigg).
\end{align*}
Along with $\langle W_{n} \rangle (\tau) \le b_n^2$, the above display gives the same exponential bound on ${\rm P}(\,|W_{n}(\tau)| \ge \epsilon_n\,)$.
\end{proof}

\begin{lemma}\label{lemma:properties_estCHF}
For $\mathcal{I}_n = \{(j,k,u): j \in \{q_n,\ldots,n-1\}, k \in \mathcal{K}_n, u \in [-1,1]\}$, $\tilde{e}_{nj}$ and $e_{nj}$ are as defined in \eqref{eq:random_func_en} and \eqref{eq:random_func_enj}. Under the conditions of Theorem \ref{Thm:stab_one_step}, with probability tending to one,
\begin{lemmaitemenum}
  \item \label{eq:property_estCHF_1} $\sup_{(j,k,u) \in \mathcal{I}_n}\big| \int_{\mathcal{T}}\tilde{e}_{nj}(u,s,k)\{d\hat{\Lambda}_n(s)-d\Lambda(s)\}\big| \le \sqrt{\log(n \lor p_n)}/q_n^{1/4}$;

  \item \label{eq:property_estCHF_2}
   $\sup_{(j,k,u) \in \mathcal{I}_n}\big|\int_{\mathcal{T}}e_{nj}(u,s,k)\{d\hat{\Lambda}_n(s)-d\Lambda(s)\}\big| \lesssim \sqrt{\log(n \lor p_n)}/q_n^{1/4}$, relying on an appropriately specified (non-random) constant that does not depend on $n$.
\end{lemmaitemenum}
\end{lemma}
\begin{proof}
We prove \ref{eq:property_estCHF_1} and \ref{eq:property_estCHF_2} sequentially.
By the decomposition of $\hat{\Lambda}_n-\Lambda$ in \eqref{eq:martingale_est_CHF.1}, we have that for $(j,k,u) \in \mathcal{I}_n$,
\begin{align*}  
  \int_{\mathcal{T}}\tilde{e}_{nj}(u,s,k)\{d\hat{\Lambda}_n(s)-d\Lambda(s)\}
  = \int_{\mathcal{T}}\tilde{e}_{nj}(u,s,k) \frac{1(Y_n(s)>0)}{Y_n(s)}d\bar{M}(s),
\end{align*}
where $d\bar{M}(s) = dN_n(s)-Y_n(s)d\Lambda(s)$.
Note that by Lemma \ref{lemma:inclusion_E0_Ehat} and \eqref{eq:random_func_en},
there exist functions $a_k, b_k \in {\cal BV}({\cal T}, \tilde{M}_0, \tilde{M}_1)$ such that $E_0(u,s,k)=a_k(s)+b_k(s)u$ and 
\begin{align}
  \tilde{e}_{nj}(u,s,k) = [a_k(s)+b_k(s)u]1(X_{j+1} \geq s)1(Y_n(s) \ge \sqrt{n}),
\end{align}
where $a_k$ and $b_k$ are uniformly bounded by $\tilde{M}_0$ on $(k,s)$ and left-continuous in $s$, inheriting from the properties of $E_0$ that are assumed in \ref{assump:Conditional_mean_E0}.
Moreover, we see that $a_k(s)$, $b_k(s)$ and $Y_n(s)$ are predictable with respect to ${\cal F}'_{s}$ defined in \eqref{eq:filtration_Fprime.2}, because they are left-continuous in $s$ and adapted to ${\cal F}'_{s}$. Therefore, \ref{assump:Covariates} enables us to suppose that $U$ takes values in $[-1,1]$ without loss of generality, leading to
\begin{align} \label{eq:int_upper_bound}
  &\sup_{(j,k,u) \in \mathcal{I}_n}\Big|\int_{\mathcal{T}}\tilde{e}_{nj}(u,s,k) \frac{1(Y_n(s)>0)}{Y_n(s)}d\bar{M}(s) \Big|\\
  & = \sup_{(j,k,u) \in \mathcal{I}_n}\Big|\int_{\mathcal{T}}[a_k(s)+b_k(s)u]1(X_{j+1} \geq s)1(Y_n(s) \ge \sqrt{n}) \frac{1(Y_n(s)>0)}{Y_n(s)}d\bar{M}(s) \Big| \nonumber \\
  &\le \max_{(j,k)}\Big|\int_{\mathcal{T}}a_k(s)\frac{1(X_{j+1} \geq s,Y_n(s) \ge \sqrt{n})}{Y_n(s)}d\bar{M}(s) \Big| \nonumber \\
  & \hspace{0.5cm} + \max_{(j,k)}\Big|\int_{\mathcal{T}}b_k(s)\frac{1(X_{j+1} \geq s,Y_n(s) \ge \sqrt{n})}{Y_n(s)}d\bar{M}(s) \Big|. \nonumber
\end{align}
So to show the desired result, it suffices to show that each term on the right-hand-side of \eqref{eq:int_upper_bound} is bounded above by $\sqrt{\log(n \lor p_n)}/q_n^{1/4}$ with probability tending to one. Here we only tackle the first term on the right-hand-side of \eqref{eq:int_upper_bound}, and the second term can be handled using nearly identical arguments. 

To use Lemma \ref{lemma:martingale_exp_ineq} for showing the desired result, we define the required notations as follows, especially here we have the martingale $W_n$ further be indexed by $(j,k)$ (with $n$ omitted) and the function $w_k(\cdot) \equiv a_k(\cdot)$, where $w$ now is indexed by $k$. 
For $(j,k) \in \{q_n,\ldots,n-1\} \times {\cal K}_n$, let 
\begin{align*}
  W_{jk}(t) \equiv \int_{-\infty}^t w_k(s)\frac{1(X_{j+1} \geq s, Y_n(s) \ge \sqrt{n})}{Y_n(s)}d\bar{M}(s),
\end{align*}
and $\Delta W_{jk}(t)$ be the jump of $W_{jk}$ at time $t$:
\begin{align*}
  \Delta W_{jk}(t) = W_{jk}(t)-W_{jk}(t-) = w_k(t)\frac{1(X_{j+1} \geq s, Y_n(t) \ge \sqrt{n})}{Y_n(t)}\Delta\bar{M}(t);
\end{align*}
together with $|w_k(t)| \le \tilde{M}_0$ for all $t$,
\begin{align*}
  |\Delta W_{jk}(t)| \leq \Big|w_k(t)\frac{1(X_{j+1} \geq s, Y_n(t) \ge \sqrt{n})}{Y_n(t)}\Big| \le \frac{\tilde{M}_0}{\sqrt{n}} \equiv K_n.
\end{align*}
Meanwhile, the predictable quadratic variation of $W_{jk}(t)$ is
\begin{align*}
  \langle W_{jk} \rangle (t) &= \Big|\int_{-\infty}^{t} w^2_k(s)\frac{1(X_{j+1} \geq s,Y_n(t) \ge \sqrt{n})}{Y_n(s)}d\Lambda(s) \Big| \le \frac{\tilde{M}^2_0\Lambda(t)}{\sqrt{n}} \equiv b_n^2(t),
\end{align*}
and $b_n^2(t) \le b_n^2(\tau) \equiv b_n^2$. Let $\epsilon_n = \sqrt{\log(n \lor p_n)}/q_n^{1/4}$; then Lemma \ref{lemma:martingale_exp_ineq} implies that
\begin{align*}
  & {\rm P}\big(\,|W_{jk}(\tau)| \ge \epsilon_n \,\big) \le 2\exp\bigg(-\frac{\epsilon_n^2}{2\big(\epsilon_nK_n+b_n^2\big)}\bigg)\\
  & = 2\exp\bigg(-\frac{\sqrt{n}\log(n \lor p_n)}{2q_n^{1/2}}\frac{1}{\big(q_n^{-1/4}\sqrt{\log(n \lor p_n)}\tilde{M}_0+\Lambda(\tau)\tilde{M}_0^2\big)}\bigg)
\end{align*}
for $(j,k) \in \{q_n,\ldots,n-1\} \times {\cal K}_n$, leading to
\begin{align*}
  & {\rm P}\Big(\,\max_{(j,k)}|W_{jk}(\tau)| \ge \epsilon_n \,\Big) \le
  \sum_{j=q_n}^{n-1}\sum_{k=1}^{p_n}{\rm P}\Big(|W_{jk}(\tau)| \ge \epsilon_n \Big) \\
  & \quad \le 2p_n(n-q_n)\exp\bigg(-\frac{\sqrt{n}\log(n \lor p_n)}{2q_n^{1/2}}\frac{1}{\big(q_n^{-1/4}\sqrt{\log(n \lor p_n)}\tilde{M}_0+\Lambda(\tau)\tilde{M}_0^2\big)}\bigg)\\
  & \quad \le 2\exp\big(2\log(n \lor p_n)\big)\exp\big(-n^{1/4}\sqrt{\log(n \lor p_n)}\,\big)\\
  &\quad = 2\exp\Big(\Big[2\sqrt{\log(n \lor p_n)}/n^{1/4}-1\Big]n^{1/4}\sqrt{\log(n \lor p_n)}\Big).    
\end{align*}
Therefore $\limsup_{n \to \infty}{\rm P}\big(\max_{(j,k)}|W_{jk}(\tau)| \ge a_n \big)=0$, following $n/q_n = O(1)$,
$\log(n \lor p_n)/q_n^{1/4} \to 0$ and 
$n^{1/4}\sqrt{\log(n \lor p_n)} \rightarrow \infty$. Hence, we complete the proof of \ref{eq:property_estCHF_1}.

Below we present the proof of \ref{eq:property_estCHF_2}. Since the total variation of $\hat{\Lambda}_n-\Lambda$ is bounded by $\hat{\Lambda}_n(\tau)+\Lambda(\tau)$
we have
\begin{align} \label{eq:int_upper_bound.1} 
  &\sup_{(j,k,u) \in \mathcal{I}_n}\bigg|\int_{\mathcal{T}}e_{nj}(u,s,k)\{d\hat{\Lambda}_n(s)-d\Lambda(s)\}\bigg|\\
  &\quad \le \sup_{(j,k,s,u)}|e_{nj}(u,s,k)|\big|\hat{\Lambda}_n(\tau)+\Lambda(\tau)\big| \nonumber \\
  &\quad \le \sup_{(j,k,s,u)}|e_{nj}(u,s,k)|\big|\hat{\Lambda}_n(\tau)-\Lambda(\tau)\big| + 2\sup_{(j,k,s,u)}|e_{nj}(u,s,k)|\Lambda(\tau). \nonumber
\end{align}
Observing that $\{\sqrt{n}[\hat{\Lambda}_n(t)-\Lambda(t)] : t\in\mathcal{T}\}$ converges to a tight Gaussian process, and then the continuous mapping theorem gives that $\sqrt{n}\sup_{t\in\mathcal{T}}|\hat{\Lambda}_n(t)-\Lambda(t)|$ converges to the supremum of the absolute value of this Gaussian process. Therefore for any sequence $\varepsilon_n \rightarrow \infty$, ${\rm P}(\sup_{t \in \mathcal{T}}\sqrt{n}|\hat{\Lambda}_n(t)-\Lambda(t)| > \varepsilon_n \big) \rightarrow 0$, in particular, $\varepsilon_n=\sqrt{\log n}$. Consequently, $\sup_{t \in \mathcal{T}}|\hat{\Lambda}_n(t)-\Lambda(t)| \le \sqrt{\log n}/\sqrt{n}$ with probability tending to one. Following from \eqref{eq:random_func_enj} that $\sup_{(j,k,s,u)}|e_{nj}(u,s,k)| \le K_*\max_jK_{nj}$ for some positive constant $K_*$ that does not depend on $n$, we have
that with probability tending to one,
\begin{align*}
  \sup_{(j,k,s,u)}|e_{nj}(u,s,k)|\big|\hat{\Lambda}_n(\tau)-\Lambda(\tau)\big|  \le K_*\max_{j}\{K_{nj}\}\frac{\sqrt{\log n}}{\sqrt{n}} < K_*\max_{j}\{K_{nj}\}. 
\end{align*}
Hence from \eqref{eq:int_upper_bound.1},
\begin{align*}
  &\sup_{(j,k,u) \in \mathcal{I}_n}\bigg|\int_{\mathcal{T}}e_{nj}(u,s,k)\{d\hat{\Lambda}_n(s)-d\Lambda(s)\}\bigg| < K_*(1+2\Lambda(\tau))\max_{j}\{K_{nj}\}\\ 
  &\quad \le K_*(1+2\Lambda(\tau))\frac{\sqrt{\log(n \lor p_n)}}{q_n^{1/4}}
\end{align*}  
with probability tending to one, which gives \ref{eq:property_estCHF_2}.
\end{proof}

%In what follows, we will need a martingale tail bound given in  Proposition 2.1 of \cite{Freedman1975} (cf., Lemma 1 of \cite{Bae2010}):
%\begin{lemma} \label{lemma:martingale_tail_bound}
%  Let $\{(D_{j}, \mathcal{E}_{j}), j=1,\ldots,n\}$ be a martingale difference sequence such that $\max_{j}|D_j| \le B$ for a constant $B$, and $\sum_{j=1}^nE[D_j^2|\mathcal{E}_{j-1}] \le V$ for a constant $V$, where $\mathcal{E}_0$ is the trivial $\sigma$-field. If $0 \le \eta \le V/[2B]$, then 
%  $${\rm P}\bigg(\bigg|\sum_{j=1}^nD_j\bigg| > \eta \bigg) \le 2\exp\bigg(-\frac{\eta^2}{3V}\bigg).$$
%\end{lemma}

In upcoming lemmas, we show that $|\sum_{j=q_n}^{n-1}D_{n,\,j+1}|$ converges to zero in probability as $n$ goes to infinity, where for $j=q_n+1,\ldots, n$,
\begin{align*}
  &D_{nj} \equiv \frac{1}{\sqrt{n-q_n}} \frac{m_{j-1}(U_{j,k_{j-1}}-Q_u[U_{k_{j-1}}])}{\sigma_{n{j-1}}{\rm Var}_{Q_u}(U_{k_{j-1}})}\int_{{\cal T}}\tilde{e}_{n,\,j-1}(U_{j,k_{j-1}},s,k_{j-1})\{d\hat{\Lambda}_n(s)-d\Lambda(s)\}\\
  & = \frac{1}{\sqrt{n-q_n}}\frac{m_{j-1}(U_{j,k_{j-1}}-Q_u[U_{k_{j-1}}])}{\sigma_{n{j-1}}{\rm Var}_{Q_u}(U_{k_{j-1}})}\int_{{\cal T}}\tilde{e}_{n,\,j-1}(U_{j,k_{j-1}},s,k_{j-1})\frac{1(Y_n(s)>0)}{Y_n(s)}d\bar{M}(s),
\end{align*}
with $\tilde{e}_{n,\,j-1}$ as defined in \eqref{eq:random_func_en} with $j$ replaced by $j-1$, where the second equality holds by the decomposition of $\hat{\Lambda}_n-\Lambda$ in \eqref{eq:martingale_est_CHF.1}. Following $\tilde{e}_{n,\,j-1}(u,s,k)=E_0(u,s,k)1(X_{j} \ge s)1(Y_n(s) \ge \sqrt{n})$, we have a decomposition $D_{nj}=\widetilde{D}_{nj} + \widehat{D}_{nj}$, where
\begin{align} \label{eq:defs.martingale.diff}
  &\widetilde{D}_{nj} \equiv \frac{m_{j-1}(U_{j,k_{j-1}}-Q_u[U_{k_{j-1}}])}{\sqrt{n-q_n}\sigma_{n{j-1}}{\rm Var}_{Q_u}(U_{k_{j-1}})}\int_{{\cal T}}\big\{E_0(U_{j,k_{j-1}},s,k_{j-1})-E_0(U_{j-1,k_{j-2}},s,k_{j-2})\big\} \\
  &\hspace{6.5cm} \times \{d\hat{\Lambda}_n(s)-d\Lambda(s)\}; \nonumber \\
  &\widehat{D}_{nj} \equiv \frac{1}{\sqrt{n-q_n}}\frac{m_{j-1}(U_{j,k_{j-1}}-Q_u[U_{k_{j-1}}])}{\sigma_{n{j-1}}{\rm Var}_{Q_u}(U_{k_{j-1}})}\int_{{\cal T}}E_0(U_{j-1,k_{j-2}},s,k_{j-2})1(X_{j} \ge s) \nonumber \\
  &\hspace{7.5cm} \times \frac{1(Y_n(s) \ge \sqrt{n})}{Y_n(s)}d\bar{M}(s), \nonumber
\end{align}
so that $|\sum_{j=q_n}^{n-1}D_{n,\,j+1}| \le |\sum_{j=q_n}^{n-1}\widetilde{D}_{n,\,j+1}| + |\sum_{j=q_n}^{n-1}\widehat{D}_{n,\,j+1}|$. To show the desired result, it therefore suffices to show that both $|\sum_{j=q_n}^{n-1}\widetilde{D}_{n,\,j+1}|$ and $ |\sum_{j=q_n}^{n-1}\widehat{D}_{n,\,j+1}|$ converge to zero in probability, which will be presented in Lemmas \ref{lemma:convergence_sum_widetildeDnj} and \ref{lemma:convergence_sum_widehatDnj}, respectively. Henceforth we state that $(j,s) \in \{q_n,\ldots,n-1\} \times {\cal T}$; $(j',s) \in \{q_n,\ldots,n-1\} \times {\cal T}$ and $(k,u) \in {\cal K}_n \times \mathbb{R}$ in which the ranges will be omitted for succinct presentation in forthcoming displays.

\begin{lemma} \label{lemma:convergence_sum_widetildeDnj}
Under the conditions of Theorem \ref{Thm:stab_one_step}, $\sum_{j=q_n}^{n-1}\widetilde{D}_{n,\,j+1}$ converges to zero in probability as $n$ goes to infinity, where $\widetilde{D}_{nj}$ is as defined in \eqref{eq:defs.martingale.diff}
for $j=q_n+1,\ldots, n$.
\end{lemma}
\begin{proof}
Using \ref{assump:Covariates}, \ref{assump:Variances} and the total variation of $\hat{\Lambda}_n-\Lambda$ over ${\cal T}$ is bounded by $|\hat{\Lambda}_n(\tau)+\Lambda(\tau)|$, we have that
\begin{align*} 
  & \bigg|\sum_{j=q_n}^{n-1}\widetilde{D}_{n,\,j+1}\bigg|\\
  &\lesssim \sqrt{n-q_n}
  \max_{j}\bigg|\int_{\mathcal{T}}\big\{E_0(U_{j,{k_{j-1}}},s,k_{j-1})-E_0(U_{j-1,k_{j-2}},s,k_{j-2})\big\}\{d\hat{\Lambda}_n(s)-d\Lambda(s)\}\bigg| \\
  &\le \sqrt{n-q_n}\sup_{(j,s)}\Big|E_0(U_{j,k_{j-1}},s,k_{j-1})-E_0(U_{j-1,k_{j-2}},s,k_{j-2})\Big|\big|\hat{\Lambda}_n(\tau)+\Lambda(\tau)\big| \nonumber \\
  &\le \sqrt{n-q_n}\Big[\big|\hat{\Lambda}_n(\tau)-\Lambda(\tau)\big| + 2\Lambda(\tau)\Big]\sup_{(j,s)}\Big|E_0(U_{j,k_{j-1}},s,k_{j-1})-E_0(U_{j-1,k_{j-2}},s,k_{j-2})\Big|, \nonumber
\end{align*}
where the last line holds by the triangle inequality.

Since $\sup_{t \in \mathcal{T}}|\hat{\Lambda}_n(t)-\Lambda(t)| \le \sqrt{\log n}/\sqrt{n}$ with probability tending to one, taking the expectation on the above display gives that
\begin{align*} 
  & E\bigg[\,\bigg|\sum_{j=q_n}^{n-1}\widetilde{D}_{n,\,j+1}\bigg|\,\bigg]\\
  &\le \sqrt{n-q_n}\bigg[\frac{\sqrt{\log n}}{\sqrt{n}} + 2\Lambda(\tau)\bigg]E\bigg[\sup_{(j,s)}\big|E_0(U_{j,k_{j-1}},s,k_{j-1})-E_0(U_{j-1,k_{j-2}},s,k_{j-2})\big|\bigg] \to 0, \nonumber
\end{align*}
where the convergence follows \ref{assump:Conditional_mean_E0}. Hence we complete the proof.
\end{proof}

\begin{lemma} \label{lemma:convergence_sum_widehatDnj}
Under the conditions of Theorem \ref{Thm:stab_one_step}, $\sum_{j=q_n}^{n-1}\widehat{D}_{n,\,j+1}$ converges to zero in probability as $n$ goes to infinity, where $\widehat{D}_{nj}$ is as defined in \eqref{eq:defs.martingale.diff}
for $j=q_n+1,\ldots, n$.
\end{lemma}
\begin{proof}
First note that, with respect to the filtration ${\cal O}_{nj} \equiv \sigma(O_1,\ldots,O_j, {\cal F}'_{\tau})$ with ${\cal F}'_s$ defined in \eqref{eq:filtration_Fprime.2},
we have that $\widehat{D}_{nj}$ is a martingale difference sequence: 
\begin{align*}
  E[\widehat{D}_{n,\,j+1}|{\cal O}_{nj}] = & \frac{m_j}{\sqrt{n-q_n}\sigma_{nj}{\rm Var}_{Q_u}(U_{k_j})}\Big\{\int_{{\cal T}}E_0(U_{j, k_{j-1}},s,k_{j-1})\\
  &\times \frac{1(X_{j+1} \ge s, Y_n(s) \ge \sqrt{n})}{Y_n(s)}d\bar{M}(s)\Big\}
  E\big[U_{j+1,k_j}-Q_u[U_{k_j}]\big|{\cal O}_{nj}\big] = 0
\end{align*} 
for $j=q_n,\ldots, n-1$.

We also provide an upper bound on  the conditional variance
\begin{align*} %\label{eq:bounded_quad_var}
  &\sum_{j=q_n}^{n-1}E[\widehat{D}^2_{n,\,j+1}|{\cal O}_{nj}] = \frac{1}{(n-q_n)} 
  \sum_{j=q_n}^{n-1}\frac{1}{\sigma^2_{nj}{\rm Var}^2_{Q_u}(U_{k_j})}E\big[(U_{j+1,k_j}-Q_u[U_{k_j}])^2\big|{\cal O}_{nj}\big] \\
  & \hspace{4cm} \times \Big\{\int_{{\cal T}}E_0(U_{j, k_{j-1}},s,k_{j-1})\frac{1(X_{j+1} \ge s, Y_n(s) \ge \sqrt{n})}{Y_n(s)}d\bar{M}(s)\Big\}^2 \nonumber \\
  & \le K'_0\sup_{(j,k,u)}\Big\{\int_{{\cal T}}E_0(u,s,k)\frac{1(X_{j+1} \ge s, Y_n(s) \ge \sqrt{n})}{Y_n(s)}d\bar{M}(s)\Big\}^2, \nonumber
\end{align*}
where $K'_0 > 0$ is a constant such that for all $n$
\begin{align*}
  \max_{j}\bigg\{\frac{E\big[(U_{j+1,k_j}-Q_u[U_{k_j}])^2\big|{\cal O}_{nj}\big]}{\sigma^2_{nj}{\rm Var}^2_{Q_u}(U_{k_j})}\bigg\} \le 
  \max_{(j,k)}\bigg\{\frac{E\big[(U_{j+1,k}-Q_u[U_{k}])^2 \big|{\cal O}_{nj}\big]}{\sigma^2_{nj}{\rm Var}^2_{Q_u}(U_{k})}\bigg\} \le K'_0
\end{align*}
almost surely; this constant exists by \ref{assump:Covariates} and \ref{assump:Variances}.
We see that
\begin{align*}
  \sup_{(j,k,u)}\Big\{\int_{{\cal T}}E_0(u,s,k)\frac{1(X_{j+1} \ge s, Y_n(s) \ge \sqrt{n})}{Y_n(s)}d\bar{M}(s)\Big\}^2 \le \frac{\log(n \lor p_n)}{\sqrt{q_n}}.
\end{align*}
with probability tending to one, which can be seen from \ref{eq:property_estCHF_1} of Lemma \ref{lemma:properties_estCHF}. We therefore find that the conditional variance $\sum_{j=q_n}^{n-1}E[\widehat{D}^2_{n,\,j+1}|{\cal O}_{nj}] \to 0$ in probability. 

By the martingale central theorem given in Theorem 1.2 of \cite{Kundu2000},
$\sum_{j=q_n}^{n-1}\widehat{D}_{n,\,j+1}$ converges in probability to zero.  % Note that the limit of the conditional variance in our case is zero, whereas Pollard assumes that the limit is a positive constant, but Pollard's proof applies in our case as well \cite[cf. Theorem 1.2 in][]{Kundu2000}. 
Indeed, the conditional Lindeberg condition  holds trivially: for every $\varepsilon > 0$,
\begin{align*}
 \sum_{j=q_n}^{n-1}E\big[\widehat{D}^2_{n,\,j+1}1\big(\big|\widehat{D}_{n,\,j+1}\big| > \varepsilon \big) \big| {\cal O}_{nj} \big] \le \sum_{j=q_n}^{n-1}E[\widehat{D}^2_{n,\,j+1}|{\cal O}_{nj}] \lcrarrow{p} 0.
\end{align*}
This completes the proof.
\end{proof}

Recall that for each $(j,k)$,
\begin{align*}
  e_{nj}(u,s,k) = & [E_0(u,s,k)-\hat{E}_{j}(u,s,k)]1(X_{j+1} \ge s) 1\Big(\,\sup_{(k,s,u)}|E_0(u,s,k)-\hat{E}_{j}(u,s,k)|\\ 
  & \le \tilde{K}K_{nj},\, \inf_s Y_n(s) \ge \sqrt{n}\,\Big).
\end{align*}

\begin{lemma} \label{lemma:convergence_sum_Dnj.2}
Under the conditions of Theorem \ref{Thm:stab_one_step}, $\sum_{j=q_n}^{n-1}Q_{n,\,j+1}$ converges to zero in probability as $n$ goes to infinity, where for $j=q_n+1,\ldots, n$,
\begin{align*}
  Q_{nj} \equiv \frac{1}{\sqrt{n-q_n}} \frac{m_{j-1}(U_{j,k_{j-1}}-Q_u[U_{k_{j-1}}])}{\sigma_{n{j-1}}{\rm Var}_{Q_u}(U_{k_{j-1}})}\int_{{\cal T}}e_{n,\,j-1}(U_{j,k_{j-1}},s,k_{j-1})\{d\hat{\Lambda}_n(s)-d\Lambda(s)\},
\end{align*}
and $e_{n,\,j-1}$ is as defined in \eqref{eq:random_func_enj} with $j$ replaced by $j-1$.
\end{lemma}

\begin{proof}
First using \ref{assump:Covariates} and \ref{assump:Variances} to bound the middle part of $Q_{nj}$, together with
the expression of $\hat{\Lambda}_n-\Lambda$ in \eqref{eq:martingale_est_CHF.1}, it upper bounds $\big|\sum_{j=q_n}^{n-1}Q_{n,\,j+1}\big|$ by (up to a constant that does not depend on $(j,n)$)
\begin{align*} 
  \sqrt{n-q_n}
  \max_{(j,k)}\bigg|\int_{\mathcal{T}}e_{nj}(U_{j+1,k},s,k)1(X_{j+1} \ge s)\frac{1(Y_n(s) \ge \sqrt{n})}{Y_n(s)}d\bar{M}(s)\bigg|.
\end{align*}
Recall that $\mathcal{I}_n = \{(j,k) : j \in \{q_n,\ldots,n-1\}, k \in {\cal K}_n\}$ and $K_{nq_n} = \sqrt{\log(n \lor p_n)/q_n}$. For each $(j,k) \in \mathcal{I}_n$, Lemma \ref{lemma:inclusion_E0_Ehat}, along with \ref{assump:Covariates} and \ref{assump:Conditional_mean_E0}, indicates that the integrand in the above martingale integral, namely 
\begin{align*}
  s \mapsto e_{nj}(U_{j+1,k},s,k)\sqrt{n-q_n}\frac{1(X_{j+1} \ge s\,, Y_n(s) \ge \sqrt{n})}{Y_n(s)},   
\end{align*}
belongs to $\mathcal{H}_n(\tilde{K})$ with probability tending to one, where $\mathcal{H}_n(\tilde{K})$ is the class of c\`{a}gl\`{a}d functions in ${\cal BV}({\cal T},\tilde{K}K_{nq_n},\tilde{M}_1)$.
Therefore to show that $\sum_{j=q_n}^{n-1}Q_{n,\,j+1} \crarrow{p} 0$, it suffices to show that for $\eta > 0$,
\begin{align*}
  {\rm P}\bigg(\,\sup_{h \in \mathcal{H}_n(\tilde{K})}\bigg| \int_{\mathcal{T}} h(s)\, d\bar{M}(s) \bigg| > \eta \bigg) \to 0
\end{align*}
under the assumed conditions for Theorem \ref{Thm:stab_one_step}: $n/q_n = O(1)$ and $\log(n \lor p_n)/q_n^{1/4} \to 0$.

To show the desired result, we use the exponential inequality in Theorem 3.1 of \cite{vandeGeer1995}. For  fixed $n$, the aggregated counting process $\{N_n(t): t \in \mathcal{T}\}$ plays the role of  $\{N(t): t \in [0,T]\}$, the compensator $\{A_n(t): t \in \mathcal{T}\}$ is $\{A(t): t \in [0,T]\}$,
the integrand  $h$ in the  class $\mathcal{H}_n(\tilde{K})$ corresponds to  $g \in \mathscr{G}$, $\tau$ is $T$, $\sigma_{n\tau}^2$ is $\sigma_{T}^2$, and
$d_{n\tau}^2(h,0)$ is $d_{T}^2(g,0)$. Then we have the martingale process $\{\bar{M}(t) = N_n(t) - A_n(t): t \in \mathcal{T}\}$.
In addition, $n$ is now added as a subscript to van de Geer's various constants $L$, $K$, $\varepsilon$, $b$ and $C_1, C_2, C_3, C_4$. The upper bound $H_n(\delta, b_n, B)$ on the $\delta$-bracketing entropy of $\mathcal{H}_n(\tilde{K})$ is as defined in \cite{vandeGeer1995}, where $B$ is a measurable subset of $\{A_{n}(\tau) \le \sigma^2_{n\tau}\}$.

The compensator in our case is
\begin{align*}
  A_{n}(t) = \int_{-\infty}^tY_n(s)\, d\Lambda(s) 
\end{align*}  
so that $A_{n}(\tau) \le \sigma_{n\tau}^2 = n\Lambda(\tau)$. 
%For all $s \in \mathcal{T}$ and $h \in \mathcal{H}_n(\tilde{K})$, $h(s) \ge -\tilde{K}K_{nq_n}$, so we let $L = \tilde{K}$. 
Note also that for any $h \in \mathcal{H}_n(\tilde{K})$,
\begin{align*}
  d_{n\tau}^2(h,0) = \frac{1}{2}\int_{-\infty}^{\tau}\big[\exp(h(s))-\exp(0)\big]^2dA_{n}(s)
  \le \frac{1}{2}\exp\big(4\tilde{K}K_{nq_n}\big)n\Lambda(\tau) \equiv b_n^2,
\end{align*}
where we have used the inequality $|e^x-1|\le |x|e^{|x|} \le e^{2|x|}$. Moreover, we may take $H_n(\delta, b_n, B) = c_0\sqrt{n}/\delta$ for some constant $c_0 > 0$ \cite[Example 19.11 of][]{Vaart1998}.

We need to check condition (3.2) of \cite{vandeGeer1995}, namely that
\begin{align*}
  \frac{\varepsilon_n b^2_n}{C_{n1}} \ge \int^{b_n}_{\varepsilon_n b_n^2/(C_{n2}\sigma_{n\tau}) \wedge b_n/8} \sqrt{H_n(x, b_n, B)}\, dx \lor b_n,
\end{align*}
for appropriate choices of the various constants. For now assume  $\varepsilon_n \in (0,1]$, and take  $C_{n2} = 8\varepsilon_nb_n/\sqrt{n\Lambda(\tau)}.$
Together with the definition of $\sigma_{n\tau}$, this leads to
\begin{align*}
  \frac{\varepsilon_n b_n^2}{C_{n2}\sigma_{n\tau}} &= \varepsilon_n \frac{\sqrt{n\Lambda(\tau)}}{8\varepsilon_nb_n}b_n^2\frac{1}{\sqrt{n\Lambda(\tau)}} \ge \frac{b_n}{8} \;\mbox{ and }\\
  \int^{b_n}_{\varepsilon_n b_n^2/(C_{n2}\sigma_{n\tau}) \wedge b_n/8} \sqrt{H(x, b_n, B)} dx &= \sqrt{c_0}n^{1/4}\int^{b_n}_{b_n/8} \frac{1}{\sqrt{x}} dx = 2\sqrt{c_0}\bigg\{1-\frac{1}{\sqrt{8}}\bigg\}n^{1/4}\sqrt{b_n}.
\end{align*}
Taking  $C_{n1} = \varepsilon_n b_n^{3/2}/\big\{2\sqrt{c_0}\big[1-8^{-1/2}\big]n^{1/4}\big\}$
shows that van de Geer's condition (3.2) is  satisfied. Then, applying her result with  $B \subset \{A_{n}(\tau) \le \sigma^2_{n\tau}\}$  having  probability one, gives 
\begin{align} \label{eq:exp_prob}
  & {\rm P}\bigg(\,\sup_{h \in \mathcal{H}_n(\tilde{K})}\bigg|\int_{\mathcal{T}}h(s)d\bar{M}(s)\bigg| > \varepsilon_n b_n^2 \bigg) \\
  & \le {\rm P}\bigg(\,\bigg\{\bigg|\int_{\mathcal{T}}h(s)d(N_n-A_n)(s)\bigg| \ge \varepsilon_n b_n^2 \mbox{ and } d^2_{n\tau}(h,0) \le b_n^2 \mbox{ for some } h \in \mathcal{H}_n(\tilde{K})\bigg\} \cap B\bigg) \nonumber \\
  &\le C_{n3}\exp\bigg(-\frac{\varepsilon^2_nb_n^2}{C_{n4}}\bigg), \nonumber
\end{align}
where $C_{n3}$ and $C_{n4}$ are specified below.

From the proof of van de Geer's theorem, we can take $C_{n3}=2$ and $C_{n4} \ge 2(\varepsilon_nK_n + c_n)$, where
\begin{align*}
  c_n = \frac{4\{\exp(-L_n)-1-L_n\}}{\{\exp(-L_n)-1\}^2},
\end{align*}
$L_n=\tilde{K}K_{nq_n}$ and $K_n= \tilde{K}K_{nq_n}$. Note that $c_{n} < 0$ since $L_n>0$, so we can take
$C_{n4} = 2\varepsilon_nK_n$.
Finally, specifying $\varepsilon_n=n^{-5/4}$, we have
\begin{align*} 
  &\frac{\varepsilon_n^2b_n^2}{C_{n4}} = \frac{n^{-5/2}2^{-1}\exp\big(4\tilde{K}K_{nq_n}\big)n\Lambda(\tau)}{2n^{-5/4}K_n} = \frac{n^{-3/2}\exp\big(4\tilde{K}K_{nq_n}\big)\Lambda(\tau)}{4n^{-5/4}K_n} \nonumber \\
  & = \frac{\exp\big(4\tilde{K}\sqrt{\log(n \lor p_n)}q_n^{-1/2}\big)\Lambda(\tau)}{4\tilde{K}\sqrt{\log(n \lor p_n)}n^{1/4}q_n^{-1/2}} \\
  & = \frac{\exp\big(4\tilde{K}\sqrt{\log(n \lor p_n)}q_n^{-1/2}\big)\Lambda(\tau)}{4\tilde{K}}\frac{q_n^{1/4}}{\sqrt{\log(n \lor p_n)}}\frac{q_n^{1/4}}{n^{1/4}} \to \infty,
\end{align*}
under the conditions of our Theorem \ref{Thm:stab_one_step}: $n/q_n = O(1)$ and $\log(n \lor p_n)/q_n^{1/4} \to 0$, and note that $\tilde{K}>0$. Thus, from \eqref{eq:exp_prob}, we have 
\begin{align*}
 & \limsup_{n \to \infty}{\rm P}\bigg(\,\sup_{h \in \mathcal{H}_n(\tilde{K})}\bigg| \int_{\mathcal{T}} h(s)d\bar{M}(s) \bigg| > \varepsilon_nb_n^2 \bigg) = 0,
\end{align*}
and since  $\varepsilon_nb_n^2 = n^{-5/4}\exp\big(4\tilde{K}\sqrt{\log(n \lor p_n)}q_n^{-1/2}\big)n\Lambda(\tau)/2 \to 0$,  the proof is complete.
\end{proof}

\begin{lemma} \label{lemma:asymptotic_negligibility_(I)}
Under the conditions of Theorem \ref{Thm:stab_one_step}, $\mbox{(I)}1_{{\cal L}_n}$ is asymptotically negligible.
\end{lemma}
\begin{proof}
Recall that
\begin{align*}
  \mbox{(I)} &= \frac{1}{\sqrt{n-q_n}}\sum_{j=q_n}^{n-1}\Big\{\frac{\hat{\sigma}_{nj}^{-1}m_j}{{\rm Var}_{\Qj}(U_{k_j})}(U_{j+1,k_j}-\Qj[U_{k_j}])\\
  &\hspace{3.0cm} \times\int_{\mathcal{T}}[E_0(U_{j+1, k_j},s,k_j)-\hat{E}_j(U_{j+1, k_j},s,k_j)]d\hat{M}_{j+1}(s)\Big\}.
\end{align*}
Let $dM_{j+1}(s) = 1(X_{j+1} \in ds, \delta_{j+1}=0)-1(X_{j+1} \geq s)d\Lambda(s)$, and $d\hat{M}_{j+1}(s) - dM_{j+1}(s)
= -1(X_{j+1} \geq s)\big\{d\hat{\Lambda}_n(s)-d\Lambda(s)\big\}$, so that
\begin{align*}
  &\Big| 1_{{\cal L}_n}\int_{\mathcal{T}}[E_0(U_{j+1, k_j},s,k_j)-\hat{E}_j(U_{j+1, k_j},s,k_j)]\big\{d\hat{M}_{j+1}(s)-dM_{j+1}(s)\big\} \Big| \\
  &\quad = \Big| 1_{{\cal L}_n}\int_{\mathcal{T}}[E_0(U_{j+1, k_j},s,k_j)-\hat{E}_j(U_{j+1, k_j},s,k_j)]1(X_{j+1} \geq s) \big\{d\hat{\Lambda}_n(s)-d\Lambda(s)\big\} \Big|.
\end{align*}
Along with the triangle inequality, the above results imply that
\begin{align*}
  &|\mbox{(I)}1_{{\cal L}_n}| \le \bigg|\frac{1}{\sqrt{n-q_n}}\sum_{j=q_n}^{n-1}\frac{\hat{\sigma}_{nj}^{-1}m_j}{{\rm Var}_{\Qj}(U_{k_j})}(U_{j+1,k_j}-\Qj[U_{k_j}])\\
  &\hspace{4cm} \times 1_{{\cal L}_n}\int_{\mathcal{T}}\big[E_0(U_{j+1, k_j},s,k_j)-\hat{E}_j(U_{j+1, k_j},s,k_j)\big]dM_{j+1}(s)\bigg|\\
  &\quad + \bigg|\frac{1}{\sqrt{n-q_n}}\sum_{j=q_n}^{n-1}\frac{\hat{\sigma}_{nj}^{-1}m_j}{{\rm Var}_{\Qj}(U_{k_j})}(U_{j+1,k_j}-\Qj[U_{k_j}])\\
  &\hspace{1cm} \times 1_{{\cal L}_n}\int_{\mathcal{T}}\big[E_0(U_{j+1, k_j},s,k_j)-\hat{E}_j(U_{j+1, k_j},s,k_j)\big]1(X_{j+1} \geq s) \big\{d\hat{\Lambda}_n(s)-d\Lambda(s)\big\}\bigg|.
\end{align*}
We further apply Lemma \ref{lemma:bounded_multiplicative_term} to the above display, taking
\begin{align*}
  f_{n1}(O_{j+1,k_j}) &= 1_{{\cal L}_n}\int_{\mathcal{T}}\big[E_0(U_{j+1, k_j},s,k_j)-\hat{E}_j(U_{j+1, k_j},s,k_j)\big]dM_{j+1}(s),\mbox{ with}\\
  \sup_{o}\big|f_{n1}(o)\big| &\leq \sup_{(j,k,s,u)}\big|E_0(u,s,k)-\hat{E}_j(u,s,k)\big|(1+\Lambda(\tau)) \lesssim K_{nq_n},
\end{align*}
and
\begin{align*}
  &f_{n2}(O_{j+1,k_j}) = 1_{{\cal L}_n}\int_{\mathcal{T}}\big[E_0(U_{j+1, k_j},s,k_j)-\hat{E}_j(U_{j+1, k_j},s,k_j)\big]1(X_{j+1} \geq s)\big\{d\hat{\Lambda}_n(s)-d\Lambda(s)\big\};\\  
  &\sup_{o}\big|f_{n2}(o)\big| = \sup_{(j,k,u)}\bigg|1_{{\cal L}_n}\int_{\mathcal{T}}\big[E_0(u,s,k)-\hat{E}_j(u,s,k)\big]1(X_{j+1} \ge s)\big\{d\hat{\Lambda}_n(s)-d\Lambda(s)\big\}\bigg|.
\end{align*}
By the definition of $e_{nj}$ in \eqref{eq:random_func_enj} that implies
\begin{align} \label{eq:random_func_enj.1}
  e_{nj}(u,s,k) = 1_{{\cal L}_n}\big[E_0(u,s,k)-\hat{E}_j(u,s,k)\big]1(X_{j+1} \geq s),
\end{align}
so we see from \ref{eq:property_estCHF_2} of Lemma \ref{lemma:properties_estCHF} that $\sup_{o}\big|f_{n2}(o)\big| \lesssim \sqrt{\log(n \lor p_n)}/q_n^{1/4}$ with probability tending to one.
Using Lemma \ref{lemma:bounded_multiplicative_term}, it is implied that 
\begin{align*}
  &|\mbox{(I)}1_{{\cal L}_n}| \lesssim
  \bigg|\frac{1}{\sqrt{n-q_n}}\sum_{j=q_n}^{n-1}\frac{\sigma_{nj}^{-1}m_j}{{\rm Var}_{Q_u}(U_{k_j})}(U_{j+1,k_j}-Q_u[U_{k_j}])\\
  &\hspace{4cm} \times 1_{{\cal L}_n}\int_{\mathcal{T}}\big[E_0(U_{j+1, k_j},s,k_j)-\hat{E}_j(U_{j+1, k_j},s,k_j)\big]dM_{j+1}(s)\bigg|\\
  &+ \bigg|\frac{1}{\sqrt{n-q_n}}\sum_{j=q_n}^{n-1}\frac{\sigma_{nj}^{-1}m_j}{{\rm Var}_{Q_u}(U_{k_j})}(U_{j+1,k_j}-Q_u[U_{k_j}])1_{{\cal L}_n} \\
  &\quad \times \int_{\mathcal{T}}\big[E_0(U_{j+1, k_j},s,k_j)-\hat{E}_j(U_{j+1, k_j},s,k_j)\big]1(X_{j+1} \geq s) \big\{d\hat{\Lambda}_n(s)-d\Lambda(s)\big\}\bigg| + o_p(1).
\end{align*}
Moreover, note that $1_{{\cal L}_n} \le 1\big( \,\sup_{(j,k,s,u)}|E_0(u,s,k)-\hat{E}_j(u,s,k)|\lesssim K_{nq_n}\,\big)$, and then define
\begin{align} \label{eq:random_func_bar_enj}
  &\bar{e}_{nj}(u,s,k) = [E_0(u,s,k)-\hat{E}_j(u,s,k)]1(X_{j+1} \geq s)\\
  &\qquad \times 1\bigg(\,\sup_{(j,k,s,u)}\big|E_0(u,s,k)-\hat{E}_j(u,s,k)\big| \lesssim K_{nq_n} \,\bigg). \nonumber
\end{align}
Then by \eqref{eq:random_func_enj.1} and \eqref{eq:random_func_bar_enj}, we have that
\begin{align*}
  &|\mbox{(I)}1_{{\cal L}_n}| \lesssim
  \bigg|\frac{1}{\sqrt{n-q_n}}\sum_{j=q_n}^{n-1}\frac{m_j(U_{j+1,k_j}-Q_u[U_{k_j}])}{\sigma_{nj}{\rm Var}_{Q_u}(U_{k_j})}\int_{\mathcal{T}}\bar{e}_{nj}(U_{j+1,k_j},s,k_j)dM_{j+1}(s)\bigg|\\
  & \quad + \bigg|\frac{1}{\sqrt{n-q_n}}\sum_{j=q_n}^{n-1}\frac{m_j(U_{j+1,k_j}-Q_u[U_{k_j}])}{\sigma_{nj}{\rm Var}_{Q_u}(U_{k_j})} \int_{\mathcal{T}}e_{nj}(U_{j+1,k_j},s,k_j) \{d\hat{\Lambda}_n(s)-d\Lambda(s)\}\bigg|\\
  & \quad + o_p(1),
\end{align*}
where the middle term converges to zero in probability by Lemma \ref{lemma:convergence_sum_Dnj.2}. 

Now we deal with the first term on the right-hand-side of the above inequality. Fix $n$ and for $j \in \{q_n+1,\ldots, n\}$,
\begin{align*}
  \bar{H}_{nj} \equiv \frac{m_{j-1}(U_{j,k_{j-1}}-Q_u[U_{k_{j-1}}])}{\sqrt{n-q_n}\sigma_{n,\,j-1}{\rm Var}_{Q_u}(U_{k_{j-1}})}\int_{\mathcal{T}}\bar{e}_{n,\,j-1}(U_{j,k_{j-1}},s,k_{j-1})dM_{j}(s).
\end{align*}
From \eqref{eq:random_func_bar_enj} and by \ref{assump:Covariates},
\begin{align}\label{eq:upper_bound_enj.1}
  \sup_{(j,s)}|\bar{e}_{nj}(U_{j+1,k_{j}},s,k_{j})| \le 
  \sup_{(j,k,s,u)}\big|E_0(u,s,k)-\hat{E}_j(u,s,k)\big| \lesssim
  K_{nq_n} \; \mbox{ a.s.}
\end{align}
Along with the uniform boundedness of $U_k$ almost surely in \ref{assump:Covariates} and that $\sigma_{nj}$ and ${\rm Var}_{Q_u}(U_k)$ are uniformly bounded away from zero in \ref{assump:Variances}, \eqref{eq:upper_bound_enj.1} implies that
there exists a $B_n\equiv K'\sqrt{\log(n \lor p_n)}/\sqrt{(n-q_n)q_n}$ for some constant $K'>0$ such that $\max_j|\bar{H}_{nj}| \le B_n$ almost surely. 

Define the filtration ${\cal O}_{nj} \equiv \sigma(O_1,\ldots,O_j, \bs{U}_{j+1})$.
We know that $\{(\bar{H}_{nj}, {\cal O}_{nj}), j=q_n+1,\ldots, n\}$ is a martingale difference sequence because $E|\bar{H}_{nj}| < \infty$; $\bar{H}_{nj}$ is ${\cal O}_{nj}$-measurable, and for $j=q_n,\ldots, n-1$,
\begin{align*}  
  & E\big[\bar{H}_{n,\,j+1}\big|{\cal O}_{nj}\big]\\
  & = \frac{1}{\sqrt{n-q_n}}\frac{m_{j}(U_{j+1,k_{j}}-Q_u[U_{k_{j}}])}{\sigma_{nj}{\rm Var}_{Q_u}(U_{k_{j}})}\int_{\mathcal{T}}\bar{e}_{nj}(U_{j+1,k_{j}},s,k_{j})E\big[dM_{j+1}(s)\big|{\cal O}_{nj}\big]\\
  &=\frac{1}{\sqrt{n-q_n}}\frac{m_j(U_{j+1,k_j}-Q_u[U_{k_j}])}{\sigma_{nj}{\rm Var}_{Q_u}(U_{k_j})} \int_{\mathcal{T}}\bar{e}_{nj}(U_{j+1, k_j},s,k_j)E[1(T_{j+1} \ge s)|\bs{U}_{j+1}] \\
  &\hspace{6cm} \times E\left[1(C_{j+1}\in ds)-1(C_{j+1}\ge s)d\Lambda(s)\right]\\
  & = 0,
\end{align*}
where the second equality holds by the independent censoring assumption, and the last step follows from the definition $d\Lambda(s) = {\rm P}(C \in ds)/{\rm P}(C \ge s)$. Then for $\varepsilon > 0$,
\begin{align*}
  & {\rm P}\bigg(\,\bigg|\sum_{j=q_n}^{n-1}\bar{H}_{n,\,j+1}\bigg| \ge \varepsilon \bigg) \le \varepsilon^{-2}E\bigg[\bigg(\sum_{j=q_n}^{n-1}\bar{H}_{n,\,j+1}\bigg)^2\,\bigg] \\
  &= \varepsilon^{-2}\bigg(\sum_{j=q_n}^{n-1}E\big[\bar{H}_{n,\,j+1}^2\big] + 2\sum_{q_n\le i< j\le n-1}E\Big[\bar{H}_{n,\,i+1}E\Big[\bar{H}_{n,\,j+1} \Big| \mathcal{O}_{nj}\Big]\Big]\bigg) \\
  &= \varepsilon^{-2}\sum_{j=q_n}^{n-1}E\big[ \bar{H}_{n,\,j+1}^2 \big]
  \le \varepsilon^{-2}(n-q_n)B_n^2 \to 0.
\end{align*}
As $|\mbox{(I)}1_{{\cal L}_n}| \lesssim |\sum_{j=q_n}^{n-1}\bar{H}_{n,\,j+1}| + o_p(1)$, we conclude that $\mbox{(I)}1_{{\cal L}_n}=o_p(1)$.
\end{proof}

\begin{lemma} \label{lemma:asymptotic_negligibility_(II)}
Under the conditions of Theorem \ref{Thm:stab_one_step}, $\mbox{(II)}1_{{\cal L}_n}$ is asymptotically negligible.
\end{lemma}
\begin{proof}
First we have 
\begin{align*}
  &\mbox{(II)} = \frac{1}{\sqrt{n-q_n}}\sum_{j=q_n}^{n-1}\frac{m_j(U_{j+1,k_j}-\Qn[U_{k_j}])}{\hat{\sigma}_{nj}{\rm Var}_{\Qn}(U_{k_j})} \Big[\delta_{j+1}X_{j+1} \Big(\frac{1}{\hat{G}_n(X_{j+1})}-\frac{1}{G(X_{j+1})}\Big)\\ 
  &\hspace{3.5cm} + \int_{\mathcal{T}} E_0(U_{j+1, k_j},s,k_j)1(X_{j+1} \geq s)\{d\hat{\Lambda}_{n}(s)-d\Lambda(s)\} \Big].
\end{align*}
By \ref{assump:Survival function}, we can fix an $\tilde{\epsilon}$ such that $0 < \tilde{\epsilon} < G(\tau)$. 
Fix $0 < r < 1$, and let ${\cal G}_r$ be the collection of monotone nonincreasing c\`{a}dl\`{a}g functions $\tilde{G} \colon \mathcal{T} \rightarrow [0,1]$ such that $\tilde{G}(\tau) > \tilde{\epsilon}$, and $\sup_{s\in\mathcal{T}}|\tilde{G}(s)/G(s)-1| \le r$. Note that $G\in {\cal G}_{r} \subset {\cal G}$ that was defined right before \eqref{eq:class_tilde_F}, and ${\rm P}(\hat{G}_n \in {\cal G}_{r}) \to 1$, using the argument involving $\mathcal{R}_{n1}$
in the proof of Lemma \ref{lemma:An_Bn_prob_to_one}.

We first give an upper-bound of $|\mbox{(II)}1_{{\cal L}_n}|$ by using Lemma \ref{lemma:bounded_multiplicative_term}, taking
\begin{align*}
  &f_n(O_{j+1,k_j}) = \delta_{j+1}X_{j+1}
  \Big(\frac{1}{\hat{G}_n(X_{j+1})}-\frac{1}{G(X_{j+1})}\Big)\\
  & \qquad + \int_{\mathcal{T}} E_0(U_{j+1, k_j},s,k_j)1(X_{j+1} \geq s)\{d\hat{\Lambda}_n(s)-d\Lambda(s)\},  
\end{align*}
and showing that $\sup_{o}|f_n(o)1_{{\cal L}_n}| \lesssim \sqrt{\log(n \lor p_n)}/q_n^{1/4}$ below. First, using a Taylor expansion of $z\mapsto \delta x1_{\mathcal{L}_n}/z$ around $z=G(x)$ gives that
\begin{align*} %\label{eq:taylor_expansion_hatGn}
  &\sup_{o}\Big|\delta x\Big(\frac{1}{\hat{G}_n(x)}-\frac{1}{G(x)}\Big)1_{{\cal L}_n}\Big| \leq
  \frac{\tau}{G(\tau)}\bigg[\sum_{r=1}^{\infty}\Big(\sup_{t}\Big|\frac{\hat{G}_n(t)}{G(t)}-1\Big|\Big)^{r}\bigg]1_{{\cal L}_n}\\
  & \quad \leq \frac{\tau}{G(\tau)}\sum_{r=1}^{\infty}\Big(\sqrt{\frac{\log n}{n}}\,\Big)^{r} = \frac{\tau}{G(\tau)}\frac{\sqrt{\log n}}{\sqrt{n}-\sqrt{\log n}} \lesssim \sqrt{\frac{\log n}{n}}, \nonumber
\end{align*}
where the result follows by the occurrence of ${\cal B}_n(\tilde{K}) \supset {\cal L}_n$, $\sqrt{\log n/n} \leq 1$ and $G(\tau) > 0$ in \ref{assump:Survival function}. Along with the above result and \ref{eq:property_estCHF_1} of Lemma \ref{lemma:properties_estCHF}, the triangle inequality gives that with probability tending to one,
\begin{align*}
  &\sup_{o}|f_n(o)1_{{\cal L}_n}|
  \lesssim \sqrt{\frac{\log n}{n}} + \frac{\sqrt{\log(n \lor p_n)}}{{q_n^{1/4}}} \leq 2\frac{\sqrt{\log(n \lor p_n)}}{{q_n^{1/4}}}.
\end{align*}
Therefore along with the above results, Lemma \ref{lemma:bounded_multiplicative_term} implies that
\begin{align*}
  &|\mbox{(II)}1_{{\cal L}_n}| 
  \lesssim \bigg| \frac{1}{\sqrt{n-q_n}} \sum_{j=q_n}^{n-1}\frac{m_j(U_{j+1,k_j}-Q_u[U_{k_j}])}{\sigma_{nj}{\rm Var}_{Q_u}(U_{k_j})} \bigg[\delta_{j+1}X_{j+1}\Big(\frac{1}{\hat{G}_n(X_{j+1})}-\frac{1}{G(X_{j+1})}\Big) \nonumber \\ 
  &\hspace{3cm} + \int_{\mathcal{T}} E_0(U_{j+1, k_j},s,k_j)1(X_{j+1} \ge s)\{d\hat{\Lambda}_n(s)-d\Lambda(s)\}\bigg]\bigg|1_{{\cal L}_n} + o_p(1).
\end{align*}  

According to the definition of ${\cal G}_r$, the events contained in ${\cal L}_n$ and the definition of $\tilde{e}_{nj}$ as in \eqref{eq:random_func_en},
the above display further leads to
\begin{align} \label{eq:upper_bound_(II)_on_Ln_new}
  &|\mbox{(II)}1_{{\cal L}_n}| \le \bigg|\frac{1}{\sqrt{n-q_n}} \sum_{j=q_n}^{n-1}  \frac{m_j(U_{j+1,k_j}-Q_u[U_{k_j}])}{\sigma_{nj}{\rm Var}_{Q_u}(U_{k_j})}\int_{{\cal T}}\tilde{e}_{nj}(U_{j+1, k_j},s,k_j)\{d\hat{\Lambda}_n(s)-d\Lambda(s)\}\bigg| \\
  &\quad + \frac{1}{\sqrt{n-q_n}}\sup_{\tilde{G} \in {\cal G}_r} \bigg|\sum_{j=q_n}^{n-1}\frac{m_j(U_{j+1,k_j}-Q_u[U_{k_j}])}{\sigma_{nj}{\rm Var}_{Q_u}(U_{k_j})} \delta_{j+1}X_{j+1}\Big(\frac{1}{\tilde{G}(X_{j+1})}-\frac{1}{G(X_{j+1})}\Big)\bigg|1_{{\cal L}_n} \nonumber \\
  &\quad + o_p(1), \nonumber
\end{align}
where the first term converges to zero in probability, applying Lemmas \ref{lemma:convergence_sum_widetildeDnj} and \ref{lemma:convergence_sum_widehatDnj}. Therefore it remains to show that the middle term on the right-hand-side converges to zero in probability, using the properties of martingale difference arrays.

Fix $n$. Define a process $\{\tilde{S}_n(\tilde{G}): \tilde{G} \in {\cal G}_r\}$ by
$\tilde{S}_n(\tilde{G}) \equiv \sum_{j=q_n+1}^{n}V_{nj}(\tilde{G})$, where
\begin{align*}
  V_{nj}(\tilde{G})&=\frac{m_{j-1}(U_{j,k_{j-1}}-Q_u[U_{k_{j-1}}])}{\sqrt{n-q_n}\sigma_{n,\,j-1}{\rm Var}_{Q_u}(U_{k_{j-1}})} \delta_{j}X_{j}\Big(\frac{1}{\tilde{G}(X_{j})}-\frac{1}{G(X_{j})}\Big).
\end{align*}
We see that $E|V_{nj}(\tilde{G})| < \infty$ by \ref{assump:Covariates}, \ref{assump:Survival function}
and \ref{assump:Variances}.
Define the filtration ${\cal O}_{nj} = \sigma(O_1,\ldots,O_j,\delta_{j+1},X_{j+1})$ and we have that, for each $\tilde{G}$, $E[V_{n,\,j+1}(\tilde{G})|{\cal O}_{nj}] = 0$.
Therefore $\{(V_{nj}(\tilde{G}),{\cal O}_{nj}):j=q_n+1,\ldots,n,\, \tilde{G} \in {\cal G}_r\}$ is an array of martingale-differences of adapted
processes indexed by ${\cal G}_r$. Note that the class ${\cal G}_r$ has a finite uniform entropy integral \citep[see Lemma~2.6.13 in][]{Vaart1996}.
For all $\tilde{G}, G' \in {\cal G}_r$,
\begin{align*}
  & \tilde{\sigma}_n^2(\tilde{G}, G') \equiv \sum_{j=q_n}^{n-1}E\big[\{V_{n,\,j+1}(\tilde{G})-V_{n,\,j+1}(G')\}^2|{\cal O}_{nj}\big]\\
  & = \frac{1}{(n-q_n)}\sum_{j=q_n}^{n-1}X_{j+1}^2\Big(\frac{1}{\tilde{G}(X_{j+1})}-\frac{1}{G'(X_{j+1})}\Big)^2E\bigg[\frac{(U_{j+1,k_j}-Q_u[U_{k_j}])^2}{\sigma^2_{nj}{\rm Var}^2_{Q_u}(U_{k_j})}\bigg|{\cal O}_{nj}\bigg]\\
  & \le \frac{\bar{K}_0^2\tau^2}{\tilde{\epsilon}^4}\frac{1}{(n-q_n)}\sum_{j=q_n}^{n-1}\big\{\tilde{G}(X_{j+1})-G'(X_{j+1})\big\}^2, %\le \frac{\bar{K}_0^2\tau^2}{\tilde{\epsilon}^4}\sup_{s \in {\cal T}}\big|\tilde{G}(s)-G'(s)\big|^2,
\end{align*}
where the inequality holds because $X \le \tau$; for each $j$, $G(X_{j+1}) \ge G(\tau) > \tilde{\epsilon} > 0$ by \ref{assump:Survival function}; for each $j$ and any $\bar{G} \in {\cal G}_r$, $\bar{G}(X_{j+1}) \ge \bar{G}(\tau) > \tilde{\epsilon} > 0$ and $\tilde{G}, G' \in {\cal G}_r$;
for each $j$, $|U_{j+1,k_j}-Q_u[U_{k_j}]|/[\sigma_{nj}{\rm Var}_{Q_u}(U_{k_j})] \le \bar{K}_0$ almost surely by \ref{assump:Covariates} and \ref{assump:Variances}, for some positive constant $\bar{K}_0$.

Let $\mu_n$ be the empirical distribution of $\{O_{q_n},\ldots,O_{n-1}\}$ normalized by $\sup_{s \in {\cal T}}\big|\hat{G}_n(s)-G(s)\big|$. Following the notation of Theorem 1 of  \cite{Bae2010}, it gives that
for any $\tilde{G}, G' \in {\cal G}_r$, 
%\begin{align*}
%  d^{(2)}_{\mu_n}(\tilde{G}, G')^2 \equiv \frac{\sup_{s \in {\cal T}}\big|\tilde{G}(s)-G'(s)\big|^2}{\sup_{s \in {\cal T}}\big|\hat{G}_n(s)-G(s)\big|};
%\end{align*}
\begin{align*}
 d^{(2)}_{\mu_n}(\tilde{G}, G')^2 \equiv \frac{1}{(n-q_n)}\sum_{j=q_n}^{n-1}\frac{\big\{\tilde{G}(X_{j+1})-G'(X_{j+1})\big\}^2}{\sup_{s \in {\cal T}}\big|\hat{G}_n(s)-G(s)\big|}.    
\end{align*}
Therefore, checking condition (3) of Theorem 1 of \cite{Bae2010} in our case, we have that for any positive constant $L$,
\begin{align*}\color{red}
  &{\rm P}\bigg(\sup_{\tilde{G},G' \in {\cal G}_r}\frac{\tilde{\sigma}_n^2(\tilde{G},G')}{d^{(2)}_{\mu_n}(\tilde{G}, G')^2} \ge L\bigg) \le {\rm P}\bigg(\mathcal{L}_n,\sup_{\tilde{G},G' \in {\cal G}_r}\frac{\tilde{\sigma}_n^2(\tilde{G},G')}{d^{(2)}_{\mu_n}(\tilde{G}, G')^2} \ge L\bigg) + {\rm P}\bigg(\mathcal{L}_n^c\bigg) \\
  &\quad \le \frac{1}{L}E\bigg[\sup_{\tilde{G},G' \in {\cal G}_r}\frac{\tilde{\sigma}_n^2(\tilde{G},G')}{d^{(2)}_{\mu_n}(\tilde{G}, G')^2}1_{\mathcal{L}_n}\bigg] + {\rm P}\bigg(\mathcal{L}_n^c\bigg)\\
  &\quad \le \frac{\bar{K}_0^2\tau^2}{L\tilde{\epsilon}^4}E\bigg[\sup_{s \in {\cal T}}\Big|\hat{G}_n(s)-G(s)\Big|1_{{\cal L}_n}\bigg] + {\rm P}\bigg(\mathcal{L}_n^c\bigg) \\
  &\quad \le \frac{\bar{K}_0^2\tau^2}{L\tilde{\epsilon}^4}E\bigg[\sup_{s \in {\cal T}}\bigg|\frac{\hat{G}_n(s)}{G(s)}-1\bigg|1_{{\cal L}_n}\bigg] + {\rm P}\bigg(\mathcal{L}_n^c\bigg) \\
  &\quad \le \frac{\bar{K}_0^2\tau^2}{L\tilde{\epsilon}^4}\sqrt{\frac{\log n}{n}}
  + {\rm P}\bigg(\mathcal{L}_n^c\bigg) \to 0
  \mbox{ as } n \to \infty.
\end{align*}
The Lindeberg condition holds trivially in our case: note that $|V_{n,\,j+1}(\tilde{G})| \lesssim 1/\sqrt{n-q_n}$, so
for any fixed $\epsilon > 0$ and for all $n$ sufficiently large we have
\begin{align*}
 &\frac{1}{\epsilon}\sum_{j=q_n}^{n-1}E[V^2_{n,\,j+1}(\tilde{G})1(V_{n,\,j+1}(\tilde{G})>\epsilon)]=0.
\end{align*}
Now appealing to Theorem 1 of \cite{Bae2010}, for given $\gamma > 0$ and $\epsilon > 0$, there exists an $\eta > 0$ for which
\begin{align} \label{eq:limsup_prob}
  \limsup_{n \to \infty}{\rm P}\bigg(\sup_{\tilde{G} \in {\cal G}_r \,:\,  d^{(2)}_{\mu_n}(\tilde{G}, G) \le \eta}|\tilde{S}_n(\tilde{G})| > 5\gamma \bigg) \le 3\epsilon.
\end{align}
Note by the arguments in the proof of Lemma~\ref{lemma:An_Bn_prob_to_one},
${\rm P}(\hat{G}_n \in {\cal G}_r) \to 1$, and
\begin{align*}
  & d^{(2)}_{\mu_n}(\hat{G}_n, G) \le \Big\{\sup_{s \in {\cal T}}\big|\hat{G}_n(s)-G(s)\big|\Big\}^{1/2} \le \bigg\{\sup_{s \in {\cal T}}\bigg|\frac{\hat{G}_n(s)}{G(s)}-1\bigg|\bigg\}^{1/2} \to 0
\end{align*}
with probability tending to one.
Hence, $\hat{G}_n\in \{\tilde{G}\in\mathcal{G}_r : d^{(2)}_{\mu_n}(\tilde{G}, G) \le \eta\}$ with probability tending to one, and so
\begin{align*}
  &\limsup_{n \to \infty}{\rm P}\Big(|\tilde{S}_n(\hat{G}_n)| > 5\gamma \Big) \le \limsup_{n \to \infty}{\rm P}\bigg(\sup_{\tilde{G} \in {\cal G}_{r}\,: \,d^{(2)}_{\mu_n}(\tilde{G}, G) \le \eta}|\tilde{S}_n(\tilde{G})| > 5\gamma \bigg) 
  \le 3\epsilon.
\end{align*}
As $\epsilon>0$ was arbitrary, $\limsup_{n \to \infty}{\rm P}\Big(|\tilde{S}_n(\hat{G}_n)| > 5\gamma \Big)=0$ and, as $\gamma>0$ was arbitrary, this shows that $\tilde{S}_n(\hat{G}_n)=o_p(1)$. The argument following \eqref{eq:upper_bound_(II)_on_Ln_new} then shows that $\mbox{(II)}1_{{\cal L}_n}=o_p(1)$.
\end{proof}

To prove the next lemma, we need to develop a decomposition involving three types of martingale differences. The filtrations for these martingale differences are
\begin{align} \label{eq:filtrations}
  %{\cal O}_{nj} &\equiv \sigma\big(\{O_1,\ldots,O_{j-1}\},U_{j,k_{j-1}},\tilde{Y}_{j}\big);\\
  {\cal O}^{\dagger}_{nj} &\equiv \sigma \big(O_1,\ldots,O_j,U_{j+1,k_j}\big);\\
  {\cal O}^{*}_{nj} &\equiv \sigma \big(O_1,\ldots,O_{j-1},k_j, m_j, E\big[\tilde{Y}\big|U_{j+1,k_j}\big],{\rm Var}_{\Qj}(U_{k_j}), U_{j,k_{j-1}}\big); \nonumber \\
  \tilde{{\cal O}}_{nj} &\equiv \sigma(O_1,\ldots,O_j, \bs{U}_{j+1}), \nonumber
\end{align}
$j=q_n,\ldots,n-1$. 
Define
\begin{align} \label{eq:def_hnj}
  h_{nj}(u) \equiv & \bigg[\frac{(u-\Qj[U_{k_j}])}{{\rm Var}_{\Qj}(U_{k_j})}-\frac{(u-Q_u[U_{k_j}])}{{\rm Var}_{Q_u}(U_{k_j})}\bigg]1\Big(\max_k\big|(\Qj-Q_u)[U_k]\big| \lesssim K_{nj}, \nonumber \\ 
  &\quad \max_{k}\big|{\rm Var}^{-1}_{\Qj}(U_{k})-{\rm Var}^{-1}_{Q_u}(U_{k})\big| \lesssim K_{nj}\Big),
\end{align}
where in $\lesssim$ the implicit constants are independent of $(j,n)$. 
The martingale differences to be used in the proof are then defined by 
\begin{align} \label{eq:decomp_representations}
  H^{\dagger}_{nj} &\equiv \frac{1}{\sqrt{n-q_n}}\frac{m_{j-1}}{\sigma_{n,\,j-1}}h_{n,\,j-1}(U_{j,k_{j-1}})\big\{\tilde Y_{j}-E\big[\tilde{Y}\big|U_{j,k_{j-1}}\big]\big\},\\
  H^*_{nj} &\equiv \frac{1}{\sqrt{n-q_n}}\frac{m_{j-1}}{\sigma_{n,\,j-1}}h_{n,\,j-1}(U_{j,k_{j-1}})E\big[\tilde{Y}\big|U_{j,k_{j-1}}\big], \mbox{ and } \nonumber \\
  \tilde{H}_{nj} &\equiv \frac{1}{\sqrt{n-q_n}}\frac{m_{j-1}}{\sigma_{n,\,j-1}}h_{n,\,j-1}(U_{j,k_{j-1}})
  \int_{\mathcal{T}}E_0(U_{j, k_{j-1}},s,k_{j-1})dM_{j}(s). \nonumber
\end{align}

Note that $(H^{\dagger}_{nj}, {\cal O}^{\dagger}_{nj})$, $(H^*_{nj}, {\cal O}^*_{nj})$, $(\tilde{H}_{nj}, \tilde{{\cal O}}_{nj})$, $j= q_n+1,\ldots,n$, are martingale difference sequences.
In particular,
\begin{align} \label{eq:conditional_mean_zero.1}
  &\sqrt{n-q_n}\,E\big[H^{\dagger}_{n,\,j+1} \big| {\cal O}^{\dagger}_{nj}\big] = E\Big[\frac{m_{j}}{\sigma_{n,j}}h_{n,j}(U_{j+1,k_j})\Big[\tilde Y_{j+1}-E\big[\tilde{Y}\big|U_{j+1,k_j}\big]\Big]\Big|{\cal O}^{\dagger}_{nj}\Big]\\
  & = \frac{m_{j}}{\sigma_{nj}}h_{n,j}(U_{j+1,k_j})\Big[E\big[\tilde{Y}_{j+1}\big|U_{j+1,k_j}\big]-E\big[\tilde{Y}\big|U_{j+1,k_j}\big]\Big] = 0, \nonumber
\end{align}
and
\begin{align} \label{eq:conditional_mean_zero.2}
  &\sqrt{n-q_n}E\big[H^*_{n,\,j+1} \big| {\cal O}^*_{nj}\big] = E\Big[\frac{m_{j}}{\sigma_{n,j}}h_{n,j}(U_{j+1,k_j})E\big[\tilde{Y}\big|U_{j+1,k_j}\big]\Big|{\cal O}^*_{nj}\Big]\\
  & = \frac{m_jE\big[\tilde{Y}\big|U_{j+1,k_j}\big]}{\sigma_{nj}}
  \bigg\{\frac{E[(Q_u-\Qj)[U_{k_j}]|{\cal O}^*_{nj}]}{{\rm Var}_{\Qj}(U_{k_j})} \nonumber \\
  &\quad + E\big[U_{j+1,k_j}-Q_u[U_{k_j}]\big|{\cal O}^*_{nj}\big]\Big[\frac{1}{{\rm Var}_{\Qj}(U_{k_j})}-\frac{1}{{\rm Var}_{Q_u}(U_{k_j})}\Big]\bigg\} \nonumber \\
  & = 0, \nonumber
\end{align}
where the first step of \eqref{eq:conditional_mean_zero.2} holds by
\begin{align} \label{eq:decomp_for_(III)}
  &\frac{(U_{j+1,k_j}-\Qj[U_{k_j}])}{{\rm Var}_{\Qj}(U_{k_j})}-\frac{(U_{j+1,k_j}-Q_u[U_{k_j}])}{{\rm Var}_{Q_u}(U_{k_j})}\\
  &\quad = \frac{(Q_u-\Qj)[U_{k_j}]}{{\rm Var}_{\Qj}(U_{k_j})} + (U_{j+1,k_j}-Q_u[U_{k_j}])\Big[\frac{1}{{\rm Var}_{\Qj}(U_{k_j})}-\frac{1}{{\rm Var}_{Q_u}(U_{k_j})}\Big]. \nonumber
\end{align}
In addition,
\begin{align} \label{eq:conditional_mean_zero.3} 
  & \sqrt{n-q_n}E\big[\tilde{H}_{n,\,j+1}\big|\tilde{{\cal O}}_{nj}\big]
  = \frac{m_{j}}{\sigma_{nj}}h_{n{j}}(U_{j+1,k_{j}})\int_{\mathcal{T}}E_0(U_{j+1,k_{j}},s,k_{j})E\big[dM_{j+1}(s)\big|\tilde{{\cal O}}_{nj}\big] \\
  & = \frac{m_{j}}{\sigma_{nj}}h_{n{j}}(U_{j+1,k_{j}}) \int_{\mathcal{T}}E_0(U_{j+1, k_j},s,k_j)E[1(T_{j+1} \ge s)|\bs{U}_{j+1}] \nonumber \\
  &\hspace{6cm} \times E\left[1(C_{j+1}\in ds)-1(C_{j+1}\ge s)d\Lambda(s)\right] \nonumber \\
  & = 0, \nonumber 
\end{align}
where the first step holds by the independent censoring assumption, and the second step follows from the definition $d\Lambda(s) = {\rm P}(C \in ds)/{\rm P}(C \ge s)$.

\begin{lemma} \label{lemma:asymptotic_negligibility_(III)}
Under the conditions of Theorem \ref{Thm:stab_one_step}, $\mbox{(III)}1_{{\cal L}_n}$ is asymptotically negligible.
\end{lemma}
\begin{proof}
Note that $dM_{j+1}(s) = 1(X_{j+1} \in ds, \delta_{j+1}=0)-1(X_{j+1} \geq s)d\Lambda(s)$; we re-express  (III) (from \eqref{eq:decomposition_rootn_Sn_star} in the main text) as
\begin{align*} %\label{eq:decomposition_(III)} 
 &\frac{1}{\sqrt{n-q_n}}\sum_{j=q_n}^{n-1}m_j\Big[\frac{1}{\hat{\sigma}_{nj}}-\frac{1}{\sigma_{nj}} + \frac{1}{\sigma_{nj}}\Big] \Big[\frac{(U_{j+1,k_j}-\Qj[U_{k_j}])}{{\rm Var}_{\Qj}(U_{k_j})}-\frac{(U_{j+1,k_j}-Q_u[U_{k_j}])}{{\rm Var}_{Q_u}(U_{k_j})}\Big] \\
 &\hspace{6.2cm} \times \Big[\tilde Y_{j+1} -\int_{\mathcal{T}}E_0(U_{j+1, k_j},s,k_j)dM_{j+1}(s)
 \Big]. \nonumber
\end{align*} 

Along with the fact that $U_k$ is uniformly bounded in \ref{assump:Covariates}, that $\sigma_{nj}$ is uniformly bounded away from zero in \ref{assump:Variances},
\ref{eq:event_Cn_2} and \ref{eq:event_Cn_4} of Lemma \ref{lemma:An_Bn_Cn_prob_to_one}, and
\ref{eq:event_En_2} of Lemma \ref{lemma:An_Bn_Cn_Dn_En_prob_to_one},
the above display further gives that
\begin{align*}
 &\bigg|\frac{1}{\sqrt{n-q_n}}\sum_{j=q_n}^{n-1}m_j\Big[\frac{1}{\hat{\sigma}_{nj}}-\frac{1}{\sigma_{nj}}\Big] \Big[\frac{(U_{j+1,k_j}-\Qj[U_{k_j}])}{{\rm Var}_{\Qj}(U_{k_j})}-\frac{(U_{j+1,k_j}-Q_u[U_{k_j}])}{{\rm Var}_{Q_u}(U_{k_j})}\Big]\\
 &\hspace{5cm} \times \Big[\tilde Y_{j+1} -\int_{\mathcal{T}}E_0(U_{j+1, k_j},s,k_j)dM_{j+1}(s)
 \Big]1_{{\cal L}_n}\bigg| \nonumber\\
 &\quad \lesssim \Big[\frac{\tau}{G(\tau)}+(1+\Lambda(\tau))\sup_{(k,s,u)}\big|E_0(u,s,k)\big|\Big]\frac{\log(n \lor p_n)}{\sqrt{n-q_n}}\sum_{j=q_n}^{n-1}\frac{1}{j} \nonumber\\
 &\quad \le \Big[\frac{\tau}{G(\tau)}+(1+\Lambda(\tau))\sup_{(k,s,u)}\big|E_0(u,s,k)\big|\Big]\frac{\log(n \lor p_n)}{\sqrt{n-q_n}}\log\Big(\frac{n}{q_n}\Big)
 \to 0, \nonumber
\end{align*} 
where the convergence to zero follows by $G(\tau) > 0$ in \ref{assump:Survival function}, that $\sup_{(k,s,u)}\big|E_0(u,s,k)\big|$ is bounded in \ref{assump:Conditional_mean_E0}, and the conditions: $n/q_n=O(1)$ and $q_n^{1/4}/\log(n \lor p_n) \to \infty$.

Combining all the above results, we have that
\begin{align*}
 |\mbox{(III)}1_{{\cal L}_n}| \le &  \bigg|\frac{1}{\sqrt{n-q_n}}\sum_{j=q_n}^{n-1}\frac{m_j}{\sigma_{nj}} h_{nj}(U_{j+1,k_j})\Big[\tilde Y_{j+1} -
 \int_{\mathcal{T}}E_0(U_{j+1, k_j},s,k_j)dM_{j+1}(s)\Big]\bigg|.
\end{align*}
Therefore, we have that
\begin{align}\label{eq:upper_bound_(III)_on_Ln}
  |\mbox{(III)}1_{{\cal L}_n}|
  \le \bigg|\sum_{j=q_n}^{n-1}H^{\dagger}_{n,\,j+1}\bigg| + \bigg|\sum_{j=q_n}^{n-1}H^*_{n,\,j+1}\bigg|
  + \bigg|\sum_{j=q_n}^{n-1}\tilde{H}_{n,\,j+1}\bigg|,
\end{align}
where $H^{\dagger}_{nj}$, $H^*_{nj}$ and $\tilde{H}_{nj}$ are defined in \eqref{eq:decomp_representations}.
To complete the proof, it suffices to show that the three terms on the right-hand-side of \eqref{eq:upper_bound_(III)_on_Ln} are $o_p(1)$.

First note  that for all $n$,
\begin{align}\label{eq:upper_bound_hnj}
  \max_{j \in \{q_n,\ldots,n-1\}}|h_{nj}(U_{j+1,k_{j}})| \lesssim \frac{\sqrt{\log(n \lor p_n)}}{\sqrt{q_n}} \; \mbox{ a.s.}
\end{align}
from the decomposition in \eqref{eq:decomp_for_(III)}, the definition of $h_{nj}$ in \eqref{eq:def_hnj}, and \ref{assump:Covariates}; here the implicit constant in $\lesssim$ does not depend on $n$.
Also, note that $\tilde{Y}_{j+1}-E[\tilde{Y}|U_{j+1,k_j}]$ is bounded almost surely using \ref{assump:Covariates} and \ref{assump:Survival function}. Then, with $\sigma_{nj}$ uniformly bounded away from zero in \ref{assump:Variances}, \eqref{eq:upper_bound_hnj} further implies that
\begin{align*}
  \max_{j}|H^{\dagger}_{nj}| \lesssim \sqrt{\frac{\log(n \lor p_n)}{q_n(n-q_n)}} \equiv B_n,
\end{align*}
%\begin{align*}
  %&|H^{\dagger}_{n,\,j+1}| \lesssim \sqrt{\frac{\log(n \lor p_n)}{q_n(n-q_n)}} < q_n^{-1/4}\sqrt{\frac{\log(n \lor p_n)}{(n-q_n)}} \equiv B_n;
  %&\sum_{j=q_n}^{n-1}E[(H^{\dagger}_{n,\,j+1})^2|{\cal O}^{\dagger}_{nj}] \lesssim \frac{\log(n \lor p_n)}{(n-q_n)}\sum_{j=q_n}^{n-1}\frac{1}{j} < \frac{\log(n \lor p_n)}{(n-q_n)}\log\bigg(\frac{n}{q_n}\bigg) \equiv V_n,
%\end{align*}
where in $\lesssim$ the implicit constant is also independent of $n$. Then for $\varepsilon > 0$,
\begin{align*}
  & {\rm P}\bigg(\,\bigg|\sum_{j=q_n}^{n-1}H^{\dagger}_{n,\,j+1}\bigg| \ge \varepsilon \bigg) \le \varepsilon^{-2}E\bigg[\bigg(\sum_{j=q_n}^{n-1}H^{\dagger}_{n,\,j+1}\bigg)^2\,\bigg] \\
  &= \varepsilon^{-2}\bigg(\sum_{j=q_n}^{n-1}E\big[\big(H^{\dagger}_{n,\,j+1}\big)^2\big] + 2\sum_{q_n\le i< j\le n-1}E\Big[H^{\dagger}_{n,\,i+1}E\Big[H^{\dagger}_{n,\,j+1} \Big| \mathcal{O}_{nj}\Big]\Big]\bigg) \\
  &= \varepsilon^{-2}\sum_{j=q_n}^{n-1}E\big[ \big(H^{\dagger}_{n,\,j+1}\big)^2 \big]
  \lesssim \varepsilon^{-2}(n-q_n)B_n^2 \to 0.
\end{align*}
This shows the first term on the right-hand-side of \eqref{eq:upper_bound_(III)_on_Ln} converges to zero in probability. The second and the last terms on the right-hand-side of \eqref{eq:upper_bound_(III)_on_Ln} can be handled, using similar arguments.
\end{proof}

\begin{lemma} \label{lemma:asymptotic_negligibility_(V)}
Under the conditions of Theorem \ref{Thm:stab_one_step}, $\mbox{(V)}1_{{\cal L}_n}$ is asymptotically negligible.
\end{lemma}
\begin{proof}
It is trivial to see that $\mbox{(V)}1_{{\cal L}_n}=o_p(1)$ under the null. To verify it under the alternative, we first have
\begin{align*}
  &|\mbox{(V)}1_{{\cal L}_n}| = \bigg|\frac{1}{\sqrt{n-q_n}}\sum_{j=q_n}^{n-1}\frac{m_j}{\hat{\sigma}_{nj}}1_{{\cal L}_n}[\Psi_{k_j}(P) - \Psi(P)]\bigg|\\
  &\quad \leq \max_{j \in \{q_n,\ldots,n-1\}}\left|\frac{\sigma_{nj}}{\hat{\sigma}_{nj}}1_{{\cal L}_n}\right| 
  \frac{1}{\sqrt{n-q_n}}\sum_{j=q_n}^{n-1}\frac{1}{\sigma_{nj}}|\Psi_{k_j}(P) - \Psi(P)|1_{{\cal L}_n}.
\end{align*}
Because $\sigma_{nj}$ is assumed to be bounded away from zero in \ref{assump:Variances} and $\max_{j}\left|\sigma_{nj}1_{{\cal L}_n}/\hat{\sigma}_{nj}\right|$ is bounded above using \ref{eq:event_En_2} of Lemma \ref{lemma:An_Bn_Cn_Dn_En_prob_to_one}, it suffices to show that 
\begin{align*}
  \frac{1}{\sqrt{n-q_n}}\sum_{j=q_n}^{n-1}|\Psi_{k_j}(P) - \Psi(P)|1_{{\cal L}_n} = o_p(1).
\end{align*}
Recall that 
\begin{align*}
  k_j = {\rm arg}\max_{k \in \mathcal{K}_n}\left|\frac{{\rm 
  Cov}_{\mathbb{P}_{j}}(U_k, \delta X/\hat G_j(X))}{{\rm Var}_{\mathbb{P}_{j}}(U_k)}\right| \:\: \mbox{and} \:\:
  m_j = \mbox{sgn}\left[\frac{{\rm  
  Cov}_{\mathbb{P}_{j}}(U_k, \delta X/\hat G_j(X))}{{\rm Var}_{\mathbb{P}_{j}}(U_k)}\right],
\end{align*}
and let 
\begin{align*}
  k_0 = {\rm arg}\max_{k \in \mathcal{K}_n}\left|\frac{{\rm Cov}(U_{k}, T)}{{\rm Var}(U_{k})}\right| \:\: \mbox{and} \:\:
  m_0 = \mbox{sgn}\left[\frac{{\rm Cov}(U_{k_0}, T)}{{\rm Var}(U_{k_0})}\right].
\end{align*}
Because $m{\rm Cov}_{\mathbb{P}_{j}}(U_k, \delta X/\hat G_j(X))/{\rm Var}_{\mathbb{P}_{j}}(U_k)$ is maximized at $(k_j, m_j)$, we observe that
\begin{align} \label{eq:bounded_condition_for_(V)}
  0 &\geq m_0\frac{{\rm 
  Cov}_{\mathbb{P}_{j}}(U_{k_0}, \delta X/\hat G_j(X))}{{\rm Var}_{\mathbb{P}_{j}}(U_{k_0})}-m_j\frac{{\rm 
  Cov}_{\mathbb{P}_{j}}(U_{k_j}, \delta X/\hat G_j(X))}{{\rm Var}_{\mathbb{P}_{j}}(U_{k_j})}\\
  & =\bigg[m_0\frac{{\rm Cov}(U_{k_0}, T)}{{\rm Var}(U_{k_0})}-m_j\frac{{\rm Cov}(U_{k_j}, T)}{{\rm Var}(U_{k_j})}\bigg] + m_0\bigg[\frac{{\rm 
  Cov}_{\mathbb{P}_{j}}(U_{k_0}, \delta X/\hat G_j(X))}{{\rm Var}_{\mathbb{P}_{j}}(U_{k_0})} \nonumber\\
  &\hspace{0.8cm} - \frac{{\rm Cov}(U_{k_0}, T)}{{\rm Var}(U_{k_0})}\bigg]
  - m_j\bigg[\frac{{\rm 
  Cov}_{\mathbb{P}_{j}}(U_{k_j}, \delta X/\hat G_j(X))}{{\rm Var}_{\mathbb{P}_{j}}(U_{k_j})}-\frac{{\rm Cov}(U_{k_j}, T)}{{\rm Var}(U_{k_j})}\bigg] \nonumber\\
  &\geq \Psi(P) - \Psi_{k_j}(P) - 2\max_{k \in \mathcal{K}_n}|\Psi_{\hat{G}_j,k}(\mathbb{P}_j)-\Psi_{G,k}(P)|,\nonumber
\end{align}
where $\Psi_{k}(P)=\Psi_{G,k}(P) \equiv {\rm Cov}(U_k, \delta X/G(X))/{\rm Var}(U_{k})$ by ${\rm Cov}(U_k, T)={\rm Cov}(U_k, \delta X/G(X))$.
Moreover,
\begin{align} \label{eq:decomp_Psi_hatG}
  &|\Psi_{\hat{G}_j,k}(\mathbb{P}_j)-\Psi_{G,k}(P)|1_{{\cal L}_n}
  = \left|\frac{{\rm Cov}_{\mathbb{P}_{j}}(U_{k}, \delta X/\hat G_j(X))}{{\rm Var}_{\mathbb{P}_{j}}(U_{k})} -
  \frac{{\rm Cov}(U_{k}, \delta X/G(X))}{{\rm Var}(U_{k})} \right|1_{{\cal L}_n} \\
  & = \frac{1}{{\rm Var}_{\mathbb{P}_j}(U_{k})}1_{{\cal L}_n}\Big| 
  {\rm Cov}_{\mathbb{P}_j}(U_{k}, \delta X/\hat G_j(X))
  - {\rm Cov}_{\mathbb{P}_j}(U_{k}, \delta X/G(X)) \nonumber \\
  & \hspace{4cm} 
  + {\rm Cov}_{\mathbb{P}_j}(U_{k}, \delta X/G(X))
  - {\rm Cov}(U_{k}, \delta X/G(X))
  \Big| \nonumber \\
  & \quad + 1_{{\cal L}_n}\Big|{\rm Cov}(U_{k}, \delta X/G(X)) \Big|\left|\frac{1}{{\rm Var}_{\mathbb{P}_j}(U_{k})}-\frac{1}{{\rm Var}(U_{k})}\right| \nonumber \\
  & \leq 1_{{\cal L}_n}\frac{| 
  {\rm Cov}_{\mathbb{P}_j}(U_{k}, \delta X/\hat G_j(X)-\delta X/G(X))|}{{\rm Var}_{\mathbb{P}_j}(U_k)} \nonumber \\
  &\quad + 1_{{\cal L}_n}\frac{| 
  {\rm Cov}_{\mathbb{P}_j}(U_{k}, \delta X/G(X))
  - {\rm Cov}(U_{k}, \delta X/G(X))|}{{\rm Var}_{\mathbb{P}_j}(U_k)} \nonumber \\
  &\quad + 1_{{\cal L}_n}|{\rm Cov}(U_{k}, \delta X/G(X))|\left|\frac{1}{{\rm Var}_{\mathbb{P}_j}(U_{k})}-\frac{1}{{\rm Var}(U_{k})}\right|,\nonumber
\end{align}
where the second equality holds by using the identity $a_nb_n-ab=(a_n-a)b_n+(b_n-b)a$, and the ensuing step follows by the triangle inequality.

Below we further tackle each term in the upper bound of $|\Psi_{\hat{G}_j,k}(\mathbb{P}_j)-\Psi_{G,k}(P)|1_{{\cal L}_n}$
from \eqref{eq:decomp_Psi_hatG}. To address the first term,
\begin{align} \label{eq:first_term_decomp}
  &1_{{\cal L}_n}\frac{|{\rm Cov}_{\mathbb{P}_j}(U_{k}, \delta X/\hat G_j(X)-\delta X/G(X))|}{{\rm Var}_{\mathbb{P}_j}(U_k)} \\
  &\quad \leq 1_{{\cal L}_n}\frac{1}{{\rm Var}_{\mathbb{P}_j}(U_k)}
  \bigg[\mathbb{P}_j\bigg|U_k\tilde{Y}\sum_{r=1}^{\infty}\Big(\sup_{t}\Big|\frac{\hat{G}_j(t)}{G(t)}-1\big|\,\Big)^{r}\bigg| \nonumber \\
  &\hspace{3.5cm} + 
  \mathbb{P}_j|U_k|\mathbb{P}_j\bigg|\tilde{Y}\sum_{r=1}^{\infty}\Big(\sup_{t}\Big|\frac{\hat{G}_j(t)}{G(t)}-1\Big|\,\Big)^{r} \bigg|\;\bigg] \nonumber \\
  &\quad \leq \frac{1}{{\rm Var}_{\mathbb{P}_j}(U_k)}
  \bigg[\mathbb{P}_j\bigg|U_k\tilde{Y}\sum_{r=1}^{\infty}\bigg(\sqrt{\frac{\log n}{j}}\;\bigg)^{r}\bigg| + 
  \mathbb{P}_j|U_k|\mathbb{P}_j\bigg|\tilde{Y}\sum_{r=1}^{\infty}\bigg(\sqrt{\frac{\log n}{j}}\;\bigg)^{r}\bigg|\;\bigg] \nonumber \\
  &\quad \leq \frac{1}{{\rm Var}_{\mathbb{P}_j}(U_k)}\bigg[\sum_{r=1}^{\infty}
  \bigg(\sqrt{\frac{\log n}{q_n}}\;\bigg)^r\,\bigg]\big[\mathbb{P}_j|U_k\tilde{Y}| + \mathbb{P}_j|U_k|\mathbb{P}_j|\tilde{Y}|\;\big] \nonumber \\
  &\quad \lesssim \bigg(\frac{\sqrt{\log n}}{\sqrt{q_n}-\sqrt{\log n}}\bigg)
  \max_{j,k}\bigg[\mathbb{P}_j\Big|U_k\frac{\delta X}{G(X)}\Big|
  + \mathbb{P}_j|U_k|\mathbb{P}_j\Big|\frac{\delta X}{G(X)}\Big| \bigg] \to 0 \mbox{ a.s.}, \nonumber 
\end{align}
where the first inequality holds by the triangle inequality and Taylor expansion with respect to $\hat{G}_j$ around $G$; the second inequality results from $\sup_{t}\big|\hat{G}_j(t)/G(t)-1\big| \leq \sqrt{\log n /j}$ given by the occurrence of ${\cal L}_n$ and
Lemma \ref{lemma:Hn_prob_to_one}; the third inequality holds by
$j \ge q_n$; the last inequality follows from the fact that ${\rm Var}_{\mathbb{P}_j}(U_k)={\rm Var}_{\mathbb{Q}_j}(U_k)$ that we have showed bounded away from zero in \ref{eq:event_Cn_2} of Lemma \ref{lemma:An_Bn_Cn_prob_to_one}, and the final convergence to zero results from the uniform boundedness of $\mathbb{P}_j|U_k\delta X/G(X)|$, $\mathbb{P}_j|U_k|$ and $\mathbb{P}_j|\delta X/G(X)|$ almost surely, which is implied by \ref{assump:Covariates}, $X \le \tau$, $G(\tau)>0$ in \ref{assump:Survival function} and the condition $q_n^{1/4}/\log(n \lor p_n) \to \infty$.
This gives the first term on the right-hand-side of \eqref{eq:decomp_Psi_hatG} is $o_p(1)$ uniformly in $(j, k)$.

To tackle the second term on the right-hand-side of \eqref{eq:decomp_Psi_hatG}, again using the identity $a_nb_n-ab=(a_n-a)b_n+(b_n-b)a$ and the triangle inequality gives that for some positive finite $(j,n)$-independent constant $\zeta'_1$,
\begin{align} \label{eq:second_term_decomp}
  &1_{{\cal L}_n}\frac{|{\rm Cov}_{\mathbb{P}_j}(U_{k}, \delta X/G(X))
  - {\rm Cov}(U_{k}, \delta X/G(X))|}{{\rm Var}_{\mathbb{P}_j}(U_k)}\\
  &\leq
  \frac{1_{{\cal L}_n}}{\min_{(j,k)}{\rm Var}_{\mathbb{P}_j}(U_k)}
  \max_{k}\Big\{\big|(\mathbb{P}_j-P)[U_{k}\delta X/G(X)]\big| +
  \big|(\mathbb{Q}_j-Q_u)[U_{k}]P[\delta X/G(X)]\big| \nonumber\\
  &\hspace{4.5cm} + \big|(\mathbb{P}_j-P)[\delta X/G(X)]\mathbb{Q}_j[U_{k}]\big| \Big\} \nonumber \\
  &\le \frac{\zeta}{\tilde{\varepsilon}}\Big[1 + \big|P[\delta X/G(X)]\big| + \max_{(j,k)}\big\{\big|\mathbb{Q}_j[U_{k}]\big|\big\}\Big] K_{nj} \equiv \zeta'_1 K_{nj}, \nonumber
\end{align}
where the second inequality holds given the occurrence of 
${\cal A}_n \supset {\cal L}_n$ for a sufficiently large constant $\zeta$ that does not depend on $(j,n)$, and that
$\min_{(j,k)}{\rm Var}_{\mathbb{P}_j}(U_k)=\min_{(j,k)}{\rm Var}_{\mathbb{Q}_j}(U_k)$ is bounded away from zero by \ref{eq:event_Cn_2} of Lemma \ref{lemma:An_Bn_Cn_prob_to_one}, so that $\min_{(j,k)}{\rm Var}_{\mathbb{P}_j}(U_k) > \tilde{\varepsilon}$ for some sufficiently small positive constant $\tilde{\varepsilon}$ that is independent of $(j,n)$.
Similarly for some positive $(j,n)$-independent constant $\zeta'_2$, the third term on the right-hand-side of \eqref{eq:decomp_Psi_hatG} is 
\begin{align} \label{eq:third_term_decomp}
  &1_{{\cal L}_n}\big|{\rm Cov}(U_{k}, \delta X/G(X))\big|\left|\frac{1}{{\rm Var}_{\mathbb{P}_j}(U_{k})}-\frac{1}{{\rm Var}(U_{k})}\right|\\
  &\quad = 1_{{\cal L}_n}\big|{\rm Cov}(U_{k}, \delta X/G(X))\big|\frac{|{\rm Var}_{\mathbb{Q}_j}(U_{k})-{\rm Var}_{Q_u}(U_{k})|}{{\rm Var}_{Q_u}(U_{k}){\rm Var}_{\mathbb{Q}_j}(U_{k})} \nonumber \\
  &\quad \leq 1_{{\cal L}_n}\max_{k}\big|{\rm Cov}(U_{k}, \delta X/G(X))\big|\frac{\max_{k}|{\rm Var}_{\mathbb{Q}_j}(U_{k})-{\rm Var}_{Q_u}(U_{k})|}{\min_k{\rm Var}_{Q_u}(U_{k})\min_{(j,k)}{\rm Var}_{\mathbb{Q}_j}(U_{k})} \nonumber\\
  &\quad \le \zeta'_2K_{nj}, \nonumber
\end{align}
where the last step holds by the presence of
$1_{{\cal L}_n}$ in which $\min_{(j,k)}{\rm Var}_{\mathbb{Q}_j}(U_{k})$ is bounded away from zero by \ref{eq:event_Cn_2} of Lemma \ref{lemma:An_Bn_Cn_prob_to_one}, along with that $\max_{k}\big|{\rm Cov}(U_{k}, \delta X/G(X))\big|$ is assumed to be bounded using \ref{assump:Covariates} and \ref{assump:Survival function}, and
$\min_k{\rm Var}_{Q_u}(U_{k})$ is assumed to be bounded away from zero in \ref{assump:Variances}.

Let $\zeta'=\zeta'_1+\zeta'_2$.
Collecting  the above  results in \eqref{eq:decomp_Psi_hatG}--\eqref{eq:third_term_decomp} gives that
$|\Psi_{\hat{G}_j,k}(\mathbb{P}_j)-\Psi_{G,k}(P)|$ is bounded above by $\zeta'K_{nj}$ and then $\zeta'\sqrt{\log(n \lor p_n)/q_n}$ for $j \ge q_n$.
Inserting this result back into \eqref{eq:bounded_condition_for_(V)} leads to
\begin{align*}
  |\Psi_{k_j}(P)-\Psi(P)|1_{{\cal L}_n} \leq 2\zeta' K_{nj} \le 2\zeta'K_{nq_n}.
\end{align*}
Together with $\Psi(P) \equiv \max_{k}|\Psi_{k}(P)|$ in \eqref{eq:def_parameter}, the above display implies that on the event ${\cal L}_n$,
\begin{align} \label{eq:deviance_condition_for_(V)}
  0 \le \Psi(P)-|\Psi_{k_j}(P)| \le |\Psi_{k_j}(P)-\Psi(P)| < \epsilon K_{nq_n},
\end{align}
where $\epsilon >2\zeta' $ is chosen in connection with \ref{assump:Signal strength}. 
Recall that $k_0$ is the label of the predictor that attains $\Psi(P)$ under the alternatives, so it is easy to see that $k_0 \in \mathcal{K}^*_n$, where $\mathcal{K}^*_n$ contains the predictors that have stronger association with $T$ than the other predictors in $\mathcal{K}_n$, as indicated in \ref{assump:Signal strength}. Therefore \eqref{eq:deviance_condition_for_(V)} implies that $k_j \in \mathcal{K}^*_n$, because $\Psi(P)-|\Psi_{k_j}(P)|$ is then under the threshold specified in \ref{assump:Signal strength}.
By \ref{assump:Signal strength}, we conclude that
\begin{align*}
  \frac{1}{\sqrt{n-q_n}}\sum_{j=q_n}^{n-1}|\Psi_{k_j}(P) - \Psi(P)|1_{{\cal L}_n}
  = O_p(\sqrt{n-q_n}\,{\rm Diam}(\mathcal{K}^*_n)\,) = o_p(1).
\end{align*}
\end{proof}

\section{Supplementary results for simulation studies and real data application} \label{sec:simulation_data_analysis_results}
In Tables \ref{tab:most_correlated_features}-\ref{tab:features_codes} we report the names of the identified features based on the various competing approaches, separately  for the Subtype B and C datasets. For the stabilized one-step estimator, in Subtype B the most correlated binary feature is the interaction of $hxb2.677.K.1mer$ and $hxb2.460.sequon\, actual.1mer$; the most correlated count feature is the interaction of $sequons.total.gp120$ and $sequons.total.v5$.
In other words, the presence of specific amino acid (coded by $K$) at position $677$ and the presence of some enzymatic processes starting at position $460$ are found to  have a synergistic  influence on $\mbox{IC}_{50}$.
Similarly, the change in $\mbox{IC}_{50}$ appears to be simultaneously affected by the total numbers of observed chemical reactions in the region of gp120 and of V5.

\begin{table}[htb]
\centering
\caption{The most correlated features identified by the competing methods, according to data type. The interactions are coded as in Table \ref{tab:features_codes}, with $\alpha \times \beta$ denoting the interaction between $\alpha$ and $\beta$.}
\begin{tabular}[t]{p{2cm}ccc} %  
\toprule
\parbox{2cm}{} & Method & \multicolumn{2}{c}{Feature} \\ 
\cmidrule{3-4}
 &  & {\centering Binary predictors} & {\centering Count predictors} \\
\midrule
\multirow{2}{*}{Subtype B} &
%Bonferroni KSV & $\alpha_3 \times \beta_2$ & $\gamma_2 \times \gamma_4$ \\
%\cmidrule{2-4}
 Bonferroni One-Step & $\alpha_2 \times \alpha_6$ & $\gamma_3 \times \delta$ \\
\cmidrule{2-4}
 & Stabilized One-Step & $\alpha_3 \times \beta_2$ & $\gamma_2 \times \gamma_4$\\
\midrule 
\multirow{2}{*}{Subtype C} &
%Bonferroni KSV & $\alpha_5 \times \beta_3$ & $\gamma_1 \times \epsilon$ \\
%\cmidrule{2-4}
 Bonferroni One-Step & $\alpha_1 \times \alpha_4$ & $\delta \times \zeta$\\
\cmidrule{2-4}
 & Stabilized One-Step & $\alpha_5 \times \beta_1$ & $\gamma_1 \times \epsilon$\\
\bottomrule
\end{tabular}
\label{tab:most_correlated_features}
\end{table}

\begin{table}[htb]
\centering
\caption{Coding of feature names used in Table \ref{tab:most_correlated_features}
based on  source code for \cite{Magaret2019}.}
\begin{tabular}{|ll|}
\hline
$\alpha_1 \colon hxb2.389.G.1mer$ & $\alpha_2 \colon hxb2.130.K.1mer$ \\
$\alpha_3 \colon hxb2.677.K.1mer$ & $\alpha_4 \colon hxb2.462.N.1mer$ \\
$\alpha_5 \colon hxb2.363.S.1mer$ & $\alpha_6 \colon hxb2.132.T.1mer$\\ \hline
$\beta_1 \colon hxb2.142.sequon\, actual.1mer$ & $\beta_2 \colon hxb2.460.sequon\, actual.1mer$ \\ \hline
$\gamma_1 \colon sequons.total.env$ & $\gamma_2 \colon sequons.total.gp120$\\
$\gamma_3 \colon sequons.total.loop.e$ & $\gamma_4 \colon sequons.total.v5$ \\ \hline 
$\delta \colon cysteines.total.gp120$ & $\epsilon \colon length.env$ \\
$\zeta \colon taylor.small.total.cd4$ & \\ \hline
\end{tabular}
\label{tab:features_codes}
\end{table}

\begin{figure}[H]
\begin{center}
 \includegraphics[width=0.9\textwidth]{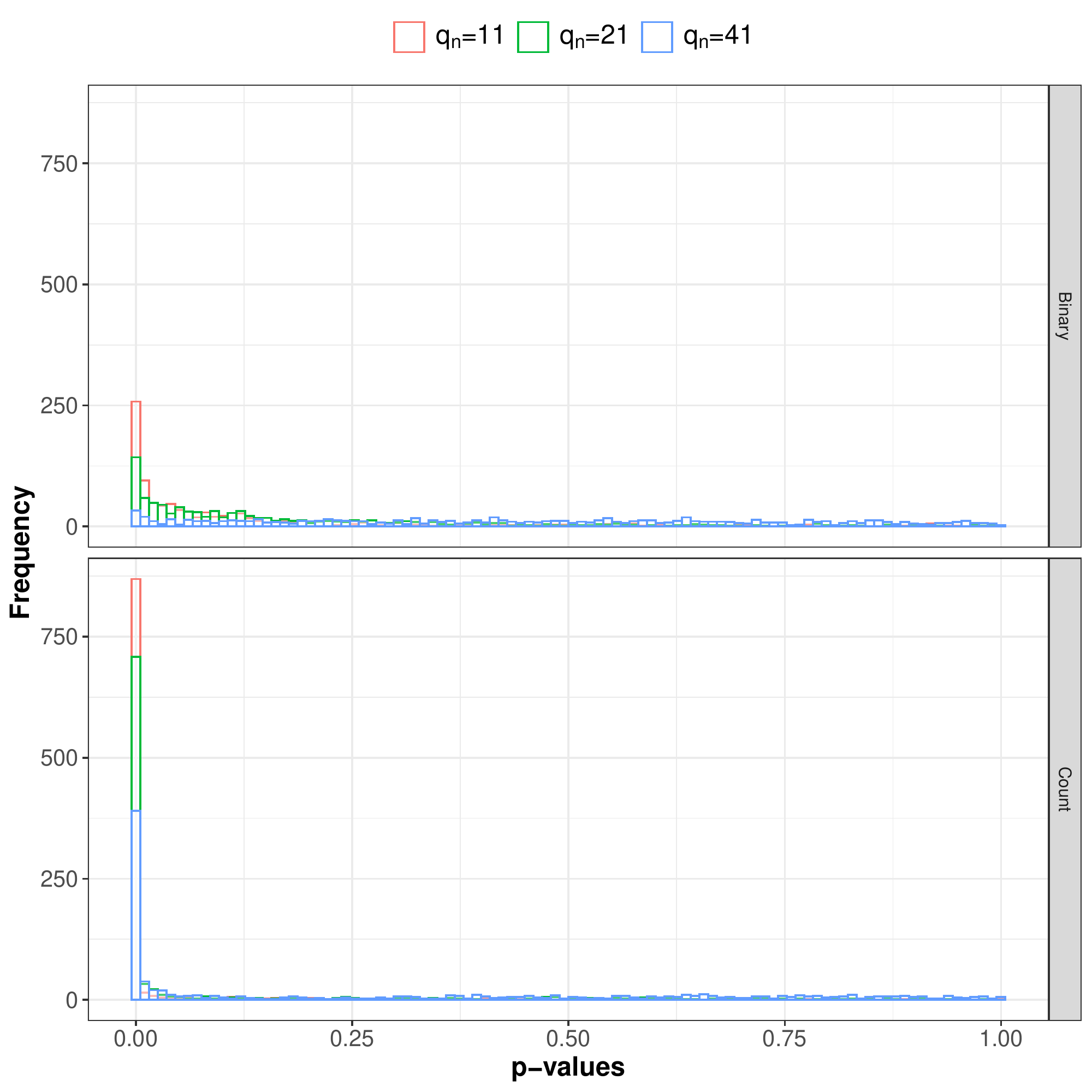}
 \caption{Histogram of p-values obtained by applying the stabilized one-step estimator to $1000$ random orderings of the Subtype B data for various values of $q_n$,  separated according to binary and count predictors.}
 \label{fig:histogram_1000pvalues_subtypeB}
 \end{center}
\end{figure}

\begin{figure}[H]
\begin{center}
 \includegraphics[width=0.9\textwidth]{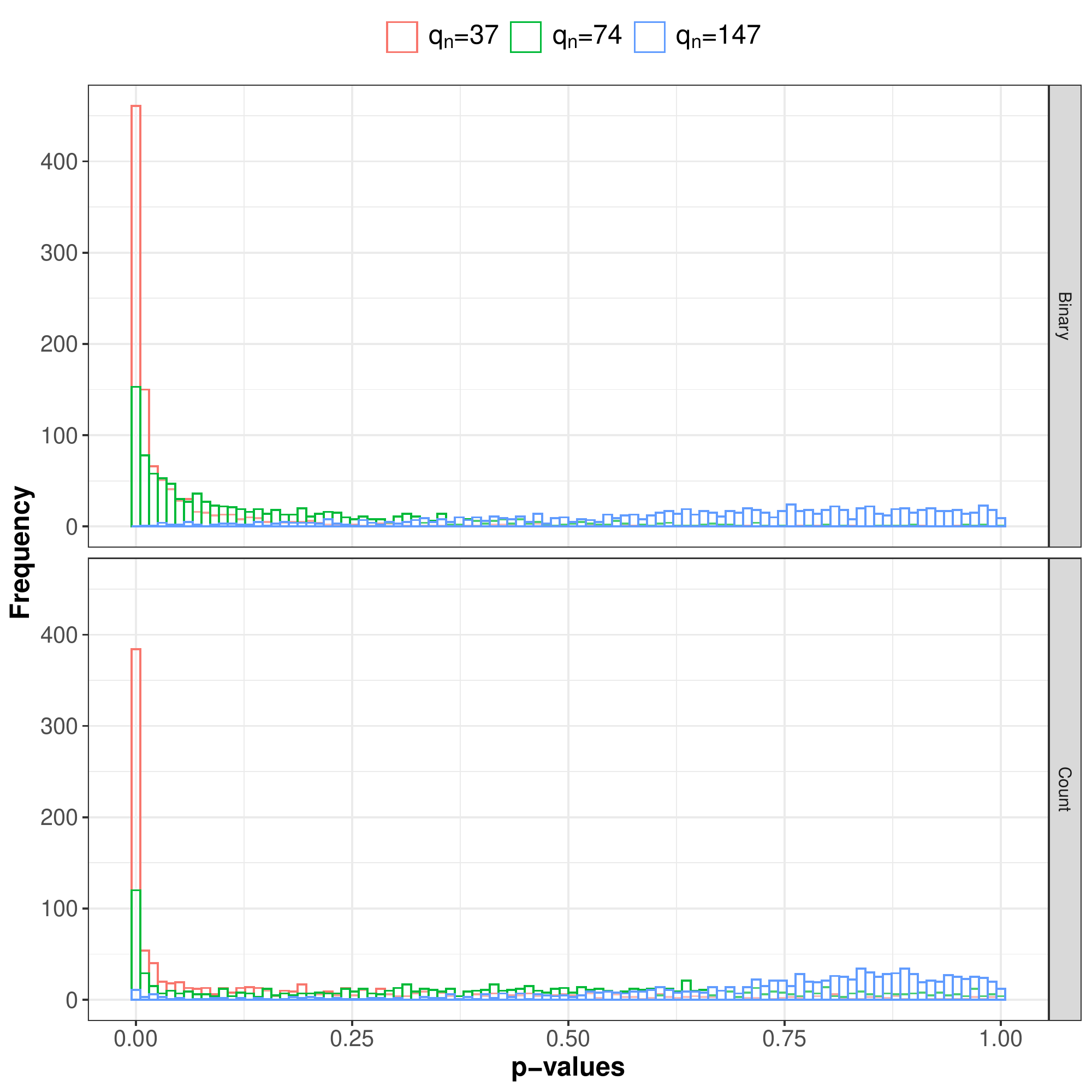}
\caption{As in Figure \ref{fig:histogram_1000pvalues_subtypeB}, except for the data on virus subtype C.}
\label{fig:histogram_1000pvalues_subtypeC}
 \end{center}
\end{figure}

\begin{figure}[H]
\begin{center}
 \includegraphics[width=0.9\textwidth]{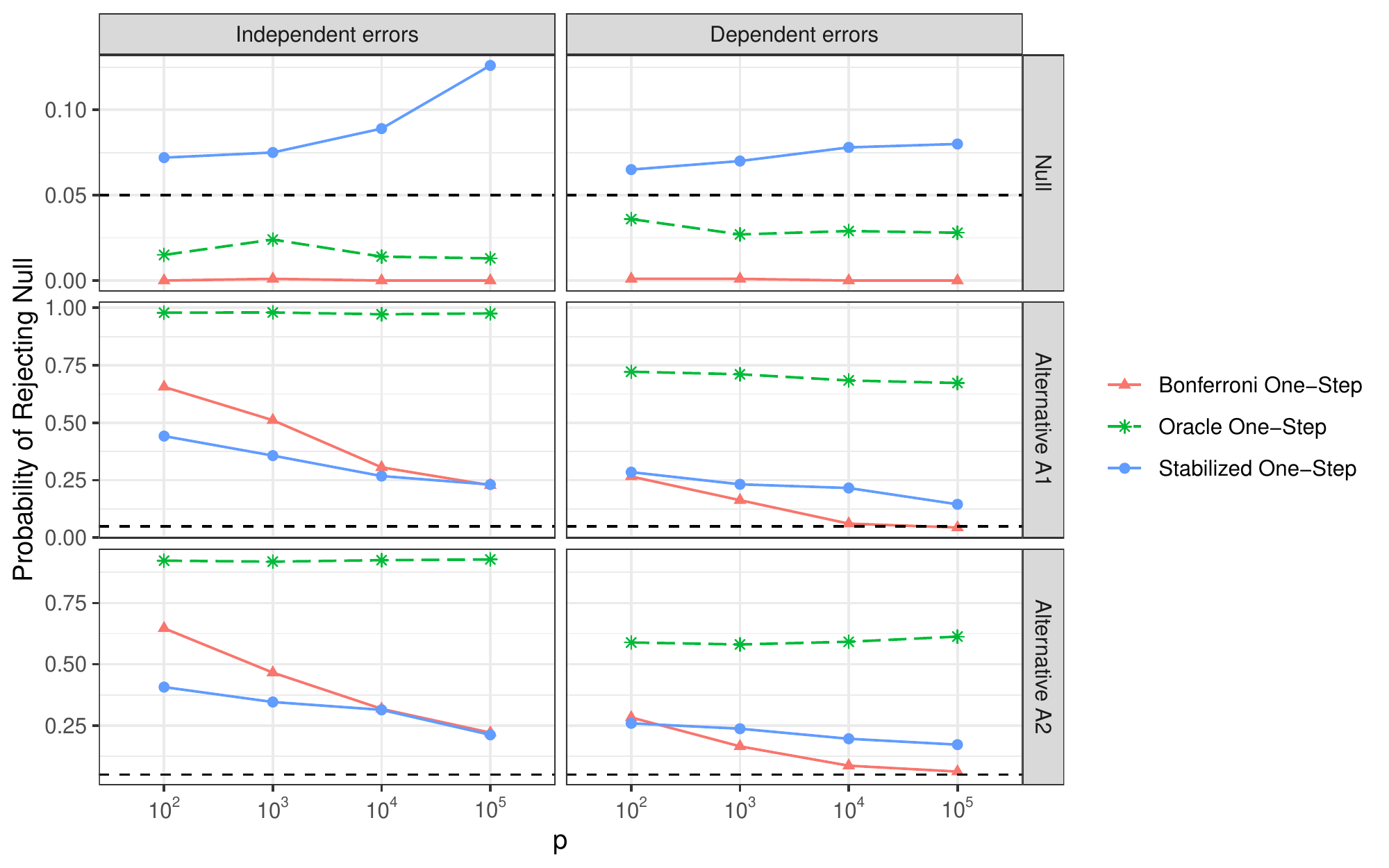}
 \caption{As in Figure \ref{fig:empirical_rejectionrate_lightcensoring_partialsample}, except under heavy censoring ($30 \%$).}
 \label{fig:empirical_rejectionrate_heavycensoring_partialsample}
\end{center}
\end{figure}

\begin{figure}[H]
\begin{center}
 \includegraphics[width=0.9\textwidth]{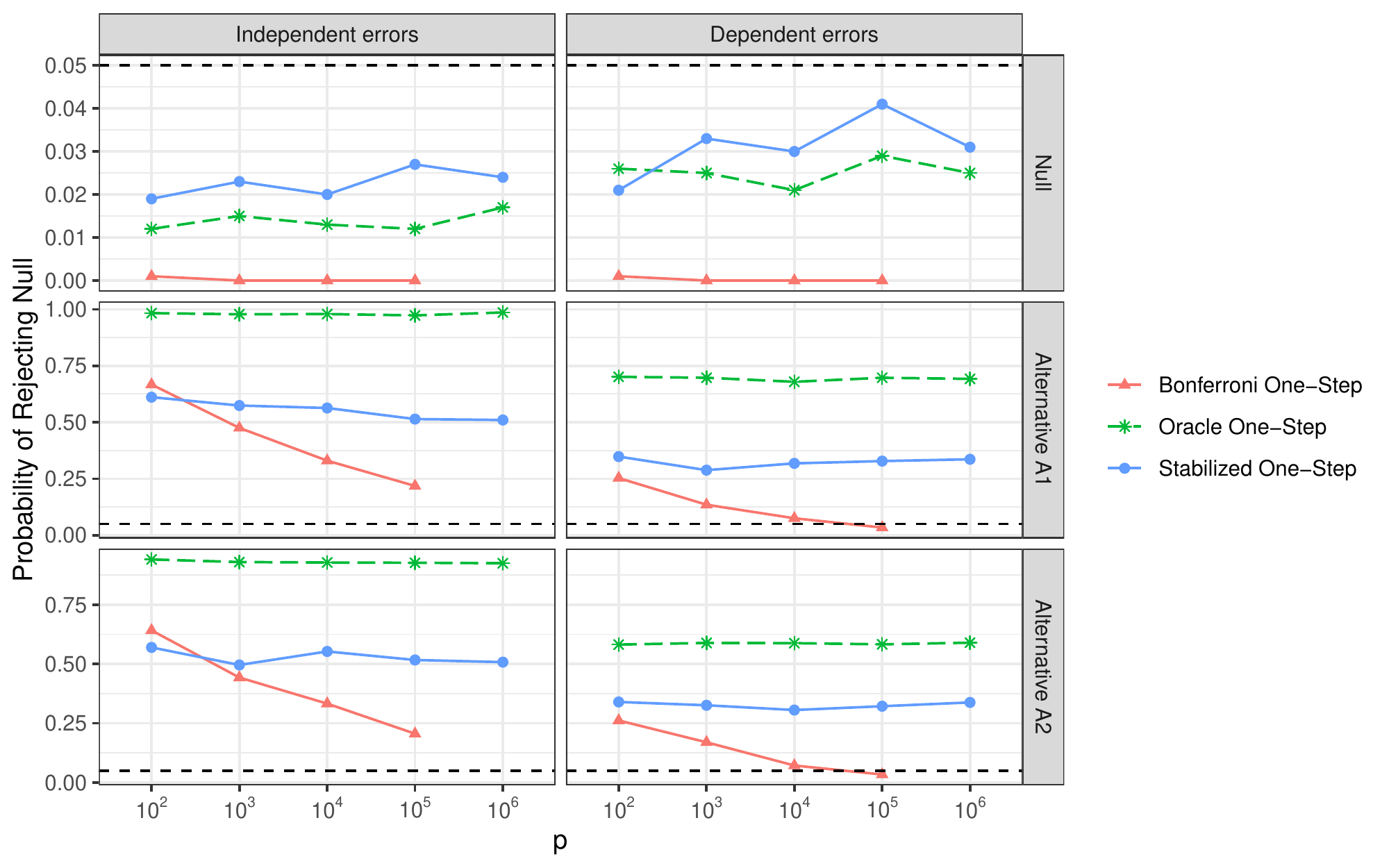}
 \caption{As in Figure \ref{fig:empirical_rejectionrate_heavycensoring_partialsample}, except using the whole sample to estimate nuisance parameters $E_0$ and $Q_u$ that are estimated by the partial sample $q_n=n/2$ in Fig \ref{fig:empirical_rejectionrate_heavycensoring_partialsample}.}
 \label{fig:empirical_rejectionrate_heavycensoring_wholesample}
\end{center}
\end{figure}
\end{appendix}

\vspace{1in}
\section*{Acknowledgements}  
We thank Peter Gilbert for suggesting the application in Section \ref{sec:real_data_application}.
AL was supported by the National Institutes of Health (NIH) under award number DP2-LM013340. IWM was supported by NIH under award 1R01 AG062401 and by the National Science Foundation (NSF) under award DMS-2112938. The content is solely the responsibility of the authors and does not necessarily represent the official views of NIH or NSF.

% \newpage
% \vskip 2in
\renewcommand\bibname{}
%\vspace{-1in}
\bibliographystyle{abbrvnat}
\bibliography{refs_maxCorrSurv}
\end{document}